\documentclass[aps,showpacs,twocolumn,superscriptaddress]{revtex4-2}
\usepackage{amsfonts,amsmath,amssymb,amsthm}
\usepackage{braket}
\usepackage{physics}
\usepackage{bm}
\usepackage{dcolumn}
\usepackage{graphicx} % when submitting, turn on this
\usepackage{color}
\usepackage{natbib}
\usepackage[colorlinks=true]{hyperref}
\usepackage{algorithmic}
\usepackage{algorithm}
\usepackage{blkarray}
\usepackage{quantikz}
\usepackage{multirow}
\usepackage{listings}
\usepackage{booktabs}
\usepackage{makecell}

\theoremstyle{plain}
\newtheorem{thm}{Theorem}
\newtheorem{proposition}{Proposition}
\newtheorem{lem}{Lemma}

\newtheorem*{thm*}{Theorem}
\newtheorem*{lem*}{Lemma}
\newtheorem*{cor*}{Corollary}

\theoremstyle{definition}
\newtheorem{dfn}{Definition}
\newtheorem{remark}[dfn]{Remark}
\newtheorem{remark*}{Remark}

\hypersetup{
           breaklinks=true,   % splits links across lines
           colorlinks=true,   % displays links as colored text instead of blocks
           pdfusetitle=true, 
           citecolor=blue,
           urlcolor=blue
           }

\newcommand{\bbra}[1]{\langle\!\bra{#1}}
\newcommand{\kett}[1]{\ket{#1}\!\rangle}

\begin{document}
\title{Exponentially accurate open quantum simulation via randomized dissipation\\
with minimal ancilla}

\author{Jumpei Kato}
\email{jumpei\_kato@keio.jp}
\affiliation{Mitsubishi UFJ Financial Group, Inc.~and MUFG Bank, Ltd., 4-10-2 Nakano, Nakano-ku, Tokyo 164-0001, Japan}
\affiliation{Quantum Computing Center, Keio University, Hiyoshi 3-14-1, Kohoku, Yokohama, Kanagawa 223-8522, Japan}
\affiliation{Graduate School of Science and Technology, Keio University, 3-14-1 Hiyoshi, Kohoku, Yokohama, Kanagawa 223-8522, Japan}

\author{Kaito Wada}
\email{wkai1013keio840@keio.jp}
\thanks{J.K. and K.W. contributed equally to this work.}
\affiliation{Graduate School of Science and Technology, Keio University, 3-14-1 Hiyoshi, Kohoku, Yokohama, Kanagawa 223-8522, Japan}

\author{Kosuke Ito}
\affiliation{Advanced Material Engineering Division, Toyota Motor Corporation, 1200 Mishuku, Susono, Shizuoka 410-1193, Japan}
\affiliation{Quantum Computing Center, Keio University, Hiyoshi 3-14-1, Kohoku, Yokohama, Kanagawa 223-8522, Japan}

\author{Naoki Yamamoto}
\email{yamamoto@appi.keio.ac.jp}
\affiliation{Quantum Computing Center, Keio University, Hiyoshi 3-14-1, Kohoku, Yokohama, Kanagawa 223-8522, Japan}
\affiliation{Department of Applied Physics and Physico-Informatics, Keio University, Hiyoshi 3-14-1, Kohoku, Yokohama, Kanagawa 223-8522, Japan}

% \date{\today}

\begin{abstract}
Simulating open quantum systems is an essential technique for understanding complex physical phenomena and advancing quantum technologies.
Some quantum algorithms simulate Lindblad dynamics exponentially accurately, i.e., they achieve
logarithmically short circuit depth in terms of accuracy, but they need to coherently encode all possible jump operators with a large ancilla consumption.
Minimizing the gate and ancilla counts while achieving such a logarithmic scaling in accuracy remains an important challenge.
In this work, we present two randomized quantum algorithms for simulating general Lindblad dynamics with multiple jump operators aimed at an observable estimation that achieves a circuit depth with not only logarithmic scaling in accuracy but also either partial or complete independence from the parameters specifying the Lindbladian.
This is based on a novel random circuit compilation method that leverages dissipative processes with only a single jump operator, leading to the proposed methods using minimal ancilla qubits---$4+\lceil\log_2 M\rceil$ in the first case and $7$ in the other, where each single jump operator has at most $M$ Pauli strings.
In addition, we numerically demonstrate the practical advantage over existing approaches by providing a detailed analysis of the required gate and ancilla counts.
This work represents a significant step toward making open quantum system simulations more feasible on early fault-tolerant quantum computing devices.
\end{abstract}

\maketitle
\section{Introduction}
\subsection{Background}
Analysis and simulation of open quantum systems~\cite{Breuer2007-textbook,
Wiseman2009-rr,Daley2014-fz,Weiss2012-iq, ollitrault2021molecular} are indispensable for understanding complex quantum phenomena found in many-body systems, including transport in boundary-driven systems, transient dynamics, and thermalization.
It can also be used for synthesizing quantum systems having useful functionalities in the framework of reservoir engineering \cite{poyatos1996quantum,Haroche2006-ee}, and moreover for devising noise mitigation techniques in quantum computation \cite{van2023probabilistic}.
Such studies are usually based on the Lindblad equation or Gorini--Kossakowski--Sudarshan--Lindblad equation~\cite{lindblad1976generators, gorini1976completely}:
\begin{equation}\label{main:lindbladeq}
    \frac{{\rm d}\rho}{{\rm d}t}=\mathcal{L}(\rho):=-i[H,\rho]+\sum_{k=1}^{K} \left(L_k\rho L_k^\dagger - \frac{1}{2}
    \left\{L_k^\dagger L_k,\rho\right\}\right),
\end{equation}
where $H$ and $\{L_k\}_{k=1}^K$ are the Hamiltonian and jump operators of a target system, respectively. 

Quantum computation offers a promising approach for simulating open-system dynamics~\cite{Lloyd1996-kq}.
This is particularly important for many-body systems with higher spatial dimensions, where classical simulation methods, such as tensor networks, become computationally demanding~\cite{cirac2021review, weimer2021review, schuch2007peps}.
In such settings, even the estimation of physically relevant observables can already be hard for classical methods, and this difficulty is expected to arise not only in asymptotically long-time dynamics but already in comparatively short-time dynamics~\cite{kashyap2025oqsgurantee}.
This highlights the need for efficient quantum algorithms for Lindblad simulation.

There are several efficient approaches to simulate the Lindblad equation via unitary dynamics of quantum computation,
based on Trotter-type decomposition~\cite{PhysRevLett.107.120501,PhysRevLett.108.230504,childs2017efficient,PhysRevLett.127.020504,PhysRevA.64.062302,PhysRevA.91.062308,jo2022simulating}, vectorized Lindblad equation~\cite{Kamakari2022-QITE,Watad2024-ix, Schlimgen2022-vectorizationTalor}, %PhysRevResearch.4.023216,
and sample-based simulation~\cite{Patel2023-li, Patel2023-so}.
Some studies~\cite{di2015quantum, Hu2020-nagy, hu2022fenna, Kamakari2022-QITE, Suri2023-iq} focus on observable estimation rather than state preparation because physical information of interest, such as state populations, can often be revealed via observable estimation without state tomography.
These works need the first-order depth scaling $\mathcal{O}(T^2/\varepsilon)$ in accuracy (meaning the simulation error) $\varepsilon$; moreover, they typically require a large sampling overhead.

In contrast, advanced approaches~\cite{Cleve2016-yj, Li2023-DysonseriesOSS, Chen2023-thermalstate} achieve a logarithmically shorter circuit depth with respect to the accuracy for state preparation, as in the case of closed systems via the state-of-the-art methods for Hamiltonian simulation (HS)~\cite{Berry2015-truncatedtaylor,low2017optimal, Low2019hamiltonian}.
However, they need to simultaneously encode \textit{all} possible jump processes in a quantum circuit, via linear combination of unitaries (LCU) method followed by oblivious amplitude amplification (OAA), leading to a large consumption of ancilla qubits and complicated controlled operations. 
Very recently, Ref.~\cite{Ding2024-SDE} introduces a HS-based method achieving ${\rm Q}$-th order approximation without OAA, while it still requires a large amount of ancilla qubits in addition to a highly accurate simulation of a dilated Hamiltonian comprised of an exponentially large number of terms $\mathcal{O}(K^{\rm Q})$.

On the other hand, recently, the importance of reducing ancilla qubits has been demonstrated in various quantum algorithms~\cite{Lin2022-wc, Wan2022-tx, An2023-lchs,Wada2024-nb,Wang2024-nx,Katabarwa2024-EFTQC}, especially in the early stage of fault-tolerant quantum computers (FTQC), where only a limited amount of quantum resources is available.
In addition, in the context of Hamiltonian simulation,
recent studies have explored the trade-off between the circuit depth and the number of ancilla qubits while maintaining the exponentially high accuracy of FTQC algorithms~\cite{Berry2015-truncatedtaylor,low2017optimal,Low2019hamiltonian, gilyen2019quantum}. Some of these algorithms~\cite{Wan2022-tx, chakraborty2024implementing, zeng2025simple, watson2025-exponentially} achieve only a few or even zero ancilla qubits, by sacrificing the optimal scaling in time.
However, for open quantum system simulations, only limited attention has been paid to improving space complexity while guaranteeing high accuracy. 
This is possibly because it is challenging to achieve both the reduction of ancilla qubits and boosting the simulation accuracy for non-unitary dynamics with advanced coherent techniques. 
Note that we may use near-term quantum algorithms~\cite{Endo2020-generalprocess,haug2022generalized,Liu2022-bo} that work with fewer ancilla qubits, but they typically require a large sampling overhead to achieve an accurate simulation, leading to longer end-to-end runtime.
Therefore, significant progress in reducing space overhead, together with high accuracy and constant sampling overhead, has yet to be achieved.
\begin{table*}[tbh]
    \centering
    \makebox[\textwidth]{
    \begin{tabular}{l|ccc}
        \makecell[l]{Algorithm} & Gate  complexity & Additional qubits & \makecell[l]{Sampling overhead}\\ \hline \hline
        \\[-0.8em]
        Channel LCU~\cite{Cleve2016-yj}&  $\tilde{\mathcal{O}}\left(\tau \log^2 \left(\dfrac{\tau}{\varepsilon}\right)K^2 \mathrm{poly}(m, M)\right)$ & ${\mathcal{O}} \left(\log\left(\dfrac{\tau}{\varepsilon}\right) \log\left(\dfrac{\tau K(m + M)}{\varepsilon}\right) \right)$
        & 1 \\[1.2em]
        
        \makecell[l]{Higher order HS \\ ($q$th order)~\cite{Ding2024-SDE}} & $\tilde{\mathcal{O}}\left(\tau \left(\dfrac{\tau}{\varepsilon}\right)^{\frac{1}{q}} K^{q} \mathrm{poly}(m, M) \right)$ & ${\Omega}(q \log ((m+M)K))$ & 1\\[1.2em]
        
        Randomization~\cite{Peng2024-Random} & $\mathcal{O}\left(\dfrac{\tau^2}{\varepsilon}M \right)$ & $\mathcal{O}\left(\dfrac{\tau}{\varepsilon}\log M \right)$ & 1 \\[1.2em]\hline
        \\[-0.8em]
        
        \makecell[l]{Vectorization \\ ($k$-local)~\cite{Kamakari2022-QITE}} & $\mathcal{O} \left(\dfrac{\tau^2}{\varepsilon}\mathrm{poly}(K,M,m) k^2 4^k \right)$ & $N_{\rm q}+1$ & -
        \\[1.2em]
        
        \makecell[l]{Our work \\ (Theorem~\ref{thm: main2})} & $\mathcal{O}  \left(\tau^2 \left(\log \left(\dfrac{\tau}{\varepsilon}\right) + M\log M\right) \right)$ & $4 + \lceil \log_2 M \rceil$ & $e^2$ \\[1.2em]
        
        \makecell[l]{Our work \\ (Theorem~\ref{thm:mMKindep_simulator})}& $\mathcal{O}  \left(\tau^4  \log^3 \left(\dfrac{\tau}{\varepsilon}\right)\right)$ & $7$ & $e^2$
    \end{tabular}
    }
    \caption{
    Gate complexity, additional qubit requirements, and sampling overhead for $N_{\rm q}$-qubit Lindblad simulation algorithms for the rescaled simulation time $\tau =t\|\mathcal{L}\|_{\rm pauli}$ and accuracy $\varepsilon$.
    The Hamiltonian and the jump operators, specified by Eq.~\eqref{main:H_L}, consist of $m$ and $M$ Pauli strings, respectively. 
    Here, $K$ denotes the number of jump operators, and $q, k>0$ are integers.
    The sampling overhead is defined as the multiplicative overhead in the number of shots relative to the standard sampling cost $\mathcal{O}(\|O\|^2/\varepsilon^{2})$.
    Our algorithms achieve the gate complexity $\mathcal{O}(\log(1/\varepsilon))$ for $\varepsilon$, together with the significant reduction of the qubit requirements.
    Furthermore, gate complexity in Theorem~\ref{thm: main2} is independent of $K$ and $m$, and that in  Theorem~\ref{thm:mMKindep_simulator} is independent of $K$, $m$, and $M$.
    In the complexity estimate for Ref.~\cite{Ding2024-SDE}, we assume that the method uses an optimal Hamiltonian simulation for block-encoded Hamiltonians~\cite{gilyen2019quantum}, 
    which could be improved with more efficient block encoding techniques.
    Ref.~\cite{Kamakari2022-QITE} and our work mainly focus on the expectation value estimation.
    The sampling overhead of Ref.~\cite{Kamakari2022-QITE} is not comparable to the others, as it requires additional tomographic sampling. 
    The worst-case gate complexity is presented for probabilistic algorithms.}
    \label{tab:complexity}
\end{table*}

\subsection{Summary of results}
In this work, we present two quantum algorithms for estimating the physical properties such as expectation value ${\rm Tr}[O\rho(t)]$ of a given observable $O$ with respect to $\rho(t)$ evolved by Eq.~\eqref{main:lindbladeq}, rather than preparing the full state $\rho(t)$; the main algorithm (Theorem~\ref{thm: main2}) and its extended version (Theorem~\ref{thm:mMKindep_simulator}).
Given a Hamiltonian $H$ and jump operators $\{L_k\}_{k=1}^K$ specified by ${N_{\rm q}}$-qubit Pauli strings $\{P_{kj}\}$ as
\begin{equation}\label{main:H_L}
    H = \sum_{j=1}^m \alpha_{0j} P_{0j},~~~L_k =\sum_{j=1}^M \alpha_{kj} P_{kj},
\end{equation}
for $\alpha_{0j}\in \mathbb{R},\alpha_{kj}\in \mathbb{C}$, our algorithms require $\mathcal{O}(\|O\|^2/\varepsilon^2)$ samples from a set of quantum circuits that are all with the logarithmic depth $\mathcal{O}(\log(1/\varepsilon))$, for a given estimation error $\varepsilon$.
Furthermore, the number of additional ancilla qubits is only $4+\lceil \log_2 M\rceil$ (Theorem~\ref{thm: main2}) for the number of Pauli strings $M$ of each single jump operator in Eq.~\eqref{main:H_L}; this required ancilla count is minimal-size to exactly (block) encode a single jump operator in a circuit.
Our ancilla count depends solely on $M$, whereas other proposals~\cite{Cleve2016-yj,Ding2024-SDE, Kamakari2022-QITE} have the dependence of $M$, $K$, $\varepsilon$, or $N_{\rm q}$.
Table~\ref{tab:complexity} presents the explicit gate complexities (as a detailed breakdown of the circuit depth), the ancilla qubit requirements,
and the sampling overhead, which is defined as the multiplicative overhead in the number of shots relative to the standard sampling cost $\mathcal{O}(\|O\|^2/\varepsilon^2)$. 
The total time complexity is obtained as the product of the gate complexity, the sampling overhead, and the standard sampling cost.
Notably, while the gate complexities of Refs.~\cite{Cleve2016-yj,Ding2024-SDE} have a better scaling in $t$,
our gate complexity (Theorem~\ref{thm: main2}) can be significantly smaller than the previous methods; this is mainly because it does not depend on either $K$ or the number $m$ of Pauli strings in the system Hamiltonian.
Moreover, if we allow for $\mathcal{\tilde{O}}(t^4)$ scaling in time $t$, the remaining $M$-dependency in the ancilla/gate counts of our algorithm can be further removed, achieving the complete independence from $(m,M,K)$ together with the logarithmic accuracy scaling (Theorem~\ref{thm:mMKindep_simulator}).

In particular, the independence from $K$ offers a significant advantage for open systems with large $K$; e.g., in the presence of collective decay, $K$ scales exponentially regarding the number of system qubits $N_{\rm q}$~\cite{Peng2024-Random}.
A similar independence of the parameter $m$ can be found in random compilers for closed systems such as qDRIFT and subsequent studies~\cite{Campbell2019-qDRIFT,Wan2022-tx}.
Although very recent randomization approaches for Lindblad simulation in Refs.~\cite{Peng2024-Random, Chen2024-RandomThermal} have similar features, they do not achieve the logarithmic gate complexity for $\varepsilon$.

Beyond the asymptotic parameter scaling in Table~\ref{tab:complexity},
our numerical analysis in Section~\ref{sec:main-numeric} provides explicit fault-tolerant resource estimates in terms of both T-count (and T-depth) and ancilla count.
In Section~\ref{sec:gatecomplexity-analysis}, we show that these improvements are not merely asymptotic, but already translate into clear practical advantages over previous approaches in realistic parameter regimes.
In Section~\ref{sec:benchmark}, we further connect fault-tolerant resource estimates to physically motivated benchmark tasks in open-system settings that are expected to be classically challenging, demonstrating that our algorithm can access physically meaningful open problems with realistic execution-time estimates and clear resource advantages over prior quantum algorithms.

\subsection{Comparison to prior works}\label{sec:comparison-to-prior-work}
In this work, we focus on reducing the number of ancilla, while pursuing a better scaling in the gate complexity regarding as many parameters as possible. 
Consequently, our algorithms (Theorem~\ref{thm: main2} and Theorem~\ref{thm:mMKindep_simulator}) achieve the following features at the same time:
(1) exponentially short depth $\mathcal{O}(\log (1/\varepsilon))$ for accuracy $\varepsilon$, (2) small number of ancilla qubits, (3) both the gate count and the number of ancilla are independent of $m$ and $K$ (and further $M$ for Theorem~\ref{thm:mMKindep_simulator}).
Here, we present the comparison of the gate and additional ancilla requirements with most related works in Table~\ref{tab:complexity}.

Compared with the fully fault-tolerant algorithms~\cite{Cleve2016-yj, Ding2024-SDE},
our algorithm (Theorem~\ref{thm: main2}) significantly reduces ancilla qubits, and achieves better scaling of $\varepsilon$ and the parameter independence of $(m, K)$ in the gate complexity,
by sacrificing the scaling of the rescaled time $\tau\propto t$.
We justify this trade-off through the numerical analysis in Section~\ref{sec:gatecomplexity-analysis}.

Compared with the recent randomized algorithm that aims to remove the parameter-dependence~\cite{Peng2024-Random}, we significantly improve the depth scaling and ancilla qubits for $\varepsilon$, keeping the parameter independence of $m$ and $K$.

Compared with the observable estimation algorithm that employs the transfer matrix formulation~\cite{Kamakari2022-QITE}, we improve the depth scaling of $\varepsilon$ and avoid the large tomographic sampling cost.

Lastly, we give some remarks. Our algorithm (Theorem~\ref{thm:mMKindep_simulator}) uniquely attains the $(m,M,K)$-independence in both gate and ancilla requirements with $\mathcal{\tilde{O}}(\tau^4)$ depth scaling. Also we note that our algorithm is universally free from $\exp(\mathcal{O}(N_{\rm q}))$ gates and classical computation for circuit synthesis, unlike the previous methods, e.g.,~\cite{Hu2020-nagy, hu2022fenna}.

\subsection{Technical contributions}
Our algorithms are based on Taylor expansion of the dynamical map $e^{t\mathcal{L}}$ via the transfer matrix representation of superoperators (i.e., a linear map on the space of linear operators)~\cite{wood2011tensor,Nielsen2021gatesettomography}.
Notably, we naturally expand the sum of the Hamiltonian $H$ (Eq.~\eqref{main:H_L}) and $K$ dissipation processes, with the help of the matrix representation.
Then, we write the Taylor expansion as
a linear combination of superoperators,  
each of which is composed of at most $\mathcal{O}(\log(1/\varepsilon))$-th power of minimal components; i.e., a single Pauli rotation and a single jump process, by truncating Taylor series.
Note that, although each of the minimal components can be prepared at a cost independent of $m$ and $K$, the total number of components is prohibitively large.
To avoid this issue, which can be interpreted as the curse of dimensionality, we employ Monte Carlo sampling that randomly compiles superoperators according to their coefficients.
The coefficients must be rescaled on the total norm for random sampling, which introduces an additional sampling overhead.
However, naive Taylor expansion does not support the random sampling because the total norm exponentially diverges as $e^{\mathcal{O}(t)}$, leading to an exponential sampling overhead.
One of our main contributions lies in the discovery of the new decomposition of $e^{t\mathcal{L}}$ that enables the random compilation by suppressing the norm, whose value can be close to 1, without compromising the advantages (i.e., the exponentially fast convergence) of Taylor expansion.
In addition, each component in the decomposition can be efficiently and exactly simulated on circuits without doubling the target system qubits, by the new techniques introduced below and in Section~\ref{subsubsec:exactefficient_dissipation} and~\ref{subsubsec:lcs_main}.

We introduce two key developments to realize the algorithm. First, to suppress the norm of coefficients in the Taylor expansion, we develop a method to simulate the dissipative process on a quantum circuit that is an improved version of the method via LCU combined with OAA invented by Ref.~\cite{Cleve2016-yj}.
The previous method can also compress the norm, but it introduces approximation errors that would deteriorate the circuit depth. Our improved version allows for the exact simulation of dissipative processes thanks to an error recovery operation.
Also, this method requires only the minimal-size ancilla qubits for a general single jump process, given by at most $3 + \lceil \log_2 M \rceil$ qubits. 

The second development is a new framework connecting quantum circuits with a matrix representation of superoperators.
The decomposition via the transfer matrix representation yields the superoperators that are generally not even completely positive (CP) maps. Simulating such superoperators on quantum circuits remains nontrivial.
We establish an effective simulation formalism for superoperators using quantum circuits without doubling the target system qubit size,
while the transfer matrix formalism usually introduces an extended ancilla system with the same dimension as the target system~\cite{PhysRevResearch.4.023216,Kamakari2022-QITE}.
This formalism allows us to effectively simulate an exponentially accurate Taylor expansion of the transfer matrix of $e^{t \mathcal{L}}$ on the circuits.

These two discoveries are not merely components of the algorithm for simulating Lindblad dynamics but are powerful techniques with broad applicability
to the simulation of CP maps and superoperators.
Furthermore, the first error recovery ideas can be applied to general CP maps including approximated unitary channels.
The second new formalism of superoperators offers a systematic translation from the transfer matrix into quantum circuits, with the least qubit cost. 
The formalism bridges the transfer matrix-based superoperator design and its implementation on quantum devices, providing a new insight for developing algorithms based on the transfer matrix.

Finally, we can concatenate the random compilation to further remove the remaining $M$-dependence.
The remaining $M$-dependence comes from exactly encoding jump operators $L_k$ into the circuit to realize the desired dissipation. 
Instead of the direct encoding, we can use the randomized Hamiltonian simulation by taking the real and imaginary parts of $L_k$ as Hamiltonians; this achieves an $M$-independent approximate (but highly accurate) encoding due to the randomness. 
Then, we can concatenate this alternative encoding and the random compilation of Lindblad dynamics via the technique of quantum singular value transformation~\cite{gilyen2019quantum}, leading to the complexities (Theorem~\ref{thm:mMKindep_simulator}) at the bottom of Table~\ref{tab:complexity}.

\section{main results}
\subsection{Decomposition of $e^{t\mathcal{L}}$}
\begin{figure}[tb]
    \centering
    \includegraphics[scale=0.95]{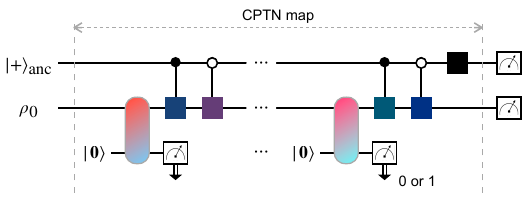}
    \caption{Quantum circuit of sampled $\widetilde{\mathcal{W}}_v$ for simulating the decomposition of $e^{t\mathcal{L}}$ in Theorem~\ref{thm: main}.
    The colored blocks denote unitary gates, which can be constructed from the input model.
    The mid-circuit measurement and qubit reset allow us to simulate the CPTN maps $\widetilde{\mathcal{W}}_v$ by the quantum circuits. The quantum register $\ket{\bm{0}}$ contains at most $3+\lceil \log_2 M\rceil$ qubits.}
    \label{fig:main_circ}
\end{figure}
We first provide a theorem giving an exact decomposition of the dynamical map $e^{t\mathcal{L}}$ in the form of a linear combination of generalized Hadamard-test circuits $\{\widetilde{\mathcal{W}}_v\}$; each circuit is illustrated in Fig.~\ref{fig:main_circ}. 
In the theorem, the upper bound of the total coefficients, Eq.~\eqref{main:c_norm}, is essential to realize an efficient random compilation for the observable estimation task.
Theorem~\ref{thm: main} is a core idea to derive our algorithms.

\begin{thm}\label{thm: main}
    Let $\mathcal{L}$ be an $N_{\rm q}$-qubit Lindblad superoperator with a Hamiltonian $H$ and jump operators $\{L_k\}_{k=1}^K$ that are specified by Eq.~\eqref{main:H_L}, and let $\|\mathcal{L}\|_{\rm pauli}:=2(\alpha_0+\sum_{k=1}^K \alpha_k^2)$ for $\alpha_k:=\sum_{j}|\alpha_{kj}|$. 
    Then, for any $t>0$ and any positive integer $r\geq t\|\mathcal{L}\|_{\rm pauli}$, there exists a linear decomposition
    \begin{equation}\label{main:dynmapdecomp}
        e^{t\mathcal{L}}(\bullet)
        = \sum_{v \in \mathrm{S}} c_v \mathrm{Tr}_{\mathrm{anc}}[(X_{\rm anc} \otimes \bm{1}) {\widetilde{\mathcal{W}}}_v (\ket{+}\bra{+}_{\rm anc} \otimes \bullet )]
    \end{equation}
    for some indices set $\mathrm{S}$, real values $c_v >0$, and $(N_{\rm q}+1)$-qubit completely positive trace non-increasing (CPTN) maps $\{ {\widetilde{\mathcal{W}}}_v \}$, such that
    \begin{equation}\label{main:c_norm}
        \sum_{v \in \mathrm{S}} c_{v} \leq e^{2 \|\mathcal{L}\|^2_{\rm pauli}t^2 / r}. 
    \end{equation}
    Furthermore, for any $v$, the $(N_{\rm q}+1)$-qubit CPTN map ${\widetilde{\mathcal{W}}}_v$ can be effectively simulated by a quantum circuit on the $(N_{\rm q}+1)$-qubit system and additional $(3+\lceil \log_2 M\rceil)$-qubit ancilla system with mid-circuit measurement and qubit reset.
\end{thm}

The constructive proof of Theorem~\ref{thm: main} is provided in Appendix~\ref{apdx:B1}.
In the following, we introduce some key techniques for this proof (Section~\ref{subsubsec:exactefficient_dissipation} and~\ref{subsubsec:lcs_main}) and then provide a sketch of the proof (Section~\ref{subsubsec:sketchofproof}).

Before proceeding to that, we here briefly introduce the framework of the \textit{transfer matrix} of superoperators~\cite{Nielsen2021gatesettomography,wood2011tensor}, which greatly simplifies involved calculations for the higher powers of superoperators.
We follow the convention that treats any $N_{\rm q}$-qubit operator $A$ as a column-stacking $4^{N_{\rm q}}$-dimensional vector $\kett{A}$. 
Under this notation, a general superoperator $\mathcal{E}$ can be written as a $4^{N_{\rm q}}\times 4^{N_{\rm q}}$ matrix $S(\mathcal{E})$, which is often called the transfer matrix of $\mathcal{E}$.
For instance, for a superoperator $A\bullet B^\dagger$, we have $S(A\bullet B^\dagger)=\overline{B}\otimes A$, where $\overline{B}$ denotes the entrywise complex conjugation of $B$.
Also, we often use the fact that $S(\Phi'\circ \Phi)=S(\Phi')S(\Phi)$ for any superoperators $\Phi$ and $\Phi'$.

In particular, the following type-(A, B) superoperators are crucial components in our decomposition: 
\begin{itemize}
    \item[(A)] Completely positive trace non-increasing (CPTN) maps
    \item[(B)] Convex combinations of asymmetric forms $U\bullet V^\dagger$ with (possibly distinct) unitaries $U$ and $V$.
\end{itemize}
For simplicity, we refer to a type-(B) superoperator composed of asymmetric forms $e^{i\theta}P\bullet Q^\dagger$ for some $\theta\in \mathbb{R}$ and $P,Q\in \{I,X,Y,Z\}^{\otimes N_{\rm q}}$ as an \textit{asymmetric Pauli mixture}.

\subsubsection{Taylor expansion of the transfer matrix of $e^{t\mathcal{L}}$}
We here derive a preliminary expansion of the dynamical map $e^{t\mathcal{L}}$
(Eq.~\eqref{eq:taylorser}-\eqref{eq:Bkl_main}),
which gives the starting point for proving Theorem~\ref{thm: main}.
The expansion clarifies the motivation to introduce the following techniques in Section~\ref{subsubsec:exactefficient_dissipation} and~\ref{subsubsec:lcs_main}.

In the transfer matrix framework, the Lindblad equation Eq.~\eqref{main:lindbladeq} can be written as
\begin{equation}
    \frac{{\rm d}}{{\rm d}t}\kett{\rho}=G\kett{\rho},
\end{equation}
where $G=S(\mathcal{L})$ is the transfer matrix of $\mathcal{L}$
\begin{align}
    G&=-i\bm{1}\otimes H+iH^{\rm T}\otimes\bm{1}\notag\\
    &~+\sum_{k=1}^K\left(\overline{L_k}\otimes L_k-\frac{1}{2}\bm{1}\otimes L_k^\dagger L_k-\frac{1}{2}L_k^{\rm T}\overline{L_k}\otimes \bm{1}\right).
\end{align}
Also, the transfer matrix of $e^{t\mathcal{L}}$ is given by $e^{tG}$.
By expanding the $r$-sliced propagator $e^{(t/r)G}$ via the Taylor expansion, 
we obtain a preliminary decomposition
\begin{equation}\label{eq:taylorser}
    e^{(t/r)G}=\sum_{l}\frac{(t/r)^{2l}}{(2l)!}{G}^{2l}G',
\end{equation}
where $G':=\bm{1}\otimes \bm{1}+tG/r(2l+1)$.
Then, defining $\alpha:=2\alpha_0+\sum_{k}\alpha_k^2$ and $\tau_l:=\alpha t/r(2l+1)$ for simplicity, we expand $G'$ into the sum of the following terms for $k=1,2,...,K$:
\begin{equation}\label{eq:main_SLL}
    \frac{\alpha_k^2}{\alpha} \left\{S(\mathcal{B}_{kl}) - \frac{\tau_l^2}{4}\frac{L_k^{\rm T} \overline{L_k}}{\alpha_k^2}\otimes \frac{L_k^\dagger L_k}{\alpha_k^2} \right\},
\end{equation}
\begin{equation}\label{eq:main_wan}
    \frac{\alpha_0}{\alpha} \left( 2\cdot \bm{1}\otimes \bm{1} + \bm{1}\otimes \frac{-i\tau_lH}{\alpha_0}+\frac{i\tau_l H^{\rm T}}{\alpha_0} \otimes \bm{1} \right),
\end{equation}
where $\mathcal{B}_{kl}$ is the CP map of the simplest dissipative operator defined as 
\begin{align}\label{eq:Bkl_main}
    \mathcal{B}_{kl}(\bullet)&=\sqrt{\tau_l}\frac{L_k}{{\alpha_k}}\bullet \left(\sqrt{\tau_l}\frac{L_k}{{\alpha_k}}\right)^\dagger\notag\\
    &+\left(\bm{1}-\frac{\tau_l}{2}\frac{L_k^\dagger L_k}{\alpha_k^2}\right) \bullet \left(\bm{1}-\frac{\tau_l}{2}\frac{L_k^\dagger L_k}{\alpha_k^2}\right)^\dagger.
\end{align}

Toward the equality Eq.~\eqref{main:dynmapdecomp} satisfying Eq.~\eqref{main:c_norm}, it is crucial to find an \textit{exact} linear decomposition of $G'$ with type-(A, B) superoperators such that the norm of coefficients scales as $1+\mathcal{O}((t/r)^2)$~\footnote{Although we can directly decompose $G'$ into an asymmetric Pauli mixture (type-(B) superoperators) up to the normalization factor $1+\mathcal{O}(t/r)$, this normalization factor is too large to achieve Eq.~\eqref{main:c_norm}. After the $r$-repetitions, this $1+\mathcal{O}(t/r)$ scaling naively provides an upper bound $e^{\|\mathcal{L} \|_{\rm pauli} t}$ instead of Eq.~\eqref{eq:c_norm_derivation}.}.
Below, we provide such a decomposition for the dissipative process $S(\mathcal{B}_{kl})$ in Section~\ref{subsubsec:exactefficient_dissipation} and then describe how to simulate the resulting superoperators with quantum circuits in Section~\ref{subsubsec:lcs_main}.

\subsubsection{Exact and efficient simulation of dissipative process}\label{subsubsec:exactefficient_dissipation}
We need an exact and efficient simulation of the dissipative process $S(\mathcal{B}_{kl})$ with a single jump operator in Eq.~\eqref{eq:Bkl_main}.
Then, more generally, we here derive Lemma~\ref{lemma_main:generalCP} for the exact simulation of CP maps.

To exactly simulate the CP map $\mathcal{B}_{kl}$, we cannot use the standard method for general quantum channel implementation~\cite{Cleve2016-yj} that relies on the LCU for channels followed by OAA.
This is because the method provides an (explicit) quantum circuit to effectively simulate a CPTN map $\mathcal{B}^{\mathrm{(approx)}}$ that only approximates a target CP map $\mathcal{B}$ when it lacks the trace-preserving property.

To overcome this difficulty, we develop a new technique for the exact simulation of {\it general} CP maps, by specifying and further eliminating the error of the OAA procedure, i.e., the error between a target CP map $\mathcal{B}$ and the CPTN map $\mathcal{B}^{\rm (approx)}$.
\begin{lem}[Exact decomposition of general CP maps]\label{lemma_main:generalCP}
    Let $\mathcal{B}=\sum_{\mu} B_{\mu}\bullet B_{\mu}^\dagger$ be a CP map for $B_\mu=\sum_{i}c_{\mu i}U_{\mu i}$ ($\mu=1,2,...$), each of which can be written as a linear combination of unitaries $\{U_{\mu i}\}$ with positive coefficients $\{c_{\mu i}\}$.
    Suppose that $\sum_{\mu} (\sum_i c_{\mu i})^2 \leq 4$ holds. 
    Then, for a Hermitian operator $D$ ($\|D\|_{\infty}\leq 1$) such that $$\sum_\mu B_{\mu}^\dagger B_{\mu}=\bm{1}+\delta D$$ 
    holds for some $\delta >0$, we can exactly simulate the following CPTN map $\mathcal{B}^{(\rm approx)}$ by using a quantum circuit constructed from $\{c_{\mu i}\}$ and (controlled) $\{U_{\mu i}\}$:
    \begin{equation}
    \mathcal{B}^{\mathrm{(approx)}}(\bullet) = \sum_{\mu} B_{\mu}' \bullet (B_{\mu}')^\dagger,
    ~B_{\mu}' := B_{\mu}\left( \bm{1} - \frac{\delta}{2}D  \right).
    \end{equation}
    Furthermore, if we know that $D$ can be described as a convex combination of unitaries, then the difference, which we call the correction superoperator, 
    $$
    \mathcal{R}:=\mathcal{B}-\mathcal{B}^{\rm (approx)}
    $$
    can be decomposed into a type-(B) superoperator up to the following normalization factor:
    \begin{equation}\label{eq:C_0_def}
        C_0:=\left(\delta+\frac{\delta^2}{4}\right)\left[\sum_{\mu}\left(\sum_i c_{\mu i}\right)^2\right].
    \end{equation}
\end{lem}

This lemma can be directly proved by Lemma~\ref{lemma:oaa_non_isometry} and~\ref{lemma:generalRsampling} in Appendix~\ref{apdx:exact_efficient_oaa}.
Importantly, the correction type-(B) superoperator $\mathcal{R}/C_0$ can be effectively simulated by the method in Section~\ref{subsubsec:lcs_main}.
As a result, we have an exact decomposition of a target CP map $\mathcal{B}$ with the superoperators that can be simulated on quantum circuits:
\begin{equation}\label{eq:exact_decomp_main}
    \mathcal{B}=(1+C_0)\times\left(\frac{1}{1+C_0}\mathcal{B}^{(\rm approx)}+\frac{C_0}{1+C_0}\frac{\mathcal{R}}{C_0}\right).
\end{equation}
By using this expression, we can effectively and exactly simulate a target CP map $\mathcal{B}$ on quantum circuits. 
That is, we can estimate ${\rm Tr}[O\mathcal{B}(\rho)]$ for any initial state $\rho$ by the following procedure: (i) select $\mathcal{B}^{\rm (approx)}$ or $\mathcal{R}/C_0$ with probability $1/(1+C_0)$ or $C_0/(1+C_0)$, (ii) apply the selected operation to $\rho$, (iii) measure the output state by $O$, and (iv) multiply $1+C_0$ with the measurement output.

The efficiency of the decomposition is determined by the normalization factor $1+C_0$, which corresponds to the sampling overhead to effectively simulate a general CP map $\mathcal{B}$.
In fact, in the procedure below of Eq.~\eqref{eq:exact_decomp_main}, $1+C_0$ increases the variance of the observable estimation; see also the discussion around Eq.~\eqref{eq:sample_estimator}.
From the definition of $C_0$ in Eq.~\eqref{eq:C_0_def}, we can efficiently and exactly simulate a general CP map when the factor $\delta$, which indicates how the target map is close to a CPTP map, is sufficiently small.
Actually, our target map $\mathcal{B}_{kl}$ in Eq.~\eqref{eq:Bkl_main} is a good example of this efficient implementation, as follows.

The Lemma~\ref{lemma_main:generalCP} allows us to simulate $\mathcal{B}_{kl}$ with the help of random sampling of \textit{correction} superoperator $\mathcal{R}_{kl}$.
That is, we construct a superoperator $\mathcal{R}_{kl}$ to recover the approximation error such that
\begin{equation}\label{eq:b_bapprox_r}
    \mathcal{B}_{kl} = \mathcal{B}^{(\rm approx)}_{kl}+\mathcal{R}_{kl}
\end{equation}
holds.
Here, we know that $D=\alpha_{k}^{-4} (L_k^\dagger L_k)^2$ is a convex combination of Pauli strings from the direct calculation.
Importantly,
while $\mathcal{R}_{kl}$ is generally not a CP map, it can be written as an asymmetric Pauli mixture up to a normalization factor of the order $\delta=\mathcal{O}(\tau_l^2)$. 
As a result, we have a linear decomposition of Eq.~\eqref{eq:main_SLL} with a small sum of coefficients \begin{equation}\label{eq:norm_scaling_1}
    \frac{\alpha^2_{k}}{\alpha}\left({1}+\mathcal{O}(\tau_l^2)\right),
\end{equation}
using the exactly simulatable type-(A) superoperator $\mathcal{B}^{\mathrm{(approx)}}_{kl}$ and asymmetric Pauli mixtures (in type-(B)).

Finally, we mention a quantum circuit $U_{kl}$ for realizing the CPTN map $\mathcal{B}_{kl}^{\rm (approx)}$; this circuit can be obtained from the standard method~\cite{Cleve2016-yj} (or see Appendix~\ref{apdx:exact_efficient_oaa} for the circuit diagram).
The circuit $U_{kl}$ has at most $\mathcal{O}(M \log M)$ depth and $3+\lceil \log_2 M\rceil$ additional ancilla qubits beyond the target system.
In addition, $U_{kl}$ satisfies
\begin{align}\label{eq:effectivesim_CPTN_main}
    \mathcal{B}^{\rm (approx)}_{kl}(\rho)&=1\times {\rm Tr}_{\overline{\rm sys}}\left[\tilde{\Pi} U_{kl}(\ketbra{\bm{0}}\otimes\rho) U_{kl}^\dagger\right]\notag\\
    &~~~+0\times {\rm Tr}_{\overline{\rm sys}}\left[\left(\bm{1}-\tilde{\Pi}\right) U_{kl}(\ketbra{\bm{0}}\otimes\rho) U_{kl}^\dagger\right]
\end{align}
for any $\rho$, where ${\rm Tr}_{\overline{{\rm sys}}}$ denotes the partial trace over all the qubits except for the target system, and $\tilde{\Pi}:=\bm{1}\otimes  \ket{0}\bra{0}^{\otimes 2+\lceil \log_2 M\rceil}\otimes I$ and $\ket{\bm{0}}=\ket{0}^{\otimes 3+\lceil \log_2 M\rceil}$.
Thus, by multiplying measurement outputs by 0 or 1, we can exactly simulate the CPTN map $\mathcal{B}_{kl}^{(\rm approx)}$.

\subsubsection{Linear combination of superoperators}\label{subsubsec:lcs_main}
After substituting Eq.~\eqref{eq:b_bapprox_r} into Eq.~\eqref{eq:taylorser}, we can rewrite $e^{tG}=(e^{(t/r)G})^r$ as a linear combination of (the transfer matrix of) superoperators $\mathcal{W}_v$. 
Here, each $\mathcal{W}_v$ can be written as a composition of type-(A, B) superoperators; please take a quick look at the first few paragraphs of the proof sketch in Section~\ref{subsubsec:sketchofproof}.
Even when the linear combination of superoperators $\sum_v c_v \mathcal{W}_v$ is determined, its circuit simulation remains nontrivial, particularly for type-(B) superoperators e.g., $U\bullet V^\dagger$.
Although we could obtain $\kett{U\rho V^\dagger}$ by applying $\overline{V}\otimes U$ to $\kett{\rho}$, this approach requires $(2N_{\rm q})$-qubit system and non-trivial final measurements with sampling overhead~\cite{Kamakari2022-QITE}.
To avoid large space and sampling overheads, we need to establish an effective simulation protocol for superoperators $\mathcal{W}_v$ by using quantum circuits without doubling the target system qubit size.

Here, we introduce a formalism to simulate superoperators, thereby ensuring that the superoperators of our interest can be 
simulated by quantum circuits with a minimal number of ancilla qubits.
When we have a composite map $\mathcal{W}_v$ of type-(A, B) superoperators,
by adding a single ancilla qubit, we introduce a map $\widetilde{\mathcal{W}}_v$ on the $N_{\rm q}+1$ qubits:
\begin{equation}
            \widetilde{\mathcal{W}}_v:
            \begin{pmatrix}
                A_{00}&A_{01}\\    
                A_{10}&A_{11}
            \end{pmatrix}
            \mapsto 
            \begin{pmatrix}
                *&\mathcal{W}_v(A_{01})\\    
                \mathcal{J}\circ \mathcal{W}_v \circ \mathcal{J}(A_{10})&*
            \end{pmatrix},
\end{equation}
where $A_{ij}:=\bra{i}_{\rm anc} A \ket{j}_{\rm anc}$  with an ancilla state $\ket{j}_{\rm anc}$, and $\mathcal{J}$ is an anti-linear map that $\mathcal{J}:A\mapsto A^\dagger$.
We can construct such $\widetilde{\mathcal{W}}_v$ from $\mathcal{W}_v$ by the following simple replacement;
\begin{itemize}
    \item for type-(A), the CPTN map $\Phi(\bullet)$ is replaced with $\mathcal{I}_{\rm anc}\otimes \Phi$,
    \item for type-(B), $\sum_i p_i U_i\bullet V_i^\dagger$ is replaced with the mixed unitary channel $\sum_i p_i \mathcal{U}_i$ where $\mathcal{U}_i$ is defined by the unitary
    \begin{equation}
    \ket{0}\bra{0}_{\rm anc}\otimes U_i + \ket{1}\bra{1}_{\rm anc}\otimes V_i.
    \end{equation}
\end{itemize}
Obviously, $\mathcal{U}_i$ is implementable in quantum circuits even though $U_i\bullet V^\dagger_i$ is not.
The translation is summarized in Table~\ref{tab:translationlist}.

Using the translated superoperators $\widetilde{\mathcal{W}}_v$, we provide a circuit simulation protocol of the linear combination of superoperators (LCS)
$\sum_v c_v \mathcal{W}_v$ as follows.
 \begin{table}
        \centering
        \begin{tabular}{c|c}
                Superoperator & Circuit Diagram \\ 
                \hline
                \hline
                \begin{tabular}{c}
                     (A) CPTN map $\mathcal{B}$
                \end{tabular}
                ~~~&
                \begin{quantikz}[]
                \\
                &\qwbundle{1} & \qw & \qw\\
                &\qwbundle{} & \gate{\mathcal{B}} & \qw
                \\
                \end{quantikz} \\ \hline
                \begin{tabular}{c}
                     (B) Convex combination \\
                     of asymmetric forms\\ 
                     $\sum_i p_i U_i \bullet V^\dagger_i $
                \end{tabular} &~~$\sum_i p_i$
                \begin{quantikz}[]
                \\
                &\qwbundle{1} & \octrl{1} & \ctrl{1} &\qw\\
                &\qwbundle{}  & \gate{U_i}  & \gate{V_i} & \qw
                \\
                \end{quantikz} \\ \hline
        \end{tabular}
        \caption{Translation rule from type-(A,B) superoperators into quantum circuits. The top (bottom) line represents the ancilla (system) qubit(s).}
        \label{tab:translationlist}
    \end{table}
\begin{lem}[LCS, simplified]\label{lem:lcs_simplified}
    Let $\{\mathcal{W}_v\}$ be composite maps of type-(A) and type-(B) superoperators,
    and $\{c_v\}$ be positive values,
    such that $\sum_v c_v \mathcal{W}_v$ is a Hermitian-preserving map, i.e., $\mathcal{J} \circ (\sum_v c_v \mathcal{W}_v) \circ  \mathcal{J} = \sum_v c_v \mathcal{W}_v$.
    Also, let $\widetilde{\mathcal{W}}_v$ be a CPTN map obtained from $\mathcal{W}_v$ by the aforementioned replacement of type-(A) and type-(B) superoperators.
    Then, we can simulate the linear combination of superoperators $\sum_v c_v \mathcal{W}_v$ by
    \begin{multline}\label{eq:lcu_lemmma_statement}
        \sum_v c_v \mathcal{W}_v(\bullet) \\
        = {\rm Tr}_{\rm anc}\left[(X_{\rm anc}\otimes \bm{1})\cdot \left(\sum_v c_v \widetilde{\mathcal{W}}_v\right)\left(|+\rangle\langle +|_{\rm anc}\otimes \bullet\right)\right].
    \end{multline}
\end{lem}
A more general statement and its proof are deferred to Proposition~\ref{prop:lcs_alternative_seq} in Appendix~\ref{apdx:construct_circuits}.
It allows us to simulate LCS $\{\mathcal{W}_v\}$ by quantum circuits specified by $\{\widetilde{\mathcal{W}}_v\}$ via Eq.~\eqref{eq:lcu_lemmma_statement}.
For clarity, we mention the simplest version of this lemma.
That is, we can simulate any Hermitian-preserving map $\Phi=\sum_{i} p_i U_i \bullet V_i^\dagger$ with distinct unitaries $U_i, V_i$ and probability $p_i$, by using the quantum circuits $\{\mathcal{U}_i\}$. 
This can be confirmed by
\begin{align}
    \widetilde{\Phi}(\ketbra{+} \otimes \rho) 
        &=\frac{1}{2}\sum_{i} p_i
        \begin{pmatrix}
            * & U_i \rho V_i^\dagger\\
            V_i \rho U_i^\dagger& *
        \end{pmatrix}\notag\\
        &= \frac{1}{2}
        \begin{pmatrix}
            * & \Phi(\rho)\\
            \mathcal{J}\circ\Phi\circ \mathcal{J}(\rho)& *
        \end{pmatrix}.
\end{align}
Note that the simple case can be generalized for maps without Hermitian-preserving property, by effectively measuring $(X-iY)_{\rm anc}$ instead of $X_{\rm anc}$ as stated in Proposition~\ref{prop:lcs} in Appendix~\ref{apdx:construct_circuits}.

To conclude, we highlight the utility of the formalism.
This formalism provides a quantum algorithm for simulating LCS, which can be described as composite maps of (A) CPTN maps and (B) convex combinations of $U_i \bullet V^\dagger_i $. The composite maps are general and can express a wide range of tasks, for example, the randomized approach for Hamiltonian simulation~\cite{chakraborty2024implementing}.

\subsubsection{Sketch of the proof}\label{subsubsec:sketchofproof}
Now, we are ready to prove Theorem~\ref{thm: main}.
Figure~\ref{fig:proof_of_thm1} provides a schematic overview of the proof.
\begin{figure*}[htbp]
    \centering
    \includegraphics[width=0.9\textwidth]{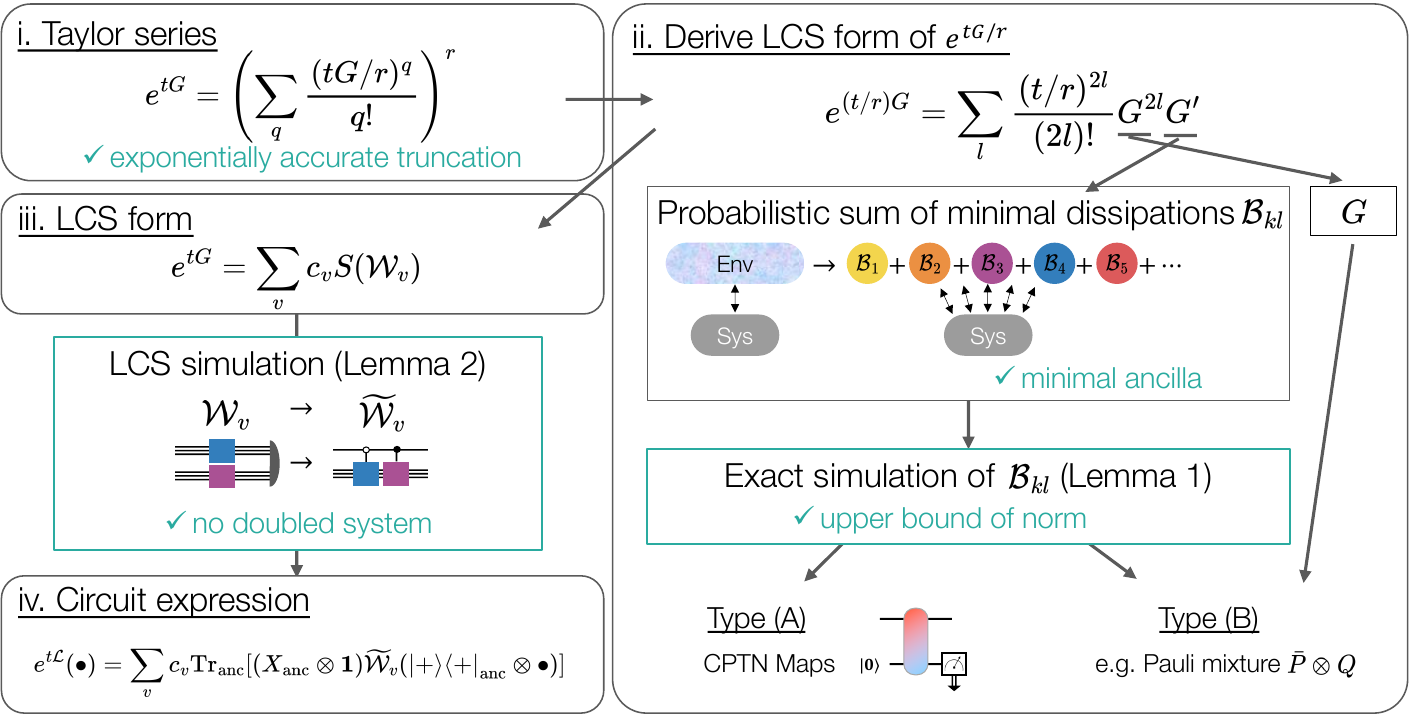}
    \caption{A flowchart for the derivation of Theorem~\ref{thm: main}.
     (i) Expansion of the Taylor series of the transfer matrix $e^{tG}$ that can be truncated exponentially accurately.
     (ii) Decomposition of Eq.~\eqref{eq:taylorser} into type-(A, B) superoperators. For simplicity, the figure shows the flowchart for the dissipation part only. First, we can find the convex combination of minimal dissipative processes $\mathcal{B}_{kl}$ by Eq.~\eqref{eq:main_SLL}. Then, using Lemma~\ref{lemma_main:generalCP} for the circuit simulation of $\mathcal{B}_{kl}$, we obtain type-(A) superoperators. All other terms including those from Hamiltonian dynamics can be described as type-(B) superoperators.   
     (iii) The desirable linear combination of superoperators $\sum_v c_v S(\mathcal{W}_v)$, obtained by repeating $e^{(t/r)G}$.
     $\mathcal{W}_v$ is a composite map of type-(A, B) superoperators.
     (iv) Translation of LCS form $\sum_v c_v S(\mathcal{W}_v)$ into the circuit expression Eq.~\eqref{main:dynmapdecomp}. Lemma~\ref{lem:lcs_simplified} provides the translation protocol.
     }
    \label{fig:proof_of_thm1}
\end{figure*}

\begin{proof}[Sketch of the proof]
We aim to have the decomposition of Eq.~\eqref{main:dynmapdecomp} through that of $e^{tG}$ on the extended space.
As the first step, we decompose $e^{tG}$ as a linear combination of (the transfer matrix of) superoperators $\mathcal{W}_v$ that can be written as a product of type-(A, B) superoperators.
Then, we translate the composite map $\mathcal{W}_v$ into the form ${\rm Tr}_{\rm anc}[\cdots]$ of Eq.~\eqref{main:dynmapdecomp}.

We focus on the preliminary decomposition Eq.~\eqref{eq:taylorser} of $e^{(t/r)G}$ and provide decompositions of $G$ and $G'$ which consists of Eqs.~\eqref{eq:main_SLL} and~\eqref{eq:main_wan}.
We can represent $G/\|\mathcal{L}\|_{\rm pauli}$ as an asymmetric Pauli mixture, from the assumption of the access model \eqref{main:H_L}.
Eq.~\eqref{eq:main_SLL} can be written by the type-(A) superoperator $\mathcal{B}_{kl}^{(\rm approx)}$ and asymmetric Pauli mixtures, using Eq.~\eqref{eq:b_bapprox_r}.
As for the decomposition of Eq.~\eqref{eq:main_wan}, we can apply the decomposition method for Hamiltonian simulation in Refs.~\cite{Wan2022-tx, chakraborty2024implementing}, which provides
\begin{align}
    &\bm{1} \otimes \bm{1} + \bm{1}\otimes \frac{-i\tau_lH}{\alpha_0}\notag\\
    &=\sqrt{1 + \tau_l^2} \sum_{j=1}^{m}
    \frac{|\alpha_{0j}|}{\alpha_0}
    \cdot \bm{1} \otimes e^{-i \theta_{l} \mathrm{sgn}(\alpha_{0j}) P_{0j}},
\end{align}
where $\theta_l:=\arccos(\{1+\tau_l^2\}^{-1/2})$.
By applying the similar decomposition for the term of $\bm{1} \otimes \bm{1} + (i\tau_l H^{\rm T}/\alpha_0) \otimes  \bm{1}$, Eq.~\eqref{eq:main_wan} with the normalization factor \begin{equation}\label{eq:norm_scaling_2}
    \frac{2\alpha_0}{\alpha}\sqrt{1+\tau_l^2}=\frac{2\alpha_0}{\alpha} \left({1}+\mathcal{O}(\tau_l^2)\right)
\end{equation}
can be written as a type-(B) superoperator with asymmetric forms of $\bm{1}\otimes U$ or $\overline{V}\otimes \bm{1}$ with some Pauli rotation gates $U$ and $V$.

Therefore, combining the above results together with Eq.~\eqref{eq:taylorser}, we arrive at a linear decomposition of $e^{(t/r){G}}$ with composite maps of type-(A, B) superoperators. 
This immediately leads to a linear decomposition of $e^{t{G}}$ as $e^{t{G}}=\sum_{v}c_v S(\mathcal{W}_v)$ for some indices $v$, $c_{v}\geq 0$, and composite maps $\mathcal{W}_v$. 
Since the norm of coefficients in the decomposition scales as $1 + \mathcal{O}((t/r)^2)$ in Eqs.~\eqref{eq:norm_scaling_1} and~\eqref{eq:norm_scaling_2}, the upper bound of the sum of coefficients is obtained as
\begin{align}\label{eq:c_norm_derivation}
\sum_{v \in \mathrm{S}} c_v &= \left(\sum_{l=0}^{\infty} \frac{(\tau/r)^{2l}}{(2l)!}\left(1 + \mathcal{O}(\tau_l^2)\right)\right)^r \le e^{2 \tau^2 / r},
\end{align}
where $\tau =t\|\mathcal{L}\|_{\rm pauli}$. 
Thus, we have Eq.~\eqref{main:c_norm}.

In the final step toward obtaining Eq.~\eqref{main:dynmapdecomp}, we construct the CPTN map $\widetilde{\mathcal{W}}_v$ on the $N_{\rm q}+1$ qubits from each superoperator $\mathcal{W}_v$, via the translation rule in Table~\ref{tab:translationlist}.
Then, we complete the proof of Eq.~\eqref{main:dynmapdecomp} by directly using Lemma~\ref{lem:lcs_simplified}.
Here, we note that $e^{t\mathcal{L}}=\sum_v c_v \mathcal{W}_v$ and $e^{t\mathcal{L}} = \mathcal{J}\circ e^{t\mathcal{L}} \circ \mathcal{J}$ from the Hermitian-preserving property of $e^{t\mathcal{L}}$, thereby satisfying the assumption of Lemma~\ref{lem:lcs_simplified}.

To simulate the $(N_{\rm q}+1)$-qubit CPTN map $\widetilde{\mathcal{W}}_v$, we need to introduce additional $3+\lceil \log_2 M\rceil$ qubits.
This additional qubits are required for quantum circuits $U_{kl}$ in Eq.~\eqref{eq:effectivesim_CPTN_main}, obtained from the standard method~\cite{Cleve2016-yj}, to effectively simulate the CPTN maps $\mathcal{B}^{(\rm approx)}_{kl}$.
These additional qubits are measured by the computational basis in the middle of the entire circuit and then reset for the next $\mathcal{B}_{kl}^{(\rm approx)}$.
Performing a classical post-processing on the final outputs in order to drop an unnecessary part of the CPTP process by unitary circuits (see Eq.~\eqref{eq:effectivesim_CPTN_main}), we can effectively realize CPTN maps  $\widetilde{\mathcal{W}}_v$ using the circuits illustrated in Fig.~\ref{fig:main_circ}.
\end{proof}

\subsubsection{Gate complexity and logarithmic circuit depth of $\widetilde{\mathcal{W}_{v}}$}\label{sec:gatecomplexityofw}

In Theorem~\ref{thm: main},
the complexity to simulate a target open system is well reflected by that of $\widetilde{\mathcal{W}}_v$, which can be effectively simulated on a quantum device.
The additional space overhead and the total gate complexity for each circuit for $\widetilde{\mathcal{W}}_v$ does not depend on either the number $m$ of Pauli strings in a Hamiltonian $H$ or the number $K$ of jump operators, while the collection of $\widetilde{\mathcal{W}}_v$ recovers the full dynamics $e^{t\mathcal{L}}$.
This parameter independence is attributed to the primitive type-(A, B) superoperators forming each $\mathcal{W}_v$ in our decomposition; each superoperator is clearly independent from $m$ and $K$.
For the type-(B) superoperators, the independence is evident.
For the type-(A) superoperators, Lemma~\ref{lemma:oaa_for_B0B1} in Appendix~\ref{apdx:exact_efficient_oaa} provides the explicit circuit construction and detailed resource analysis of $\mathcal{B}_{kl}^{(\rm approx)}$.
The dominant contribution arises from encoding the jump operator $L_k$, which requires $\mathcal{O}(M\log M)$ gates using LCU.
As a result, the gate complexity of our type-(A) superoperators depends solely on $M$.

Toward developing a practically effective algorithm for estimating the expectation value $\mathrm{Tr}[O\rho(t)]$, let us discuss the maximal circuit depth for $\widetilde{\mathcal{W}}_v$. 
The decomposition \eqref{main:dynmapdecomp} is attributed to the Taylor series expansion Eq.~\eqref{eq:taylorser}; the CPTN map $\widetilde{\mathcal{W}}_v$ contains $2l$ sequential applications of the mixed unitary channel, which may result in a large depth circuit without any truncation.
However, the Taylor expansion is exponentially accurate with respect to the truncation order $\mathrm{Q}$, meaning that we only need circuits with
\begin{equation}
    l\leq \mathrm{Q}=\mathcal{O}\left(\frac{\log(r/\Delta)}{\log\log(r/\Delta)}\right)=\mathcal{O}\left(\log(r/\Delta)\right)    
\end{equation}
to achieve accuracy $\Delta/r$ for each segment.
Therefore, we have a finite subset $\mathrm{S}_{\Delta}$ of $\mathrm{S}$ satisfying the following properties: (i) Eq.~\eqref{main:dynmapdecomp} holds up to error $\Delta$, and (ii) the circuit depth for $\widetilde{\mathcal{W}}_v$ is $\mathcal{O}(r\mathrm{Q})=\mathcal{O}(r \log(r/\Delta))$ for all $v\in \mathrm{S}_{\Delta}$.
In addition, we can efficiently sample the explicit circuit for $\widetilde{\mathcal{W}}_v$ according to the distribution proportional to $c_v$ for $v\in \mathrm{S}_{\Delta}$; see Algorithm~\ref{alg_circuit_generation} in Appendix~\ref{apdx:C_alg}.
With these preliminaries, we now describe our main algorithm in what follows.

\subsection{Quantum algorithm for expectation value estimation}\label{sec:algorithm}
Here we derive a first algorithm (Theorem~\ref{thm: main2} and Algorithm~\ref{alg:obs_estimaion}).
Our quantum algorithm for estimating $\mathrm{Tr}[O\rho(t)]$ works as follows.
First, we randomly generate a quantum circuit from the set $\{ \widetilde{\mathcal{W}}_v \}$ according to the probability distribution $\{ c_v / C \} $, where 
\begin{equation}\label{eq:normalization_factor_main}
    C:=\sum_{v\in \mathrm{S}_{\Delta}} c_v \leq \sum_{v\in \mathrm{S}} c_v \leq e^{2\|\mathcal{L}\|_{\rm pauli}^2t^2/r}.
\end{equation}
The sampled circuit contains (possibly multi-round) mid-circuit measurement with binary outcome(s) $b\in\{0,1\}$ and qubit reset for realizing the CPTN property of $\widetilde{\mathcal{W}}_v$. 
Then, running the quantum circuit with the initial state $\ket{+}\bra{+}_{\rm anc}\otimes \rho_0$ followed by the measurement for observable $X_{\rm anc}\otimes O$ at the end of the circuit,
we record measurement results $(b_X,b_O)$ and a collection of mid-circuit measurement outcomes $\bm{b}=(b_1, b_2,...)$ ($b_i\in \{0,1\}$).
By repeating the above procedure $N$ times independently, we calculate the average of the obtained results as
\begin{equation}\label{eq:sample_estimator}
    \varphi_N = \frac{C}{N}\sum_{i=1}^N b_X^{(i)} b_O^{(i)} \delta_{\bm{b}^{(i)},\bm{0}}
\end{equation}
where $i$ denotes the trial index; this quantity serves as an estimator for $\mathrm{Tr}[O\rho(t)]$. 
Note that the CPTN map is effectively realized by dropping an unnecessary part of the CPTP process by the quantum circuit, through the classical post-processing of $\delta_{\bm{b}^{(i)},\bm{0}}$; see the end of the sketch of proof for Theorem~\ref{thm: main}.

Using the construction of $\varphi_N$, we prove the existence of a quantum algorithm to achieve our goal in the following theorem.
\begin{thm}\label{thm: main2}
    For any Hamiltonian $H$ and jump operators $\{L_k\}_{k=1}^K$ specified by Eq.~\eqref{main:H_L}, there exists a quantum algorithm that estimates the expectation value of an observable $O$ for the $N_{\rm q}$-qubit Lindblad dynamics Eq.~\eqref{main:lindbladeq} with the use of additional $4+\lceil \log_2 M \rceil$ ancilla qubits, where $M$ is the number of Pauli strings contained in a single jump operator $L_k$.
    For given additive error $\varepsilon$, $\delta$, and simulation time $t$, this algorithm outputs an $\varepsilon$-close estimate for the expectation value with at least $1-\delta$ probability, using $\mathcal{O}(\|O\|^2\log(1/\delta)/\varepsilon^2)$ samples from the set of quantum circuits such that their gate complexity is
    \begin{equation}\label{eq:gatecomplexity}
    \mathcal{O}\left(\tau^2 \left(\frac{\log(\|O\|\tau/\varepsilon)}{\log \log(\|O\|\tau/\varepsilon)} + M \log M \right) \right),
    \end{equation}
    where $\tau =\|\mathcal{L}\|_{\rm pauli} t$.
    Importantly, the gate complexity is independent of $m$ and $K$.
\end{thm}

The sketch of the proof is as follows; see the full proof in Appendix~\ref{apdx:C_alg}.
The estimator $\varphi_N$ for the true expectation value has a bias of magnitude at most $\Delta\|O\|$, where $\|O\|$ denotes the operator norm of $O$.
Thus, taking $\Delta=\mathcal{O}(\varepsilon/\|O\|)$, we conclude that 
\begin{equation}
    N=\mathcal{O}\left(C^2\|O\|^2 \log(1/\delta)/\varepsilon^2\right)    
\end{equation} 
samples are sufficient to assure that $\varphi_N$ has an additive error $\varepsilon$ with a high probability due to Hoeffding's inequality. 
Here, the factor $C$ represents the sampling overhead in the proposed algorithm, because our problem can be solved by $\mathcal{O}(\|O\|^2 \log(1/\delta) /\varepsilon^2)$ calls of some quantum algorithm that directly prepares $\rho(t)$ (when we use no amplitude estimation).
The sampling overhead $C$, which is upper bounded as Eq.~\eqref{eq:normalization_factor_main}, is controllable in our method by adjusting the circuit depth or equivalently the time slicing $r$.
In particular, we can make $C=\mathcal{O}(1)$ by taking $r=\mathcal{O}(\|\mathcal{L}\|_{\rm pauli}^2 t^2)$.
Therefore, this choice of $r$ combined with Theorem~\ref{thm: main} (more precisely, the truncated version of the decomposition) completes the proof of Theorem~\ref{thm: main2}.

Summarizing the above discussions, we arrive at our solution, Algorithm~\ref{alg:obs_estimaion}, for probing the physical properties of the Lindblad dynamics. 
Note that we have focused on expectation value estimation in the discussion, it is noteworthy that the proposed Lindblad dynamics simulation is not limited to observable estimation, but can also be applied to other tasks such as the estimation of non-linear functions of the final state. See Appendix~\ref{apdx:nonlinear} for details.

Note that the gate complexity is based on the worst-case.
To evaluate the average-case gate complexity, two sources of randomness are considered.
One is a random compilation of Hamiltonian simulation (Eq.~\eqref{eq:main_SLL}) or 
dissipation (Eq.~\eqref{eq:main_wan}) according to the dissipation intensity, and the other is the order of approximation $l$.
The first randomness governs the $\mathcal{O}(M \log (M))$ term and the second one governs the 
$\frac{\log(\|O\|\tau/\varepsilon)}{\log \log(\|O\|\tau/\varepsilon)}$
term in our gate complexity.
Particularly, if the jump intensity is low, the average gate complexity can be significantly reduced by avoiding more costly implementation of the dissipation. We will provide numerical examples in Section~\ref{sec:benchmark}.

\begin{figure}[htbp]
    \begin{algorithm}[H]
    \caption{Estimation of $\mathrm{Tr}[O e^{t \mathcal{L}}(\rho_0)]$}
    \label{alg:obs_estimaion}
    \begin{algorithmic}[1]
    \REQUIRE Hamiltonian $H$ and jump operators $\{L_k\}_{k=1}^K$ specified by Eq.~\eqref{main:H_L}, observable $O$ with the spectral norm $\|O\|$, initial state $\rho_0$,
    additive errors $\varepsilon$, and simulation time $t>0$.
    \ENSURE An estimate $\varphi_{N}$ that is $\varepsilon$-close to $\mathrm{Tr}[O e^{t \mathcal{L}}(\rho_0)]$ with a high probability.
    \STATE Set the number of time segments $r \leftarrow \lceil 2 \|\mathcal{L}\|_{\rm pauli}^2 t^2\rceil$.
    \STATE Calculate the constant $C$ given by Eq.~\eqref{eq:normalization_factor_main} with $\Delta  \leftarrow \mathcal{O}(\varepsilon / \| O \|)$, by Algorithm~\ref{alg_circuit_generation} in Appendix \ref{apdx:C_alg}.
    \STATE Set the number of samples $N \leftarrow \mathcal{O}(C^2 \|O\|^2 /\varepsilon^2)$.
    \FOR{$i = 1$ to $N$}
        \STATE Get a quantum circuit $ \widetilde{\mathcal{W}}_i$ by running Algorithm~\ref{alg_circuit_generation}.
        \STATE Measure $X_{\rm anc} \otimes O$ on the circuit with $\widetilde{\mathcal{W}}_i$ in Fig.~\ref{fig:main_circ} and obtain a measurement outcome $(b^{(i)}_X$, $b^{(i)}_O$, $\bm{b}^{(i)})$.
    \ENDFOR
    \RETURN An estimate $\varphi_N:= \frac{C}{N}\sum_{i=1}^N b_X^{(i)} b_O^{(i)} \delta_{\bm{b}^{(i)},\bm{0}}$.
    \end{algorithmic}
    \end{algorithm}
\end{figure}

\subsection{$(m,M,K)$-independent Lindblad simulation}
The remaining $M$-dependence in both gate and ancilla counts can be removed at the cost of worsening the scaling of $\tau$ and $\log(1/\varepsilon)$ polynomially in the gate count.
The $M$-dependence comes from the deterministic implementation using $U_{kl}$ (see Eq.~\eqref{eq:effectivesim_CPTN_main}) for the selected single dissipation $\mathcal{B}_{kl}^{(\rm approx)}$. 
However, this is not a unique way to simulate the dissipation in our algorithm. 
Indeed, to prove Theorem~\ref{thm: main}, it is sufficient to have a linear map $\Upsilon_{kl}$ that maps
\begin{equation}\label{main:eq_target_upsilon}
        \begin{pmatrix}
            A_{00}&A_{01}\\    
            A_{10}&A_{11}
        \end{pmatrix}
        \mapsto 
        \begin{pmatrix}
            *&\mathcal{B}^{(\rm approx)}_{kl}(A_{01})\\    
            \mathcal{B}^{(\rm approx)}_{kl} (A_{10})&*
        \end{pmatrix}.
\end{equation}
Making use of the degree of freedom, we may remove the $M$-dependency in our algorithm.

The key technique removing $M$-dependency is the randomized Hamiltonian simulation~\cite{Wan2022-tx,chakraborty2024implementing}, which is equivalent to our case when $L_k\to 0$.
We take the real and imaginary part $L_{k}^{\rm R,I}/\alpha_k$ of $L_k/\alpha_k$ as the Hamiltonian.
By taking the time slicing $r'$, the time evolution of $L_{k}^{\rm R,I}/\alpha_k$, i.e., $ e^{-i (L_{k}^{\rm R,I}/\alpha_k) t'}$ can be effectively simulated by circuits with the sampling overhead at most $e^{2(t')^2/r'}$, only a single-ancilla qubit (denoted by anc), and $\mathcal{\tilde{O}}(r')$ gate complexity, which is independent of the number of terms $M$ in $L_{k}$.
Here, the simulation time $t'$ for the Hamiltonians $L_{k}^{\rm R,I}/\alpha_k$ is taken as $t'=1/2$ in the following and differs from the target time $t$ for the Lindblad simulation.

By combining the randomized Hamiltonian simulation for $L_{k}^{\rm R,I}/\alpha_k$ and quantum singular value transformation~\cite{gilyen2019quantum}, we can construct a mixed unitary channel that approximately simulates Eq.~\eqref{main:eq_target_upsilon} with the sampling overhead $\leq e^{\mathcal{O}({\rm log} (1/\varepsilon'))/r'}$ for the approximation error $\varepsilon'$.
This mixed unitary channel acts on a $7$-ancilla-qubit system (including the single-qubit for anc) and has the (non-Clifford) gate complexity $\mathcal{O}(r'\log(1/\varepsilon'))$; see Lemma~\ref{lem:anotherBkl_imple} in Appendix~\ref{apdx:D_mMK_indep}.
The gate and ancilla counts are now independent of $M$ due to the randomization for $L_k/\alpha_k$ (more precisely, $L_k^{\rm R,I}/\alpha_k$).

Algorithm~\ref{alg:obs_estimaion} still returns the estimate for the target expectation value, after replacing all $U_{kl}$ in the circuit sampled in Algorithm~\ref{alg:obs_estimaion} with the mixed unitary channel followed by the multiplication of the additional sampling overhead.
As a result, by carefully choosing the additional parameters $\varepsilon'$ and $r'$, we can prove the following theorem.

\begin{thm}[$(m,M,K)$-independent Lindblad simulation]\label{thm:mMKindep_simulator}
    The same problem as Theorem~\ref{thm: main2} can be solved by using $\mathcal{O}(\|O\|^2\log(1/\delta)/\varepsilon^2)$ samples from a set of quantum circuits that use $7$ ancilla qubits and have the gate complexity $$\mathcal{O}\left(\tau^4\frac{\log^3(\|O\|\tau/\varepsilon)}{\log\log(\|O\|\tau/\varepsilon)}\right),
    $$
    where the prefactor hidden in $\mathcal{O}(\cdots)$ is independent of all parameters in $(m,M,K)$.
\end{thm}

The proof is provided in Appendix~\ref{apdx:D_mMK_indep}; the $\tau^4$ scaling comes from the concatenation of the randomized techniques.
Several Lindblad dynamics for example, thermal simulation~\cite{chen2025efficient} and open systems with quasi-local jumps~\cite{nachtergaele2019quasi} may have a significantly large $M$ with the norm $\|\mathcal{L}\|_{\rm pauli}$ being bounded.
In such cases, our entirely randomized algorithm offers a further reduction in gate counts compared to the algorithm in Theorem~\ref{thm: main2}.
We leave for future work the detailed comparison with existing works in gate counts for such highly complicated dissipative systems.

\section{Numerical simulation}\label{sec:main-numeric}
Here, we provide three types of numerical validations: the dynamics simulation to validate the algorithm, gate complexity analysis to reveal the utility of the algorithm (Algorithm~\ref{alg:obs_estimaion} or Theorem~\ref{thm: main2}), and benchmarking in more practical and classically challenging setups.
\subsection{Dynamics simulation}
The tested instance is a two-level system with decay, described by the following Hamiltonian and the jump operator:
\begin{equation}\label{eq:numerical_tls}
    H = - \frac{\delta}{2} Z - \frac{\Omega}{2}X,  ~~~~ L = \sqrt{\gamma} \frac{X - i Y}{2},
\end{equation}
where $\delta$, $\Omega$, and $\gamma$ are the detuning, the Rabi frequency, and the emission rate, respectively.
This model is simple but well-studied \cite{Kamakari2022-QITE, Schlimgen2022-vectorizationTalor}.
We apply the proposed algorithm for the Lindblad dynamics defined by Eq.~\eqref{eq:numerical_tls}, and the observable $O:=\ketbra{0}$,
which corresponds to the state population of the excited state.

The results of the numerical simulation are shown in Fig.~\ref{fig:excited_state_population}.
Through the demonstration, the numerical results are obtained by Qiskit density matrix simulation.
Figure~\ref{fig:excited_state_population} (a) shows the comparison between the simulation output and an exact solution obtained by QuTiP~\cite{qutip1, qutip2}.
It can be verified that our results are in good agreement with the exact solution.

Furthermore, Fig.~\ref{fig:excited_state_population} (b) shows that 
the total norm $C$, which corresponds to the output of Algorithm~\ref{alg_circuit_generation}, can be maintained at a constant by an appropriate choice of the number of segments $r$,
where we set $r = \max[\lceil 2 \| \mathcal{L} \|_{\mathrm{pauli}}^2 t^2 \rceil ,1 ]$.
This indicates that we need only a few overheads ($\approx 1.4$) for the rescaling of the sample mean.
In addition, it can be seen that the sampling overhead is much smaller than the upper bound shown by Theorem~\ref{thm: main}.
It means that we can expect further cost reduction from the theoretical bound for practical setups.

\begin{figure*}[htbp]
    \centering
    \begin{tabular}{cc}
         \includegraphics[width=0.45\textwidth]{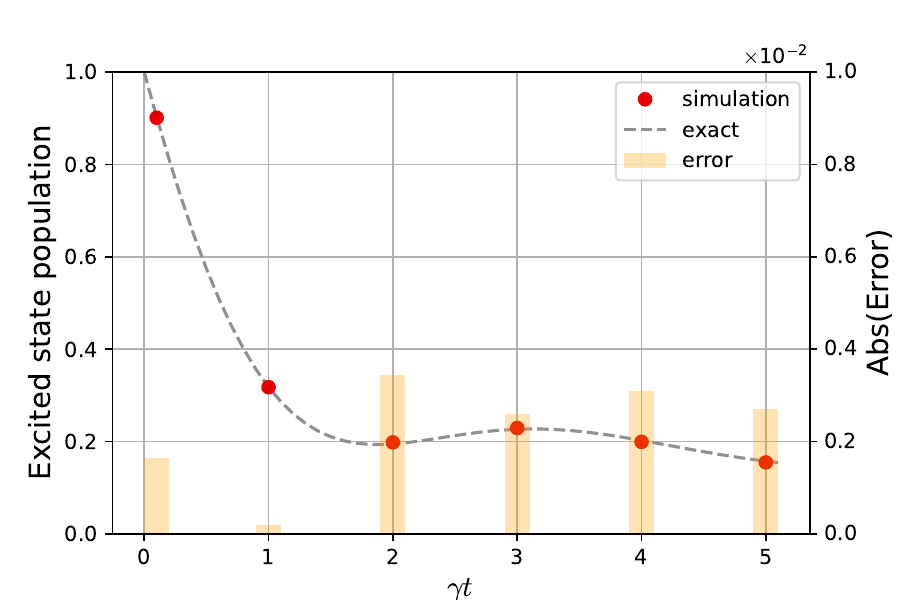}&\includegraphics[width=0.45\textwidth]{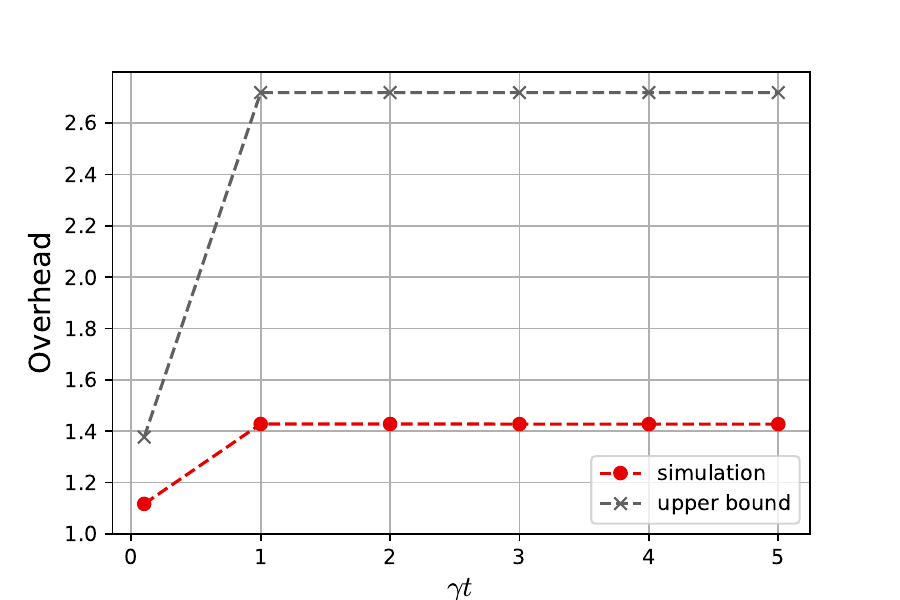}\\
         (a) Excited state population &  (b) Total norm $C$ 
         \\
         \includegraphics[width=0.45\textwidth]{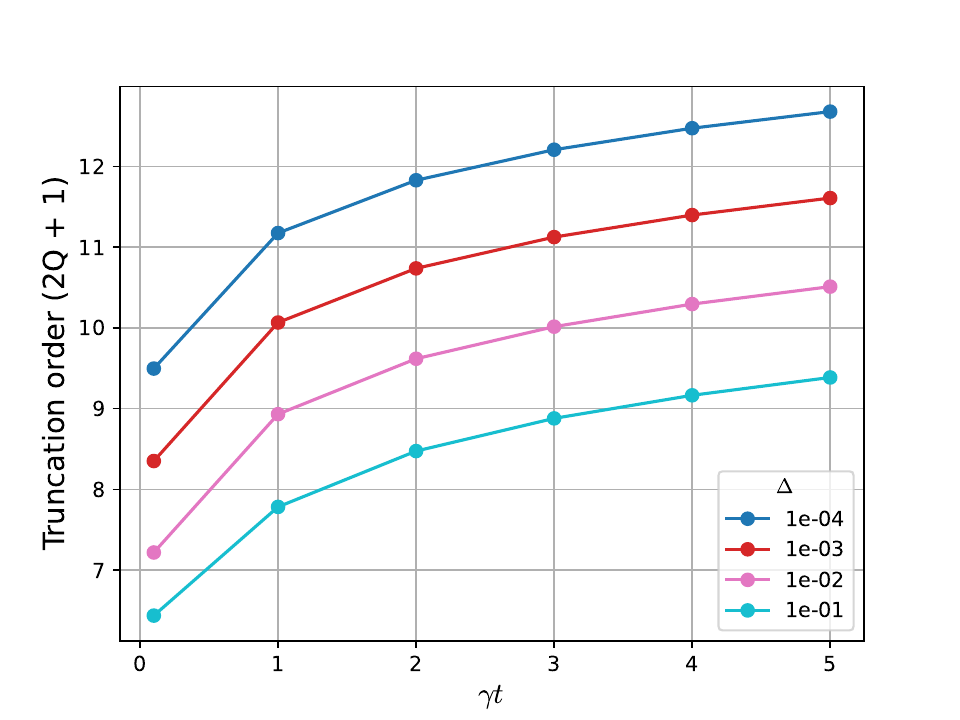} &
         \raisebox{9em}{
             \begin{tabular}{ll}\hline
                        Items & Values\\ \hline
                dynamics parameters  &  $\delta = \Omega = \gamma = 1$ \\
                simulation time & $t \in \{ 0.1, 1, 2, 3, 4, 5 \}$\\
                initial state  & $\ketbra{0}$\\
                observable  & $\ketbra{0}$\\
                the number of segments & $r = \max
            [\lceil 2 \| \mathcal{L} \|_{\mathrm{pauli}}^2 t^2 \rceil ,1 ]$\\
                target error level& $\Delta = 10^{-2}$\\
                the number of samples & $N = 2 \times 10^4$ \\ \hline
            \end{tabular}
        } \\
     (c) Truncation order & (d) Simulation setup
    \end{tabular}
    \caption{Numerical simulation of the two-level system with the decay.
    (a) The comparison of the excited state population between the exact solution and the simulation by the proposed algorithm.
    The \textit{exact} solution is obtained by QuTiP,
    and the \textit{simulation} is obtained with Qiskit density matrix simulation.
    The \textit{error} indicates the absolute error between \textit{exact} and \textit{simulation} for the single experiment on the right axis.
    For simplicity, we directly sampled $\mathcal{B}^{\mathrm{(approx)}}$ instead of OAA circuits.
    (b) The total norm $C$ with respect to $\gamma t$. 
    $C$ is less than $1.5$ at each time $t$ due to the appropriate choice of $r= \max[\lceil 2 \| \mathcal{L} \|_{\mathrm{pauli}}^2 t^2 \rceil ,1 ]$.
    The upper bound equals to $e$ for $2 \|\mathcal{L}\|_{\mathrm{pauli}}^2 t^2 \ge 1$.
    (c) Taylor series truncation dependence on the required $\Delta$. This varies in time but is bounded at most 11 for our demonstration. (d) The simulation setup.}
    \label{fig:excited_state_population}
\end{figure*}

\subsection{Gate complexity analysis}\label{sec:gatecomplexity-analysis}
As summarized in Table~\ref{tab:complexity},
the gate complexity of our algorithm (Theorem~\ref{thm: main2}) offers better scaling with respect to all parameters ($\varepsilon, K, m$) except for $\tau$,
compared with the efficient previous works, channel LCU~\cite{Cleve2016-yj} and Hamiltonian simulation-based algorithm~\cite{Ding2024-SDE}.
Consequently, there are parameter regions where our algorithm outperforms the prior works.
The key interest is identifying the parameter regions, as well as examining which approach is likely to be more effective in practical problem settings.
In this section, we numerically investigate these regimes.

For the numerical analysis, we introduce two common instances: (i) a dissipative one-dimensional transverse field Ising model (TFIM), and (ii) a two-dimensional
Fermi-Hubbard model (FHM) with the two-body loss.
TFIM is described by the following Hamiltonian and jump operators:
\begin{gather}
    H_{\mathrm{TFIM}}  = -J \sum_{i=1}^n Z_i Z_{i+1} - h \sum_{i=1}^{n} X_i,  \notag \\
    L_{\mathrm{TFIM}, k} = \sqrt{\gamma} \frac{X_k - iY_k}{2},
    \label{eq:tfim-definition}
\end{gather}
where $n$ is the number of sites, $J$ is the nearest-neighbor coupling, $h$ is a transverse field, $\gamma$ is a decay rate,
FHM is described by the following Hamiltonian and jump operators:
\begin{gather}
    H_{\mathrm{FHM}}  = -J \sum_{\Braket{i,j}, \alpha}\left( c_{i, \alpha}^\dagger c_{j, \alpha}  + c_{j, \alpha}^\dagger c_{i, \alpha} \right)
    + U \sum_{i, \alpha \ne \beta}  n_{i, \alpha} n_{i, \beta},
    \notag \\
    L_{\mathrm{FHM}, k} = \sqrt{\gamma} \sum_{\alpha \ne \beta} c_{k, \alpha }c_{k, \beta},
    \label{eq:fhm-definition}
\end{gather}
where $c_{i,\sigma}$ and $c_{i,\sigma}^\dagger$ are fermionic annihilation and creation operators at site $i$ on a square lattice with the length $\sqrt{n}$, and for the spin state $\sigma = 1, \cdots, s$.
$n_{i,\sigma} = c_{i,\sigma}^\dagger c_{i,\sigma}$ is a number operator.
$\sum_{\Braket{i,j}}$ denotes summation over the nearest-neighbor pairs $i,j$. Let $J$ and $\gamma$ be the hopping amplitude and the two-body loss rate, respectively.
The parameters are summarized in Table~\ref{tab:lindbladian-parameters}.

\begin{table}[htbp]
    \centering
    \begin{tabular}{ r| r |r}\hline
        Parameters  &  1D-TFIM &  2D($\sqrt{n} \times \sqrt{n}$)-FHM\\ \hline
        $m$ &  $2n$ & $12n$\\
        $M$ & $2$ & $4$\\
        $K$ & $n$ & $n$\\
        $\| \mathcal{L} \|_{\mathrm{pauli}}$ &  ~~~$2(|J| + |h| + \gamma )n $ &  ~~~$2(8|J| + |U| + \gamma)n$ \\
        \hline
    \end{tabular}
    \caption{Algorithm parameters for $n$-site TFIM and FHM. We assume $s=2$ in FHM. For the implementation of fermionic operators, the Jordan-Wigner transformation is employed.}
    \label{tab:lindbladian-parameters}
\end{table}

We investigate the gate and ancilla qubit requirements for simulating the two instances.
Figure~\ref{fig:gatecomplexity} illustrates the $\tau$- and $\varepsilon$-dependence of the requirements for the channel LCU~\cite{Cleve2016-yj}, HS-based algorithm~\cite{Ding2024-SDE}, and ours (Theorem~\ref{thm: main2}). 
Here, we calculate the number of costly logical gates (i.e., T gates) as the gate count; a detailed setup for the numerical results and a further analysis are provided in Appendix~\ref{apdx:numerical}.
From Fig.~\ref{fig:gatecomplexity} (a,b), we observe distinct scaling behavior of these methods with respect to $\varepsilon$; 
our method provides an exponential advantage for the accuracy over the first order HS-based algorithm.
Furthermore, Fig.~\ref{fig:gatecomplexity} (e,f) shows that
in the wide $\tau$-regime below the crossing points around $\tau \approx n \times 10^4= 10^6$, our method achieves the lower gate requirement than that in the channel LCU.
Also, we can find the great ancilla qubit reduction from Fig.~\ref{fig:gatecomplexity} (c,d); our algorithm requires only a small constant number of ancilla qubits.
Note that the two instances we examined restrict the ability to take advantage of randomization in our method, since all components of the Hamiltonian and the jump operators are equally weighted.
Thus, we expect further advantage in the gate counts for a more realistic setup such as molecules as observed in qDRIFT~\cite{Campbell2019-qDRIFT}.
Also, in the two instances, the number of jump operators $K$ scales linearly in the system size $n$.
In such cases, the polynomial dependence on $K$ of the previous works directly translates into a higher-order dependence on the system size $n$.
Hence, $K$-independent scaling of our algorithm has a substantial impact on the scalability in $n$.

The observed advantages of our method become important when implementing it on
early FTQC devices, which are still limited in both the total number of executable gates and the number of available logical qubits.
From Fig.~\ref{fig:gatecomplexity}, reaching the crossing point requires around $\approx 10^{15}$ gate operations.
For larger or more complex systems, which are often of greater practical interest, the number of required gate operations to reach the crossing points grows substantially; also see Appendix~\ref{apdx:numerical}.
In addition, implementing the channel LCU algorithm typically requires on the order of $10^2$ logical ancilla qubits in addition to the system register.
When FTQC devices are first introduced, they are expected to be small-scale in terms of both gate count and logical qubit capacity.
For such early FTQC devices, our algorithm offers a practical and viable option.

\begin{figure*}[htbp]
    \centering
    \begin{tabular}{ccc}
        \includegraphics[width=0.4\textwidth]{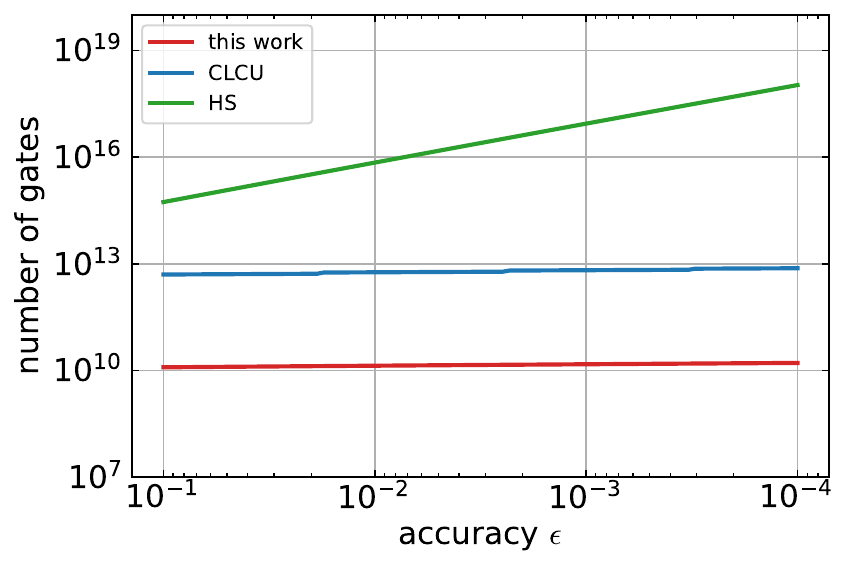} &~~~~~~& \includegraphics[width=0.4\textwidth]{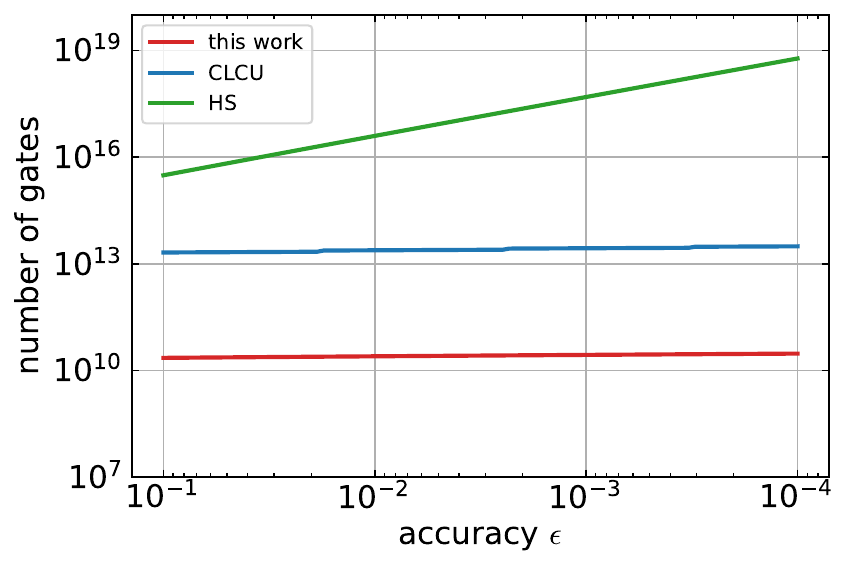} \\
             (a) Gate counts vs $\varepsilon$ of TFIM &~~~~~~& (b) Gate counts vs $\varepsilon$ of FHM \\ \\
          \includegraphics[width=0.4\textwidth]{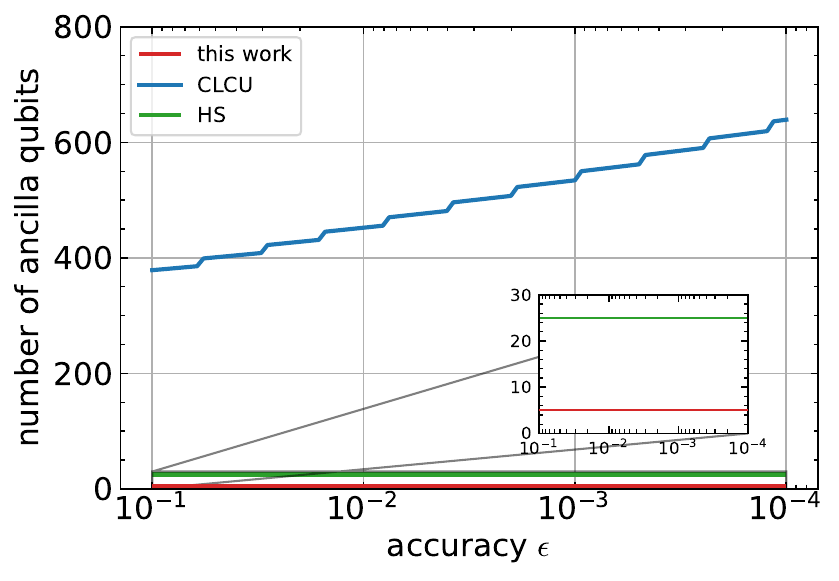}&~~~~~~&\includegraphics[width=0.4\textwidth]{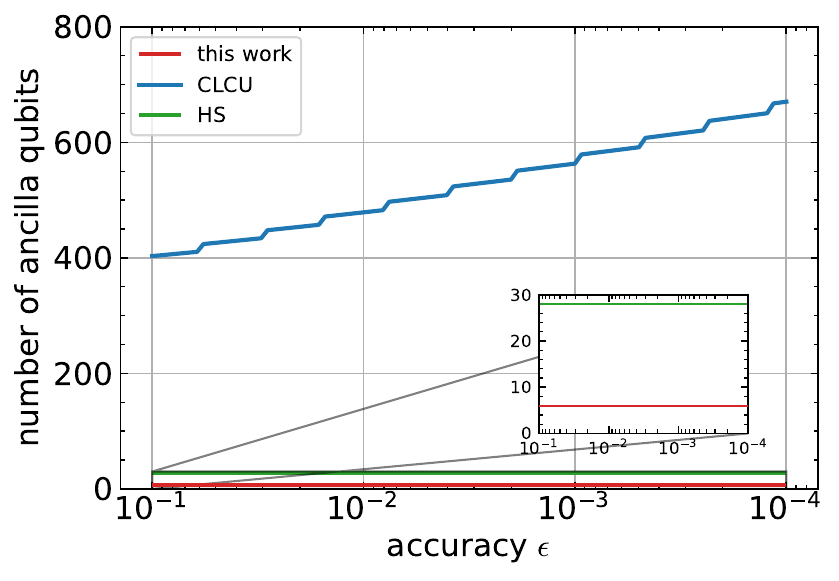}\\
         (c) Ancilla counts vs $\varepsilon$ of TFIM  &~~~~~~& (d) Ancilla counts vs $\varepsilon$ of FHM\\ \\
         \includegraphics[width=0.4\textwidth]{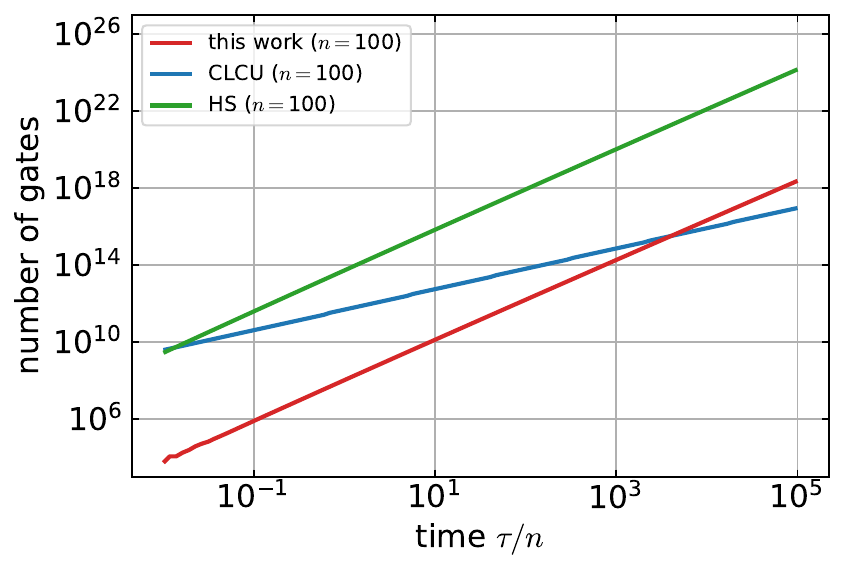} &~~~~~~& \includegraphics[width=0.4\textwidth]{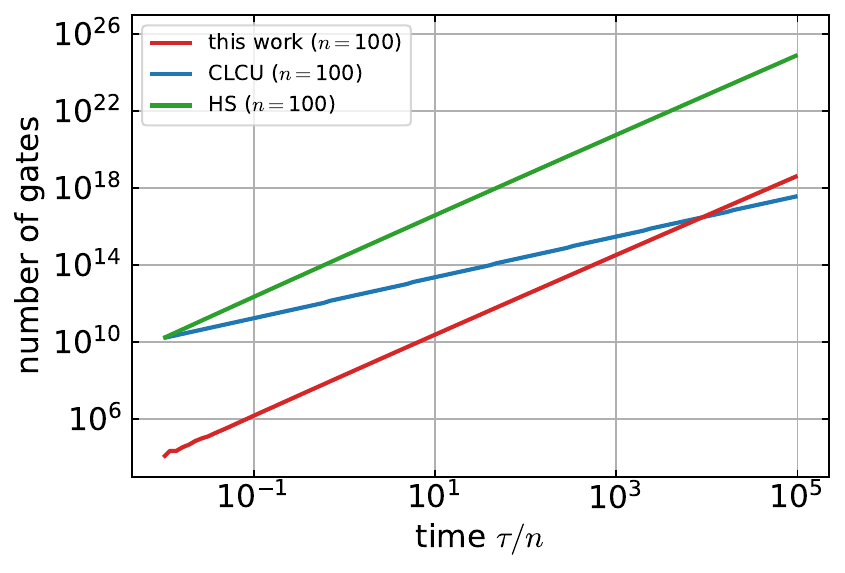} \\
         (e) Gate counts vs $\tau/n$ of TFIM &~~~~~~& (f) Gate counts  vs $\tau/n$ of FHM \\ \\
          \includegraphics[width=0.4\textwidth]{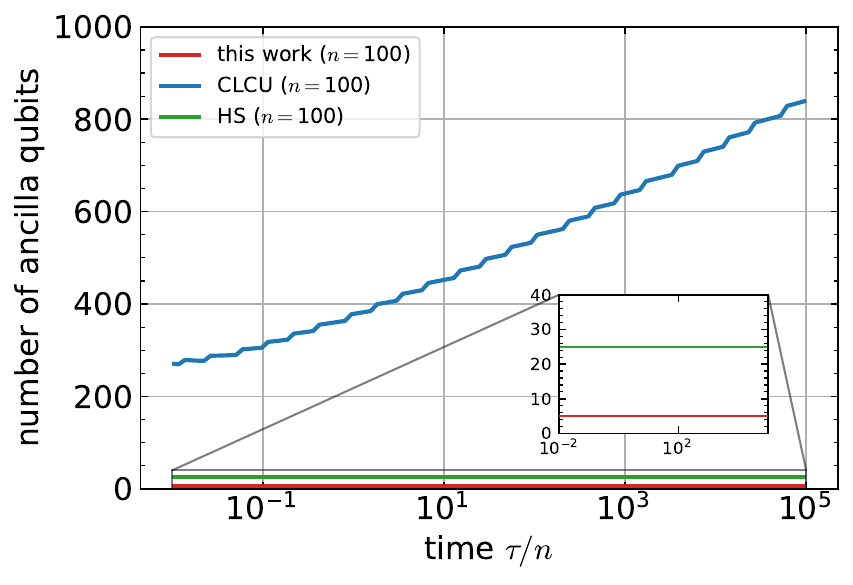}&~~~~~~&\includegraphics[width=0.4\textwidth]{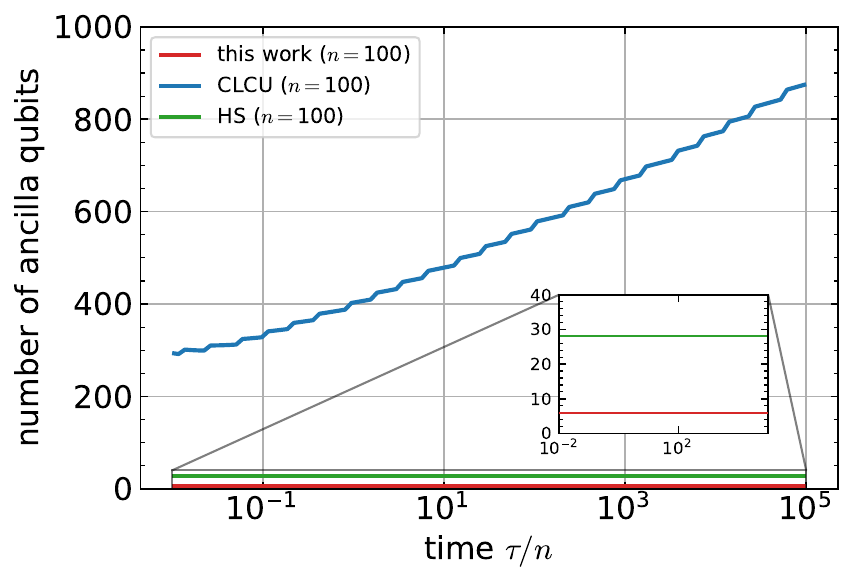}\\
         (g) Ancilla counts vs $\tau/n$ of TFIM  &~~~~~~& (h) Ancilla counts vs $\tau/n$ of FHM\\
    \end{tabular}
    \caption{T-gate counts and additional ancilla qubit requirements of our algorithm of Theorem~\ref{thm: main2} (\textit{this work} in red), channel LCU (\textit{CLCU} in blue)~\cite{Cleve2016-yj}, and first order HS-based (\textit{HS} in green)~\cite{Ding2024-SDE} for simulating $n$-site TFIM and FHM.
    Since $\|\mathcal{L}\|_{\mathrm{paili}} \propto n$, we employ $\tau/n$ as a horizontal axis instead of $\tau$. We set parameters $(n, \tau/n) = (10^2, 10^2)$ for (a-d), and $(n, \varepsilon) = (10^2, 10^{-4})$ for (e-h).}
    \label{fig:gatecomplexity}
\end{figure*}

\subsection{Benchmarking}\label{sec:benchmark}
    We now connect our circuit complexities to two concrete and physically relevant setups.
    The first setup is a dissipative system with two-body loss, which is specified by the same Lindbladian Eq.~\eqref{eq:fhm-definition}.
    This is inspired by a cold-atom experiment on short-time transient dynamics~\cite{sponselee2019ultracold}.
    The second setup is a boundary-driven dissipative system, recently used as a benchmark for searching metal-insulator transition material~\cite{mozgunov2024benchmark}.
    Below we show that, within a realistic execution time, our algorithm can access physically meaningful quantities in dissipative many-body systems that are expected to be classically inaccessible, while offering clear resource advantages over previous quantum algorithms.
\begin{table}[h]
    \centering
    \begin{tabular}{c c c}
        \begin{tabular}{r | r} \hline
            Parameters & Values \\ \hline
            $U/J$     & $3.8$    \\
            $\gamma/J$ & $1.1$      \\
            $J t$      & $7.6$  \\ 
            $\varepsilon$ & $10^{-1}$ \\ \hline
        \end{tabular}
        & ~ &
        \begin{tabular}{r | r} \hline
            Parameters & Values \\ \hline
            $U/t_{\mathrm{hop}}$     & $6.7$    \\
            $\gamma/t_{\mathrm{hop}}$ & $1$      \\
            $t_{\mathrm{hop}}'/t_{\mathrm{hop}}$ & $-0.15$      \\
            $t_{\mathrm{hop}} t$      & $36$  \\
            $\varepsilon$ & $10^{-1}$ \\ \hline
        \end{tabular}
        \vspace{0.3em}
        \\
        (a) Parameters for TBL.
        &~&(b) Parameters for BD.\\\\
    \end{tabular}
    \begin{tabular}{c}
        \begin{tabular}{ l | rr }\hline
            Algorithm &  T-count & ~~Additional qubits \\ \hline
            Our work (Theorem~\ref{thm: main2}) &  $1.1\times 10^{13} $  & $6$ \\
            Channel LCU~\cite{Cleve2016-yj} &  $9.7\times 10^{14}$ & $518$ \\ \hline
        \end{tabular}
        \vspace{0.3em}
        \\
        (c) T-count and ancilla qubit requirements for TBL.
        \\\\
        \begin{tabular}{ l | rr }\hline
            Algorithm & T-count & ~~Additional qubits \\ \hline
            Our work (Theorem~\ref{thm: main2}) &  $2.2\times 10^{14}$  & $5$\\
            Channel LCU~\cite{Cleve2016-yj} &  $2.4\times 10^{15}$ & $580$\\ \hline
        \end{tabular}
        \vspace{0.3em}
        \\
        (d) T-count and ancilla qubit requirements for BD.
        \\\\
        \begin{tabular}{ l | r r}\hline
           \multirow{2}{*}{Assumption}  &  \multicolumn{2}{c}{T-count (T-depth)} \\
                 &  TBL  & BD\\ \hline
                Baseline (Theorem~\ref{thm: main2}) &  $1.1\times 10^{13} $ & $2.2\times 10^{14} $  \\
                + (1) Optimized single-jump circuit & $1.3\times 10^{12}$  & $4.4\times 10^{13}$ \\
                + (2) Average-case assumption & $1.4\times 10^{11}$   & $3.6\times 10^{12}$  \\ \hline
                + (3) Efficient rotation synthesis  &  $3.8\times 10^{9}$~  &  $2.2\times 10^{10}$ \\
                \hline
        \end{tabular}
        \vspace{0.3em}
        \\
        (e) T-count or T-depth after successive optimizations.
    \end{tabular}
    \caption{Benchmark results for 
    spin-$1/2$ fermionic dissipative systems on a 100-site lattice.
    (a)(b)
    The physical parameters for simulating the dissipative systems with two-body loss (TBL) and the boundary-driven dissipative system (BD), defined by Eqs.~\eqref{eq:fhm-definition} and \eqref{eq:bd-definition}.
    The unit $J=1$ and $t_{\mathrm{hop}}=1$ is applied.
    (c)(d) The resource requirements for simulating TBL and BD.
    (e) T-count (T-depth) requirements.
    The baseline corresponds to the results in (c) and (d). Each row, except for the baseline, includes all optimizations listed above. 
    The T-depth is reported only in the final line; this may be further improved, for instance, by exploiting locality.
    Assuming a T-gate rate of 1~MHz, each shot takes approximately 1.1 hours for TBL and 6.1 hours for BD, respectively.
    }
    \label{tab:benchmark}
\end{table}

\subsubsection{Dissipative systems with two-body loss}
    For the benchmarking, we revisit the Fermi-Hubbard model with two-body loss analyzed in Section~\ref{sec:gatecomplexity-analysis}. 
    This numerical analysis is inspired by a cold-atom experiment~\cite{sponselee2019ultracold}.
    Ref.~\cite{sponselee2019ultracold} investigates the one-dimensional (1D) Fermi-Hubbard model with two-body loss, and observes the dynamics of the number of particles.
    The corresponding Lindbladian is the 1D version of our benchmarking model in Section~\ref{sec:gatecomplexity-analysis}, given by Eq.~\eqref{eq:fhm-definition}.
    The target observable is the total number operator $N = \sum_{i,\alpha} n_{i,\alpha} =\sum_{i,\alpha} c^\dagger_{i, \alpha}c_{i, \alpha}$.
    The experimental result in Ref.~\cite{sponselee2019ultracold} demonstrates that after an initial fast transient dynamics, $\langle N\rangle (t)$ becomes a constant value,
    showing a discrepancy from a simple constant-rate model.
    Then, a numerical fitting for this experimental result is validated by numerical simulations of the 1D Lindblad dynamics.

    Our quantum algorithm can simulate the dynamics of $\langle N\rangle (t)$ in a similar setting.
    While Ref.~\cite{sponselee2019ultracold} considers the 1D system with $\sim 10^4$ sites, we here analyze the 2D system with $10^2$ sites as one of the natural extensions of the experimental setup, in mind that simulating 2D Lindbladian can be a hard task for classical computation as we will discuss in Section~\ref{sec:conclusion}.
    Although the experiment itself is performed in 1D, the physical mechanism responsible for the rapid suppression of the time variation of $\langle N\rangle (t)$ after the initial transient dynamics is essentially local, and therefore a similar behavior is naturally expected also in 2D.
    Therefore, we consider the same simulation timescale $(Jt=7.6\approx 2U/J)$ as that studied in the 1D experiment, while noting that the precise relaxation behavior may differ in higher dimensions.
    The remaining physical parameters for simulation are summarized in Table~\ref{tab:benchmark}~(a), which are equivalent to those of Ref.~\cite{sponselee2019ultracold}.

    For this task, we evaluate the gate and ancilla requirements in Table~\ref{tab:benchmark}~(c).
    Our algorithm requires $1.1\times 10^{13}$ T gates and $6$ additional ancilla qubits, 
    compared to $9.7 \times 10^{14}$ T gates and $518$ additional ancilla qubits for the previous channel LCU approach with almost linear $\tau$-scaling~\cite{Cleve2016-yj}.
    Therefore, this example demonstrates that our algorithm enables the simulation of transient dynamics in experimentally relevant regimes that are expected to be challenging for classical simulation (see Section~\ref{sec:conclusion}), while achieving a clear computational gain over the previous approach~\cite{Cleve2016-yj}.

\subsubsection{Boundary-driven dissipative systems}
Ref.~\cite{mozgunov2024benchmark} proposes a benchmark task of simulating the 2D Fermi-Hubbard model with the boundary dissipation to study the metal-insulator transition~\cite{imada1998metal-insulator}.
Notably, the boundary-driven dissipative system, which is coupled to different baths at the edges, itself contains a rich set of open problems for open quantum systems~\cite{landi2022boundary}, including the transport properties of steady states.
Hence, we here focus on this benchmark task as another physically motivated task in open quantum systems.

We evaluate the explicit gate and additional qubit complexity of this benchmark task proposed by Ref.~\cite{mozgunov2024benchmark}.
Specifically, the benchmark setup utilizes the following $10\times 10$ square lattice Fermi-Hubbard model with annihilation and creation operators on opposite edges of the lattice:
\begin{gather}
    H_{\mathrm{BD}}  = - \sum_{i,j, \alpha} t_{i,j}\left( c_{i, \alpha}^\dagger c_{j, \alpha}  + c_{j, \alpha}^\dagger c_{i, \alpha} \right)
    + U \sum_{i, \alpha \ne \beta}  n_{i, \alpha} n_{i, \beta},
    \notag \\
    L_{\mathrm{BD}, k, \alpha} =
    \begin{cases}
        \sqrt{\gamma} c_{k,\alpha} & k \text{ on left edge} \\
        \sqrt{\gamma} c^\dagger_{k,\alpha} & k \text{ on right edge}
     \end{cases}
    \label{eq:bd-definition}
\end{gather}
where the hopping $t_{i,j}$ includes the nearest-neighbor terms ($t_{i,j} = t_{\mathrm{hop}}$) and next-nearest-neighbor terms ($t_{i,j} = t_{\mathrm{hop}}'$),
and all other terms are zero. $\alpha=1,2$ denotes the spin state.
By simulating this Lindbladian until the system reaches its steady state, one measures the local part of the Hamiltonian, the particle density, and the current operator at some region of the system.
These measured quantities can be used for estimating the physically relevant properties such as the electric field and the conductivity~\cite{mozgunov2024benchmark}.
Here, it may not be necessary to simulate the dynamics until the steady state, when the expectation value converges in an earlier time.

The resulting requirements and the corresponding parameters are summarized in Table~\ref{tab:benchmark}~(d) and (b).
In particular, we assume a target accuracy $\varepsilon = 10^{-1}$, and simulation time $t_{\mathrm{hop}} t=36$, which are given by Ref.~\cite{mozgunov2024benchmark}.
Our algorithm requires only $2.2\times 10^{14}$ gates and a mere 5 additional ancilla qubits. In contrast, the channel LCU approach~\cite{Cleve2016-yj} requires $2.4\times 10^{15}$ gates and $580$ ancilla qubits.
In this setup, we demonstrate the advantage in both gate and ancilla requirements over the previous work~\cite{Cleve2016-yj}.

Finally, we mention the classical hardness of this benchmark task. As we will discuss in Section~\ref{sec:conclusion}, simulating 2D dissipative systems is generally considered unattainable with classical simulation techniques~\cite{kashyap2025oqsgurantee}. 
Also, Ref.~\cite{mozgunov2024benchmark} notes that, while $4 \times 4$ grid simulation has been done classically via exact diagonalization, the $10 \times 10$ grid is large enough to avoid inappropriate finite size effects but is too large for classical approaches.

\subsubsection{Further gate reduction}
Table~\ref{tab:benchmark}~(c,~d) present the baseline complexities evaluated by our cost model described in Appendix~\ref{apdx:gata-complexity-procedure}.
This baseline comparison is performed under a common and uniform resource estimation framework, so that the advantage reported in Table~\ref{tab:benchmark}~(c,~d) reflects a fair comparison between algorithms.
Beyond this, however, our algorithm has an additional practical benefit: our circuit structure Fig.~\ref{fig:main_circ} is much simpler than those of previous approaches with sophisticated multiple ancilla uses (e.g., see Fig.~\ref{fig:cleve-circuit} in Appendix~\ref{apdx:gata-complexity-procedure}).
Because of this simplicity, standard circuit-level optimization techniques can be incorporated much more directly. 
This circuit-level optimization cannot be applied to all methods directly, as such simplifications become substantially harder or non-trivial when the underlying block encodings and multi-qubit controlled operations are complicated.

Specifically, by exploiting this simplicity, we can directly use the techniques of
(1) optimized implementation of a single-jump process, (2) assumptions based on average-case estimation, and (3) improved rotation synthesis.
Here, we recall that our circuit is a gate sequence in which each step is randomly chosen to be either a single-jump process or a single Pauli rotation. Thus, the first two techniques are naturally incorporated. Also, the minimal ancilla usage of our algorithm justifies the third technique requiring a small additional ancilla system.
As a result, these techniques provide a further reduction of our total T-gate count (depth) to
$3.8\times 10^{9}$ for the dissipative system with two-body loss, and $2.2\times 10^{10}$ for the boundary-driven dissipative system; see Table~\ref{tab:benchmark}~(e).
Assuming a T-gate rate of 1~MHz, the T-depth directly translates to execution times of approximately 1.1 and 6.1 hours, respectively.
The detailed description is provided in Appendix~\ref{apdx:t-gate-reduction}.

\section{Conclusion and discussion}\label{sec:conclusion}
We presented a randomized quantum algorithm for estimating the physical properties of the general Lindblad dynamics,
which has (1) an exponentially short circuit depth $\mathcal{O}(\log (1/\varepsilon))$ to achieve an accuracy $\varepsilon$, (2) a nearly constant ancilla qubits $4 + \lceil \log_2 (M) \rceil$, and (3) the independence from $m$ and $K$ in the circuits.
Further concatenating randomization, we can remove the remaining $M$-dependence in both gate and ancilla counts, resulting in a Lindblad simulator having the independence from all $(m,M,K)$ in gate complexity.
As demonstrated in the numerical analysis, this algorithm has a potential to reveal undiscovered physical phenomena even with early-FTQC devices.
We also would like to emphasize the importance of newly developed techniques for superoperator simulation, such as the translation method of transfer matrices and the exact CP map simulation via OAA together with the recovery operation. 
Actually, they are widely applicable techniques for many problems beyond open system simulations.
\\

Our algorithms improve the scaling with respect to the Lindbladian parameters and the target accuracy, while their scaling in the simulation time $\tau$ is not linear. Nevertheless, the independence from the Lindbladian parameters can still make the overall cost smaller in physically meaningful regimes, especially when the simulation time is not too large relative to the parameter-dependent overheads. 
Such regimes include, for example, the computation of (i) transient observables and (ii) steady-state properties in rapid-mixing systems. As we discuss below, these tasks are not only physically relevant, but can also be classically challenging even at comparatively short simulation times.

In higher-dimensional open systems, 
the classical hardness of computing physically relevant quantities is supported both theoretically and practically.
On the theoretical side, Ref.~\cite{kashyap2025oqsgurantee} shows that, estimating the expectation values of time-evolved local observables in geometrically local 2D open systems can be classically hard.
On the practical side, tensor network methods, which are one of the most powerful classical tools for many-body simulation,
face significant challenges in general simulation tasks, as the spatial dimension increases~\cite{cirac2021review, weimer2021review}.
Consequently, the classical methods will exhibit a large gap in computational hardness between 1D systems and higher-dimensional systems. This is in contrast to quantum algorithms, including ours, whose qubit and gate complexities depend only weakly on the spatial dimension of the underlying lattice.

Relatively short-time quantum simulation is relevant not only for transient dynamics, but also for extracting steady-state properties in rapid-mixing systems.
In general, extracting steady-state properties may require a long simulation time.
However, this is not the case for important systems exhibiting fast-relaxation or rapid mixing, for which the simulation time required to access steady-state properties can remain independent of or scales only logarithmically with system size $n$.
Moreover, Ref.~\cite{kashyap2025oqsgurantee} explicitly includes the estimation of fixed-point expectation values of rapidly mixing local observables in 2D systems as a classically hard target.
The rapid mixing is also important from the viewpoint of many-body physics.
For local observables evolved under spatially local Lindbladians, rapid mixing implies a well-defined thermodynamic limit in the steady state; this makes the estimation of such thermodynamic-limit local observables a natural and physically meaningful computational task~\cite{kashyap2025oqsgurantee}.
At the same time, whether rapid mixing occurs is itself a nontrivial question.
In this sense, the short-time quantum simulation may be used as a probe of rapid mixing.
    
In rapid-mixing systems, Theorem~\ref{thm:mMKindep_simulator}, stating that the gate complexity is independent of Lindbladian parameters ($m, M, K$), becomes particularly attractive. 
This is because the practical relevance of Theorem~\ref{thm:mMKindep_simulator} would emerge in tasks with complicated Lindbladians, such as dissipative ground state~\cite{zhan2026rapidGS} and Gibbs state preparation~\cite{rouze2026gibss, chen2025efficient} for rapid-mixing systems including higher-dimensional physical models~\cite{vsmid2025rapid}.
The construction makes the resulting Lindbladian complicated, and its fully coherent implementation with nearly linear time scaling is not necessarily the best option in all parameter regimes, especially when the required simulation time is short or intermediate due to rapid mixing.
Although the engineered jump operators in those works have a slightly different form from our access model, our algorithm can likely be adapted to these cases; a detailed cost analysis remains an important future direction.

Another important direction is the further improvement of gate complexity shown in Theorem~\ref{thm: main2}.
A natural next step would be to achieve the higher-order scaling $\mathcal{O}(t^{1+1/p} \log (1/\epsilon))$ while keeping the small ancilla requirements independent of $K$ and $m$.
However, constructing efficient higher-order schemes for Lindblad dynamics is nontrivial, unlike the product formula in unitary dynamics, and likely requires new conceptual advances.
Furthermore, it is also highly nontrivial to construct an algorithm that simultaneously achieves scaling $\mathcal{O}(t^{1+1/p} \log (1/\epsilon))$ and independence of the system parameters such as $(m,M,K)$.
Even for the case of Hamiltonian simulation, no algorithm is currently known to achieve these properties simultaneously.

An interesting question is on the relationship between our algorithm and the formulation via stochastic Schr\"{o}dinger equation (SSE)~\cite{gardiner2004quantum}, a widely applicable stochastic simulation method for open quantum systems. 
In fact, our algorithm also uses stochastic applications of unitary or jump operations to simulate the Lindblad master equation. 
However, our dissipative process is not physical unlike the case of SSE, because it is described as the composition of the minimal dissipation $\mathcal{B}_{kl}$ and (possibly non-physical) superoperators $U\bullet V^\dagger$. 
Exploring the connection to SSE could offer some physical interpretation of our algorithm and thereby enable improvement of it in both implementability and computational complexity perspectives. 

\textit{Note added.}---
Shortly after we submitted our first preprint paper on arXiv, Yu et al.~\cite{Yu2024-exponentially} proposed a Lindblad simulation algorithm for observable estimation, with a logarithmically short depth and constant number of ancilla qubits.
The fundamental difference to our method is that they rely on Trotter decomposition and subsequent error compensation.
While Trotter decomposition introduces the dependence on the number of terms including the number of jump terms $K$ and the Pauli counts in Hamiltonian $m$ in gate complexity, our work eliminates the dependence on those two important factors 
by entirely random compilation derived from Taylor expansion of $e^{t\mathcal{L}}$; this desirable feature is of particularly advantageous for large systems.

\begin{acknowledgments}
K.W. was supported by JSPS KAKENHI Grant Number JP24KJ1963.
J.K. acknowledges support by SIP Grant Number JPJ012367.
This work was supported by MEXT Quantum Leap Flagship Program Grants No. JPMXS0118067285 and No. JPMXS0120319794. 
\end{acknowledgments}

\bibliography{reference}

\clearpage
\onecolumngrid
\setcounter{equation}{0}
\setcounter{section}{0}

% --------------------  Supplementary Materials  --------------------

\appendix
\renewcommand{\thesection}{\Alph{section}}
\renewcommand{\theequation}{\thesection.\arabic{equation}}

\section{Preliminary}
\subsection{Notation, norm, and vectorization}

Let us consider a finite-dimensional Hilbert space $\mathrm{H}$. 
We write $\mathrm{L}(\mathrm{H})$ as the set of all linear operators on $\mathrm{H}$, which is also a finite-dimensional Hilbert space with the Hilbert-Schmidt inner product, i.e., $\langle A, B \rangle:={\rm Tr}[A^\dagger B]$ for 
$A, B\in \mathrm{L}(\mathrm{H})$.
The Schatten $p$-norm of an operator $A\in {\rm L(H)}$ is defined as
\begin{equation}
    \left\|A\right\|_{p}:=\left({\rm Tr}\left[|A|^p\right]\right)^{1/p},~~~|A|:=\sqrt{A^\dagger A},~~~p\in [1,\infty].
\end{equation}
When $p=\infty$, this norm corresponds to the operator norm of $A$, denoted by $\|A\|$.
In particular, for two operators $A,B\in \mathrm{L}(\mathrm{H})$ and parameters $p,q\in [1,\infty]$ such that $1/p+1/q=1$, 
the H\"{o}lder's inequality holds:
\begin{equation}
    \|AB\|_1\leq \|A\|_p\|B\|_q.
\end{equation}
Under the same assumption, the following matrix H\"{o}lder's inequality also holds~\cite{baumgartner2011inequality}:
\begin{equation}
    \left|{\rm Tr}\left[A^\dagger B\right]\right|\leq \|A\|_p\|B\|_q.
\end{equation}

For a given superoperator $\mathcal{L}$, which is a linear map acting on the space of operators ${\rm L(H)}$, we define the induced $1\to 1$ norm~\cite{10.5555/2011608.2011614} as
\begin{equation}
    \|\mathcal{L}\|_{1\to1}:=\sup_{A\in {\rm L(H)},~A\neq 0} \frac{\|\mathcal{L}(A)\|_1}{\|A\|_1}.
\end{equation}
It is clear from the definition that $\|\mathcal{L}'\circ\mathcal{\mathcal{L}}\|_{1\to 1}\leq \|\mathcal{L}'\|_{1\to 1}\|\mathcal{\mathcal{L}}\|_{1\to 1}$ holds for any superoperators $\mathcal{L}$ and $\mathcal{L}'$.
Especially, the $1\to 1$ norm of completely positive and trace-preserving (CPTP) maps satisfies the following property.
\begin{lem}\label{lem_CPTP_norm}
    For any CPTP map $\Phi$, $\|\Phi\|_{1\to 1}=1$ holds.
\end{lem}
\begin{proof}
As shown in~\cite{10.5555/2011608.2011614}, for any CPTP map $\Phi$, the supremum in $\|\Phi\|_{1\to 1}$ can be restricted to self-adjoint operators as
\begin{equation}
     \|\Phi\|_{1\to 1} = \sup_{A\in \mathrm{L(H)}, A=A^{\dagger},\|A\|_1=1} \|\Phi(A)\|_1.
\end{equation}
Observing that a self-adjoint $A$ with $\|A\|_1=1$ has the spectral decomposition $A = \sum_i \lambda_i \ketbra{v_i}$ with $\sum_i |\lambda_i|=1$, we obtain
\begin{equation}
    \|\Phi(A)\|_1
     \le \sum_i |\lambda_i| \|\Phi(\ketbra{v_i})\|_1
     \le \max_i \|\Phi(\ketbra{v_i})\|_1
     = \max_i \Tr [\Phi(\ketbra{v_i})]
     = \max_i \Tr \ketbra{v_i}
     = 1,
\end{equation}
where the first equality follows from the fact that $\Phi(\ketbra{v_i})\ge 0$ due to the positivity of $\Phi$, and the second equality follows from the trace-preserving property of $\Phi$.
Especially, $A = \ketbra{v}$ with a unit vector $\ket{v}$ attains $\|\Phi(\ketbra{v})\|_1 = \Tr\ketbra{v} = 1$ in the same way. 
Therefore, the proof is completed.
\end{proof}

\noindent
To distinguish operators and superoperators, we use the italic font and a calligraphic font to denote them, respectively, e.g., $A$ is an operator and $\mathcal{L}$ is a superoperator. 
When it is clear from the context, we omit the subscript $1\to 1$ of the norm $\|\bullet\|_{1\to 1}$ for superoperators. 

Throughout this work, we use \textit{vectorization} to represent operators in ${\rm L(H)}$ as vectors~\cite{wood2011tensor}.
Specifically, taking an orthonormal basis $\{\sigma_{\alpha}\}_{\alpha}$ in the Hilbert space ${\rm L(H)}$, the vectorization (with respect to the base $\{\sigma_{\alpha}\}$) maps any operator $A\in {\rm L(H)}$ to a vector $\kett{A}_{\sigma}$ as
\begin{equation}
    A\mapsto\kett{A}_{\sigma}:=\sum_{\alpha} {\rm Tr}[\sigma_{\alpha}^\dagger A]\ket{\alpha},
\end{equation}
where $\ket{\alpha}$ is an orthonormal basis in a $(\dim \mathrm{H})^2$ dimensional Hilbert space.
The inner product of vectorized operators $\kett{A}_{\sigma}$ and $\kett{B}_{\sigma}$ matches the Hilbert-Schmidt inner product as
\begin{equation}
    \langle\!\langle B\kett{A}_{\sigma}={\rm Tr}[B^\dagger A]=\langle B,A\rangle.
\end{equation}
In this formalism, a superoperator $\mathcal{E}$ can be written as a matrix $S(\mathcal{E})$ defined as
\begin{equation}
    S(\mathcal{E}):=\sum_{\alpha,\beta} {\rm Tr}[\sigma_{\alpha}^\dagger\mathcal{E}(\sigma_{\beta})]\ket{{\alpha}}\bra{{\beta}}.
\end{equation}
This matrix representation is often called the \textit{transfer matrix} of $\mathcal{E}$~\cite{Nielsen2021gatesettomography}.
For instance, the expectation value of an observable $O$ with respect to a quantum state $\rho$ evolved by a superoperator $\mathcal{E}$ can be written as
\begin{equation}
    {\rm Tr}\left[O\mathcal{E}(\rho)\right]= \bbra{O}S(\mathcal{E})\kett{\rho}=\bbra{O}{\mathcal{E}(\rho)}\rangle\!\rangle.
\end{equation}
We often use the fact that for any superoperators $\Phi$ and $\Phi'$ acting on ${\rm L(H)}$, $S(\Phi\circ\Phi')=S(\Phi)S(\Phi')$ holds.
In particular, we use the map from the orthonormal basis $\ket{i}\bra{j}$ in ${\rm L(H)}$ to $\ket{j}\ket{i}\equiv \ket{j}\otimes \ket{i}$ via vectorization, where $\ket{i}$ is the computational basis of $\mathrm{H}$. That is, we use the column-stacking convention for vectorization. Then, we can write (omitting the subscript of vectorization)
\begin{equation}\label{suppleeq:cvec}
    \kett{A}=\sum_{i,j}\langle i|A|j\rangle \ket{j}\ket{i}=\bm{1}\otimes A \kett{\bm{1}},
\end{equation}
where the vectorization $\kett{\bm{1}}$ of the identity operator $\bm{1}$ corresponds to the unnormalized maximally entangled state:
\begin{equation}
    \kett{\bm{1}}=\sum_{i}\ket{i}\ket{i}.
\end{equation}
For a superoperator $\mathcal{E}$, its transfer matrix $S(\mathcal{E})$ via the column-stacking vectorization is given by
\begin{equation}\label{suppleeq:cvec_supop}
    S(\mathcal{E}) = \sum_{ij,kl}\bra{k}\mathcal{E}(\ket{i}\bra{j})\ket{l} \, \ket{l}\bra{j}\otimes  \ket{k}\bra{i}.
\end{equation}

\subsection{Lindblad equation}
In this work, we aim to estimate the physical properties of (finite-dimensional) density matrices $\rho(t)$ whose dynamics is governed by the Lindblad equation or Gorini--Kossakowski--Sudarshan--Lindblad equation~\cite{lindblad1976generators,gorini1976completely}:
\begin{equation}\label{suppleeq:lindblad_eq_original}
    \frac{{\rm d}}{{\rm d}t}\rho=\mathcal{L}(\rho):=-i[H,\rho]+\sum_{k=1}^{K} \left(L_k\rho L_k^\dagger - \frac{1}{2}
    \left\{L_k^\dagger L_k,\rho\right\}\right),
\end{equation}
where $H$ denotes the Hamiltonian of a target system, and $\{L_k\}_{k=1}^K$ are linear operators on the target system. 
The operators $\{L_k\}$ are called jump operators.
The superoperator $\mathcal{L}$ acts on the space of linear operators.
In addition, the operation in the last term is an anti-commutator i.e., $\{A, B\}:=AB+BA$ for operators $A$ and $B$. 
The resulting superoperator $e^{t\mathcal{L}}$ generated by $\mathcal{L}$ is a quantum channel (i.e., CPTP map) for any $t\geq 0$, where
\begin{equation}
    e^{t\mathcal{L}}:=\sum_{q=0}^\infty \frac{t^q \mathcal{L}^q }{q!}.
\end{equation}
The superoperator $\mathcal{L}^q:=\mathcal{L}\circ \cdots \circ \mathcal{L}$ is defined as $q$ sequential applications of $\mathcal{L}$.
We note that $e^{t\mathcal{L}}$ satisfies the semigroup property: for any $t_1,t_2\geq 0$, 
\begin{equation}\label{suppleeq:semigroup}
    e^{t_2\mathcal{L}}\circ e^{t_1\mathcal{L}} = e^{(t_2+t_1)\mathcal{L}}.
\end{equation}

Using the vectorization Eqs.~\eqref{suppleeq:cvec} and~\eqref{suppleeq:cvec_supop}, the Lindblad equation can be written as
\begin{equation}\label{suppleeq:vectorized_Lindblad}
    \frac{{\rm d}}{{\rm d}t}\kett{\rho}=G\kett{\rho},
\end{equation}
where $G$ is the transfer matrix of the superoperator $\mathcal{L}$:
\begin{equation}
    G:=S(\mathcal{L})=-i\bm{1}\otimes H+iH^{\rm T}\otimes\bm{1}+\sum_{k=1}^K\left(\overline{L_k}\otimes L_k-\frac{1}{2}\bm{1}\otimes L_k^\dagger L_k-\frac{1}{2}L_k^{\rm T}\overline{L_k}\otimes \bm{1}\right).
\end{equation}
The matrix $\overline{L_k}$ denotes the entrywise complex conjugation of $L_k$.
Note that the solution of Eq.~\eqref{suppleeq:vectorized_Lindblad} for an initial state $\kett{\rho(0)}$ can be formally written as
\begin{equation}
    \kett{\rho(t)}=e^{tG}\kett{\rho(0)},~~\mbox{where}~~e^{tG}=\sum_{q=0}^{\infty}\frac{t^q G^q}{q!}=S(e^{t\mathcal{L}}).
\end{equation}

\subsection{Linear combination of unitaries (LCU) and oblivious amplitude amplification (OAA)}\label{suppli_subsec:lcu_oaa}
Here we introduce two quantum algorithms: linear combination of unitaries (LCU)~\cite{childs2012hamiltonian,berry2015hamiltonian} and oblivious amplitude amplification (OAA)~\cite{berry2014exponential}, which take crucial roles in our method.

Let $A:=\sum_{i=1}^m c_iU_i$ be a linear combination of unitary operators $\{U_i\}_{i=1}^m$ with complex coefficients $c_i\in \mathbb{C}$. 
Without loss of generality, we assume $c_i>0$ because the complex phase can be absorbed into $U_i$.
In order to implement $A$, we use the following two unitary operations.
The first one, called PREPARE, encodes the positive coefficients $\{c_i\}_{i=1}^{m}$ as
$$
{\rm PRE}:\ket{\bm{0}} \mapsto \sum_{i=1}^{m} \sqrt{\frac{c_i}{\|c\|_1}}|i\rangle,
$$ 
where $\|c\|_1$ denotes the $L^1$-norm of the vector $c=(c_1, \ldots, c_m)$, and $|\bm{0}\rangle$ and $|i\rangle$ denote an initial state and the computational basis in a $\lceil {\log_2 m\rceil}$-qubit ancilla system, respectively.
The second one, called SELECT, encodes the unitary operators $U_i$ conditioned by the $\lceil {\log_2 m\rceil}$-qubit ancilla system:
$$
{\rm SEL}=\sum_{i=1}^m|i\rangle\langle i|\otimes U_i.
$$
Using the two operations ${\rm PRE}$ and ${\rm SEL}$, 
we can show that the unitary operator 
\begin{equation}\label{suppleeq:LCUop}
    W_A:=({\rm PRE}^\dagger\otimes \bm{1})\cdot {\rm SEL}\cdot({\rm PRE}\otimes \bm{1})
\end{equation} 
implements the target operator $A=\sum_{i=1}^m c_i U_i$ as
\begin{equation}\label{suppleeq:lcu_action}
    W_{A}\ket{\bm{0}}\ket{\psi}=\frac{1}{\|c\|_1}\ket{\bm{0}}A\ket{\psi}+|{\tilde{\perp}_{\psi}}\rangle,
\end{equation}
where $\ket{\psi}$ is an arbitrary system state. 
$|{\tilde{\perp}_{\psi}}\rangle$ denotes an unnormalized state that depends on $\ket{\psi}$ and satisfies $\bra{\bm{0}}\otimes \bm{1}\cdot |{\tilde{\perp}_{\psi}}\rangle=0$.

According to Ref.~\cite{Cleve2016-yj}, the above LCU method can be further extended to the implementation of a quantum channel $\mathcal{A}(\bullet)=\sum_k A_k \bullet A_k^\dagger$, where 
each operator $A_k~(k=0,\ldots,K-1)$ is given by a linear combination of unitary operators as
\begin{equation}
    A_k:=\sum_{i=1}^m c_{ki}U_{ki},~~~c_{ki}>0.
\end{equation}
Using the LCU method, we can construct a unitary operator $W_k:=W_{A_k}$ satisfying Eq.~\eqref{suppleeq:lcu_action} for the operator $A_k$, where the normalization factor $\|c\|_1$ is replaced by $\sum_i c_{ki}$.
Now, we introduce an additional ancilla system P, and define the following unitary operator 
\begin{equation}\label{suppleeq:ext_sel}
    \sum_{k} \ket{k}\bra{k}_{\rm P}\otimes W_k,
\end{equation}
where $\{\ket{k}_{\rm P}\}$ denotes the computational basis set on P.
Also, we define the unitary operator $U_{\rm R,P}$ preparing a quantum state $\ket{\rm R}$ on the ancilla system P as
\begin{equation}
    U_{\rm R,P}\ket{{0}}_{\rm P}=\ket{\rm R}_{\rm P}:=\frac{1}{\sqrt{\sum_k (\sum_i c_{ki})^2}}\sum_{k} (\sum_i c_{ki})\ket{k}_{\rm P}.
\end{equation}
Using the two operations Eq.~\eqref{suppleeq:ext_sel} and $U_{\rm R,P}$, we can directly confirm that for any input state $\ket{\psi}$,
\begin{align}\label{suppleeq:channellcuisometry}
    \bm{1}_{\rm P}\otimes \bra{\bm{0}}\otimes \bm{1}\cdot \left(\sum_{k} \ket{k}\bra{k}_{\rm P}\otimes W_k\right)U_{\rm R,P}\cdot \ket{{0}}_{\rm P}\ket{\bm 0}\ket{\psi} 
    = \frac{1}{\sqrt{\sum_k (\sum_i c_{ki})^2}} \sum_k \ket{k}_{\rm P}\otimes 
     A_k\ket{\psi}.
\end{align}
In particular, when $\{A_k\}$ are Kraus operators (i.e., $\sum_k A_k^\dagger A_k=\bm{1}$), 
we can implement the quantum channel  $\mathcal{A}(\ket{\psi}\bra{\psi})=\sum_k A_k \ket{\psi}\bra{\psi} A_k^\dagger$ by measuring $\ket{\bm{0}}$ on the quantum state $\ket{{0}}_{\rm P}\ket{\bm 0}\ket{\psi}$ evolved by $U_{\rm R,P}$ and Eq.~\eqref{suppleeq:ext_sel}.
Note that we can also take the input state as an arbitrary density matrix.
The success probability for this implementation of a quantum channel is given by $1/{\sum_k(\sum_i c_{ki})^2}$.

To amplify the above success probability of the LCU method for a quantum channel, we can use the quantum algorithm called the oblivious amplitude amplification (OAA) for isometries~\cite{Cleve2016-yj,gilyen2019quantum}, as follows.
Let $U$ be a unitary operator, and let $\tilde{\Pi}$ and $\Pi$ be orthogonal projectors, i.e. $\Pi^2 = \Pi=\Pi^\dagger$ and $\tilde{\Pi}^2 = \tilde{\Pi}=\tilde{\Pi}^\dagger$.
In addition, we assume that $\tilde{\Pi}U\Pi$ takes the following form:
\begin{equation}\label{suppleeq:OAA4I_input}
    \tilde{\Pi}U\Pi = \sqrt{p} W
\end{equation}
with some $p\in (0,1)$ and a partial isometry $W$ (meaning that $W^\dagger W$ is a projection). 
Note that $W W^{\dagger} W = W$ holds in this case.
We focus on the case of $p={1}/{4}$, which is required in our method as explained later.
In this special case, the procedure of OAA is much simpler than the general case as follows.
Using two unitary operations $2\Pi-\bm{1}$ and $2\tilde{\Pi}-\bm{1}$, we have
\begin{equation}
    \tilde{\Pi}U(2\Pi-\bm{1})U^\dagger (2\tilde{\Pi}-\bm{1})U\Pi= 4\sqrt{p}^3 WW^\dagger W-3\sqrt{p}W=(4\sqrt{p}^3-3\sqrt{p}) W = -W,
\end{equation}
where we used the property of partial isometry $W W^\dagger W= W$ in the second equality.

To clarify the effect of this amplification, we here give an example.
Assuming 
\begin{equation}\label{suppleeq:oaa_p_condition}
    \frac{1}{\sqrt{\sum_{k}(\sum_i c_{ki})^2}} \geq \frac{1}{2},
\end{equation}
we add a single-qubit register to the state on the left side of Eq.~\eqref{suppleeq:channellcuisometry} as
\begin{align}
    R_y\ket{0}\otimes\left(\sum_{k} \ket{k}\bra{k}_{\rm P}\otimes W_k\right)U_{\rm R,P}\cdot \ket{{0}}_{\rm P}\ket{\bm 0}\ket{\psi}=\frac{1}{2}\ket{0}\otimes \left(\sum_k \ket{k}_{\rm P}\otimes \ket{\bm{0}}\otimes A_k\ket{\psi}\right)+|\tilde{\perp}'\rangle
\end{align}
where $R_y$ is a single-qubit rotation such that $R_y:\ket{0}\mapsto (\sqrt{\sum_{k}(\sum_i c_{ki})^2}/2)\ket{0}+\sqrt{1-{\sum_{k}(\sum_i c_{ki})^2}/4}\ket{1}$, and $|\tilde{\perp}'\rangle$ is an unnormalized quantum state satisfying $\bra{0}\otimes \bra{\bm{0}}\cdot|\tilde{\perp}'\rangle=0$.
If we take the following $U$ as the unitary operator in Eq.~\eqref{suppleeq:OAA4I_input}:
\begin{equation}\label{suppleeq:oaa_u_example}
    U=R_y\otimes \left(\sum_{k=0}^{K-1} \ket{k}\bra{k}_{\rm P}\otimes W_k\right)U_{\rm R,P}
\end{equation}
and set (assuming the Kraus operators $\{A_k\}_{k=0}^{K-1}$ in $\{W_k\}$ act on the $n$-qubit system)
\begin{align*}
    &\tilde{\Pi}=\ket{0}\bra{0}\otimes \bm{1}_{\rm P}\otimes\ket{\bm{0}}\bra{\bm{0}}\otimes  \bm{1}=\sum_{k=0}^{K-1}\sum_{\mu=1}^{2^n}  \ket{0}\bra{0}\otimes \ket{k}\bra{k}\otimes \ket{\bm{0}}\bra{\bm{0}}\otimes\ket{\mu}\bra{\mu},\\
    &{\Pi}=\ket{0}\bra{0}\otimes \ket{0}\bra{0}_{\rm P}\otimes \ket{\bm{0}}\bra{\bm{0}}\otimes  \bm{1}=\sum_{\mu=1}^{2^n} \ket{0}\bra{0}\otimes \ket{0}\bra{0}_{\rm P}\otimes\ket{\bm{0}}\bra{\bm{0}}\otimes  \ket{\mu}\bra{\mu},
\end{align*}
then $\tilde{\Pi}U\Pi$ takes the form of Eq.~\eqref{suppleeq:OAA4I_input} with $p=1/4$ and the isometry $W$ defined as 
\begin{equation}\label{suppleeq:explicitW}
    W=\sum_{\mu,\nu=1}^{2^n}\sum_{k=0}^{K-1} \bra{\mu}A_k\ket{\nu} \ket{0}\ket{k}_{\rm P}\ket{\bm{0}}\ket{\mu}\bra{0}\bra{0}_{\rm P}\bra{\bm{0}}\bra{\nu}.
\end{equation}
Here, we can directly check that $W$ in Eq.~\eqref{suppleeq:explicitW} satisfies $W^{\dagger} W = \Pi$ if $\{A_k\}$ are Kraus operators of a quantum channel.
Therefore, applying the following unitary sequence involving the unitary operator $U$ in Eq.~\eqref{suppleeq:oaa_u_example}
\begin{equation}
    U(2\Pi-\bm{1})U^\dagger (2\tilde{\Pi}-\bm{1})U,
\end{equation}
to the input state
$\ket{0}\ket{0}_{\rm P}\ket{\bm{0}}\ket{\psi}$
for any $n$-qubit state $\ket{\psi}$, we can \textit{deterministically} obtain the target quantum state 
\begin{equation}
    \sum_k \ket{k}_{\rm P}\otimes A_k\ket{\psi}
\end{equation}
by measuring the other qubits in the computational basis.
This is because the probability of obtaining the measurement result for $\ket{0}\otimes \ket{\bm{0}}$, or equivalently the probability of projection onto $\tilde{\Pi}$, is just 1:
\begin{equation}
    \left\|\tilde{\Pi}U(2\Pi-\bm{1})U^\dagger (2\tilde{\Pi}-\bm{1})U\ket{0}\ket{0}_{\rm P}\ket{\bm{0}}\ket{\psi}\right\|^2=\left\|-W\ket{0}\ket{0}_{\rm P}\ket{\bm{0}}\ket{\psi}\right\|^2=1.
\end{equation}

\subsection{Review of randomized LCU for Hamiltonian simulation}\label{subsection:wan}

In this section, we review 
a random-sampling implementation of a LCU decomposition for the time propagator $e^{-iHt}$~\cite{Wan2022-tx,chakraborty2024implementing}.
First, we provide the decomposition for $e^{-iHt}$ derived in~\cite{Wan2022-tx}, which is one of the key components of the proof of our main theorem.
Let $H = \sum_{i} \alpha_{i} P_{i}$ with Pauli strings $P_i$ and real coefficients $\alpha_{i}$, and let $\alpha = \sum_i |\alpha_i|$.
Splitting the total evolution time $t$ into $r$ segments, we expand each segment $e^{-iHt/r}$ via the Taylor expansion as
\begin{align}\label{eq:wan_expansion_grouping}
    e^{-iHt/r} &= \sum_{n=0}^{\infty} \frac{(-i H t/r)^n}{n!} = \sum_{k =0}^{\infty} \frac{(-iHt/r)^{2k}}{(2k)!} \left( I+ \frac{-iHt/r}{2k+1} \right) \notag\\
    &= \sum_{k =0}^{\infty} \frac{(-iHt/r)^{2k}}{(2k)!} \sum_{i} \frac{|\alpha_{i}|}{\alpha} \left( I  -i \frac{\alpha t/r}{2k+1}\mathrm{sgn}(\alpha_i) P_i  \right) \notag\\
    &= \sum_{k =0}^{\infty} \frac{(-iHt/r)^{2k}}{(2k)!} \sum_{i} \frac{|\alpha_{i}|}{\alpha} \sqrt{ 1  + \left(\frac{\alpha t/r}{2k+1} \right)^2} \exp [-i \theta_k \mathrm{sgn}(\alpha_i) P_i],
\end{align}
where $\theta_k = \arccos(\{1+(\frac{\alpha t/r}{2k+1})^2\}^{-1/2})$.
At the last equality, we have combined the two Pauli gates $I$ and $P_i$ into a single non-Clifford gate.
Note that the controlled version of $e^{-i\theta P}$ for any Pauli string $P$ can be implemented with a single-qubit controlled rotation and multiple Clifford gates (CNOT gates).
Furthermore, defining the probability distribution $p_i:=|\alpha_i|/\alpha$, we have
\begin{align}\label{eq:wan_exp_fin}
    e^{-iHt/r} 
    &=\sum_{k =0}^{\infty} \frac{(-i\alpha t/r)^{2k}}{(2k)!} \sum_{i} \frac{|\alpha_{i}|}{\alpha} \sqrt{ 1  + \left(\frac{\alpha t/r}{2k+1} \right)^2} \left(\frac{H}{\alpha}\right)^{2k} \exp [-i \theta_k \mathrm{sgn}(\alpha_i) P_i]\notag\\
    &=\sum_{k =0}^{\infty} \frac{(-i\alpha t/r)^{2k}}{(2k)!} \sqrt{ 1  + \left(\frac{\alpha t/r}{2k+1} \right)^2} \sum_{i} p_i \left(\sum_{j} p_j \mathrm{sgn}(\alpha_j)P_j\right)^{2k} \exp [-i \theta_k \mathrm{sgn}(\alpha_i) P_i].
\end{align}
Thus, we obtain a LCU decomposition of the form $e^{-iHt/r}=\sum_m c_m W_m$, where $c_m> 0$ and $W_m$ are implicitly defined through Eq.~\eqref{eq:wan_exp_fin}.
Moreover, it has been shown that the $L^1$-norm of the LCU coefficient $\{c_m\}$ is upper bounded by $ e^{(\alpha t/r)^2}$~\cite{Wan2022-tx}.
This scaling is attributed to the grouping of a linear combination of two Paulis into a single unitary in the final line of Eq.~\eqref{eq:wan_expansion_grouping}, which improves the corresponding $L^1$ norm from $1+\mathcal{O}(t/r)$ to $\sqrt{1 + \mathcal{O} ((t/r)^2)} = 1 + \mathcal{O}((t/r)^2)$.
This improvement critically reduces the total $L^1$-norm.
Therefore, we obtain the LCU decomposition of $e^{-iHt}=(e^{-iHt/r})^r$ that has the total $L^1$-norm at most $e^{(\alpha t)^2/r}$, which can be summarized as follows. 
\begin{lem}
    [A slightly modified version of Lemma 2 in Ref.~\cite{Wan2022-tx}]\label{lemma:wan_random_compiler}
    Let $H = \sum_{j} \alpha_{j} P_{j}$ be a Hermitian operator
    that is specified as a linear combination of Pauli operators with $\alpha_j\in \mathbb{R}$.
    For any $t \in \mathbb{R}$ and $r \in  \mathbb{N}$, there exists a linear decomposition
    $$ 
    e^{-iHt} = \sum_{k \in \mathrm{S}} b_{k} U_{k} 
    $$
    for some index set $\mathrm{S}$, $b_k >0$, and unitaries $U_k$,
    such that
    $$ \sum_{k \in \mathrm{S}} b_{k} \le \mathrm{exp}\left[{\left(\sum_i |\alpha_i|\right)^2 \frac{t^2 }{r} }\right].
    $$ 
    For all $k \in \mathrm{S}$, the non-Clifford cost of controlled-$U_k$ is that of $r$ controlled single-qubit Pauli rotations.
\end{lem}

\noindent
Now, we briefly summarize key ideas of the randomization for efficient simulation of the LCU decomposition of $e^{-iHt}$.
Let $U_i$ and $U_j$ be independent random unitaries following the distribution $\{p_i = b_i /b \}$ with $b:= \sum_{k \in \mathrm{S}} b_k$.
Then, the expectation value of an observable $O$ can be estimated with the expectation of the random unitaries with the multiplier $b^2$:
$$
\mathrm{Tr}\left[ O \left(\sum_{i} b_i U_i \right) \rho \left(\sum_{j} b_j^\dagger U_j^\dagger \right)\right]
= b^2 \sum_{i,j} p_i p_j \mathrm{Tr}(O U_i \rho U_j^\dagger)
=b^2  \mathbb{E}_{i,j} \left[\mathrm{Tr}(O U_i \rho U_j^\dagger) \right]
$$
Observing the right hand side of the equation, 
we randomly and independently sample $U_i$ and $U_j$ according to the distribution $\{p_i\}$ and run the generalized Hadamard test circuit (Fig.~\ref{fig:generalized_Hadamard_test}). 
Then, the sampling mean of the measurement outcomes multiplied by $b^2$ serves as an estimator of the expectation value of the observable.
\begin{figure}[htbp]
    \centering
    \begin{quantikz}
        \lstick{$\ket{+}$} & \octrl{1} & \ctrl{1} & \gate[wires=2]{X \otimes O} \\
        \lstick{$\rho$} & \gate{U_i}  & \gate{U_j} & 
    \end{quantikz}
    \caption{The circuit for the generalized Hadamard test. $U_i$ and $U_j$ are the unitary operators randomly sampled. }
    \label{fig:generalized_Hadamard_test}
\end{figure}

As shown in the previous work~\cite{chakraborty2024implementing}, by appropriately truncating the Taylor series, this estimation scheme requires $\mathcal{O}(\|O\|^2 b^4\log(1/\delta)/\varepsilon^2)$ samples
of the quantum circuits
with the depth at most $\mathcal{O}(\alpha^2 t^2 \log(\alpha t\|O\|/\varepsilon))$, in order to achieve an additive error $\varepsilon$ of the final estimation with at least $1-\delta$ probability, where $\alpha := \sum_i |\alpha_i|$.
Remarkably, the circuit depth achieves the logarithmic dependence of $\varepsilon$.
The Hamiltonian simulation based on the Taylor series, proposed by \cite{Berry2015-truncatedtaylor},
also achieves the logarithmic dependence of $\varepsilon$ together with a linear dependence on $t$.
However, this algorithm employs LCU and OAA techniques with many additional ancilla qubits and complicated multi-controlled operations.
In contrast, the randomized approach uses only one ancilla qubit because it simulates the LCU with the help of random sampling of multiple quantum circuits.

The randomized approach involves a trade-off between the circuit depth and the number of samples.
The factor $b^2$ in the estimator increases the sampling cost because the number of samples scales as $\mathcal{O}(\|O\|^2 b^4 \log(1/\delta)/\varepsilon^2)$. 
In the Hamiltonian simulation with the truncated Taylor series proposed in Ref.~\cite{Berry2015-truncatedtaylor},
the corresponding $L^1$ norm of coefficients scales as $e^{\alpha t}$, leading to an exponential increase in the sampling cost when the randomized approach is employed.
However, according to Lemma~\ref{lemma:wan_random_compiler}, if $r$ is chosen as $r =\mathcal{O} (\alpha^2 t^2)$, the sampling overhead $b^2$ can be suppressed to $\mathcal{O}(1)$ with respect to $\alpha$ and $t$.

\section{Decomposition of $e^{t\mathcal{L}}$}
In this section, we prove Theorem~\ref{thm: main} to derive the decomposition of $e^{t\mathcal{L}}$.
Subsection~\ref{apdx:B1} presents a comprehensive proof, while the essential technical components are described in the subsequent subsections.

\begin{itemize}
    \item Subsection~\ref{apdx:B1} provides the proof of Theorem~\ref{thm: main}. In the proof, we utilize Lemma~\ref{lemma:oaa_for_B0B1}, \ref{lemma:correction_superoperator}, \ref{lem:wan_gammma_inequality} and the framework described by Proposition~\ref{prop:lcs}, \ref{prop:lcs_alternative_seq}.
    \item Subsection~\ref{apdx:exact_efficient_oaa} introduces a method to simulate general CP maps.
    The method allows the exact and efficient simulation for the dissipation. Lemma~\ref{lemma:oaa_non_isometry}, \ref{lemma:generalRsampling} are the statements for general CP maps,
    and Lemma~\ref{lemma:oaa_for_B0B1}, \ref{lemma:correction_superoperator} are for the special case of our interest.
    \item Subsection~\ref{apdx:construct_circuits} introduces a general simulation framework for a linear combination of $n$-qubit superoperators with a form of $\sum c_i U_i \bullet V_i^\dagger$ using $(n+1)$-qubit quantum circuits. This framework is well summarized in Proposition~\ref{prop:lcs}, \ref{prop:lcs_alternative_seq}.
    \item Subsection~\ref{subsec:technical_lemmas} shows technical lemmas. In particular, Lemma~\ref{lem:wan_gammma_inequality} gives the upper bound of the norm of the decomposition that we derive in Subsection~\ref{apdx:B1}.
\end{itemize}

\subsection{Proof of Theorem~\ref{thm: main}}\label{apdx:B1}

\theoremstyle{plain}
\newtheorem*{T0}{Theorem~\ref{thm: main}}

\begin{T0}
    Let $\mathcal{L}$ be an $N_{\mathrm{q}}$-qubit Lindblad superoperator with a Hamiltonian $H$ and jump operators $\{L_k\}_{k=1}^K$ that are specified by a linear combination of Pauli strings as 
    \begin{equation}\label{eq:def_H_and_L}
        H = \sum_{j=1}^m \alpha_{0j} P_{0j},~~~L_k =\sum_{j=1}^M \alpha_{kj} P_{kj},
    \end{equation}
    for some coefficients $\alpha_{0j} \in \mathbb{R}$, $\alpha_{kj} \in \mathbb{C}$. 
    Also, let $\|\mathcal{L}\|_{\rm pauli}:=2(\alpha_0+\sum_{k=1}^K \alpha_k^2)$ for $\alpha_k:=\sum_{j} |\alpha_{kj}|$.
    Then, for any $t >0$ and any positive integer $r\geq \|\mathcal{L}\|_{\rm pauli}t$,
    there exists a linear decomposition
    \begin{equation}
    \label{suppleeq:explicit decomposition}
    e^{t\mathcal{L}}(\bullet)
    = \sum_{v \in \mathrm{S}} c_v \mathrm{Tr}_{\mathrm{anc}}[(X_{\rm anc} \otimes \bm{1}) \widetilde{\mathcal{W}}_v (\ket{+}\bra{+}_{\rm anc} \otimes \bullet )]
    \end{equation}
    for some index set $\mathrm{S}$, real values $c_v >0$, and $(N_{\mathrm{q}}+1)$-qubit completely positive trace non-increasing (CPTN) maps $\{\widetilde{\mathcal{W}}_v\}$ such that
    \begin{equation}\label{suppleeq:thm_norm_bound}
    \sum_{v \in \mathrm{S}} c_{v} \leq e^{2{{\|\mathcal{L}\|_{\rm pauli}^2t^2}/{r}}}.
    \end{equation}
    Furthermore, for any $v$, the $(N_{\mathrm{q}}+1)$-qubit CPTN map $\widetilde{\mathcal{W}}_v$ can be effectively simulated by a quantum circuit on the $(N_{\mathrm{q}}+1)$-qubit system and additional $3+\lceil \log_2 M\rceil$ qubits ancilla system with mid-circuit measurement and qubit reset.
\end{T0}

In the following, we write $N_{\rm q}$ simply as $n$ when the context is clear.
\begin{proof}
    [Proof of Theorem~\ref{thm: main}]
    This proof consists of two parts: (i) proof of a decomposition of the dynamical map $e^{tG}=S(e^{t\mathcal{L}})$ in the form of a linear combination of (the transfer matrix of) superoperators and (ii) proof of an efficient simulation of the superoperators $\widetilde{\mathcal{W}}_v$ using a generalized Hadamard test circuit on the target $n$-qubit system and a constant-size ancilla system, which does not require doubling the target system qubit size.
    In the proof, we employ Lemmas~\ref{lemma:oaa_for_B0B1},~\ref{lemma:correction_superoperator},~\ref{lemma:combex_combination}, and~\ref{lem:wan_gammma_inequality} which are proved in the next subsections.
    Note that from the assumption, 
    we write $H/\alpha_0$ and $L_k/\alpha_k$ as
    \begin{equation}\label{eq:H_L_convexexpandion}
        \frac{H}{\alpha_0}=\sum_{j=1}^{m} p_{0j}e^{i\theta_{0j}}P_{0j},~~~
        \frac{L_k}{\alpha_k}=\sum_{j=1}^{M} p_{kj}e^{i\theta_{kj}}P_{kj},
    \end{equation}
    where $p_{kj}:=|\alpha_{kj}|/\alpha_k$ is a probability distribution satisfying $\sum_j p_{kj}=1$ for all $k=0,1,...,K$, and the phases $e^{i\theta_{kj}}$ are defined as
    \begin{equation}
        e^{i\theta_{kj}}:=\frac{\alpha_{kj}}{|\alpha_{kj}|},~~~\forall~k,j.
    \end{equation}
    As for the case $k=0$, the phase $e^{i\theta_{0j}}$ is equivalent to ${\rm sgn}(\alpha_{0j})$.

    Now, we construct a decomposition of $e^{tG}=S(e^{t\mathcal{L}})$.
    Our intermediate goal is to expand the transfer matrix $e^{tG/r}$ of the $r$-sliced dynamical map $e^{t\mathcal{L}/r}$ in the form 
    \begin{equation}
        e^{tG/r} = \sum_{u} c_u S(\mathcal{W}_u)
    \end{equation}
    for some positive coefficients $c_u$ and some superoperators $\mathcal{W}_u$ acting on the target $n$-qubit system.
    In particular, we aim to construct such a decomposition with the coefficients $\{c_u\}$ satisfying
    \begin{equation}
        \sum_{u} c_u \leq e^{\mathcal{O}(t^2/r^2)}
    \end{equation}
    in order that the total weight for $r$-repetition of this decomposition scales as $e^{\mathcal{O}(t^2/r)}$.
    To this end, we start by expanding the transfer matrix $e^{tG/r}$ via the Taylor expansion as
    \begin{equation}\label{suppleeq:dynamical_map}
        e^{(t/r)G} =\sum_{q=0}^{\infty} \frac{(t/r)^q G^q}{q!}=\sum_{l=0}^{\infty} \frac{(t/r)^{2l} }{(2l)!}G^{2l}\left(\bm{1} \otimes \bm{1} +\frac{(t/r)}{2l+1}G\right),
    \end{equation}
    where $G=S(\mathcal{L})$ can be written as
    \begin{align}\label{eq:G_transfer_L}
        G&=-i\bm{1}\otimes H+iH^{\rm T}\otimes\bm{1}+\sum_{k=1}^K\left(\overline{L_k}\otimes L_k-\frac{1}{2}\bm{1}\otimes L_k^\dagger L_k-\frac{1}{2}L_k^{\rm T}\overline{L_k}\otimes \bm{1}\right)\notag\\
        &= \alpha_0\cdot \bm{1}\otimes \frac{-iH}{\alpha_0}
        +\alpha_0\cdot \frac{iH^{\rm T}}{\alpha_0}\otimes \bm{1}+\sum_{k=1}^K \alpha_k^2\cdot \left(\frac{\overline{L_k}}{\alpha_k}\otimes \frac{L_k}{{\alpha_k}}+\frac{1}{2}\bm{1}\otimes \frac{(-1)L_k^\dagger L_k}{\alpha_k^2}+\frac{1}{2}\frac{(-1)L_k^{\rm T}\overline{L_k}}{\alpha_k^2}\otimes \bm{1}\right).
    \end{align}
    
    In what follows, we decompose $G^{2l}$ and $\bm{1} \otimes \bm{1}+\frac{(t/r)}{2l+1}G$ in Eq.~\eqref{suppleeq:dynamical_map} toward the final form \eqref{eq:expansion_explicit_form_tr}. 
    Firstly, we provide a decomposition of $G^{2l}$ up to normalization.
    Since $H/\alpha_0$ and $L_k/\alpha_k$ are the convex combinations of Pauli strings with complex phases as in Eq.~\eqref{eq:H_L_convexexpandion}, Lemma~\ref{lemma:combex_combination} says that there is a probability distribution $p_{G}(\mu)$ over some index set $\mathrm{S}_G(\ni \mu)$ and two sets $\{U_{G}(\mu)\}$, $\{V_{G}(\mu)\}$ of $n$-qubit Pauli strings with complex phases (e.g., $e^{i\theta}P$ for some $\theta\in \mathbb{R}$) such that
    \begin{equation}\label{eq:G_convex_pauli}
        \frac{G}{\|\mathcal{L}\|_{\rm pauli}}=\sum_{\mu \in \mathrm{S}_{G}} p_{G}(\mu) \overline{U_{G}(\mu)}\otimes {V_{G}}(\mu)
        ={S\left(\sum_{\mu\in \mathrm{S}_G}p_{G}(\mu)\cdot V_G(\mu)\bullet U^\dagger_{G}(\mu)\right)}.
    \end{equation}
    Here, $\|\mathcal{L}\|_{\rm pauli}$ defined by
    \begin{equation}\label{eq:G_norm_pauli}
        \|\mathcal{L}\|_{\rm pauli} = 2\alpha_0 + \sum_{k=1}^K \alpha_k^2\cdot (1+\frac{1}{2}\times 2) = 2\left(\alpha_0+\sum_{k=1}^K\alpha_k^2\right)
    \end{equation}
    is the normalization factor introduced to guarantee $\sum_\mu p_G(\mu)=1$.
    Note that we can efficiently sample the random unitary $\overline{U_G(\mu)}\otimes V_{G}(\mu)$ according to the probability distribution $p_{G}(\mu)$ by Algorithm~\ref{alg:sample_G}.
    Also, we note that $G/\|\mathcal{L}\|_{\rm pauli}$ is the transfer matrix of the convex combination of the superoperators $V_{G}(\mu)\bullet U^\dagger_{G}(\mu)$ acting on the target system as
    \begin{equation}
        V_{G}(\mu)\bullet U^\dagger_{G}(\mu):A\mapsto    V_{G}(\mu)A U^\dagger_{G}(\mu) 
    \end{equation} 
    for any operator $A$.
    Thus, we straightforwardly obtain the decomposition of $G^{2l} / \|\mathcal{L}\|_{\rm pauli}^{2l}$ from Eq.~\eqref{eq:G_convex_pauli}.

    Next, we decompose the term $\bm{1} \otimes \bm{1}+\frac{(t/r)}{2l+1}G$ in Eq.~\eqref{suppleeq:dynamical_map} into a linear combination of (the transfer matrix of) superoperators.
    Defining
    \begin{equation}
        \alpha := 2 \left( \alpha_{0}  + \frac{1}{2}\sum_{k=1}^{K} \alpha_{k}^2\right)~~~\mbox{and}~~~\tau_l := \frac{\alpha (t/r)}{2l+1},
    \end{equation}
    we rewrite the term $\bm{1} \otimes \bm{1}+\frac{(t/r)}{2l+1}G=\bm{1} \otimes \bm{1}+\tau_l{G}/{\alpha}$.
    To satisfy Eq.~\eqref{suppleeq:thm_norm_bound}, we derive the decomposition whose normalization scales as $1 + \mathcal{O}(\tau_l^2)$.
    \begin{align}\label{suppleeq:I+Gt}
        \bm{1} \otimes \bm{1}+\tau_l \frac{G}{\alpha}
        &=\frac{\alpha_0}{\alpha} \left(\bm{1} \otimes \bm{1} + (-i\tau_l)\bm{1}\otimes \frac{H}{\alpha_0} \right)
        + \frac{\alpha_0}{\alpha} \left(\bm{1} \otimes \bm{1} + (+i\tau_l)\frac{H^{\rm T}}{\alpha_0} \otimes \bm{1} \right)\notag \\
        &~~~+ \sum_k \frac{\alpha_k^2}{\alpha} \left(\bm{1} \otimes \bm{1} + \tau_l \left( \frac{\overline{L_k}}{\alpha_k}\otimes \frac{L_k}{{\alpha_k}}-\frac{1}{2}\bm{1}\otimes \frac{L_k^\dagger L_k}{\alpha_k^2}-\frac{1}{2}\frac{L_k^{\rm T}\overline{L_k}}{\alpha_k^2}\otimes \bm{1}\right)\right)\notag\\
        & = \frac{\alpha_0}{\alpha} \left(\bm{1} \otimes \bm{1} +  (-i\tau_l)\bm{1}\otimes \frac{H}{\alpha_0} \right)
        + \frac{\alpha_0}{\alpha} \left(\bm{1} \otimes \bm{1} + (+i\tau_l)\frac{H^{\rm T}}{\alpha_0} \otimes \bm{1} \right)\notag \\
        &~~~+ \sum_k \frac{\alpha_k^2}{\alpha} \left\{
        \overline{\left(\bm{1}-\frac{\tau_l}{2}\frac{L_k^\dagger L_k}{\alpha_k^2}\right)}\otimes \left(\bm{1}-\frac{\tau_l}{2}\frac{L_k^\dagger L_k}{\alpha_k^2}\right) + \tau_l \frac{\overline{L_k}}{\alpha_k}\otimes \frac{L_k}{{\alpha_k}} - \frac{\tau_l^2}{4}\frac{L_k^{\rm T} \overline{L_k}}{\alpha_k^2}\otimes \frac{L_k^\dagger L_k}{\alpha_k^2} \right\}.
    \end{align}
    To have further detailed decomposition, we apply different treatments for the first and second terms that come from the Hamiltonian $H$, and the third term that comes from the jumps $\{L_k\}$.
    For the first and second terms, we follow the technique developed in Ref.~\cite{Wan2022-tx} (reviewed in Section~\ref{subsection:wan}),
    to combine the Pauli operators into a single non-Clifford gate. That is, the first term can be written as
    \begin{align}\label{suppleeq:vectorized_rotaion_gate}
        \bm{1} \otimes \bm{1} + (-i\tau_l)\bm{1}\otimes \frac{H}{\alpha_0}
        & = \sum_{j=1}^{m} p_{0j}\left( \bm{1} \otimes \bm{1} -i \tau_l e^{i\theta_{0j}} \bm{1} \otimes P_{0j} \right) \notag \\
        & =\sqrt{1 + \tau_l^2} \sum_{j=1}^{m} p_{0j} \left( \bm{1} \otimes \exp [-i \theta_{l} \mathrm{sgn}(\alpha_{0j}) P_{0j} ] \right)\notag\\
        & =\sqrt{1 + \tau_l^2} \cdot S\left(\sum_{j=1}^{m} p_{0j} \cdot \exp [-i \theta_{l} \mathrm{sgn}(\alpha_{0j}) P_{0j} ]\bullet \bm{1} \right),
        \end{align}
    where $\theta_{l} = \arccos(\{1+\tau_l^2\}^{-1/2})$.
    Note that the superoperator $\exp [-i \theta_{l} \mathrm{sgn}(\alpha_{0j}) P_{0j} ]\bullet \bm{1}$ acts as 
    \begin{equation}
        A\mapsto \exp [-i \theta_{l} \mathrm{sgn}(\alpha_{0j}) P_{0j} ]A
    \end{equation}
    for any operator $A$.
    Similarly, we can write the second term as
    \begin{equation}\label{suppleeq:vectorized_rotaion_gate_T}
        \bm{1} \otimes \bm{1} + (+i\tau_l)\frac{H^{\rm T}}{\alpha_0} \otimes \bm{1}  = \sqrt{1 + \tau_l^2} \cdot S\left( \sum_{j=1}^{m} p_{0j} \cdot \bm{1}\bullet \exp [i \theta_{l} \mathrm{sgn}(\alpha_{0j}) {P_{0j}} ]\right).
    \end{equation}
    Note that the scaling of the normalization factor $\sqrt{1+\tau_l^2}=1+\mathcal{O}(\tau_l^2)$ with respect to $\tau_l$ is crucial to achieve the goal $\sum_{u}c_u\leq e^{\mathcal{O}(t^2/r^2)}$ as described later.
    
    Thus, as for the terms where the jumps $\{L_k\}$ appear in Eq.~\eqref{suppleeq:I+Gt},
    we also want to have an expression with overhead $1+\mathcal{O}(\tau_l^2)$; a naive approach is as follows; 
    \begin{align}
         &\overline{\left(\bm{1}-\frac{\tau_l}{2}\frac{L_k^\dagger L_k}{\alpha_k^2}\right)}\otimes \left(\bm{1}-\frac{\tau_l}{2}\frac{L_k^\dagger L_k}{\alpha_k^2}\right) + \tau_l \frac{\overline{L_k}}{\alpha_k}\otimes \frac{L_k}{{\alpha_k}} - \frac{\tau_l^2}{4}\frac{L_k^{\rm T} \overline{L_k}}{\alpha_k^2}\otimes \frac{L_k^\dagger L_k}{\alpha_k^2} \notag\\
         &=\left(1+\frac{\tau_l^2}{2}\right)\times \left\{\frac{1+{\tau_l^2}/{4}}{1+{\tau_l^2}/{2}}\cdot S\left(\frac{1}{1+{\tau_l^2}/{4}}\mathcal{B}_{kl}\right)+\frac{\tau_l^2/4}{1+{\tau_l^2}/{2}}\cdot (-1)\cdot S\left(\frac{L_k^\dagger L_k}{\alpha_k^2}\bullet \frac{L_k^\dagger L_k}{\alpha_k^2}\right)\right\},
    \end{align}
    where the superoperator $\mathcal{B}_{kl}$ is the CP map defined as
    \begin{align}
        \mathcal{B}_{kl}(\bullet)&:={B}_{kl,0}\bullet {B}_{kl,0}^\dagger+{B}_{kl,1}\bullet {B}_{kl,1}^\dagger\notag\\
        &\equiv \left(\bm{1}-\frac{\tau_l}{2}\frac{L_k^\dagger L_k}{\alpha_k^2}\right) \bullet \left(\bm{1}-\frac{\tau_l}{2}\frac{L_k^\dagger L_k}{\alpha_k^2}\right)^\dagger + \sqrt{\tau_l}\frac{L_k}{{\alpha_k}}\bullet \left(\sqrt{\tau_l}\frac{L_k}{{\alpha_k}}\right)^\dagger.
    \end{align}
    Here, we defined the operators $B_{kl,0}$ and $B_{kl,1}$ in the equation. 
    A straightforward approach to realize $\mathcal{B}_{kl}$ is to find a CPTN map $\mathcal{K}$ such that $\mathcal{K}+(1+\tau_l^2/4)^{-1}\mathcal{B}_{kl}$ becomes a CPTP map because the superoperator $(1+\tau_l^2/4)^{-1}\mathcal{B}_{kl}$ is a CPTN map.
    We obtain $(1+\tau_l^2/4)^{-1}\mathcal{B}_{kl}$ by partially measuring a quantum circuit with (possibly a large number of) additional qubits and/or gates, which is in principle possible using, e.g., the Stinespring dilation theorem~\cite{Suri2023-iq}.
    However, such a naive approach to simulate $S(\mathcal{B}_{kl})$ is computationally hard in general, especially for finding a companion CPTN map $\mathcal{K}$.

    To overcome this difficulty, we develop a new technique that allows us to exactly simulate the CP map $\mathcal{B}_{kl}$ much easier than the above approach, with the help of random sampling of \textit{correction} superoperators. 
    We summarize this technique in Lemma~\ref{lemma:oaa_for_B0B1} and Lemma~\ref{lemma:correction_superoperator}, which will be formally proven in the next subsection. 
    Using Lemma~\ref{lemma:oaa_for_B0B1} and Lemma~\ref{lemma:correction_superoperator} (note that from the assumption of Theorem~\ref{thm: main}, we have $\tau_l\in [0,3]$ for any $l$), the CP map $\mathcal{B}_{kl}$ can be decomposed as
    \begin{equation}\label{eq:B_BpR}
        S(\mathcal{B}_{kl}) = S(\mathcal{B}^{(\rm approx)}_{kl})+S(\mathcal{R}_{kl}),
    \end{equation}
    where $\mathcal{R}_{kl}$ is the correction superoperator with the magnitude of $\mathcal{O} (\tau_l^2)$ whose transfer matrix is given by
    \begin{align}\label{eq:correctionop_kl}
        R_{kl}\equiv S(\mathcal{R}_{kl})&:=\sum_{\lambda=0}^1 \left(\frac{\tau_l^2}{8}\overline{B_{kl,\lambda}}\otimes B_{kl,\lambda}D_k+\frac{\tau_l^2}{8}\overline{B_{kl,\lambda}D_k}\otimes B_{kl,\lambda}-\frac{\tau_l^4}{64}\overline{B_{kl,\lambda}D_k}\otimes B_{kl,\lambda}D_k\right)\notag\\
        &=\sum_{\lambda=0}^1 \gamma_{l,\lambda}^2\left(\frac{\tau_l^2}{8}\overline{\frac{B_{kl,\lambda}}{\gamma_{l,\lambda}}}\otimes \frac{B_{kl,\lambda}}{\gamma_{l,\lambda}}D_k+\frac{\tau_l^2}{8}\overline{\frac{B_{kl,\lambda}}{\gamma_{l,\lambda}}D_k}\otimes \frac{B_{kl,\lambda}}{\gamma_{l,\lambda}}+\frac{\tau_l^4}{64}(-1)\overline{\frac{B_{kl,\lambda}}{\gamma_{l,\lambda}}D_k}\otimes \frac{B_{kl,\lambda}}{\gamma_{l,\lambda}}D_k\right),
    \end{align}
    for $\gamma_{l,\lambda}>0$ defined as
    \begin{equation}
        \gamma_{l,0}:=1+\frac{\tau_l}{2},~~~\gamma_{l,1}:=\sqrt{\tau_l}.
    \end{equation}
    Also, the CP map $\mathcal{B}^{(\rm approx)}_{kl}$ and the operator $D_k$ are defined as
    \begin{equation}\label{apdx:B_kl_approx_def}
        \mathcal{B}^{(\rm approx)}_{kl}(\bullet):=\sum_{\lambda=0,1} B'_{kl,\lambda}\bullet (B'_{kl,\lambda})^\dagger,~~~
        B'_{kl,\lambda} := B_{kl,\lambda}\left(\bm{1}-\frac{\tau_l^2}{8}D_k\right),
    \end{equation}
    and
    \begin{equation}
        \frac{\tau_l^2}{4}D_k:=B_{kl,0}^\dagger B_{kl,0} + B_{kl,1}^\dagger B_{kl,1} - \bm{1}=\frac{\tau_l^2}{4}\frac{(L_k^\dagger L_k)^2}{\alpha_k^4}.
    \end{equation}
    As described in Lemma~\ref{lemma:oaa_for_B0B1}, we can explicitly construct a quantum circuit with the use of additional $3+\lceil \log_2 M\rceil$ ancilla qubits that simulates the CP (more precisely, CPTN) map $\mathcal{B}^{(\rm approx)}_{kl}$ by measuring or discarding the ancilla qubits at the end of this circuit; See Remark~\ref{rem:cptn}.
    On the other hand, since $B_{kl,\lambda}/\gamma_{l,\lambda}$ and $D_k$ are the convex combinations of Pauli strings with complex phases, Lemma~\ref{lemma:correction_superoperator} yields the following decomposition
    \begin{equation}\label{eq:R_convex_pauli}
        \frac{R_{kl}}{\|R_{kl}\|_{\rm pauli}}=\sum_{\mu\in\mathrm{S}_{R_{kl}}}p_{R_{kl}}(\mu) \overline{U_{R_{kl}}(\mu)}\otimes V_{R_{kl}}(\mu),
    \end{equation}
    where the normalization factor $\|R_{kl}\|_{\rm pauli}$ is defined as
    \begin{equation}\label{eq:correction_norm_kl}
        \|R_{kl}\|_{\rm pauli}:=\sum_{\lambda=0}^1 \gamma_{l,\lambda}^2 \frac{\tau_l^2}{4}\left(1+\frac{\tau_l^2}{16}\right)=\frac{\tau_l^2}{4}+\frac{\tau_l^3}{2}+\frac{5\tau_l^4}{64}+\frac{\tau_l^5}{32}+\frac{\tau_l^6}{256}.
    \end{equation}
    Here, $p_{R_{kl}}(\mu)$ is a probability distribution over the index set $\mathrm{S}_{R_{kl}}$, and $U_{R_{kl}}(\mu),V_{R_{kl}}(\mu)$ are $n$-qubit Pauli strings with complex phases for all $\mu\in \mathrm{S}_{R_{kl}}$.
    The random unitary $\overline{U_{R_{kl}}(\mu)}\otimes V_{R_{kl}}(\mu)$ can be efficiently sampled according to the probability distribution $p_{R_{kl}}(\mu)$ by Algorithm~\ref{alg:sample_Rkl}. 
    Therefore, we now have an algorithm for exactly simulating the CP map $\mathcal{B}_{kl}$.
    Using the correction superoperators, we have
    \begin{align}\label{eq:BR-LL}
         \overline{\left(\bm{1}-\frac{\tau_l}{2}\frac{L_k^\dagger L_k}{\alpha_k^2}\right)}\otimes \left(\bm{1}-\frac{\tau_l}{2}\frac{L_k^\dagger L_k}{\alpha_k^2}\right) + \tau_l \frac{\overline{L_k}}{\alpha_k}\otimes \frac{L_k}{{\alpha_k}} - \frac{\tau_l^2}{4}\frac{L_k^{\rm T} \overline{L_k}}{\alpha_k^2}\otimes \frac{L_k^\dagger L_k}{\alpha_k^2}
         =S(\mathcal{B}_{kl}^{\rm (approx)}) +
         S(\mathcal{R}_{kl}) - \frac{\tau_l^2}{4}\frac{L_k^{\rm T} \overline{L_k}}{\alpha_k^2}\otimes \frac{L_k^\dagger L_k}{\alpha_k^2}.
    \end{align}

    Having shown the grouped form of all components in Eq.~\eqref{suppleeq:I+Gt}, we combine the Eqs.~\eqref{suppleeq:I+Gt}, \eqref{suppleeq:vectorized_rotaion_gate}, \eqref{suppleeq:vectorized_rotaion_gate_T}, and \eqref{eq:BR-LL} and proceed to derive a new decomposition of $e^{(t/r)G}$.
    Combining the Eqs.~\eqref{suppleeq:I+Gt}, \eqref{suppleeq:vectorized_rotaion_gate}, \eqref{suppleeq:vectorized_rotaion_gate_T}, and \eqref{eq:BR-LL}, we can write $\bm{1}\otimes \bm{1}+\tau_l G/\alpha$ as
    \begin{align}
        \bm{1}\otimes \bm{1} +\tau_l \frac{G}{\alpha}&= \frac{\alpha_0}{\alpha} \left(\bm{1} \otimes \bm{1} +  (-i\tau_l)\bm{1}\otimes \frac{H}{\alpha_0} \right)
        + \frac{\alpha_0}{\alpha} \left(\bm{1} \otimes \bm{1} + (+i\tau_l)\frac{H^{\rm T}}{\alpha_0} \otimes \bm{1} \right)\notag \\
        &~~~+ \sum_k \frac{\alpha_k^2}{\alpha} \left\{
        \overline{\left(\bm{1}-\frac{\tau_l}{2}\frac{L_k^\dagger L_k}{\alpha_k^2}\right)}\otimes \left(\bm{1}-\frac{\tau_l}{2}\frac{L_k^\dagger L_k}{\alpha_k^2}\right) + \tau_l \frac{\overline{L_k}}{\alpha_k}\otimes \frac{L_k}{{\alpha_k}} - \frac{\tau_l^2}{4}\frac{L_k^{\rm T} \overline{L_k}}{\alpha_k^2}\otimes \frac{L_k^\dagger L_k}{\alpha_k^2} \right\}\notag\\
        &= 2\frac{\alpha_0}{\alpha} 
        \sqrt{1 + \tau_l^2} \frac{S\left(\sum_{j=1}^{m} p_{0j} \cdot \exp [-i \theta_{l} \mathrm{sgn}(\alpha_{0j}) P_{0j} ]\bullet \bm{1} \right)+S\left( \sum_{j=1}^{m} p_{0j} \cdot \bm{1}\bullet \exp [-i \theta_{l} \mathrm{sgn}(\alpha_{0j}) {P_{0j}} ]^\dagger\right)}{2}
        \notag \\
        &~~~+ \sum_k \frac{\alpha_k^2}{\alpha} \left( S(\mathcal{B}_{kl})-\frac{\tau_l^2}{4}S\left(\frac{L_k^\dagger L_k}{\alpha_k^2}\bullet \frac{L_k^\dagger L_k}{\alpha_k^2}\right)\right)\notag\\
        &= 2\frac{\alpha_0}{\alpha} 
        \sqrt{1 + \tau_l^2} \frac{S\left(\sum_{j=1}^{m} p_{0j} \cdot \exp [-i \theta_{l} \mathrm{sgn}(\alpha_{0j}) P_{0j} ]\bullet \bm{1} \right)+S\left( \sum_{j=1}^{m} p_{0j} \cdot \bm{1}\bullet \exp [-i \theta_{l} \mathrm{sgn}(\alpha_{0j}) {P_{0j}} ]^\dagger\right)}{2}
        \notag \\
        &~~~+ \sum_k \frac{\alpha_k^2}{\alpha} \left( S(\mathcal{B}^{(\rm approx)}_{kl})+\|R_{kl}\|_{\rm pauli} S\left(\frac{\mathcal{R}_{kl}}{\|R_{kl}\|_{\rm pauli}}\right)+\frac{\tau_l^2}{4}S\left((-1)\frac{L_k^\dagger L_k}{\alpha_k^2}\bullet \frac{L_k^\dagger L_k}{\alpha_k^2}\right)\right)\notag\\
        &=\left(2\frac{\alpha_0}{\alpha} \sqrt{1 + \tau_l^2}+\sum_{k=1}^K \frac{\alpha_k^2}{\alpha}\left(1+\|R_{kl}\|_{\rm pauli}+\frac{\tau_l^2}{4}\right)\right)\sum_{k=0}^K \sum_{\nu=1}^3 q_{kl} p_{\Gamma,kl,\nu} \Gamma_{kl,\nu},
    \end{align}
    where we define the probability distributions $\{q_{kl}\}, \{p_{\Gamma,kl,\nu}\}$ as
    \begin{equation}\label{eq:def_qkl}
        q_{0l}\propto 2\frac{\alpha_0}{\alpha}\sqrt{1+\tau_l^2},~~~q_{kl}\propto \frac{\alpha_k^2}{\alpha}\left(1+\|R_{kl}\|_{\rm pauli}+\frac{\tau_l^2}{4}\right)~~~(\mbox{for}~k=1,2,...,K),
    \end{equation}
    and 
    \begin{equation}\label{eq:def_pgamma_kl}
        p_{\Gamma,0l,\nu}:=\frac{1}{2}\left(\delta_{1\nu}+\delta_{2\nu}\right),~~~p_{\Gamma,kl,\nu}:=\frac{\delta_{1\nu}+\|R_{kl}\|_{\rm pauli}\delta_{2\nu}+(\tau_l^2/4)\delta_{3\nu}}{1+\|R_{kl}\|_{\rm pauli}+\tau_l^2/4}~~~(\mbox{for}~k=1,2,...,K).
    \end{equation}
    The transfer matrix of superoperator, which we denote as $\Gamma_{kl,\nu}$, is given by 
    \begin{align}\label{eq:gamma_0}
        \Gamma_{0l,1}:=S\left(\sum_{j=1}^{m} p_{0j}\cdot\exp [-i \theta_{l} \mathrm{sgn}(\alpha_{0j}) P_{0j} ]\bullet \bm{1} \right),~~~
        \Gamma_{0l,2}:=S\left(\sum_{j=1}^{m} p_{0j} \cdot  \bm{1}\bullet {\exp [-i \theta_{l} \mathrm{sgn}(\alpha_{0j}) {P_{0j}} ]^\dagger}\right),
    \end{align}
    and for $k=1,2,...,K$,
    \begin{align}\label{eq:Gamma12}
        \Gamma_{kl,1}:=S(\mathcal{B}^{(\rm approx)}_{kl}),~~~\Gamma_{kl,2}:=S\left(\frac{\mathcal{R}_{kl}}{\|R_{kl}\|_{\rm pauli}}\right)=S\left(\sum_{\mu\in\mathrm{S}_{R_{kl}}}p_{R_{kl}}(\mu) \cdot V_{R_{kl}}(\mu)\bullet U^\dagger_{R_{kl}}(\mu)\right),
    \end{align}
    \begin{equation}\label{eq:G3_convex_pauli}
        \Gamma_{kl,3}:=S\left((-1)\frac{L_k^\dagger L_k}{\alpha_k^2}\bullet \frac{L_k^\dagger L_k}{\alpha_k^2}\right)
        =S\left(\sum_{j_1,j_2,j_3,j_4=1}^M p_{kj_1}p_{kj_2}p_{kj_3}p_{kj_4}e^{i(-\theta_{kj_1}+\theta_{kj_2}-\theta_{kj_3}+\theta_{kj_4}+\pi)}P_{kj_1}P_{kj_2}\bullet P_{kj_3}P_{kj_4}\right).
    \end{equation}
    Therefore, we arrive at the following decomposition of the dynamical map $e^{(t/r)G}$:
    \begin{align}\label{eq:expansion_explicit_form_tr}
        e^{(t/r)G} &=\sum_{l=0}^{\infty} \frac{(t/r)^{2l} }{(2l)!}G^{2l}\left(\bm{1} \otimes \bm{1} +\frac{(t/r)}{2l+1}G\right)\notag\\
        &=\sum_{l=0}^{\infty} \frac{(t/r)^{2l} \|\mathcal{L}\|^{2l}_{\rm pauli}}{(2l)!}\left(\frac{G}{\|\mathcal{L}\|_{\rm pauli}}\right)^{2l}\left(\bm{1} \otimes \bm{1} +\frac{(t/r)}{2l+1}G\right)\notag\\
        &=\sum_{l=0}^{\infty} \frac{(t/r)^{2l} \|\mathcal{L}\|^{2l}_{\rm pauli}}{(2l)!}\left(2\frac{\alpha_0}{\alpha} \sqrt{1 + \tau_l^2}+\sum_{k'=1}^K \frac{\alpha_{k'}^2}{\alpha}\left(1+\|R_{k'l}\|_{\rm pauli}+\frac{\tau_l^2}{4}\right)\right)\notag\\
        &~~~~~\times \sum_{k=0}^K \sum_{\nu=1}^3 q_{kl} p_{\Gamma,kl,\nu} 
        \left(\frac{G}{\|\mathcal{L}\|_{\rm pauli}}\right)^{2l}
        \Gamma_{kl,\nu}.
    \end{align}
    This means that the transfer matrix of the dynamical map $e^{(t/r)\mathcal{L}}$ can be written in the form of $\sum_{u} c_u S(\mathcal{W}_u)$ by implicitly defining $c_u$ and $S(\mathcal{W}_u)$ as
    \begin{equation}
        c_u\to \frac{(t/r)^{2l} \|\mathcal{L}\|^{2l}_{\rm pauli}}{(2l)!}\left(2\frac{\alpha_0}{\alpha} \sqrt{1 + \tau_l^2}+\sum_{k'=1}^K \frac{\alpha_{k'}^2}{\alpha}\left(1+\|R_{k'l}\|_{\rm pauli}+\frac{\tau_l^2}{4}\right)\right)q_{kl} p_{\Gamma,kl,\nu}\geq 0,
    \end{equation}
    \begin{equation}
        S(\mathcal{W}_u)\to \left(\frac{G}{\|\mathcal{L}\|_{\rm pauli}}\right)^{2l}
        \Gamma_{kl,\nu}.
    \end{equation}
    In addition, the sum of $c_u\geq 0$ can be evaluated as
    \begin{align}\label{suppleeq:sum_of_cu}
        \sum_{u}c_u&=\sum_{l=0}^{\infty} \frac{(t/r)^{2l} \|\mathcal{L}\|^{2l}_{\rm pauli}}{(2l)!}\left(2\frac{\alpha_0}{\alpha} \sqrt{1 + \tau_l^2}+\sum_{k'=1}^K \frac{\alpha_{k'}^2}{\alpha}\left(1+\|R_{k'l}\|_{\rm pauli}+\frac{\tau_l^2}{4}\right)\right) \notag\\ 
        &=\sum_{l=0}^{\infty} \frac{(t/r)^{2l} \|\mathcal{L}\|^{2l}_{\rm pauli}}{(2l)!}\left(\frac{2\alpha_0}{\alpha} \sqrt{1 + \tau_l^2}+\left(\sum_{k'=1}^K \frac{\alpha_{k'}^2}{\alpha}\right)\left(1+\frac{\tau_l^2}{2}+\frac{\tau_l^3}{2}+\frac{5\tau_l^4}{64}+\frac{\tau_l^5}{32}+\frac{\tau_l^6}{256}\right)\right) \notag\\ 
        &\leq \sum_{l=0}^{\infty} \frac{(t/r)^{2l} \|\mathcal{L}\|^{2l}_{\rm pauli}}{(2l)!}\left(1+\frac{(\tau'_l)^2}{2}+\frac{(\tau'_l)^3}{2}+\frac{5(\tau'_l)^4}{64}+\frac{(\tau'_l)^5}{32}+\frac{(\tau'_l)^6}{256}\right)\notag\\ 
        &\leq {\rm exp}\left(\frac{c\|\mathcal{L}\|^2_{\rm pauli} t^2}{r^2}\right)
    \end{align}
    for any constant value $c\geq 1.66$.
    The first inequality follows from the facts: $\sqrt{1+x^2}\leq 1+x^2/2$ for any $x\geq 0$ and $\tau_l':=\|\mathcal{L}\|_{\rm pauli}(t/r)/(2l+1)\geq \tau_l$. 
    The final inequality follows from Lemma~\ref{lem:wan_gammma_inequality}. 
    Using the notation $e^{(t/r)G}=\sum_u c_u S(\mathcal{W}_u)$, we obtain the full decomposition of $e^{tG}=S(e^{t\mathcal{L}})$ as
    \begin{equation}
    \label{suppleeq:exp t G}
        e^{tG}=\left(e^{(t/r)G}\right)^r=\sum_{u_1,...,u_r} c_{u_1}\cdots c_{u_r} S(\mathcal{W}_{u_1})\cdots S(\mathcal{W}_{u_r}).
    \end{equation}
    The sum of the coefficients $c_{u_1}\cdots c_{u_r}$ satisfies 
    \begin{equation}
        \sum_{u_1,...,u_r} c_{u_1}\cdots c_{u_r} = \left(\sum_u c_u \right)^r\leq {\rm exp}\left(\frac{c\|\mathcal{L}\|^2_{\rm pauli} t^2}{r}\right).
    \end{equation}
    This completes the proof of Eq.~\eqref{suppleeq:thm_norm_bound}, where we particularly choose $c=2>1.66$.

    Finally, we prove that the transfer matrix form Eq.~\eqref{suppleeq:exp t G} can be reduced to the expression with the
    $(n+1)$-qubit CPTN maps as Eq.~\eqref{suppleeq:explicit decomposition}.
    At the same time, this will also reveal a concrete simulation method of the above decomposition \eqref{suppleeq:exp t G} on quantum circuits without doubling the target system qubit size.
    The key observation is that for any $u$, the superoperator $\mathcal{W}_u$ can be simulated by using a generalized Hadamard test.
    The main components of constructing a quantum circuit for $\mathcal{W}_u$ are the following type-I, II, III circuits illustrated in Fig.~\ref{fig:three_type_circuits}.
    We here summarize the properties of the quantum circuits; see Appendix~\ref{apdx:construct_circuits} for more details.

    \begin{figure}
        \centering
        \begin{tabular}{c|c}
                Superoperator & Circuit Diagram \\ 
                \hline
                \hline
                ${G}/{\|\mathcal{L}\|_{\rm pauli}},\Gamma_{kl,2},\Gamma_{kl,3}$~~~&
                \begin{quantikz}[]
                \\
                \lstick{anc} &\qwbundle{1} & \gate[2]{\mathcal{U}_{\rm I}} & \qw\\
                \lstick{sys} &\qwbundle{n} & \ghost{\mathcal{U}_{\rm I}} & \qw
                \\
                \end{quantikz}
                $=\sum_j p_j~\times$
                \begin{quantikz}
                \\
                \lstick{anc}  & \octrl{1} & \ctrl{1} & \gate{{\rm Phase}(e^{i\theta_j})} &\qw \\
                \lstick{sys}  & \gate{P_j}  & \gate{Q_j} & \qw &\qw
                \\
                \end{quantikz}\\ 
                \hline
                $\Gamma_{kl,1}=S(\mathcal{B}_{kl}^{(\rm approx)})$~~~&
                \begin{quantikz}
                \\
                \lstick{anc} &\qwbundle{1} & \qw & \qw & \qw & \qw & \qw\\
                \lstick{sys} &\qwbundle{n} & \qw & \qw & \gate[2]{\mathcal{U}_{\rm II}} & \qw & \qw\\
                \lstick{$\ket{\bm{0}}$} & \qwbundle{3+\lceil \log_2 M\rceil} & \qw  & \qw & \ghost{\mathcal{U}_{\rm II}} & \qw & \meter{} & \cw
                \\
                \end{quantikz}
                $=$
                \begin{quantikz}
                \\
                \lstick{anc} & \qw & \qw \\
                \lstick{sys} & \gate[1]{{\rm CPTN:}~\mathcal{B}_{kl}^{(\rm approx)}} & \qw
                \\
                \end{quantikz}\\
                \hline
                $\Gamma_{0l,\nu}~(\nu=1,2)$~~~&
                \begin{quantikz}[]
                \\
                \lstick{anc} &\qwbundle{1} & \gate[2]{\mathcal{U}_{\rm III}} & \qw\\
                \lstick{sys} &\qwbundle{n} & \ghost{\mathcal{U}_{\rm III}} & \qw
                \\
                \end{quantikz}
                $=\sum_{j=1}^m p_{0j}~\times$
                \begin{quantikz}
                \\
                \lstick{anc}  & \gate{X^{\nu}} & \ctrl{1}  &\gate{X^{\nu}} &\qw\\
                \lstick{sys}  &\qw & \gate{e^{-i\theta_l {\rm sgn}(\alpha_{0j})P_{0j}}}  & \qw &\qw
                \\
                \end{quantikz}\\
                \hline
        \end{tabular}
        \caption{Quantum circuits for the superoperators in $\mathcal{W}_u$.
        The phase gate is defined as $e^{i\theta_j}\ket{0}\bra{0}+\ket{1}\bra{1}$. The index $k$ runs from 1 to $K$.
        The diagram of the CPTN map $\mathcal{B}_{kl}^{(\rm approx)}$ is detailed in Remark~\ref{rem:cptn}.}
        \label{fig:three_type_circuits}
    \end{figure}

    \begin{itemize}
        \item The type-I circuit $\mathcal{U}_{\rm I}$ is a mixed unitary channel that has the following action for any linear operator $A$ on the $(n+1)$-qubit system
        \begin{equation}
        \mathcal{U}_{\rm I}:
        \begin{pmatrix}
            A_{00}&A_{01}\\    
            A_{10}&A_{11}
        \end{pmatrix}
        \mapsto 
        \begin{pmatrix}
            *&\Phi(A_{01})\\    
            \mathcal{J}\circ \Phi \circ \mathcal{J}(A_{10})&*
        \end{pmatrix},
        \end{equation}
        where $A_{ij}:=(\bra{i}_{\rm anc}\otimes \bm{1}) A (\ket{j}_{\rm anc}\otimes \bm{1})$, and $\mathcal{J}$ is the anti-linear map defined as
        \begin{equation}
        \mathcal{J}:A\mapsto \mathcal{J}(A):=A^\dagger.
        \end{equation}
        The $n$-qubit superoperator $S(\Phi)$ becomes $S(\mathcal{L})/\|\mathcal{L}\|_{\rm pauli}$, $\Gamma_{kl,2}$, $\Gamma_{kl,3}$ by choosing the corresponding probability $p_j$, complex phase ${e^{i\theta_j}}$, and $n$-qubit Pauli strings $(P_j,Q_j)$ according to Eqs.~\eqref{eq:G_convex_pauli}, \eqref{eq:R_convex_pauli}, \eqref{eq:G3_convex_pauli}, respectively.

        \item The type-II circuit 
        $\mathcal{I}_{\rm anc}\otimes \mathcal{U}_{\rm II}$ effectively simulates the CPTN map $\mathcal{I}_{\rm anc}\otimes \mathcal{B}^{(\rm approx)}_{kl}$ by measuring (or discarding) the additional ancilla qubits by the computational basis measurement.
        $\mathcal{I}_{\rm anc}$ denotes the identity channel.
        The TN property arises from the multiplication of 0 to the final outputs if the measurement result is not all zero.
        Also, the action of $\mathcal{I}_{\rm anc}\otimes \mathcal{B}^{(\rm approx)}_{kl}$ can be written as
        \begin{equation}
        \mathcal{I}_{\rm anc}\otimes \mathcal{B}^{(\rm approx)}_{kl}:
        \begin{pmatrix}
            A_{00}&A_{01}\\    
            A_{10}&A_{11}
        \end{pmatrix}
        \mapsto 
        \begin{pmatrix}
            *&\mathcal{B}^{(\rm approx)}_{kl}(A_{01})\\    
            \mathcal{B}^{(\rm approx)}_{kl} (A_{10})&*
        \end{pmatrix}.
        \end{equation}
        The quantum circuit $\mathcal{U}_{\rm II}$ with the classical post-processing instruction can be constructed by Lemma~\ref{lemma:oaa_for_B0B1} from the description of $\mathcal{B}_{kl}$.
        This circuit needs to introduce an (at most) additional $3+\lceil\log_2 M\rceil$ qubits for any $k,l$.
        After the (mid-circuit) measurement for the additional ancilla qubits, these qubits are reset to the initial state for the next type-II circuit.
        
        \item The type-III circuit $\mathcal{U}_{\rm III}$ is a mixed unitary channel that has the following action for any linear operator $A$ on the $(n+1)$-qubit system
        \begin{equation}
        \mathcal{U}_{\rm III}:
        \begin{pmatrix}
            A_{00}&A_{01}\\    
            A_{10}&A_{11}
        \end{pmatrix}
        \mapsto 
        \begin{pmatrix}
            *&\Phi'(A_{01})\\    
            \mathcal{J}\circ \Phi' \circ \mathcal{J}(A_{10})&*
        \end{pmatrix},
        \end{equation}
        where the $n$-qubit superoperator $S(\Phi')$ becomes $\Gamma_{0l,\nu}$ $(\nu=1,2)$ by choosing the index $\nu=1,2$.
        
    \end{itemize}
    Thus, 
    $\mathcal{U}_{\rm I}$, $\mathcal{U}_{\rm II}$, or $\mathcal{U}_{\rm III}$ according to the superoperator $\Gamma_{kl,\nu}$ followed by $2l$ sequential applications of $\mathcal{U}_{\rm I}$ for the superoperator $\mathcal{L}/\|\mathcal{L}\|_{\rm pauli}$ yields a quantum circuit $\widetilde{\mathcal{W}}_u$ such that
    \begin{equation}
        \widetilde{\mathcal{W}}_u:
        \begin{pmatrix}
            A_{00}&A_{01}\\    
            A_{10}&A_{11}
        \end{pmatrix}
        \mapsto 
        \begin{pmatrix}
            *&\mathcal{W}_u(A_{01})\\    
            \mathcal{J}\circ \mathcal{W}_u \circ \mathcal{J}(A_{10})&*
        \end{pmatrix},
    \end{equation}
    where we used the facts $\mathcal{J}^2=\mathcal{I}$ and $\mathcal{J}\circ\mathcal{B}_{kl}^{(\rm approx)}\circ\mathcal{J}=\mathcal{B}_{kl}^{(\rm approx)}$ for any $k,l$.
    An example of $\widetilde{\mathcal{W}}_u$ is shown in Fig.~\ref{fig:example_Wu}.
    This result immediately leads to
    \begin{equation}
        \widetilde{\mathcal{W}}_{u_1} \circ\cdots \circ\widetilde{\mathcal{W}}_{u_r}:
        \begin{pmatrix}
            A_{00}&A_{01}\\    
            A_{10}&A_{11}
        \end{pmatrix}
        \mapsto 
        \begin{pmatrix}
            *&\mathcal{W}_{u_1} \circ \cdots \circ \mathcal{W}_{u_r}(A_{01})\\    
            \mathcal{J}\circ \mathcal{W}_{u_1} \circ \cdots \circ \mathcal{W}_{u_r} \circ \mathcal{J}(A_{10})&*
        \end{pmatrix},
    \end{equation}
    where we used $\mathcal{J}^2=\mathcal{I}$ again.
    Therefore, we obtain 
    \begin{align}
        &\sum_{u_1,u_2,...,u_r}c_{u_1}c_{u_2}\cdots c_{u_r} {\rm Tr}_{\rm anc}\left[\left(X_{\rm anc}\otimes \bm{1}\right)\widetilde{\mathcal{W}}_{u_1}\circ\cdots \circ\widetilde{\mathcal{W}}_{u_r}\left(|+\rangle\langle +|_{\rm anc}\otimes \bullet\right)\right]\notag\\
        &=\sum_{u_1,u_2,...,u_r}c_{u_1}c_{u_2}\cdots c_{u_r} \frac{1}{2}\left(\sum_{q=0}^1 \mathcal{J}^q\circ \mathcal{W}_{u_1} \circ \cdots \circ \mathcal{W}_{u_r}\circ \mathcal{J}^q\right)(\bullet)\notag\\
        &=\frac{1}{2}\left\{\sum_{q=0}^1 \sum_{u_1,u_2,...,u_r}c_{u_1}c_{u_2}\cdots c_{u_r} \mathcal{J}^q\circ \mathcal{W}_{u_1} \circ \cdots \circ \mathcal{W}_{u_r}\circ \mathcal{J}^q\right\}(\bullet)\notag\\
        &=\frac{1}{2}\left\{\sum_{q=0}^1 \mathcal{J}^q\circ e^{t\mathcal{L}}\circ \mathcal{J}^q\right\}(\bullet)=e^{t\mathcal{L}}(\bullet),
    \end{align}
    where for the final equality  we used $\mathcal{J}\circ e^{t\mathcal{L}} \circ \mathcal{J} = e^{t\mathcal{L}} $ due to the Hermitian preserving property of $e^{t\mathcal{L}}$. 
    This completes the proof of Theorem~\ref{thm: main}. 

    \begin{figure}[tbp]
    \centering
    \begin{quantikz}
        \\
        \lstick{anc} &\qwbundle{1} & \qw & \qw & \qw & \qw  & \gate[2]{\mathcal{U}^{2l}_{\rm I}} & \qw & \qw \\
        \lstick{sys} &\qwbundle{n} & \qw & \qw & \gate[2]{\mathcal{U}_{\rm II}} & \qw & \ghost{{\mathcal{U}^{2l}_{\rm I}}} & \qw & \qw\\
        \lstick{$\ket{\bm{0}}$} & \qwbundle{3+\lceil \log_2 M\rceil} & \qw  & \qw & \ghost{\mathcal{U}_{\rm II}} & \qw & \meter{} & \cw \push{~~\ket{\bm{0}}~} &\qw
        \\
    \end{quantikz}
    \caption{Quantum circuit $\widetilde{\mathcal{W}}_u$ for $S(\mathcal{W}_u)= \left(\frac{G}{\|\mathcal{L}\|_{\rm pauli}}\right)^{2l}\Gamma_{kl,\nu}$ when $\Gamma_{kl,\nu}=S(\mathcal{B}^{\rm (approx)}_{kl})$. The $\ket{\bm{0}}$ after the measurement denotes the qubit reset operation.}
    \label{fig:example_Wu}
    \end{figure}    
\end{proof}

We here remark on the effective simulation of CPTN maps using quantum circuits and classical post-processing.

\begin{remark}
    [Effective simulation of CPTN map $\mathcal{B}_{kl}^{(\rm approx)}$]\label{rem:cptn}
    From the Lemma~\ref{lemma:oaa_for_B0B1}, the quantum circuit $\mathcal{U}_{\rm II}$ with measurement for $\mathcal{B}_{kl}^{(\rm approx)}$ has the following form
    \begin{center}
    \begin{quantikz}
        \\
        \lstick{sys} &\qwbundle{n} & \qw & \qw & \gate[2]{\mathcal{U}_{\rm II}} & \qw & \qw\\
        \lstick{$\ket{\bm{0}}$} & \qwbundle{3+\lceil \log_2 M\rceil} & \qw  & \qw & \ghost{\mathcal{U}_{\rm II}} & \qw &\meter{}
        \\
    \end{quantikz}
    =
    \begin{quantikz}
        \\
        \lstick{sys}            & \qwbundle{n}                       & \qw  & \qw & \gate[3]{U}& \qw & \gate[3]{U^\dagger} &\qw & \gate[3]{U} & \qw \\
        \lstick{$\ket{\bm{0}}$} & \qwbundle{2+\lceil \log_2 M\rceil} & \qw  & \qw & \ghost{U}  & \gate{2\tilde{\Pi}-\bm{1}} &\ghost{U^\dagger} &\gate[2]{\bm{1}-2{\Pi}}& \ghost{U} &\meter{} \\
        \lstick{$\ket{0}_{\rm P}$}      & \qwbundle{1}                       & \qw  & \qw & \ghost{U}  & \qw &\ghost{U^\dagger}&\ghost{\bm{1}-2{\Pi}}& \ghost{U}&\meter{} 
        \\
    \end{quantikz}
    \end{center}
    where $\tilde{\Pi}:=\bm{1}\otimes \ket{\bm{0}}\bra{\bm{0}}\otimes I_{\rm P}$, ${\Pi}:=\bm{1}\otimes \ket{\bm{0}}\bra{\bm{0}}\otimes \ket{0}\bra{0}_{\rm P}$, and $U$ is a quantum circuit satisfying 
    \begin{equation}
        \tilde{\Pi}U\Pi = \frac{1}{2} \left(B_{kl,0}\otimes \ket{\bm{0}}\bra{\bm{0}} \otimes \ket{0}\bra{0}_{\rm P}+B_{kl,1}\otimes \ket{\bm{0}}\bra{\bm{0}} \otimes \ket{1}\bra{0}_{\rm P}\right).
    \end{equation}
    The final measurement on $3+\lceil \log_2 M\rceil$ qubits is given by the following two-outcome POVM $\{\Pi_{b}\}_{b=0,1}$
    \begin{equation}\label{eq:midcircuit_M_POVM}
        \Pi_0:=\tilde{\Pi},~~~\Pi_1:=\bm{1}-\tilde{\Pi},
    \end{equation}
    which can be simulated by the computational basis measurement.
    Let us consider a general case that we measure an observable $O$ on the target system after the process for $\mathcal{B}_{kl}^{\rm (approx)}$ i.e., $\mathcal{U}_{\rm II}$ followed by the measurement of $\{\Pi_{b}\}_{b=0,1}$.
    In a single trial of the whole process, we obtain measurement outcomes $(o,b)$ for the observable $O$ and the POVM $\{\Pi_{b}\}_{b=0,1}$, respectively.
    Using the measurement outcomes, we can estimate ${\rm Tr}[O\mathcal{B}_{kl}^{\rm (approx)}(\rho)]$ for any input state $\rho$, in an unbiased manner. 
    That is, the mean of $o\delta_{b0}$ is given by
    \begin{align}
        \mathbb{E}[o\delta_{b0}]&=\sum_{o,b}o\delta_{b0}{\rm Tr}\left[\ket{o}\bra{o} \otimes\bm{1} \cdot\Pi_{b}\cdot \mathcal{U}_{\rm II}(\rho\otimes \ket{\bm{0}}\bra{\bm{0}})\right]\notag\\
        &=\sum_{o}o\mathrm{Tr} \left[\ket{o}\bra{o}\mathrm{ Tr_{\overline{sys}}} \left[\tilde{\Pi}~\mathcal{U}_{\mathrm{ II}} (\rho \otimes \ket{\bm{0}}\bra{\bm{0}}) \right] \right] \notag\\
        &={\rm Tr}\left[O\mathcal{B}_{kl}^{(\rm approx)}(\rho)\right],
    \end{align}
    where ${\rm Tr}_{\overline{{\rm sys}}}$ denotes the partial trace over all the qubits except for the target system, and $\{\ket{o}\}$ are the eigenstates of $O$.
    The final equality follows from Eq.~\eqref{eq:bapproxc} in the next subsection.
    In this way, we can effectively simulate CPTN maps, using classical post-processing for final outputs, without an additional increase in the estimators' variance.
\end{remark}

\subsection{Exact and efficient simulation of general CP maps}\label{apdx:exact_efficient_oaa}

In the proof of Theorem~\ref{thm: main} in the previous subsection, we employed the CP map $\mathcal{B}$ in the operator form as follows:
\begin{equation}
    S(\mathcal{B}) = \overline{B_0} \otimes B_0 + \overline{B_1} \otimes B_1,
\end{equation}
where $B_0$ and $B_1$ are defined as, for $\tau>0$,
\begin{gather}\label{suppleeq:def_of_b0b1}
    B_0 = \bm{1} -  \frac{1}{2}\frac{L^\dagger L}{\alpha^2} \tau, ~~~~ B_1 = \frac{L}{\alpha} \sqrt{\tau}.
\end{gather}
Here, $L$ is a jump operator described by a linear combination of unitaries $U_i$ and coefficients $\alpha_i > 0$:
\begin{equation}\label{suppleeq:def_of_l}
    L = \sum_{i=1}^M \alpha_i U_i, ~~~~ \alpha = \sum_{i=1}^M \alpha_i.
\end{equation}

Cleve et al.~\cite{Cleve2016-yj} introduced a quantum algorithm,
using LCU for channels and OAA for isometry, to deterministically implement a target CPTP map as we reviewed in Section~\ref{suppli_subsec:lcu_oaa}.
In addition, the algorithm is available for a CP map $\mathcal{A}$ without a TP property.
Initially, LCU for channels itself can effectively simulate the target CP map exactly up to normalization $p \in [0, 1]$ as $p \mathcal{A}$, by measuring an observable in ancilla qubits (e.g., measuring the observable $\bm{1}_{\rm P}\otimes \ket{\bm{0}}\bra{\bm{0}}$ in Eq.~\eqref{suppleeq:channellcuisometry}).
It means that we require the multiplication of $1/p$ to the final outputs for compensating the normalization, which leads to classical sampling overhead.
To avoid the overhead, OAA is used to increase the normalization factor toward 1.
However, it introduces an approximation error between the resulting map $\mathcal{A}^{\mathrm{(approx)}}$ and the target map $\mathcal{A}$, i.e., $\mathcal{A}^{\mathrm{(approx)}} \approx \mathcal{A}$, unless $\mathcal{A}$ is a CPTP map.
We call this procedure a {\it standard method} for deterministic implementation of an approximate CP map $\mathcal{A}^{(\rm approx)}\approx \mathcal{A}$.

Now we seek an exact and efficient (i.e. low-overhead) simulation method of the CP map $\mathcal{B}$ to prove Theorem~\ref{thm: main}.
However, the $1/p$ overhead from LCU for channels is too large, and an approximation error by OAA is not acceptable for our case.
Here, we first identify the approximated map $\mathcal{B}^{\mathrm{(approx)}}$ and its error that is implemented by the standard method.
Then, we propose a new method to simulate $\mathcal{B}$ exactly by developing a recovery operation for the approximation error with the help of classical sampling, while benefiting from the overhead reduction by OAA.

\begin{lem}[Exact form of approximated CPTN map by the standard method]\label{lemma:oaa_non_isometry}
    Let $B_k$ be linear operators for $k=1, \cdots, K$, and we assume that each $B_k$ can be written by a linear combination of unitary operators as $B_k := \sum_{i} c_{ki} U_{ki}$ with $c_{ki} \ge 0$ satisfying $\sum_k (\sum_i c_{ki})^2 \leq 4$.
    Also, we write $\sum_{k} B^\dagger_k B_k = \bm{1} + \delta D$ for some $\delta>0$ and a Hermitian operator $D$ with $\| D \|_\infty \le 1$.
    Then the standard method, described in Section~\ref{suppli_subsec:lcu_oaa},
    for a target map $\mathcal{B} = \sum_{k} B_k \bullet B_k^\dagger$
    yields the CPTN map $\mathcal{B}^{\mathrm{(approx})}$ such that
    \begin{equation}
    \mathcal{B}^{\mathrm{(approx)}} = \sum_{k} B_k' \bullet (B_k')^\dagger,
    ~~~B_k' := B_k\left( \bm{1} - \frac{\delta}{2}D  \right).
    \end{equation}
    Furthermore, the exact form of $\mathcal{B}^{\mathrm{(approx)}}$ is given by
    \begin{equation}\label{eq:generalformofR}
        \mathcal{B}^{\mathrm{(approx)}}
        = \mathcal{B} - \mathcal{R},~~~
        \mathcal{R} = \sum_{k} \left(\frac{\delta}{2} B_k \bullet DB^\dagger_k
        + \frac{\delta}{2} B_k D \bullet B^\dagger_k
        - \frac{\delta^2}{4} B_kD  \bullet DB^\dagger_k \right).
    \end{equation}
\end{lem}
\begin{proof}
    We follow the standard method mentioned in Subsection~\ref{suppli_subsec:lcu_oaa} and carefully evaluate the approximation error.
    Since $B_{k}$ are given by linear combinations of unitaries, we can construct unitaries $W_k$ satisfying 
    \begin{equation}
        W_{k}\ket{\bm{0}}\ket{\psi}=\frac{1}{\sum_{i} c_{ki}}\ket{\bm{0}}B_k \ket{\psi}+|{\tilde{\perp}_{\psi}}\rangle,
    \end{equation}
    where $\ket{\bm{0}}$ is an initial state on an ancilla system, and $|{\tilde{\perp}_{\psi}}\rangle$ denotes an unnormalized state satisfying $\bra{\bm{0}}\cdot |{\tilde{\perp}_{\psi}}\rangle=0$.
    Let $U$, $\Pi$, $\Tilde{\Pi}$, and $W$ be given in the same manner as in the original OAA procedure described in Subsection~\ref{suppli_subsec:lcu_oaa};
    \begin{gather}
        U_{\rm R,P}\ket{{0}}_{\rm P}=\frac{1}{\sqrt{\sum_k (\sum_i c_{ki})^2}}\sum_{k} (\sum_i c_{ki})\ket{k}_{\rm P},~~~U=R_y\otimes \left(\sum_{k} \ket{k}\bra{k}_{\rm P}\otimes W_k\right)U_{\rm R,P}, \notag \\
        \tilde{\Pi}=\ket{0}\bra{0}\otimes \bm{1}_{\rm P}\otimes\ket{\bm{0}}\bra{\bm{0}}\otimes \bm{1},~~~{\Pi}=\ket{0}\bra{0}\otimes \ket{0}\bra{0}_{\rm P}\otimes \ket{\bm{0}}\bra{\bm{0}}\otimes  \bm{1},\notag \\
        W = \sqrt{p}\Tilde{\Pi} U \Pi,~~~p=1/4.
    \end{gather}
    
    Now we evaluate the error in OAA procedure when $W$ is not a partial isometry.
    The action of the OAA unitary $V_{\mathrm{OAA}}:=U(\bm{1}-2\Pi)U^\dagger (2\tilde{\Pi}-\bm{1})U$ satisfies
    \begin{align}
        \tilde{\Pi}V_{\mathrm{OAA}}\Pi
        & = -4\sqrt{p}^3 WW^\dagger W+3\sqrt{p}W
        =-4\sqrt{p}^3 W\left(\bm{1} +  \delta \tilde{D} \right)+3\sqrt{p}W \notag \\
        &=  W\left(\bm{1} - \frac{\delta}{2}  \tilde{D} \right),
    \end{align}
    where we used $W^\dagger W = \ket{0}\bra{0} \otimes \ket{0}_{\rm P} \bra{0}_{\rm P} \otimes \ket{\bm{0}} \bra{\bm{0}} \otimes (\sum_k B_k^\dagger B_k) = \Pi (\bm{1} + \delta \tilde{D})$ with $\tilde{D} := \bm{1}_{\overline{\mathrm{sys}}}\otimes D$.
    Therefore, the standard method prepares $W(\bm{1} - \delta\Tilde{D}/2)$, and thus, it yields the CPTN map $\mathcal{B}^{\mathrm{(approx)}}$ as
    \begin{equation}\label{eq:bapproxc}
    \Tr_{\overline{\mathrm{sys}}}[\tilde{\Pi}V_{\mathrm{OAA}} (\ketbra{0} \otimes \ketbra{0}_{\mathrm{P}} \otimes \ketbra{\bm{0}}\otimes \bullet )V_{\mathrm{OAA}}^{\dagger}]
    = \sum_{k} B_k' \bullet (B_k')^\dagger=\mathcal{B}^{\mathrm{(approx)}},
    \end{equation}
    where $\Tr_{\overline{\mathrm{sys}}}$ denotes the partial trace over all the qubits except for the target system, and
    \begin{equation}
        B_k' = B_k\left(\bm{1} - \frac{\delta}{2} D \right).
    \end{equation}
    The projection $\tilde{\Pi}$ and the tracing out can be realized by measuring the ancilla qubits except the purifier and the classical post-processing where we multiply the final output by zero whenever the ancilla-measurement result does not coincide with 
    $\ket{0}\ket{\bm{0}}$.
    We remark that this post-processing does not introduce additional cost differently from the post-selection since we purpose to implement a CPTN map.
    
    The difference $\mathcal{R}$ between $\mathcal{B}$ and $\mathcal{B}^{\rm (approx)}$ is determined by the direct calculation as follows:
    \begin{align}
    \mathcal{R}
        &= \mathcal{B} - \mathcal{B}^{\mathrm{(approx)}} \notag \\
        &= \sum_{k} \left(B_k \bullet B_k^\dagger - B_k \left(\bm{1} - \frac{\delta}{2}D \right) \bullet \left(\bm{1} - \frac{\delta}{2}D \right) B_k^\dagger \right) \notag \\
        &= \sum_{k} \left(\frac{\delta}{2} B_k \bullet DB^\dagger_k + \frac{\delta}{2} B_k D \bullet B^\dagger_k - \frac{\delta^2}{4} B_kD  \bullet DB^\dagger_k \right).
    \end{align}
\end{proof}

While the CPTN property of $\mathcal{B}^{\mathrm{(approx)}}$ necessarily follows from its construction \eqref{eq:bapproxc}, we remark that this CPTN property is tied with the assumption $\sum_k (\sum_i c_{ki})^2 \leq 4$ required for $R_y$ to be well-defined unitary, as shown below.
Let us observe that the assumption $\sum_k (\sum_i c_{ki})^2 \leq 4$ implies
\begin{align}
\|\bm{1} + \delta D\|_{\infty}
= \left\| \sum_k B_k^\dagger B_k \right\|_{\infty} 
\le \sum_k \| B_k^\dagger B_k \|_{\infty}
\le \sum_k \left\| \sum_i c_{ki}  U_{ki} \right\|_{\infty}^2
\le \sum_k (\sum_i c_{ki} \| U_{ki} \|_{\infty} )^2
\le \sum_k(\sum_i c_{ki})^2
\le 4.
\end{align}
Then, we have $\bm{1} + \delta D \le \|\bm{1} + \delta D\|_{\infty} \times \bm{1} \le 4\times  \bm{1}$, which implies $3 - \delta D \ge 0$. Therefore, we have
\begin{equation}
\sum_k B_k'^{\dagger} B_k' = \left(1 - \frac{\delta}{2} D\right)\sum_k B_k^{\dagger} B_k \left(1 - \frac{\delta}{2} D\right)
= \left(1 - \frac{\delta}{2} D\right)\left(1 + \delta D\right)\left(1 - \frac{\delta}{2} D\right)
= \bm{1} - \frac{\delta^2}{4} D^2 (3 - \delta D)
\le \bm{1},
\end{equation}
which is equivalent to the CPTN property.

Although the procedure in Lemma~\ref{lemma:oaa_non_isometry} follows the standard method, we additionally obtain the exact representation of the CPTN map $\mathcal{B}^{\mathrm{(approx)}}$ and the approximation-error map $\mathcal{R}$.
This identification opens the way to exactly and efficiently simulate the general CP map $\mathcal{B}$ via constructing $\mathcal{B}^{\mathrm{(approx)}}$ by the standard method and compensating its error $\mathcal{R}$ by classical sampling.
This technique is one of our key contributions.
Furthermore, this error recovery operation can be efficient as shown below.

\begin{lem}[Decomposition of $\mathcal{R}$]\label{lemma:generalRsampling}
Let $\{ B_k \}$, $D$, $\delta$, $\mathcal{R}$ be as in Lemma~\ref{lemma:oaa_non_isometry}.
In addition, we assume that $D$ can be described as a linear combination of unitaries form:
\begin{equation}
D = \sum_i q_i V_i
\end{equation}
with unitaries $V_i$ and  $q_i >0$ satisfying $\sum_i q_i = 1$.
Then, $\mathcal{R}$ can be decomposed into
a convex combination of superoperators formed as $A_i\bullet B_i^\dagger$ up to a normalization factor $C$, where $A_i$ and $B_i$ are some unitaries.
The normalization factor $C$ is determined as follows
\begin{equation}
C =\frac{1}{p}\left(\delta + \frac{\delta^2}{4}\right)
\end{equation}
where $1/p = \sum_k (\sum_i c_{ki})^2$.
\end{lem}
\begin{proof}
By defining $b_k = \sum_{i} c_{ki}$, we obtain
\begin{align}
\mathcal{R}
    &= \sum_{k} \left(\frac{\delta}{2} B_k \bullet D B^\dagger_k + \frac{\delta}{2} B_k D \bullet B^\dagger_k - \frac{\delta^2}{4} B_kD  \bullet DB^\dagger_k \right) \notag \\
    &= C \sum_{k} \frac{p b_k^2}{\delta  + \frac{\delta^2}{4}}\left(\frac{\delta }{2} \frac{B_k}{b_k}  \bullet D \frac{B_k^\dagger}{b_k}  + \frac{\delta}{2} \frac{B_k}{b_k}   D \bullet \frac{B_k^\dagger}{b_k}  +\frac{\delta^2}{4}\cdot(-1)\cdot\frac{B_k}{b_k} D  \bullet D\frac{B_k^\dagger}{b_k}  \right),
\end{align}
and $C =\frac{1}{p}\left(\delta + \frac{\delta^2}{4}\right)$. Since $B_k/b_k$ and $D$ are convex combinations of unitaries, we can write the terms in the round bracket as a convex combination of superoperators with an asymmetric form of unitaries.
\end{proof}

Importantly, the convex combination of superoperators $\mathcal{R}/C$ can be effectively simulated by the random sampling of Hadamard test circuits.
The concrete procedure for the simulation is provided in the next subsection.
In this method, $C$ corresponds to the sampling overhead to maintain the norm of $\mathcal{R}$.
As a result, we have an exact decomposition of a target CP map $\mathcal{B}$:
\begin{equation}
    \mathcal{B}=(1+C)\times\left(\frac{1}{1+C}\mathcal{B}^{(\rm approx)}+\frac{C}{1+C}\frac{\mathcal{R}}{C}\right).
\end{equation}
The total overhead in this decomposition is given by $1+C$. $C$ corresponds to the sampling overhead to maintain the norm.

The efficiency of methods depends on $p$ and $\delta$.
If we prepare the exact CP map $\mathcal{B}$ by LCU for channels like Eq.~\eqref{suppleeq:channellcuisometry},
we have
\begin{align}
    \sqrt{p} \left(\sum_{k} \ket{k}_{\mathrm{P}} \otimes \ket{\bm{0}} \otimes  B_k \ket{\psi} \right) + \ket{\perp}
\end{align}
for any state $\ket{\psi}$.
With the same post-processing as OAA, this provides a CPTN map $p \mathcal{B}$. For the rescaling, we need additional sampling overhead $1/p$.
By comparing the overheads,
we observe that the method with $\mathcal{B}^{(\rm approx)}$ and the recovery operation $\mathcal{R}$ is more efficient than the approach with only the use of the LCU method when $\delta$ satisfies
\begin{gather}
    1 + C \le \frac{1}{p}, \quad \text{i.e.}  \quad \delta \le 2(-1 + \sqrt{2-p}).
\end{gather}
The OAA with the recovery operation is available for general CP maps, and this has advantages over encoding by simply applying the LCU for channels when $\delta$ is small, i.e., $\mathcal{B}$ is close to a TP map.
This inherits the characteristics of OAA.
Our target map~\eqref{suppleeq:def_of_b0b1} is a good example where this method has a substantial effect.

Now, let us analyze the case in which the map specified by Lindblad-type operators Eqs.~\eqref{suppleeq:def_of_b0b1} and~\eqref{suppleeq:def_of_l}.
This example is a fundamental case for Lindblad simulation, with a single jump operator.
For the case, we provide a detailed description of the implementation of $\mathcal
{B}^{\mathrm{(approx)}}$ and calculate the circuit complexity to verify the efficiency.
Through the following analysis, we assume the access of a unitary gate $W_L$ such that 
    \begin{equation}\label{eq:defofWL}
    (\bra{\bm{0}}\otimes \bm{1}) \cdot W_L \cdot (\ket{\bm{0}}\otimes \bm{1}) = \frac{L}{\alpha}
    \end{equation}
using a $l$ ancilla qubits with initial state $\ket{\bm{0}}$.
Such $W_L$ can be prepared with $l=\lceil \log_2 (M) \rceil $ ancilla qubits by LCU.
We apply Lemma~\ref{lemma:oaa_non_isometry} to $B_0$ and $B_1$ and obtain the following.
\begin{lem}
    [CPTN map $\mathcal{B}^{\rm (approx)}$ for $B_0 \bullet B_0^\dagger + B_1 \bullet B_1^\dagger$]\label{lemma:oaa_for_B0B1}
    Let $L$ be a linear operator defined by Eq.~\eqref{suppleeq:def_of_l}, and $W_L$ be a unitary satisfying Eq.~\eqref{eq:defofWL} requiring $l$ ancilla qubits.
    Let $B_0$ and $B_1$ be linear operators defined by Eq.~\eqref{suppleeq:def_of_b0b1} with $ \tau \in [0,3]$.
    Let $D$ be a Hermitian operator defined by
    \begin{equation}
    \frac{\tau^2}{4} D  := \sum_{k=0}^{1} B_k^\dagger B_k - \bm{1} = \frac{\tau^2}{4}\frac{(L^\dagger L)^2}{\alpha^4}.
    \end{equation}
    Defining modified operators $B_k' = B_k(\bm{1} - \frac{\tau^2}{4} \frac{D}{2})$ for $k=0, 1$, we can effectively simulate the CPTN map $\mathcal{B}^{\rm (approx)}: = \sum_{k=0}^{1}B_k' \bullet (B_k')^\dagger$ with use of a quantum circuit that has the following non-Clifford cost: 
    \begin{itemize}
        \item Single-qubit gates: $12$,
        \item NOT gates controlled by at most $(l+2)$ qubits: $5$,
        \item Unitary gates $W_L$ and its inverse: $3$,
        \item Controlled version of $W_L$ and its inverse: $3$,
        \item Additional ancilla qubits, excluding the target system qubits: $l + 3$,
        \item The computational basis measurements for the $l+3$ ancilla qubits,
    \end{itemize}
    and a classical post-processing for the measurement outcomes.
\end{lem}

To prove the Lemma~\ref{lemma:oaa_for_B0B1}, firstly we propose the explicit circuit implementations of $W_{B_0}$ and $W_{\mathcal{B}}$ that satisfy
\begin{gather}
    (\bra{0} \otimes \bra{\bm{0}} \otimes \bm{1}) \cdot W_{B_0} \cdot (\ket{0}\otimes \ket{\bm{0}} \otimes \bm{1}) = \bm{1} - \frac{L^\dagger L}{2\alpha^2} \tau, \label{eq:BEofB0}\\    
    (\bm{1}_{\rm P}\otimes \bra{\bm{0}}\otimes \bm{1}) \cdot W_{\mathcal{B}} \cdot (\ket{{0}}_{\rm P}\ket{\bm 0}\ket{\psi}) = \frac{1}{\sqrt{1 + \tau}} \sum_{k=0}^{1} \ket{k}_{\rm P}\otimes B_k\ket{\psi},\label{eq:BEofmathcalB}
\end{gather}
respectively.
Actually, the proposed implementation is more efficient than just preparing $L^\dagger L$.

\begin{lem}\label{lemma: efficient_BE_ofB0}
Let $L$ be a linear operator defined by Eq.~\eqref{suppleeq:def_of_l}, and $W_L$ be a unitary satisfying Eq.~\eqref{eq:defofWL} requiring $l$ ancilla qubits.
$\tau$ denotes a time satisfying $ \tau \in [0,4]$.
The unitary $W_{B_0}$ satisfying Eq.~\eqref{eq:BEofB0}, can be constructed by using 
    \begin{itemize}
        \item Single-qubit non-Clifford gates: $2$,
        \item NOT gates controlled by at most $l$ qubits: $1$,
        \item Unitary gate $W_L$ and its inverse: $2$,
        \item Additional ancilla qubits, excluding the target system qubits: $l + 1$.
    \end{itemize}
\end{lem}
\begin{proof}
Let us define the PREPARE circuit for a single qubit as follows:
\begin{equation}
\mathrm{PRE} :\ket{0} \mapsto \sqrt{\frac{\tau}{4}} \ket{0} + \sqrt{1 - \frac{\tau}{4}} \ket{1},
\end{equation}
and the reflection $\mathrm{REF}$ as $\mathrm{REF} = \bm{1} - 2 \ket{\bm{0}}\bra{\bm{0}}$,
on $l$ ancilla qubits.
Then, we consider the following circuit.
\begin{center}
    \begin{quantikz}
        \lstick{$\ket{\psi}$} & & \gate[2]{W_L}& &\gate[2]{W^\dagger_L} &\\
        \lstick{$\ket{\bm{0}}$} &\qwbundle{l} &  & \gate{\mathrm{REF}} &  & \rstick{$\bra{\bm{0}}$} \\
        \lstick{$\ket{0}$} & \qwbundle{1} & \gate{\mathrm{PRE}} & \octrl{-1} & \gate{\mathrm{PRE}^\dagger} & \rstick{$\bra{0}$}
    \end{quantikz}.
\end{center}
With noting that $(\bra{\bm{0}}\otimes \bm{1}) \cdot  W^\dagger_L \cdot \mathrm{REF}\cdot  W_L \cdot  (\ket{\bm{0}}\otimes \bm{1}) = \bm{1} - 2 L^\dagger L / \alpha^2$,
we obtain
\begin{align}
    \ket{0}\ket{\bm{0}}\ket{\psi}
    & \rightarrow \left( \sqrt{\frac{\tau}{4}} \ket{0} + \sqrt{1 - \frac{\tau}{4}} \ket{1} \right)\ket{\bm{0}}\ket{\psi}\notag  \\
    & \rightarrow \sqrt{\frac{\tau}{4}}\ket{0}\left(\bm{1} - 2 \frac{L^\dagger L}{\alpha^2}\right) \ket{\psi} + \sqrt{1 - \frac{\tau}{4}} \ket{1}\ket{\psi} \notag \\
    & \rightarrow \left(\bm{1} - \frac{L^\dagger L}{2\alpha^2} \tau \right) \ket{\psi}
\end{align}
for an arbitrary state $\ket{\psi}$.
Therefore, the circuit provides the $W_{B_0}$ that requires $l+1$-ancilla qubit, and calling $W_L$ and $W^\dagger_L$ once each, and 2 PREPARE circuits.
\end{proof}

We remark that a straightforward LCU implementation of $L^\dagger L / \alpha^2$ requires at most $ \lceil \log_2 M^2 \rceil$ ancilla qubits, whereas our implementation reduces them to $ \lceil \log_2 M \rceil + 1$.
This indicates that Lemma~\ref{lemma: efficient_BE_ofB0} is a more memory-efficient implementation.

Next we prepare the $W_{\mathcal{B}}$ using the proposed $W_{B_0}$ and $W_L$.

\begin{lem}\label{lemma: efficient_BE_ofB0B1}
Let $L$ be a linear operator defined by Eq.~\eqref{suppleeq:def_of_l}, and $W_L$ be a  unitary satisfying Eq.~\eqref{eq:defofWL} requiring $l$ ancilla qubits.
Let $B_0$ and $B_1$ be linear operators defined by Eq.~\eqref{suppleeq:def_of_b0b1}.
$\tau$ denotes a time satisfying $ \tau \in [0,4]$.
The unitary $W_{\mathcal{B}}$ satisfying Eq.~\eqref{eq:BEofmathcalB},
can be constructed by using 
\begin{itemize}
    \item Single-qubit non-Clifford gates: $3$,
    \item NOT gates controlled by at most $(l+1)$ qubits: $1$,
    \item Unitary gate $W_L$: $1$,
    \item Controlled version of $W_L^\dagger$: $1$,
    \item Additional ancilla qubits, excluding system qubits: $l + 2$.
\end{itemize}
\end{lem}
\begin{proof}
Define another PREPARE circuit on a single qubit for a purifier ancilla $\mathrm{P}$ as
\begin{equation}
    \mathrm{PRE'} :\ket{0}_{\mathrm{P}}  \mapsto \frac{1}{\sqrt{1 + \tau}} (\ket{0}_{\mathrm{P}}  + \sqrt{\tau} \ket{1}_{\mathrm{P}} ).
\end{equation}
Referring to Lemma~\ref{lemma: efficient_BE_ofB0}, we design the circuit as:
\begin{center}
    \begin{quantikz}
        \lstick{$\ket{\psi}$} & & \gate[4]{W_{\mathcal{B}}} & \ghost{W} \\
        \lstick{$\ket{\bm{0}}$} &\qwbundle{l} & \ghost{W}& \rstick{$\bra{\bm{0}}$}  \\
        \lstick{$\ket{0}$} & \qwbundle{1} & \ghost{W}& \rstick{$\bra{0}$}\\
        \lstick{$\ket{0}_{\mathrm{P}}$} & \qwbundle{1} &  &
    \end{quantikz}
    $=$
    \begin{quantikz}
        \lstick{$\ket{\psi}$} & & & \gate[2]{W_L}& & &\gate[2]{W_L^\dagger} &\\
        \lstick{$\ket{\bm{0}}$} &\qwbundle{l} & & & \gate{\mathrm{REF}} &  & &\rstick{$\bra{\bm{0}}$} \\
        \lstick{$\ket{0}$} & \qwbundle{1} && \gate{\mathrm{PRE}} & \octrl{-1} & \gate{\mathrm{PRE}^\dagger} && \rstick{$\bra{0}$} \\
        \lstick{$\ket{0}_{\mathrm{P}} $} & \qwbundle{1} & \gate{\mathrm{PRE'}} & & \octrl{-1} & &\octrl{-3} & 
    \end{quantikz}.
\end{center}
The circuit can be considered as the controlled application of $W_L$ or $W_{B_0}$ depending on the ancilla's state.
Then, we obtain
\begin{align}
    \ket{0}_{\mathrm{P}} \ket{0}\ket{\bm{0}}\ket{\psi}
    & \rightarrow \frac{1}{\sqrt{1 + \tau}} (\ket{0}_{\mathrm{P}}  + \sqrt{\tau} \ket{1}_{\mathrm{P}} ) \ket{0} \ket{\bm{0}}\ket{\psi} \notag \\
    & \rightarrow \frac{1}{\sqrt{1 + \tau}} \left(\ket{0}_{\mathrm{P}}  \left(\bm{1} - \frac{L^\dagger L}{2\alpha^2} \tau \right) \ket{\psi} + \ket{1}_{\mathrm{P}}  \cdot \sqrt{\tau} \frac{L}{\alpha}\ket{\psi}\right).
\end{align}
The statement holds by directly counting the gates and qubits of the circuit diagrams above.
\end{proof}

Now we are ready to prove Lemma~\ref{lemma:oaa_for_B0B1}.

\begin{proof}[Proof of Lemma~\ref{lemma:oaa_for_B0B1}]
Define a rotation gate $R_y$ to maintain the success probability as
\begin{equation}
R_y : \ket{0} \mapsto \frac{\sqrt{1 + \tau}}{2} \ket{0} + \sqrt{\frac{3 - \tau}{4}}\ket{1}
\end{equation}
for $\tau \in [0,3]$.
Here, $W_{\mathcal{B}}$ can be constructed by the circuit in Lemma~\ref{lemma: efficient_BE_ofB0B1}, and we define $U = (R_y \otimes \bm{1})(\bm{1} \otimes W_{\mathcal{B}})$ as the circuit:
\begin{center}
\begin{quantikz}
    \lstick{$\ket{\psi}$} & &  \gate[4]{U} & \\
    \lstick{$\ket{\bm{0}}$} &\qwbundle{l+1} & &\\
    \lstick{$\ket{0}_{\mathrm{P}} $} & \qwbundle{1} & &\\
    \lstick{$\ket{0}$} & \qwbundle{1} &  &  
\end{quantikz}
$=$
\begin{quantikz}
    \lstick{$\ket{\psi}$} & & \gate[3]{W_\mathcal{B}} & \\
    \lstick{$\ket{\bm{0}}$} &\qwbundle{l+1} & &\\
    \lstick{$\ket{0}_{\mathrm{P}} $} & \qwbundle{1} & &\\
    \lstick{$\ket{0}$} & \qwbundle{1} & \gate{R_y} &  
\end{quantikz},
\end{center}
and the projectors as
\begin{equation}
    \tilde{\Pi}=\ket{0}\bra{0}\otimes \bm{1}_{\rm P}\otimes\ket{\bm{0}}\bra{\bm{0}}\otimes \bm{1},~~~{\Pi}=\ket{0}\bra{0}\otimes \ket{0}\bra{0}_{\rm P}\otimes \ket{\bm{0}}\bra{\bm{0}}\otimes  \bm{1}.
\end{equation}
This unitary $U$ satisfies
\begin{equation}
    \tilde{\Pi} U \Pi = \frac{1}{2} \sum_{k=0,1}\ket{{0}}\bra{{0}}\otimes \ket{k}\bra{0}_{\rm P}\otimes\ket{\bm{0}}\bra{\bm{0}}\otimes B_k.
\end{equation}
The remaining steps follow the procedure of Lemma~\ref{lemma:oaa_non_isometry},
and then we confirm that the CPTN map $\mathcal{B}^{\mathrm{(approx)}} = \sum_{k=0}^1 B_k' \bullet (B_k')^\dagger$ can be implemented deterministically.
The corresponding circuit is depicted as follows:
\begin{center}
\begin{quantikz}
    \lstick{$\ket{\psi}$} & &  \gate[4]{U} & & \gate[4]{U^\dagger}& & \gate[4]{U}&\\
    \lstick{$\ket{\bm{0}}$} &\qwbundle{l+1} & & \gate[2]{2 \Tilde{\Pi} - \bm{1}} & &  \gate[3]{\bm{1} - 2 \Pi} & & \rstick{$\bra{\bm{0}}$} \\
    \lstick{$\ket{0}$} & \qwbundle{1} &  &  & & & & \rstick{$\bra{0}$}\\
    \lstick{$\ket{0}_{\mathrm{P}} $} & \qwbundle{1} & & &  &  & & \meter{}
\end{quantikz}.
\end{center}
For the visualization, the bottom two registers are swapped.
Note that $2 \Tilde{\Pi} - \bm{1}$ consists of a single NOT gate controlled by at most $(l+1)$ qubits, and $\bm{1} - 2 \Pi$ consists of a NOT gate controlled by at most $(l+2)$ qubits. From the circuit above and Lemma~\ref{lemma: efficient_BE_ofB0B1}, the number of gates can be directly counted as follows:
\begin{itemize}
    \item Single-qubit gate: $12 = (3 + 1) \times 3$,
    \item NOT gate controlled by at most $(l+2)$ qubits: $5 = 3 + 2 $,
    \item Unitary gate $W_L$ and its inverse: $3$,
    \item Controlled version of $W_L$ and its inverse: $3$,
    \item Additional ancilla qubits, excluding system qubits: $l + 3$.
\end{itemize}

All $l+3$ ancilla qubits are measured by the computational basis,
and we perform classical post-processing depending on the measurement outcome.
That is, partially tracing the purifier register out, we multiply the final output by zero whenever the other $l+2$ ancilla qubits measurement result does not coincide with $\ket{\bm{0}}\ket{0}$.
\end{proof}

In the Lindblad case, we can derive a correction superoperator $\mathcal{R}$ immediately.

\begin{lem}[Correction superoperator $\mathcal{R}$ for $B_0 \bullet B_0^\dagger + B_1 \bullet B_1^\dagger$]\label{lemma:correction_superoperator}
    Let $L,\{B_k\},D,\tau$, and $\{B'_{k}\}$ be as in Lemma~\ref{lemma:oaa_for_B0B1}.
    Then, the correction superoperator $\mathcal{R}$,
    representing the difference between the target CP map $\mathcal{B}: = \sum_{k=0}^{1} B_k \bullet B_k^\dagger$
    and the approximated CPTN map $\mathcal{B}^{\rm (approx)}: = \sum_{k=0}^{1} B'_k \bullet (B_k')^\dagger$ effectively simulated by Lemma~\ref{lemma:oaa_for_B0B1}, can be written as
    \begin{align}
        \mathcal{R} & := \mathcal{B} -  \mathcal{B}^{\rm (approx)}= \sum_{k=0}^{1} \left(\frac{\tau^2}{8} B_k \bullet DB^\dagger_k
        + \frac{\tau^2}{8} B_k D \bullet B^\dagger_k
        - \frac{\tau^4}{64} B_kD  \bullet DB^\dagger_k \right).
    \end{align}
    Furthermore, the transfer matrix of $\mathcal{R}$ can be described as a linear combination of unitaries, that is,
    \begin{equation}
    S(\mathcal{R}) = \sum_{i} c_{i} \overline{V_{i}} \otimes U_{i},
    \end{equation}
    for some unitaries $U_i,V_i$ and $c_i\geq 0$.
    The sum of coefficients is determined as
    \begin{equation}
    \sum_i c_i = \frac{1}{4}\tau^2 + \frac{1}{2} \tau^3 + \frac{5}{64} \tau^4 + \frac{1}{32} \tau^5 + \frac{1}{256}\tau^6.
    \end{equation}
    Also, if $L$ is written as a linear combination of Pauli strings,
    i.e., $L  = \sum_i \alpha_i P_i$ where $\{ P_i \}$ are Pauli strings, and $\alpha_i \in \mathbb{C}$, and $\alpha = \sum_i |\alpha_i|$, 
    then $S(\mathcal{R}) / (\sum_i c_i)$ is a convex combination of Pauli strings with complex phases.
\end{lem}
\begin{proof}
    From the direct calculation, we obtain
    \begin{align}
    \mathcal{R} &=\mathcal{B} -  \mathcal{B}^{\mathrm{(approx)}}\notag\\
    & = \sum_{k=0}^{1} B_k \bullet B_k^\dagger - B_k \left( \bm{1} - \frac{\tau^2 D}{8}\right) \bullet \left(\bm{1} - \frac{\tau^2 D}{8}\right) B_k^\dagger \notag\\
    &= \sum_{k=0}^{1} \left(\frac{\tau^2}{8} B_k \bullet DB^\dagger_k
    + \frac{\tau^2}{8} B_k D \bullet B^\dagger_k
    - \frac{\tau^4}{64} B_kD  \bullet DB^\dagger_k \right)
    \end{align}
    or, 
    \begin{equation}
    S(\mathcal{R})
    = \sum_{k=0}^{1} \left(\frac{\tau^2}{8} \overline{B_k D} \otimes  B_k 
    + \frac{\tau^2}{8} \overline{B_k} \otimes  B_k D
    - \frac{\tau^4}{64} \overline{B_k D} \otimes  B_k D \right).
    \end{equation}
    Since $L/\alpha$ is a convex combination of unitaries,
    $B_0 / (1+ \tau/2)$ and $B_1/\sqrt{\tau}$ are convex combinations of unitaries,
    as is $D$.
    Then, from the Lemma~\ref{lemma:combex_combination},
    normalized $\mathcal{R}$ can be written as a linear combination of unitaries.
    The normalization factor is
    \begin{align}
    C &= \left(1+\frac{\tau}{2}\right)^2 \left(\frac{\tau^2}{8} + \frac{\tau^2}{8} + \frac{\tau^4}{64} \right)
    + \tau \left(\frac{\tau^2}{8} + \frac{\tau^2}{8} + \frac{\tau^4}{64} \right) \notag \\
    &= \frac{1}{4}{\tau}^2 + \frac{1}{2} {\tau}^3 + \frac{5}{64} {\tau}^4 + \frac{1}{32} {\tau}^5 + \frac{1}{256}{\tau}^6,
    \end{align}
    which is consistent with Lemma~\ref{lemma:generalRsampling} where $1 / p = (1+\tau/2)^2 + \tau$ and $\delta = \tau^2 / 4$.
    In addition, if $L/\alpha$ is expanded by Pauli strings, $B_0$, $B_1$, and $D$ in $\mathcal{R}$ can be written in Pauli strings. Then we obtain $\mathcal{R}$ as a convex combination of Pauli strings.
\end{proof}

To conclude the subsection, we discuss the effectiveness of our OAA and recovery operation method in this case.
We can prepare the exact map $\mathcal{B}$ using the LCU for channel $W_{\mathcal{B}}$,
\begin{equation}
    \sqrt{p} \left(\sum_{k=0}^{1} \ket{k}_{\mathrm{P}} \otimes \ket{\bm{0}} \otimes  B_k \ket{\psi} \right) + \ket{\perp}.
\end{equation}
This provides a CPTN map $p \mathcal{B}$ with normalization  factor $1/p = 1 + \tau$ from Eq.~\eqref{eq:BEofmathcalB}.
The overhead for rescaling is equal to $1 + \tau$, and the linear dependence on $\tau$ is not suitable for our random compilation purposes.
Our new OAA and recovery operation method can also
yield the exact CPTN map via
\begin{equation}
    \frac{1}{1 + C} \mathcal{B} = \frac{1}{1 + C}\mathcal{B}^{\mathrm{(approx)}}  + \frac{C}{1 + C} \frac{\mathcal{R}}{C}.
\end{equation}
Therefore, the method requires additional $1 + C = 1 + \mathcal{O}(\tau^2)$ overhead for rescaling.
Comparing the overheads, 
the OAA with the recovery operation improves the overhead from $1 +\tau$ to $1 +\mathcal{O}(\tau^2)$, and the improvement is quadratic with respect to $\tau$.   
This compression is a crucial fact that leads to our decomposition in Theorem~\ref{thm: main}.
OAA can be regarded as an essential quantum resource usage corresponding to non-Clifford gates in the case of randomized LCU for Hamiltonian simulation (Section~\ref{subsection:wan}).

\subsection{Circuit simulation of superoperators}\label{apdx:construct_circuits}
As we derived in Eq.~\eqref{eq:expansion_explicit_form_tr},
the time propagator $e^{G t/r}$ can be decomposed into a (rescaled) convex combination of products of four-type superoperators having the following forms:
$e^{i\theta}\overline{P}\otimes Q$, $S(\mathcal{B}^{(\rm approx)}_{kl})$, $\bm{1}\otimes e^{-i \theta P}$, and $\overline{e^{-i \theta P}} \otimes \bm{1}$, where $P$ and $Q$ are some $n$-qubit Pauli strings.
In the previous subsection, we showed explicit circuits for $\mathcal{B}^{({\rm approx})}$.
Next, we focus on the construction of the other superoperators,
more specifically, an asymmetric superoperator
\begin{equation}
    \mathcal{F}_{U,V}(\bullet) := U\bullet V^\dagger,~~~S(\mathcal{F}_{U,V})=\overline{V}\otimes U
\end{equation}
and its convex combination $\Phi_{p,\mathcal{F}} = \sum_{i} p_{i} \mathcal{F}_i$ with $\mathcal{F}_i(\bullet)=U_i\bullet V_i^\dagger$ for some unitaries $U$, $V$, $U_i$, and $V_i$.
In the following, we provide a framework to simulate a linear combination of superoperators (LCS) $\mathcal{F}_i$ using quantum circuits with some additional ancilla system.

The simulation of $\mathcal{F}_{U,V}$ is inspired by Hadamard test circuits.
First, by the analogy with the Hadamard test, we consider the following circuit $\widetilde{\mathcal{F}}_{U,V}$:
\begin{center}
    \begin{quantikz}
        \lstick{$\ket{+}$} &\gate[2]{\widetilde{\mathcal{F}}_{U,V}} & \\
        \lstick{$\rho$} & & 
    \end{quantikz}
    $:=$
    \begin{quantikz}
        \lstick{$\ket{+}$} & \octrl{1} & \ctrl{1} & \\
        \lstick{$\rho$} & \gate{U}  & \gate{V} &
    \end{quantikz}
    $ =  \frac{1}{2}\begin{pmatrix} * &  \mathcal{F}_{U, V}(\rho) \\ \mathcal{J} \circ \mathcal{F}_{U, V} \circ \mathcal{J}  (\rho)& * \end{pmatrix}$,
\end{center}
where $\mathcal{J}$ is an anti-linear map to transform an operator into its complex conjugate,
i.e., $\mathcal{J}:A \mapsto \mathcal{J}(A):= A^\dagger$ for any operator $A$.
The circuit provides an encoding of the map $\mathcal{F}_{U, V}$ on $(0,1)$-component of the dilated density matrix $(\ket{+}\bra{+}\otimes \rho)$ together with the conjugate $\mathcal{J} \circ  \mathcal{F}_{U, V} \circ  \mathcal{J}= \mathcal{F}_{V, U} $ on $(1, 0)$-component.
Thus, we can use the unitary circuit
$\widetilde{\mathcal{F}}_{U, V}$ on the $(n + 1)$ qubits to simulate the superoperator $\mathcal{F}_{U, V}$ on (a part of) the dilated density matrix.
Actually, measuring $X, Y$ on the top 1 qubit, we can effectively simulate the superoperator $\mathcal{F}_{U, V}$ as
\begin{equation}
    {\rm Tr}_{\rm anc}\left[(X-iY)_{\rm anc}\otimes \bm{1}\cdot \widetilde{\mathcal{F}}_{U,V}\left(|+\rangle\langle +|_{\rm anc}\otimes \bullet\right)\right]=\mathcal{F}_{U,V}(\bullet).
\end{equation}
Note that $\widetilde{\mathcal{F}}_{U, V}$ is a unitary channel, even if $\mathcal{F}_{U, V}$ is not a CPTP map.
This observation can be generalized as follows.
\begin{proposition}
    [Linear combination of superoperators (LCS)]\label{prop:lcs}
    Let $\Phi_{c,\mathcal{F}}=\sum_{i} c_i \mathcal{F}_i$ be $n$-qubit superoperators with real coefficients $c:=\{c_i\}$ and superoperators $\mathcal{F}_i(\bullet):=\mathcal{F}_{U_i,V_i}(\bullet)=U_i\bullet V_i^\dagger$ for some unitaries $U_i$ and $V_i$.
    Let $\widetilde{\Phi}_{c,\mathcal{F}}$ be a linear combination of unitary channels:
    \begin{equation}
        \widetilde{\Phi}_{c,\mathcal{F}}:=\sum_{i} c_i \widetilde{\mathcal{F}}_{i}
    \end{equation}
    for unitary channels $\widetilde{\mathcal{F}}_{i}=\widetilde{\mathcal{F}}_{U_i,V_i}$ defined as
    \begin{equation}\label{eq:defofUUV}
        \widetilde{\mathcal{F}}_{U, V} (\bullet):=
        (\ket{0}\bra{0}\otimes U + \ket{1}\bra{1}\otimes V) \bullet (\ket{0}\bra{0}\otimes U^\dagger + \ket{1}\bra{1}\otimes V^\dagger).
    \end{equation}
    Then, the action of $\widetilde{\Phi}_{c,\mathcal{F}}$ acting on $(n+1)$ qubits can be written with the superoperator $\Phi_{c,\mathcal{F}}$ as 
    \begin{equation}
    \widetilde{\Phi}_{c, \mathcal{F}} : 
        \begin{pmatrix}
            A_{00}& A_{01}\\    
            A_{10}& A_{11}
        \end{pmatrix}
        \mapsto 
        \begin{pmatrix}
            *& \Phi_{c, \mathcal{F}}(A_{01}) \\
            \mathcal{J}\circ \Phi_{c, \mathcal{F}} \circ \mathcal{J}(A_{10})&*
        \end{pmatrix},
    \end{equation}
    where $A_{ij}:=(\bra{i}_{\rm anc}\otimes \bm{1}) A (\ket{j}_{\rm anc}\otimes \bm{1})$ and $\mathcal{J}$ is the anti-linear map $\mathcal{J}:A\mapsto A^\dagger$.
    Furthermore, we can effectively simulate $\Phi_{c, \mathcal{F}}$ using the unitary channels $\{\widetilde{\mathcal{F}}_i\}$ by measuring $(X-iY)$ on the ancilla qubit as 
    \begin{equation}
        {\rm Tr}_{\rm anc}\left[(X-iY)_{\rm anc}\otimes \bm{1}\cdot \widetilde{\Phi}_{c,\mathcal{F}}\left(|+\rangle\langle +|_{\rm anc}\otimes \bullet\right)\right]=\Phi_{c,\mathcal{F}}(\bullet).
    \end{equation}
    by the circuit depicted in Fig.~\ref{fig:circ_X-iY}.
    If $\Phi_{c,\mathcal{F}}$ is a Hermitian-preserving map, then we can omit the $Y$ measurement as
    \begin{equation}
        {\rm Tr}_{\rm anc}\left[X_{\rm anc}\otimes \bm{1}\cdot \widetilde{\Phi}_{c,\mathcal{F}}\left(|+\rangle\langle +|_{\rm anc}\otimes \bullet\right)\right]=\Phi_{c,\mathcal{F}}(\bullet).
    \end{equation}
\begin{figure}[htbp]
\begin{center}
    \begin{quantikz}
        \lstick{$\ket{+}_{\mathrm{anc}}$}  &  \gate[2]{\widetilde{\Phi}_{c, \mathcal{F}}}  &  \gate{X - i Y} \\
        \lstick{sys} &   &
    \end{quantikz}
    $ =
    \Phi_{c, \mathcal{F}}(\bullet)
    $
    \caption{A circuit to effectively simulate $\Phi_{c, \mathcal{F}}$}
    \label{fig:circ_X-iY}
\end{center}
\end{figure}
\end{proposition}

\begin{proof}
    From the direct calculation, we have
    \begin{equation}
    \widetilde{\mathcal{F}}_{U,V}:
        \begin{pmatrix}
            A_{00}& A_{01}\\    
            A_{10}& A_{11}
        \end{pmatrix}
        \mapsto 
        \begin{pmatrix}
            *& \mathcal{F}_{U,V}(A_{01}) \\
            \mathcal{J}\circ \mathcal{F}_{U,V} \circ \mathcal{J}(A_{10})&*
        \end{pmatrix}
    \end{equation}
    for any unitaries $U$ and $V$.
    Thus, 
    \begin{align}
    \widetilde{\Phi}_{c, \mathcal{F}}\left( 
        \begin{pmatrix}
            A_{00}& A_{01}\\    
            A_{10}& A_{11}
        \end{pmatrix}\right)
        &=\sum_{i}c_i\widetilde{\mathcal{F}}_{i}\left( 
        \begin{pmatrix}
            A_{00}& A_{01}\\    
            A_{10}& A_{11}
        \end{pmatrix}\right)\notag\\
        &=
        \begin{pmatrix}
            *& \sum_i c_i \mathcal{F}_{i}(A_{01}) \\
            \sum_i c_i \mathcal{J}\circ \mathcal{F}_{i} \circ \mathcal{J}(A_{10})&*
        \end{pmatrix}\notag\\
        &=
        \begin{pmatrix}
            *& \Phi_{c, \mathcal{F}}(A_{01}) \\
            \mathcal{J}\circ \Phi_{c, \mathcal{F}} \circ \mathcal{J}(A_{10})&*
        \end{pmatrix}.
    \end{align}
    Particularly, for a matrix $\ketbra{+} \otimes \rho$,
    \begin{equation}
    \widetilde{\Phi}_{c, \mathcal{F}}(\ketbra{+} \otimes \rho)
    =
        \begin{pmatrix}
            *& \Phi_{c, \mathcal{F}}(\rho)/2 \\
            \mathcal{J}\circ \Phi_{c, \mathcal{F}} \circ \mathcal{J}(\rho)/2&*
        \end{pmatrix}
    \end{equation}
    holds.
    By measuring $X- iY$ on the ancilla qubit, we obtain
    \begin{align}
        {\rm Tr}_{\rm anc}\left[ ((X-iY)_{\rm anc}\otimes \bm{1}) \cdot \widetilde{\Phi}_{c,\mathcal{F}}\left(|+\rangle\langle +|_{\rm anc}\otimes \rho\right)\right]
        &=\frac{1}{2} \left( \Phi_{c,\mathcal{F}}(\rho) + 
        \mathcal{J}\circ\Phi_{c,\mathcal{F}}(\rho) \circ \mathcal{J}
        + \Phi_{c,\mathcal{F}}(\rho) - \mathcal{J}\circ\Phi_{c,\mathcal{F}} \circ \mathcal{J} (\rho \right) \notag \\
        &= \Phi_{c,\mathcal{F}}(\rho).
    \end{align}
    Similarly, if $\Phi_{c,\mathcal{F}}$ is a HP map,
    \begin{align}
        {\rm Tr}_{\rm anc}\left[(X_{\rm anc}\otimes \bm{1}) \cdot \widetilde{\Phi}_{c,\mathcal{F}}\left(|+\rangle\langle +|_{\rm anc}\otimes \rho\right)\right]
        &=\frac{1}{2} \left( \Phi_{c,\mathcal{F}}(\rho) + 
        \mathcal{J}\circ\Phi_{c,\mathcal{F}}(\rho) \circ \mathcal{J} \right)
        = \Phi_{c,\mathcal{F}}(\rho).
    \end{align}
    \end{proof}

In particular, for a probability distribution $p_i\geq 0$ satisfying $\sum_i p_i =1$, $\widetilde{\Phi}_{p,\mathcal{F}}$ becomes a mixed unitary channel that can be realized on quantum circuits $\{\widetilde{\mathcal{F}}_i\}$ with random application of the circuits according to the probability $p_i$.

\begin{remark}\label{remark:asymmetricconvexcomb}
    Let $\{ U_i \}$ and $\{V_i \}$ be $n$-qubit unitary operators.
    Define a convex combination of $\mathcal{F}_i(\bullet):=U_i\bullet V_i^\dagger$ as 
    \begin{equation}
        \Phi_{p,\mathcal{F}}(\bullet) = \sum_i p_i \mathcal{F}_i(\bullet)
    \end{equation}
    with $p_i >0$ and $\sum_i p_i =1$.
    Then, $\Phi_{p,\mathcal{F}}$ can be effectively simulated by the mixed unitary channel
    \begin{equation}
        \widetilde{\Phi}_{p,\mathcal{F}}:=\sum_i p_i \widetilde{\mathcal{F}}_i
    \end{equation}
    acting on $(n+1)$ qubits as
    \begin{equation}
        \sum_i p_i \widetilde{\mathcal{F}}_i:
        \begin{pmatrix}
                A_{00}& A_{01}\\    
                A_{10}& A_{11}
            \end{pmatrix}
            \mapsto 
            \begin{pmatrix}
                *& \Phi_{p,\mathcal{F}}(A_{01}) \\
                \mathcal{J}\circ\Phi_{p,\mathcal{F}} \circ \mathcal{J}(A_{10})&*
        \end{pmatrix}.
    \end{equation}
\end{remark}

    In the proof of Theorem~\ref{thm: main}, we consider an alternating sequence of $n$-qubit superoperators in the form of
    \begin{equation}\label{eq:alternative_sequence}
        \mathcal{W}=\mathcal{B}^{(l)}\circ \Phi_{p^{(l)},\mathcal{F}^{(l)}}\circ \cdots \circ \mathcal{B}^{(1)}\circ \Phi_{p^{(1)},\mathcal{F}^{(1)}},
    \end{equation}
    where $\{\mathcal{B}^{(l)}\}$ are CPTN maps, and $\{\Phi_{p^{(l)},\mathcal{F}^{(l)}}\}$ are convex combinations of asymmetric forms $\mathcal{F}_i^{(l)}=U^{(l)}_i\bullet V^{(l)\dagger}_i$ for some unitaries $U^{(l)}_i,V^{(l)}_i$.
    Referring to Eq.~\eqref{eq:alternative_sequence}, we define $(n+1)$-qubit CPTN map $\widetilde{\mathcal{W}}$ for the alternating sequence $\mathcal{W}$:
    \begin{equation}
        \widetilde{\mathcal{W}} := \mathcal{I}_{\rm anc}\otimes\mathcal{B}^{(l)}\circ \widetilde{\Phi}_{p^{(l)},\mathcal{F}^{(l)}}\circ \cdots \circ \mathcal{I}_{\rm anc}\otimes\mathcal{B}^{(1)}\circ \widetilde{\Phi}_{p^{(1)},\mathcal{F}^{(1)}}.
    \end{equation}
    This can be obtained by simply replacing CPTN maps $\mathcal{B}$ with $\mathcal{I}_{\rm anc}\otimes \mathcal{B}$ and $\Phi_{p,\mathcal{F}}$ with $\widetilde{\Phi}_{p,\mathcal{F}}$ in $\mathcal{W}$, respectively.
    Using $\widetilde{\mathcal{W}}$, we can simulate a convex combination of the superoperators $\mathcal{W}$, which is a key technique to prove our main Theorem~\ref{thm: main}.
    \begin{proposition}\label{prop:lcs_alternative_seq}
        For any index $u$, let $\mathcal{W}_u$ be a sequence in the form of Eq.~\eqref{eq:alternative_sequence},
        and $\Phi_{p,\mathcal{W}}:=\sum_u p_u \mathcal{W}_u$ is a convex combination of $\mathcal{W}_u$ with a probability distribution $\{p_u\}$.
        In addition, 
        we assume that
        $\Phi_{p,\mathcal{W}}$ is a Hermitian-preserving map for the probability distribution $\{p_u\}$.
        Then, defining $\widetilde{\Phi}_{p,\mathcal{W}}:=\sum_u p_u \widetilde{\mathcal{W}}_u$, we have
        \begin{equation}\label{eq:lcs_alternative_wo_MCMR}
            {\rm Tr}_{\rm anc}\left[X_{\rm anc}\otimes \bm{1}\cdot \left(\widetilde{\Phi}_{p,\mathcal{W}}\right)^r\left(|+\rangle\langle +|_{\rm anc}\otimes \bullet\right)\right]=\Phi_{p,\mathcal{W}}^r(\bullet)
        \end{equation}
        for any positive integer $r$, where the corresponding circuit diagram is shown in Fig.~\ref{fig:woMCMR}.
        Furthermore, the same map $\Phi^r_{p,\mathcal{W}}$ can be simulated by cutting the wire of the ancilla qubit at each iteration using mid-circuit measurement and qubit reset without additional sampling overheads, as shown in Fig.~\ref{fig:MCMR}, i.e., the following equation holds:
        \begin{equation}\label{eq:lcs_alternative_with_MCMR}
            \left({\rm Tr}_{\rm anc}\left[X_{\rm anc}\otimes \bm{1}\cdot \widetilde{\Phi}_{p,\mathcal{W}}\left(|+\rangle\langle +|_{\rm anc}\otimes \bullet\right)\right]\right)^r=\Phi_{p,\mathcal{W}}^r(\bullet).
        \end{equation}
    \begin{figure}[htbp]
        \centering
            \begin{quantikz}
                \lstick{$\ket{+}_{\mathrm{anc}}$}  &  \gate[2]{\widetilde{\Phi}_{p, \mathcal{W}}} &  &  \gate[2]{\widetilde{\Phi}_{p, \mathcal{W}}}  & \midstick[2, brackets=none]{$\cdots$} & \gate[2]{\widetilde{\Phi}_{p, \mathcal{W}}}  & \gate{X} \\
                \lstick{sys} &  & &  & & &  
            \end{quantikz}
        \caption{A circuit to simulate $\Phi_{p,\mathcal{W}}^r$}\label{fig:woMCMR}
    \end{figure}
    \begin{figure}[htbp]
            \begin{quantikz}
                \lstick{$\ket{+}_{\mathrm{anc}}$}  &  \gate[2]{\widetilde{\Phi}_{p, \mathcal{W}}}  &  \gate{X} &  \midstick{$\ket{+}_{\mathrm{anc}}$}  \setwiretype{n} &  \gate[2]{\widetilde{\Phi}_{p, \mathcal{W}}} \setwiretype{q}&  \gate{X} &  \midstick{$\ket{+}_{\mathrm{anc}}$}  \setwiretype{n} &  \gate[2]{\widetilde{\Phi}_{p, \mathcal{W}}} \setwiretype{q} & \gate{X} \\
                \lstick{sys} &   & & & & & \midstick{$\cdots$}  & &
            \end{quantikz}
        \caption{A circuit to simulate $\Phi_{p,\mathcal{W}}^r$ with mid-circuit measurement and qubit reset}\label{fig:MCMR}
    \end{figure}
    \end{proposition}

\begin{proof}
    The action of $\mathcal{I}_{\rm anc}\otimes \mathcal{B}$ on $(n+1)$ qubits is obviously
    \begin{equation}
    \mathcal{I}_{\rm anc}\otimes \mathcal{B}:
    \begin{pmatrix}
        A_{00}&A_{01}\\    
        A_{10}&A_{11}
    \end{pmatrix}
    \mapsto 
    \begin{pmatrix}
        * &\mathcal{B}(A_{01})\\
        \mathcal{B} (A_{10})& *
    \end{pmatrix},
    \end{equation}
    where $A_{ij}=(\bra{i}_{\rm anc}\otimes \bm{1}) A (\ket{j}_{\rm anc}\otimes \bm{1})$.
    The composition of $\mathcal{I}_{\rm anc}\otimes \mathcal{B}$ and $\widetilde{\Phi}_{p, \mathcal{F}}$ can simulate $\mathcal{B} \circ \Phi_{p, \mathcal{F}}$, and their sequence $\widetilde{\mathcal{W}}$ can also simulate $\mathcal{W}$. This can be described explicitly as
    \begin{align}
    &\widetilde{\mathcal{W}}:
    \begin{pmatrix}
        A_{00}&A_{01}\\    
        A_{10}&A_{11}
    \end{pmatrix}
    \notag \\
    &\mapsto 
    \begin{pmatrix}
        * & \mathcal{B}^{(l)}\circ \Phi_{p^{(l)},\mathcal{F}^{(l)}}\circ \cdots \circ \mathcal{B}^{(1)}\circ \Phi_{p^{(1)},\mathcal{F}^{(1)}}(A_{10})\\
        \mathcal{B}^{(l)}\circ \mathcal{J}\circ \Phi_{p^{(l)},\mathcal{F}^{(l)}}\circ\mathcal{J}\circ \cdots \circ \mathcal{B}^{(1)}\circ\mathcal{J}\circ \Phi_{p^{(1)},\mathcal{F}^{(1)}}\mathcal{J}(A_{01})& *
    \end{pmatrix} \notag \\
    &=
    \begin{pmatrix}
        * & \mathcal{W} (A_{01})\\
        \mathcal{J} \circ \mathcal{W} \circ  \mathcal{J}(A_{10})& *
    \end{pmatrix}.
    \end{align}
    At the last equality, we use the Hermitian-preserving property of CPTN maps $\mathcal{J} \circ \mathcal{B}^{(l)} \circ  \mathcal{J}=\mathcal{B}^{(l)}$ and $\mathcal{J}^2 = \mathcal{I}$.
    From the same calculation as Proposition~\ref{prop:lcs}, we obtain the action of convex combination of $\widetilde{\mathcal{W}}$ as
    \begin{equation}
    \widetilde{\Phi}_{p,\mathcal{W}}:
    \begin{pmatrix}
        A_{00}&A_{01}\\    
        A_{10}&A_{11}
    \end{pmatrix}
    \mapsto 
    \begin{pmatrix}
        * & \Phi_{p,\mathcal{W}} (A_{01})\\
        \mathcal{J} \circ \Phi_{p,\mathcal{W}} \circ \mathcal{J} (A_{10}) & *
    \end{pmatrix}
    =
    \begin{pmatrix}
        * & \Phi_{p,\mathcal{W}} (A_{01})\\
        \Phi_{p,\mathcal{W}} (A_{10}) & *
    \end{pmatrix},
    \end{equation}
    since $\Phi_{p,\mathcal{W}}$ is assumed to be a Hermitian-preserving map.
    Finally, after measuring $X$ on the ancilla qubit, we can simulate the map $\Phi_{p,\mathcal{W}}$:
    \begin{equation}\label{eq:sim_of_Phi_pW}
        {\rm Tr}_{\rm anc}\left[X_{\rm anc}\otimes \bm{1}\cdot \widetilde{\Phi}_{p,\mathcal{W}}\left(|+\rangle\langle +|_{\rm anc}\otimes \bullet\right)\right]=\Phi_{p,\mathcal{W}}(\bullet).
    \end{equation}
    By repeating the above $r$ times, Eq.~\eqref{eq:lcs_alternative_with_MCMR} is shown straightforwardly.
    Next, considering the action of the $(\widetilde{\Phi}_{p,\mathcal{W}})^r$, we obtain 
    \begin{equation}
    \widetilde{(\Phi}_{p,\mathcal{W}})^r:
    \begin{pmatrix}
        A_{00}&A_{01}\\    
        A_{10}&A_{11}
    \end{pmatrix}
    \mapsto 
    \begin{pmatrix}
        * & (\Phi_{p,\mathcal{W}})^r (A_{01})\\
        (\Phi_{p,\mathcal{W}})^r (A_{10})& *
    \end{pmatrix}.
    \end{equation}
    Therefore, by repeating $\widetilde{\Phi}_{p,\mathcal{W}}$ $r$ times and measuring $X$, we obtain Eq.~\eqref{eq:lcs_alternative_wo_MCMR}.
    \end{proof}

These Proposition~\ref{prop:lcs} and \ref{prop:lcs_alternative_seq} provide practical procedures to simulate LCS.
In practical scenarios, the circuits Fig.~\ref{fig:MCMR} with measurements and resets may be more advantageous than the original ones Fig.~\ref{fig:woMCMR} in the sense of maintaining coherence.

Here we illustrate a simple example to make the formulation clearer.
Let us consider the case of randomized LCU for Hamiltonian simulation.
Assume that we have the linear combination of unitaries of $e^{-iHt/r} \approx \sum_m c_m U_m$, where $U_m$ are unitaries, and $c_m >0$.
By defining $c = \sum_m c_m$ and joint distribution $p_{m ,m'} = c_mc_{m'}/c^2$, the map of the propagator $e^{-iHt/r}$ can be described as 
\begin{equation}\label{eq:re_expression_wan}
e^{-iHt/r} \bullet e^{iHt/r} \approx c^2 \sum_{m, m'} p_{m, m'} U_m \bullet U^\dagger_{m'}
\end{equation}

Using the LCS notations as
\begin{gather}
\mathcal{F}_{m,m'} = U_m \bullet U_{m'}^\dagger \notag \\
\Phi_{p, \mathcal{F}} = \sum_{m, m'} p_{m,m'} \mathcal{F}_{m,m'},
\end{gather}
we obtain the map for the Hamiltonian simulation below:
\begin{equation}
    {\rm Tr}_{\rm anc}\left[X_{\rm anc}\otimes \bm{1}\cdot \left(\widetilde{\Phi}_{p,\mathcal{F}}\right)^r\left(|+\rangle\langle +|_{\rm anc}\otimes \bullet\right)\right]
    \approx \Phi_{p,\mathcal{F}}^r(\bullet) = \frac{1}{c^{2r}}e^{-iHt} \bullet e^{iHt}.
\end{equation}
where we use the HP property of $\Phi_{p,\mathcal{F}}$, that is,
\begin{equation}
\mathcal{J} \circ \Phi_{p,\mathcal{F}} \circ \mathcal{J} = \sum_{m',m} p_{m', m}U_{m'} \bullet U_{m}^\dagger = \Phi_{p,\mathcal{F}}.
\end{equation}
An asymmetric superoperator $\mathcal{F}_{m, m'}$ alone does not have HP property since
$\mathcal{J}\circ \mathcal{F}_{m, m'} \circ  \mathcal{J} = \mathcal{F}_{m', m}$,
while a symmetrized combination $(\mathcal{F}_{m, m'} + \mathcal{F}_{m', m} ) /2$ does.
Note that this HP property is always satisfied for the superoperators formed as
$(\sum_m c_m U_m) \bullet (\sum_m c_m U_m^\dagger)$ with $c_m \in \mathbb{R}$.
This example is consistent with the results of prior studies that used Hadamard test such as~\cite{chakraborty2024implementing}.
The Fig.~\ref{fig:hamiltonian_sim} shows the circuit example.
Moreover, we can cut the wire at each simulation step by the HP property
as shown in Proposition~\ref{prop:lcs_alternative_seq}.
\begin{figure}[htbp]
    \begin{quantikz}
            \lstick{$\ket{+}_{\mathrm{anc}}$} &\qwbundle{1} & \octrl{1} & \ctrl{1} & \octrl{1} & \ctrl{1} &  \midstick{$\cdots$} & \octrl{1} & \ctrl{1} & \gate{X} \\
            \lstick{sys} &\qwbundle{n} & \gate{U_{m_1}} & \gate{U_{m_1'}} & \gate{U_{m_2}} & \gate{U_{m_2'}} & \midstick{$\cdots$} & \gate{U_{m_r}} & \gate{U_{m_r'}} & \qw
    \end{quantikz}
    \caption{A circuit example to effectively simulate Hamiltonian dynamics.}\label{fig:hamiltonian_sim}
\end{figure}

\subsection{Technical lemmas}\label{subsec:technical_lemmas}

\begin{lem}\label{lemma:combex_combination}
    Let $\{A_k\}$ be convex combinations of unitary operators as
    \begin{equation}
        A_k:= \sum_{j} p_{kj} U_{kj}.
    \end{equation}
    Then, $A_1 A_2$ and $A_1\otimes A_2$ can be written as a convex combination of some unitary operators.
    Also, for any $\theta_1,\theta_2,...\in \mathbb{R}$ and any probability distribution $\{q_k\}_k$, the following operator
    \begin{equation}
        B:=\sum_{k}q_ke^{i\theta_k}A_k=q_1e^{i\theta_1}A_1+q_2e^{i\theta_2}A_2+\cdots
    \end{equation}
    can be written as a convex combination of some unitary operators.
    In addition, if all of $\{U_{kj}\}$ in $A_k$ are Pauli operators with unit complex coefficients (i.e., $e^{i\theta}$ for some $\theta\in \mathbb{R}$), then the operator $A_1A_2, A_1\otimes A_2$, and $B$ defined above can be written as a convex combination of Pauli operators with unit complex coefficients.
\end{lem}
\begin{proof}
    We can prove this lemma by direct calculations. 
    
\end{proof}

Since Eq.~\eqref{eq:expansion_explicit_form_tr} already provides the explicit decomposition of $e^{(t/r)G}$, we here derive an inequality to evaluate the sum of coefficients in the decomposition.

\begin{lem}\label{lem:wan_gammma_inequality}
    For any constant value $c\geq 1.66$, 
    \begin{equation}
        \sum_{n=0}^{\infty} \frac{x^{2n}}{(2n)!} \gamma \left(\frac{x}{2n+1}\right)\le \exp(c x^2)
    \end{equation}
    holds for any $x\in[0,1]$, where $\gamma(t)$ $(t\in \mathbb{R})$ is defined as
    \begin{equation}\label{eq:gamma_props}
        \gamma(t) = 1 + \frac{1}{2}t^2 + \frac{1}{2} t^3 + \frac{5}{64} t^4 + \frac{1}{32} t^5 + \frac{1}{256}t^6.
    \end{equation}
\end{lem}
\begin{proof}
It suffices to show
\begin{equation}
    \sum_{n=0}^{\infty} \frac{x^{2n}}{(2n)!} \gamma \left(\frac{x}{2n+1}\right) \le \sum_{n=0}^{\infty} \frac{c^n x^{2n}}{n!}.
\end{equation}
We divide the series into two parts: $n=0,1$ and $n\ge 2$.
For $n = 0, 1$, 
we find a constant value $c$ such that
\begin{equation}
    \sum_{n=0}^{1} \frac{x^{2n}}{(2n)!} \gamma \left(\frac{x}{2n+1}\right)
    =  \gamma \left(x\right) + \frac{x^{2}}{2!} \gamma \left(\frac{x}{3}\right)\leq 1 + c x^2
\end{equation}
holds for any $x\in [0,1]$ or equivalently, 
\begin{align}
    c \geq \tilde{\gamma}(x):=& \left(\frac{1}{2} +\frac{1}{2} x + \frac{5}{64} x^2 +\frac{1}{32} x^3 +\frac{1}{256}x^4  \right) + \frac{1}{2!}  \gamma \left(\frac{x}{3}\right)
\end{align}
holds for any $x\in (0,1]$.
Thus, we can take $c\geq 1.66 > \tilde{\gamma}(1)$ for the inequality:
\begin{equation}
    \sum_{n=0}^{1} \frac{x^{2n}}{(2n)!} \gamma \left(\frac{x}{2n+1}\right) \le \sum_{n=0}^{1} \frac{c^n x^{2n}}{n!}.
\end{equation}

For all $n \ge 2$, $\gamma \left(\frac{x}{2n+1} \right)\leq \gamma(1/5)<2$ holds under $0 \le x \le 1$, and then we have
\begin{equation}
    \frac{x^{2n}}{(2n)!} \gamma \left(\frac{x}{2n+1}\right)
     < x^{2n}\frac{2}{(2n)!} \le \frac{x^{2n}}{n!}.
\end{equation}
Therefore, if we choose $c \ge 1.66$, the Lemma holds.
\end{proof}

\section{Proposed method for expectation value estimation}\label{apdx:C_alg}
In this section, we first show a truncated version of the decomposition of $e^{t\mathcal{L}}$ in Theorem~\ref{thm: main}. 
Then, we prove Theorem~\ref{thm: main2} for clarifying the performance of our algorithm for expectation values, using the truncated decomposition of $e^{t\mathcal{L}}$.

In Theorem~\ref{thm: main}, the index set $\mathrm{S}$ in the decomposition of $e^{t\mathcal{L}}$ has an infinitely large number of elements.
However, for any desired truncation error $\Delta\in (0,1/e)$, we can construct a finite subset $\mathrm{S}_{\Delta}\subset \mathrm{S}$ satisfying the following properties: (i) we can efficiently sample quantum circuits for CPTN maps $\widetilde{\mathcal{W}}_v$ according to the probability distribution proportional to $\{c_v\}_{v\in \mathrm{S}_{\Delta}}$, and furthermore, (ii) the sampled quantum circuits have a logarithmic depth with respect to the truncation error $\Delta$.
\begin{lem}
    [Truncated version of Theorem~\ref{thm: main}]\label{lemma:truncatedthm1}
    Let $\mathcal{L}$ be an $n$-qubit Lindblad superoperator with a Hamiltonian $H$ and jump operators $\{L_k\}_{k=1}^K$ that are specified by a linear combination of Pauli strings as 
    \begin{equation*}
        H = \sum_{j=1}^m \alpha_{0j} P_{0j},~~~L_k =\sum_{j=1}^M \alpha_{kj} P_{kj},
    \end{equation*}
    for some coefficients $\alpha_{0j} \in \mathbb{R}$, $\alpha_{kj} \in \mathbb{C}$, and let $\|\mathcal{L}\|_{\rm pauli}:=2(\alpha_0+\sum_{k=1}^K \alpha_k^2)$ for $\alpha_k:=\sum_{j} |\alpha_{kj}|$.
    Then, for any $t>0$, $\Delta \in (0,1/e)$, and any positive integer $r\geq \|\mathcal{L}\|_{\rm pauli}t$, 
    there exists an approximate decomposition of $e^{t\mathcal{L}}$ such that
    \begin{equation}\label{eq:closeness}
        \left\|e^{t\mathcal{L}}(\bullet)-\sum_{v \in \mathrm{S}_{\Delta}} c_v \mathrm{Tr}_{\mathrm{anc}}[(X_{\rm anc} \otimes \bm{1}) \widetilde{\mathcal{W}}_v (\ket{+}\bra{+}_{\rm anc} \otimes \bullet )]\right\|_{1\to 1}\leq \Delta
    \end{equation} 
    holds for some finite index set $\mathrm{S}_{\Delta}$, $(n+1)$-qubit completely positive trace non-increasing (CPTN) maps $\{\widetilde{\mathcal{W}}_v\}$, and real values $c_v >0$ satisfying 
    \begin{equation}\label{eq:delta_sampleoverhead}
        \sum_{v \in \mathrm{S}_{\Delta}} c_{v} \leq e^{2{{\|\mathcal{L}\|_{\rm pauli}^2t^2}/{r}}}.
    \end{equation}
    Furthermore, for any $v\in \mathrm{S}_{\Delta}$, the $(n+1)$-qubit CPTN map $\widetilde{\mathcal{W}}_v$ can be effectively simulated by using an
    $$n+4+\lceil\log_2 M\rceil$$
    qubits quantum circuit, including mid-circuit measurement and qubit reset, and
    \begin{equation}
    \mathcal{O}\left(r\left(\frac{\log(r/\Delta)}{\log\log (r/\Delta)} + M\log M\right)\right)
    \end{equation}
    one- or two-qubit gates.
    The corresponding quantum circuits for the CPTN maps $\widetilde{\mathcal{W}}_v$ can be efficiently sampled according to the probability distribution proportional to $c_v>0$ by Algorithm~\ref{alg_circuit_generation}.
\end{lem}

\begin{proof}
    [Proof of Lemma~\ref{lemma:truncatedthm1}]
    For a non-negative integer $\rm Q\geq 0$, let $\Lambda_{2{\rm Q}+1,t/r}$ be a superoperator defined as
    \begin{equation}
        \Lambda_{2{\rm Q}+1,t/r}:=\sum_{q=0}^{2{\rm Q}+1} \frac{(t/r)^q\mathcal{L}^q}{q!}.
    \end{equation}
    This is a truncated series of Eq.~\eqref{suppleeq:dynamical_map}.
    Thus, using the results in the proof of Theorem~\ref{thm: main}, we obtain
    \begin{align}
        &S(\Lambda_{2{\rm Q}+1,t/r})=\sum_{l=0}^{{\rm Q}} \frac{(t/r)^{2l} }{(2l)!}S(\mathcal{L})^{2l}\left(\bm{1} \otimes \bm{1} +\frac{(t/r)}{2l+1}S(\mathcal{L})\right)\notag\\
        &=\sum_{l=0}^{\rm Q} \frac{(t/r)^{2l} \|\mathcal{L}\|^{2l}_{\rm pauli}}{(2l)!}\left(2\frac{\alpha_0}{\alpha} \sqrt{1 + \tau_l^2}+\sum_{k=1}^K \frac{\alpha_k^2}{\alpha}\left(1+\|R_{kl}\|_{\rm pauli}+\frac{\tau_l^2}{4}\right)\right)\times \sum_{k=0}^K \sum_{\nu=1}^3 q_{kl} p_{\Gamma,kl,\nu} 
        \left(\frac{S(\mathcal{L})}{\|\mathcal{L}\|_{\rm pauli}}\right)^{2l}
        \Gamma_{kl,\nu},
    \end{align}
    where we used the same notation as in the proof of Theorem~\ref{thm: main}.
    For simplicity, we write the decomposition as $S(\Lambda_{2{\rm Q}+1,t/r})=\sum_{u\in \mathrm{\tilde{S}}_{\rm Q}}c_{u}S(\mathcal{W}_u)$ by implicitly defining $c_u\geq 0$, $S(\mathcal{W}_u)$, and an finite index set $\tilde{\mathrm{S}}_{\rm Q}$.
    In particular, the superoperator $S(\mathcal{W}_u)$ can always be written as 
    \begin{equation}
        S(\mathcal{W}_u)= \left(\frac{S(\mathcal{L})}{\|\mathcal{L}\|_{\rm pauli}}\right)^{2l}
        \Gamma_{kl,\nu},
    \end{equation}
    where $\Gamma_{kl,\nu}$ is the transfer matrix of a superoperator defined by Eqs.~\eqref{eq:gamma_0}--\eqref{eq:G3_convex_pauli}.
    Then, we have
    \begin{equation}
        S\left(\Lambda_{2{\rm Q}+1,t/r}^r\right)=\sum_{(u_1,u_2,...u_r)\in \mathrm{\tilde{S}}_{\rm Q}^r}  c_{u_1}\cdots c_{u_r}S(\mathcal{W}_{u_1})\cdots S(\mathcal{W}_{u_r}).
    \end{equation}
    The sum of coefficients can be evaluated as
    \begin{align}
        \sum_{(u_1,u_2,...u_r)\in \mathrm{\tilde{S}}_{\rm Q}^r}  c_{u_1}\cdots c_{u_r}&= \left\{\sum_{l=0}^{\rm Q} \frac{(t/r)^{2l} \|\mathcal{L}\|^{2l}_{\rm pauli}}{(2l)!}\left(2\frac{\alpha_0}{\alpha} \sqrt{1 + \tau_l^2}+\sum_{k=1}^K \frac{\alpha_k^2}{\alpha}\left(1+\|R_{kl}\|_{\rm pauli}+\frac{\tau_l^2}{4}\right)\right)\right\}^r\notag\\
        &\leq \left\{ \sum_{l=0}^{\infty} \frac{(t/r)^{2l} \|\mathcal{L}\|^{2l}_{\rm pauli}}{(2l)!}\left(2\frac{\alpha_0}{\alpha} \sqrt{1 + \tau_l^2}+\sum_{k=1}^K \frac{\alpha_k^2}{\alpha}\left(1+\|R_{kl}\|_{\rm pauli}+\frac{\tau_l^2}{4}\right)\right)\right\}^r\notag\\
        &\leq {\rm exp}\left(\frac{2\|\mathcal{L}\|_{\rm pauli}^2t^2}{r}\right),
    \end{align}
    where we used Eq.~\eqref{suppleeq:sum_of_cu} in the final inequality.

    As described in the proof of Theorem~\ref{thm: main}, the three-type circuits, $\mathcal{U}_{\rm I},\mathcal{U}_{\rm II}$, and $\mathcal{U}_{\rm III}$, illustrated in Fig.~\ref{fig:three_type_circuits} simulate the superoperator $\mathcal{W}_{u}$.
    More precisely, for any $\mathcal{W}_u$, we can explicitly construct a quantum circuit $\widetilde{\mathcal{W}}_{u}$ satisfying
    \begin{align}
        \Lambda_{2{\rm Q}+1,t/r}^r(\bullet)=\sum_{(u_1,u_2,...,u_r)\in \mathrm{\tilde{S}}_{\rm Q}^r}c_{u_1}c_{u_2}\cdots c_{u_r} {\rm Tr}_{\rm anc}\left[\left(X_{\rm anc}\otimes \bm{1}\right)\widetilde{\mathcal{W}}_{u_1} \circ\cdots \circ\widetilde{\mathcal{W}}_{u_r}\left(|+\rangle\langle +|_{\rm anc}\otimes \bullet\right)\right],
    \end{align}
    from the fact that $\Lambda_{2{\rm Q}+1,t/r}^r(\bullet)$ has the Hermitian preserving property.
    For any $u=(l,k,\nu)\in \mathrm{\tilde{S}}_{\rm Q}$, the corresponding quantum circuit $\widetilde{\mathcal{W}}_{u}$ consists of a single call of $\mathcal{U}_{\rm I}$, $\mathcal{U}_{\rm II}$, or $\mathcal{U}_{\rm III}$ for $\Gamma_{kl,\nu}$ and $2l\leq 2\mathrm{Q}$ calls of $\mathcal{U}_{\rm I}$ for $\mathcal{L}/\|\mathcal{L}\|_{\rm pauli}$; see Fig.~\ref{fig:example_Wu} for instance.
    $\widetilde{\mathcal{W}}_{u}$ has $\mathcal{O}(\mathrm{Q} + M \log M)$ one- or two-qubit gates in the worst case; $\mathcal{O}(\mathrm{Q})$ controlled-Pauli gates for $\mathcal{U}_{\rm I}$ and $\mathcal{O}(M \log M)$ gates for $\mathcal{U}_{\rm II}$.
    Thus, the gate complexity of quantum circuits $\widetilde{\mathcal{W}} _{u_1}\circ\cdots \circ\widetilde{\mathcal{W}}_{u_r}$ is given by
    $\mathcal{O}(r({\rm Q} + M\log M) )$ for any $(u_1,u_2,...u_r)\in \mathrm{\tilde{S}}_{\rm Q}^r$.
    Furthermore, we can efficiently sample quantum circuits $\widetilde{\mathcal{W}}_{u}$ (including classical post-processing instructions) according to the probability distribution proportional to the weight $\{c_{u}\}_{u\in \mathrm{\tilde{S}}_{\rm Q}}$ as described in Algorithm~\ref{alg_circuit_generation}.

    Finally, we show that when we take an integer $\rm Q$ as
    \begin{equation}\label{eq:defofq}
        {\rm Q}\geq\frac{\ln(3r/2\Delta)}{\ln\ln (3r/2\Delta)},
    \end{equation}
    then the approximation error of $\Lambda_{2{\rm Q}+1,t/r}^r(\bullet)$ for $e^{t\mathcal{L}}$ is upper bounded by $\Delta$ as
    \begin{equation}
        \left\|\Lambda_{2{\rm Q}+1,t/r}^r-e^{t\mathcal{L}}\right\|\leq \Delta.
    \end{equation}
    For any positive integer $\rm{Q'}$,
    \begin{align}
    \left\|\Lambda_{{\rm Q'},t/r}-e^{(t/r)\mathcal{L}}\right\|_{1\to 1}&=\left\|\sum_{q={\rm Q'}+1}^{\infty} \frac{(t/r)^q\mathcal{L}^q }{q!}\right\|_{1\to 1}\notag\\
    &\leq \sum_{q={\rm Q'}+1}^{\infty}\frac{|t/r|^q}{q!}\left\|{\mathcal{L}^q}\right\|_{1\to 1}\notag\\
    &\leq \sum_{q={\rm Q'}+1}^{\infty}\frac{1}{q!}\left(\frac{t\left\|{\mathcal{L}}\right\|_{1\to 1}}{r}\right)^q.
    \end{align}
    Here, $\|\mathcal{L}\|_{1\to 1}\leq \|\mathcal{L}\|_{\rm pauli}$ holds as follows.
    By the triangle inequality and H\"{o}lder's inequality, we have for any non-zero operator $A$,
    \begin{align}
        \|\mathcal{L}(A)\|_{1} 
        &= \left\| -i[H, A] + \sum_{k=1}^{K} \left(L_k A L_k^\dagger - \frac{1}{2} \left\{L_k^\dagger L_k, A \right\} \right) \right\|_{1} \notag \\
        &\le 2\|H\|_{\infty} \|A \|_1 + 2 \sum_{k=1}^{K} \|L_k\|_{\infty}^2 \| A \|_1 
        \le 2\|A \|_1 \left( \alpha_{0}  + \sum_{k=1}^{K} \alpha_{k}^2 \right).
    \end{align}
    Hence, $\|\mathcal{L}\|_{1\to 1}$ is bounded as 
    \begin{equation}
        \|\mathcal{L}\|_{1\to 1} \le 2 \left( \alpha_{0}  + \sum_{k=1}^{K} \alpha_{k}^2 \right)  = \| \mathcal{L}\|_{\mathrm{pauli}}.
    \end{equation}
    Thus, from the assumption of $r\geq t\|\mathcal{L}\|_{\rm pauli}\geq t \|\mathcal{L}\|_{1\to 1}$, we have
    \begin{equation}\label{suppleeq:assumption4r}
        \left\|\Lambda_{{\rm Q'},t/r}-e^{(t/r)\mathcal{L}}\right\|_{1\to 1}\leq \sum_{q={\rm Q'}+1}^{\infty}\frac{1}{q!}< \frac{1}{({\rm Q'})!}.
    \end{equation}
    Taking $\mathrm{Q'}:=\lceil 
    2\kappa/\ln\kappa\rceil$ $(\geq 1)$ for $\kappa:=\ln(r/\Delta')$ with $\Delta'\in (0,1/e)$ and using the Stirling's formula, we can derive that
    \begin{align}
        \ln \frac{1}{(\mathrm{Q'})!} &< \mathrm{Q'}-\mathrm{Q'}\ln\mathrm{Q'} = \mathrm{Q'}\ln \kappa\cdot \frac{1-\ln\mathrm{Q'}}{\ln \kappa}\leq -\frac{1}{2} \mathrm{Q'} \ln \kappa \leq -\kappa.
    \end{align}
    This leads that the approximation error of $\Lambda_{{\rm Q'},t/r}$ becomes at most $\Delta'/r$, i.e., 
    \begin{equation}\label{suppleeq:t_error_d}
        \left\|\Lambda_{{\rm Q'},t/r}-e^{(t/r)\mathcal{L}}\right\|_{1\to 1}\leq \Delta'/r.
    \end{equation}
    Due to Eq.~\eqref{suppleeq:t_error_d},
    \begin{align}
        \left\|\Lambda_{\mathrm{Q'},t/r}\right\|_{1\to 1} - \left\|e^{(t/r)\mathcal{L}}\right\|_{1\to 1}
        \le \left\|\Lambda_{\mathrm{Q'},t/r}-e^{(t/r)\mathcal{L}}\right\|_{1\to 1}= \Delta'/r
    \end{align}
    holds. Hence, we have
    \begin{align}
        \left\|\Lambda_{\mathrm{Q'},t/r}\right\|_{1\to 1} 
        \le \left\|e^{(t/r)\mathcal{L}}\right\|_{1\to 1} +\Delta'/r
        = 1 +\Delta'/r
    \end{align}
    because $\|e^{(t/r)\mathcal{L}}\|_{1\to 1}= 1$ holds from Lemma \ref{lem_CPTP_norm} since $e^{(t/r)\mathcal{L}}$ is CPTP.
    Observing that
    \begin{align}
        \|\mathcal{C}^r - \mathcal{D}^r\|_{1\to 1}
        \le \sum_{k=1}^r \|\mathcal{C}\|_{1\to 1}^{k-1}\|\mathcal{C} - \mathcal{D}\|_{1\to 1} \|\mathcal{D}\|_{1\to 1}^{r-k}
        \le r\max\{\|\mathcal{C}\|_{1\to 1}, \|\mathcal{D}\|_{1\to 1}\}^{r-1} \|\mathcal{C} - \mathcal{D}\|_{1\to 1}
    \end{align}
    for any linear maps $\mathcal{C}$ and $\mathcal{D}$,
    we obtain
    \begin{align}
        \left\|\left(\Lambda_{\mathrm{Q'},t/r}\right)^r-e^{t\mathcal{L}}\right\|_{1\to 1}
        &\leq  r \max\{\left\|\Lambda_{\mathrm{Q'},t/r}\right\|_{1\to 1} , 1\}^{r-1}\left\|\Lambda_{\mathrm{Q'},t/r}-e^{(t/r)\mathcal{L}}\right\|_{1\to 1}\nonumber\\
        &\leq \Delta' \left(1+\Delta'/r\right)^{r-1}\leq \Delta' e^{\Delta'}\notag\\
        & \leq \Delta,
    \end{align}
    where we set $\Delta':=(2/3)\Delta \in (0,1/e)$ in the final inequality.    
    Therefore, taking $\mathrm{Q}$ as $\mathrm{Q} \ge \mathrm{Q}' /2  = \frac{\ln(3r /2\Delta)}{\ln\ln(3r /2\Delta)}$, the lemma holds.
\end{proof}

\theoremstyle{plain}
\newtheorem*{T1}{Theorem~\ref{thm: main2}}

Using quantum circuits sampled by Algorithm~\ref{alg_circuit_generation}, we can estimate the target quantity ${\rm Tr}[Oe^{t\mathcal{L}}(\rho_0)]$ as follows.
First, we generate a quantum circuit $\widetilde{\mathcal{W}}$ by Algorithm~\ref{alg_circuit_generation}.
This circuit contains (possibly multi-round) mid-circuit measurement and qubit reset as in Fig.~\ref{fig:main_circ}, and the corresponding POVM $\{\Pi_{b}\}_{b=0,1}$ in each mid-circuit measurement is given by Eq.~\eqref{eq:midcircuit_M_POVM}.
Then, running the circuit with an initial state $\ket{+}\bra{+}_{\rm anc}\otimes \rho_0$ and measuring the observable $X_{\rm anc}\otimes O$ at the end of circuit, we obtain measurement outcomes $(b_X,b_O)$ and a collection of mid-circuit measurement outcomes $\bm{b}=(b_1,b_2,...)$ ($b_i\in \{0,1\}$).
By repeating the above procedure $N$ times independently, we calculate the following quantity 
\begin{equation}
    \varphi_N:=\frac{C}{N}\sum_{i=1}^{N}b_X^{(i)} b_O^{(i)} \delta_{\bm{b}^{(i)},\bm{0}},
\end{equation}
where the superscript $i$ denotes the index of trials and the positive value $C$ is also the output of Algorithm~\ref{alg_circuit_generation}.

From the construction, it can be confirmed that the mean of $\varphi_N$ becomes
\begin{equation}
    \mathbb{E}\left[\varphi_N\right]=\sum_{v \in \mathrm{S}_{\Delta}} c_v \mathrm{Tr}[(X_{\rm anc} \otimes O) \widetilde{\mathcal{W}}_v (\ket{+}\bra{+}_{\rm anc} \otimes \rho)]
\end{equation}
and this is $(\Delta \|O\|)$-close to the target value as
\begin{equation}
    \left|\mathbb{E}\left[\varphi_N\right]-{\rm Tr}\left[Oe^{t\mathcal{L}}(\rho_0)\right]\right| \leq \Delta \|O\|
\end{equation}
due to Eq.~\eqref{eq:closeness} and the tracial matrix H\"older's inequality.
In addition, from the Hoeffding's inequality and the above evaluation, we can say 
\begin{equation}
    N=\frac{2C^2\|O\|^2\log(2/\delta)}{(\varepsilon')^2}
\end{equation}
samples are sufficient in order that
\begin{equation}
    {\rm Pr}\left(\left|\varphi_N-{\rm Tr}\left[Oe^{t\mathcal{L}}(\rho_0)\right]\right|<\varepsilon'+\Delta\|O\|\right)\geq 1-\delta
\end{equation}
holds for given parameters $\varepsilon'>0$ and $\delta>0$.
Thus, for a target additive error $\varepsilon$, taking $\Delta=\varepsilon/2\|O\|$ and $\varepsilon'=\varepsilon/2$, we conclude that our algorithm estimates the target expectation value within an additive error $\varepsilon$ with at least $1-\delta$ probability, i.e., 
\begin{equation}
    {\rm Pr}\left(\left|\varphi_N-{\rm Tr}\left[Oe^{t\mathcal{L}}(\rho_0)\right]\right|<\varepsilon\right)\geq 1-\delta.
\end{equation}
Also, the sampling overhead $C=\sum_{v\in\mathrm{S}_{\Delta}} c_v$ can be $\mathcal{O}(1)$ by setting $r=\mathcal{O}(\|\mathcal{L}\|_{\rm pauli}^2 t^2)$ because of Eq.~\eqref{eq:delta_sampleoverhead}.
Consequently, the above construction and analysis complete the proof of Theorem~\ref{thm: main2}.

\clearpage

\begin{figure}[htbp]
    \begin{algorithm}[H]
    \caption{Efficient sampling of the random circuit for the approximate decomposition of $e^{t\mathcal{L}}$}
    \label{alg_circuit_generation}
    \begin{algorithmic}[1]
    \REQUIRE Hamiltonian $H$ and jump operators $\{L_k\}_{k=1}^K$ defined by~\eqref{eq:def_H_and_L},
    simulation time $t >0$, number of time segments (integer) $r \geq t\|\mathcal{L}\|_{\rm pauli}$, accuracy parameter $\Delta\in (0,1/e)$.
    \ENSURE Positive value $C$ and description of a sample from the random quantum circuit $\widetilde{\mathcal{W}}$ including classical post-processing, such that its average provides a $\Delta$-approximation of Lindblad dynamics as
    $$
    \left\|C\times \mathbb{E}_{\widetilde{\mathcal{W}}}[\mathrm{Tr_{anc}}[(X_{\rm anc}\otimes \bm{1})\widetilde{\mathcal{W}}(\ket{+}\bra{+}_{\rm anc} \otimes (\bullet))]] - e^{t\mathcal{L}} (\bullet)\right\|\leq \Delta.
    $$

    \STATE Initialize an integer $\mathrm{Q} \leftarrow \lceil{\ln(3r/2\Delta)}/{\ln\ln (3r/2\Delta)}\rceil$, a phase parameter $\phi \leftarrow 0$, and a list {QCList} $\leftarrow []$
    
    \STATE Define functions for $l\in \{0,1,...,{\rm Q}\}$
    $$\tau_l:=\frac{t/r}{2l+1}\times 2 \left( \alpha_{0}  + \frac{1}{2}\sum_{k=1}^{K} \alpha_{k}^2\right),~~~\|R_{kl}\|_{\rm pauli}:=\frac{\tau_l^2}{4}+\frac{\tau_l^3}{2}+\frac{5\tau_l^4}{64}+\frac{\tau_l^5}{32}+\frac{\tau_l^6}{256},$$
    $${C}_l:=\frac{(t/r)^{2l} \|\mathcal{L}\|^{2l}_{\rm pauli}}{(2l)!}\left(2\frac{\alpha_0}{\alpha} \sqrt{1 + \tau_l^2}+\sum_{k=1}^K \frac{\alpha_k^2}{\alpha}\left(1+\|R_{kl}\|_{\rm pauli}+\frac{\tau_l^2}{4}\right)\right)\geq 0$$
        \FOR{$i = 0$ to $r-1$}
            \STATE Sample $l\in \{0,1,...,\mathrm{Q}\}$ from the probability distribution $\{C_l/(\sum_l C_l)\}$
            \STATE Sample $k\in \{0,1,...,K\}$ from the probability distribution conditioned by $l$
            \begin{equation*}
            q_{0l}\propto 2\frac{\alpha_0}{\alpha}\sqrt{1+\tau_l^2},~~~q_{kl}\propto\frac{\alpha_k^2}{\alpha}\left(1+\|R_{kl}\|_{\rm pauli}+\frac{\tau_l^2}{4}\right)~~~(\mbox{for}~k=1,2,...,K)
            \end{equation*}
            \STATE Sample $\nu\in \{1,2,3\}$ from the probability distribution conditioned by $l,k$
            \begin{equation*}
                p_{\Gamma,0l,\nu}:=\frac{1}{2}\left(\delta_{1\nu}+\delta_{2\nu}\right),~~~p_{\Gamma,kl,\nu}:=\frac{\delta_{1\nu}+\|R_{kl}\|_{\rm pauli}\delta_{2\nu}+(\tau_l^2/4)\delta_{3\nu}}{1+\|R_{kl}\|_{\rm pauli}+\tau_l^2/4}~~~(\mbox{for}~k=1,2,...,K)
            \end{equation*}
            \IF{$k=0$}
                \STATE Sample a Pauli string $P_{0j}$ from $H$ according to the probability distribution $\{ p_{0j}:=|\alpha_{0j}|/(\sum_j |\alpha_{0j}|)\}$
                \STATE Append $\widetilde{\mathcal{F}}_{U, \bm{1}}$ ($\nu=1$) or $\widetilde{\mathcal{F}}_{\bm{1},U}$ ($\nu=0$) to QCList,
                where $U = \exp [-i \theta_l \mathrm{sgn}(\alpha_{0j}) P_{0j}]$,
                $$
                \theta_l:=\arccos(\{1+\tau_l^2\}^{-1/2}),~~~\widetilde{\mathcal{F}}_{A, B}(\bullet):=
                (\ket{0}\bra{0}\otimes A + \ket{1}\bra{1}\otimes B) \bullet (\ket{0}\bra{0}\otimes A^\dagger + \ket{1}\bra{1}\otimes B^\dagger)
                $$
            \ELSIF{$\nu=1$ with $k>0$}
            \STATE Append $\mathcal{I}_{\rm anc}\otimes \mathcal{U}_{\rm II}$, with mid-circuit measurement followed by qubit reset and classical post-processing (see Fig.~\ref{fig:three_type_circuits} and Remark~\ref{rem:cptn}) for the CPTN map $\mathcal{B}_{kl}^{\mathrm{(approx)}}$, to QCList
            
            \ELSIF{$\nu=2$ with $k>0$}
            \STATE Call Algorithm~\ref{alg:sample_Rkl} and obtain $e^{i\theta}\overline{Q}\otimes P$ with a phase $\theta$ and $n$-qubit Pauli strings $\{P,Q\}$ 
            \STATE Append $\widetilde{\mathcal{F}}_{P,Q}$ to QCList and Update $\phi \leftarrow \phi + \theta$ 
            \ELSIF{$\nu=3$ with $k>0$}
            \STATE Sample $j_1,...,j_4\in \{1,2,...,M\}$ independently from the identical distribution $\{p_{kj}:=|\alpha_{kj}|/(\sum_j |\alpha_{kj}|)\}$
            \STATE Append $\widetilde{\mathcal{F}}_{A,B}$ with $A=P_{kj_1}P_{kj_2}$ and $B=P_{kj_4}P_{kj_3}$ to QCList and Update $\phi$ as $$\phi \leftarrow \phi -\theta_{kj_1}+\theta_{kj_2}-\theta_{kj_3}+\theta_{kj_4}+\pi,~~~\theta_{kj}:={\rm arg}\left(\alpha_{kj}/|\alpha_{kj}|\right)$$
            \ENDIF
            \FOR{$j = 0$ to $2l-1$}
                \STATE Call Algorithm~\ref{alg:sample_G} and obtain $e^{i\theta}\overline{Q}\otimes P$ with a phase $\theta$ and $n$-qubit Pauli strings $\{P,Q\}$ 
                \STATE Append $\widetilde{\mathcal{F}}_{P,Q}$ to QCList and Update $\phi \leftarrow \phi + \theta$
                ~\footnote{
            We note that if we can calculate $\prod_{j=0}^{2l-1}e^{i\theta_j }\overline{Q_j} \otimes P_j =: e^{i\theta'}\overline{Q'} \otimes P'$ classically, we can simply append $\widetilde{\mathcal{F}}_{P',Q'}$ instead of $\widetilde{\mathcal{F}}_{P_{2l-1},Q_{2l-1}} \circ \cdots \circ \widetilde{\mathcal{F}}_{P_0, Q_0}$ on the line 21.
            This substitution improves the circuit depth.}
            \ENDFOR
        \ENDFOR
        \STATE Append a Phase gate ${\rm Phase}(e^{i\phi}) \otimes \bm{1}$ to QCList, where ${\rm Phase}(e^{i\phi}):=e^{i\phi}\ket{0}\bra{0}_{\rm anc}+\ket{1}\bra{1}_{\rm anc}$
    \RETURN Positive value $C:=(\sum_l C_l)^r$ and 
    $\widetilde{\mathcal{W}}=$ QCList[length(QCList)$-1$] $\circ\cdots\circ$ QCList[$0$]
    \end{algorithmic}
    \end{algorithm}
\end{figure}

\begin{figure}[htbp]
    \begin{algorithm}[H]
    \caption{Efficient sampling of the random unitary $X_{G}$ for $G=S(\mathcal{L})$}
    \label{alg:sample_G}
    \begin{algorithmic}[1]
    \REQUIRE Hamiltonian $H$ and jump operators $\{L_k\}_{k=1}^K$ defined by~\eqref{eq:def_H_and_L}.
    \ENSURE Description of a sample from the random unitary $X_{G}$ such that $\mathbb{E}[X_{G}]={G}/{\|\mathcal{L}\|_{\rm pauli}}$, where $G=S(\mathcal{L})$ and $\|\mathcal{L}\|_{\rm pauli}$ is defined by Eqs.~\eqref{eq:G_transfer_L} and~\eqref{eq:G_norm_pauli}.
    Here, $X_G$ has the form of \begin{equation}\label{alg_eq:stdform}
        {e^{i\theta}\overline{Q}\otimes P}
    \end{equation} 
    for some $\theta\in \mathbb{R}$ and $n$-qubit Pauli strings $P,Q\in \{I,X,Y,Z\}^{\otimes n}$, with probability 1.
    
        \STATE Sample $ k \in\{0,1,...,K \}$ from the probability distribution $p_k$,
        \begin{equation*}
            p_k = \frac{2\alpha_0}{\|\mathcal{L}\|_{\rm pauli}}\delta_{0k}+\sum_{m=1}^K \frac{2\alpha^2_m}{\|\mathcal{L}\|_{\rm pauli}}\delta_{mk}.
        \end{equation*}
        \IF{$k=0$}
            \STATE Sample $l \in \{1,2\}$ with equal probability.
            \STATE Sample $j\in \{ 1,...,m \}$ from the probability distribution $$p_{0j}=\frac{|\alpha_{0j}|}{\sum_{j=1}^m |\alpha_{0j}|}.$$
            \STATE Set $X_{G}$ as
            \begin{equation*}
                X_G:=
                \begin{cases}
                    \mathrm{exp}[{i(\theta_{0j}-\pi/2)}]\cdot \bm{1}\otimes P_{0j},&\mbox{if}~~~l=1\\[6pt]
                    \mathrm{exp}[{i(\theta_{0j}+\pi/2)}]\cdot P_{0j}^{\rm T}\otimes \bm{1},&\mbox{if}~~~l=2
                \end{cases}
                ~~~\mbox{and}~~~e^{i\theta_{0j}}:=\frac{\alpha_{0j}}{|\alpha_{0j}|} (\in \{1, -1\})
            \end{equation*}
        \ELSIF{$k\neq 0$}
            \STATE Sample $l \in \{1,2,3\}$ with the probability $\{1/2,1/4,1/4\}$.
            \STATE Sample $j_1,j_2\in \{ 1,..., M \}$ independently from the identical probability distribution $$p_{kj}=\frac{|\alpha_{kj}|}{\sum_{j=1}^M |\alpha_{kj}|}$$
            \STATE Set $X_G$ as
            \begin{equation*}
                X_G:=
                \begin{cases}
                    \mathrm{exp}[i(-\theta_{kj_1}+\theta_{kj_2})]\cdot \overline{P_{kj_1}}\otimes P_{kj_2},&\mbox{if}~~~l=1\\[6pt]
                    \mathrm{exp}[i(-\theta_{kj_1}+\theta_{kj_2}+\pi)]\cdot \bm{1}\otimes {P_{kj_1}}P_{kj_2},&\mbox{if}~~~l=2\\[6pt]
                    \mathrm{exp}[i(\theta_{kj_1}-\theta_{kj_2}+\pi)]\cdot P_{kj_1}^{\rm T}\overline{P_{kj_2}}\otimes \bm{1},&\mbox{if}~~~l=3,
                \end{cases}
                ~~~\mbox{and}~~~e^{i\theta_{kj}}:=\frac{\alpha_{kj}}{|\alpha_{kj}|}.
            \end{equation*}
        \ENDIF
    \STATE Modify the phase of $X_G$ to match the form of Eq.~\eqref{alg_eq:stdform}.
    \RETURN The description of $X_{G}$ with a complex phase and Pauli strings.
    \end{algorithmic}
    \end{algorithm}
\end{figure}

\begin{figure}[htbp]
    \begin{algorithm}[H]
    \caption{Efficient sampling of the random unitary $X_{R}$ for the correction superoperator $R=S(\mathcal{R})$}
    \label{alg:sample_Rkl}
    \begin{algorithmic}[1]
    \REQUIRE Hamiltonian $H$ and jump operators $\{L_k\}_{k=1}^K$ defined by~\eqref{eq:def_H_and_L}, $t>0$, $r\in\mathbb{N}$, $l\in \{0,1,...\}$, and $k\in \{1,2,...,K\}$.
    \ENSURE Description of a sample from the random unitary $X_{R}$ such that $\mathbb{E}[X_{R}]={R_{kl}}/{\|R_{kl}\|_{\rm pauli}}$, where $R_{kl}$ and $\|R_{kl}\|_{\rm pauli}$ is defined by Eqs.~\eqref{eq:correctionop_kl} and~\eqref{eq:correction_norm_kl}.
    Here, $X_R$ has the form of $${e^{i\theta}\overline{Q}\otimes P}$$
    for some $\theta\in \mathbb{R}$ and $n$-qubit Pauli strings $P,Q\in \{I,X,Y,Z\}^{\otimes n}$, with probability 1.
        \STATE Set 
        $$
        \tau_l:=\frac{t/r}{2l+1}\times 2 \left( \alpha_{0}  + \frac{1}{2}\sum_{k=1}^{K} \alpha_{k}^2\right),~~~P_{k0}:=i\bm{1},~~~\theta_{k0}:=0
        $$
        \STATE Sample $\lambda \in \{0, 1\}$ from the probability distribution
        \begin{equation*}
            p_\lambda = \frac{1+\tau_l^2/16}{1+2\tau_l+5\tau_l^2/16+\tau_l^3/8+\tau_l^4/64}\left\{\left(1+\frac{\tau_l}{2}\right)^2\delta_{0\lambda}+\tau_l\delta_{1\lambda}\right\}
        \end{equation*}
        \STATE Sample $i \in \{0, 1, 2 \}$ from the probability distribution 
        \begin{equation*}
            p_{i}=\frac{1}{2+\tau_l^2/8}\left\{\delta_{0i}+\delta_{1i}+\frac{\tau_l^2}{8}\delta_{2i}\right\}
        \end{equation*}
        \IF{$\lambda=0$}
            \STATE Sample $(j_1,j_2),(j'_1,j'_2)\in \{0,1,...,M\}^2$ independently from the identical joint probability distribution 
                \begin{equation*}
                    p_{j_1j_2}:=\frac{1}{1+\tau_l/2}\delta_{0j_1}\delta_{0j_2}+\frac{\tau_l/2}{1+\tau_l/2}\sum_{m_1,m_2=1}^{M} p_{km_1}p_{km_2} \delta_{m_1j_1}\delta_{m_2j_2}
                \end{equation*}
                \STATE Sample $j_3, ..., j_{10} \in \{1,...,M \}$ independently from the identical distribution
                $p_{kj}=\frac{|\alpha_{kj}|}{\sum_{j=1}^M |\alpha_{kj}|}$
        
                \STATE Set $X_R$ as
                \begin{equation*}
                    X_R:=
                    \begin{cases}
                    ~\overline{P_{kj_1}P_{kj_2}}\otimes P_{kj'_1}P_{kj'_2}\prod_{\nu=3}^{6}P_{kj_{\nu}}&\\
                    ~~~~~ \times \mathrm{exp}[{i(\theta_{kj_1}-\theta_{kj_2}-\theta_{kj'_1}+\theta_{kj'_2}-\theta_{kj_3}+\theta_{kj_4}-\theta_{kj_5}+\theta_{kj_6})}],&\mbox{if}~~i=0\\[8pt]
                    ~\overline{P_{kj_1}P_{kj_2}\prod_{\nu=3}^{6}P_{kj_{\nu}}}\otimes P_{kj'_1}P_{kj'_2}&\\
                    ~~~~~ \times \mathrm{exp}[{i(\theta_{kj_1}-\theta_{kj_2}-\theta_{kj'_1}+\theta_{kj'_2}+\theta_{kj_3}-\theta_{kj_4}+\theta_{kj_5}-\theta_{kj_6})}],&\mbox{if}~~i=1\\[8pt]
                    ~\overline{P_{kj_1}P_{kj_2}\prod_{\nu=3}^{6}P_{kj_{\nu}}}\otimes P_{kj'_1}P_{kj'_2}\prod_{\nu=7}^{10}P_{kj_{\nu}}&\\
                    ~~~~~ \times \mathrm{exp}[{i(\theta_{kj_1}-\theta_{kj_2}-\theta_{kj'_1}+\theta_{kj'_2}+\theta_{kj_3}-\theta_{kj_4}+\theta_{kj_5}-\theta_{kj_6}
                    -\theta_{kj_7}+\theta_{kj_8}-\theta_{kj_9}+\theta_{kj_{10}})}],&\mbox{if}~~i=2
                    \end{cases}
                \end{equation*}
        \ELSIF{$\lambda=1$}
            \STATE Sample $j_1,j'_1,j_2,..., j_9 \in \{ 1, ..., M\}$ independently from the identical probability distribution $p_{kj}=\frac{|\alpha_{kj}|}{\sum_{j=1}^M |\alpha_{kj}|}$

            \STATE Set $X_R$ as 
            \begin{equation*}
                    X_R:=
                    \begin{cases}
                    ~\overline{P_{kj_1}}\otimes P_{kj'_1}\prod_{\nu=2}^{5}P_{kj_{\nu}}&\\
                    ~~~~~ \times \mathrm{exp}[{i(-\theta_{kj_1}+\theta_{kj'_1}
                    -\theta_{kj_2}+\theta_{kj_3}-\theta_{kj_4}+\theta_{kj_5})}],&\mbox{if}~~i=0\\[8pt]
                    ~\overline{P_{kj_1}\prod_{\nu=2}^{5}P_{kj_{\nu}}}\otimes P_{kj'_1}&\\
                    ~~~~~ \times \mathrm{exp}[{i(-\theta_{kj_1}+\theta_{kj'_1}
                    +\theta_{kj_2}-\theta_{kj_3}+\theta_{kj_4}-\theta_{kj_5})}],&\mbox{if}~~i=1\\[8pt]
                    ~\overline{P_{kj_1}\prod_{\nu=2}^{5}P_{kj_{\nu}}}\otimes P_{kj'_1}\prod_{\nu=6}^{9}P_{kj_{\nu}}&\\
                    ~~~~~ \times \mathrm{exp}[{i(-\theta_{kj_1}+\theta_{kj'_1}
                    +\theta_{kj_2}-\theta_{kj_3}+\theta_{kj_4}-\theta_{kj_5}
                    -\theta_{kj_6}+\theta_{kj_7}-\theta_{kj_8}+\theta_{kj_9})
                    }],&\mbox{if}~~i=2
                    \end{cases}
                \end{equation*}
        \ENDIF
    \STATE Modify the phase of $X_R$ to match the form of ${e^{i\theta} \overline{Q}\otimes P}$
    \RETURN The description of $X_{R}$ with a complex phase and Pauli strings
    \end{algorithmic}
    \end{algorithm}
\end{figure}

\clearpage
\section{$(m,M,K)$-independent Lindblad simulation with logarithmic accuracy scaling}\label{apdx:D_mMK_indep}

We here provide an effective way to approximately simulate the CPTN map $\mathcal{B}^{(\rm approx)}_{kl}$. This dissipation $\mathcal{B}^{(\rm approx)}_{kl}$ gives the $M$-dependency of our algorithm introduced before.
For our expansion of the dynamical map $e^{\mathcal{L}t}$, it is sufficient to construct a linear map $\Upsilon_{kl}$ that can be simulated by using quantum circuits and realize $\mathcal{B}^{\rm (approx)}_{kl}$ in the following sense: 
\begin{equation}
        \Upsilon_{kl}:
        \begin{pmatrix}
            A_{00}&A_{01}\\    
            A_{10}&A_{11}
        \end{pmatrix}
        \mapsto 
        \begin{pmatrix}
            *&\mathcal{B}^{(\rm approx)}_{kl}(A_{01})\\    
            \mathcal{B}^{(\rm approx)}_{kl} (A_{10})&*
        \end{pmatrix}.
\end{equation}
We have already provided one example of $\Upsilon_{kl}$ using the unitary channel $\mathcal{U}_{kl}(\bullet)=U_{kl}\bullet U_{kl}^\dagger$ (see Eq.~\eqref{eq:effectivesim_CPTN_main}), that is, 
\begin{equation}
    {\rm tr}_{\overline{\rm anc,sys}}\left[\left(\bm{1}_{\rm anc,sys}\otimes I_{\rm P}\otimes\ket{0}\bra{0}^{\otimes 2+\lceil \log_2 M\rceil}\right)\left(\mathcal{I}_{\rm anc}\otimes \mathcal{U}_{kl}\right)(\ket{\bm{0}}\bra{\bm{0}}\otimes \bullet)\right],
\end{equation}
also see Remark~\ref{rem:cptn}.
The implementation of $U_{kl}$ requires $\mathcal{O}(\log M)$ ancilla qubits and $\mathcal{O}(M\log M)$ gates when we use LCU to encode $L_k/\alpha_k$.
In the following, we show another proposal of $\Upsilon_{kl}$ whose implementation has no dependence on $M$ in both gate and ancilla counts.

The alternative $\Upsilon_{kl}$ uses the single-ancilla qubit denoted by anc, while the previous construction has trivial action on the system anc.
The core idea to construct $\Upsilon_{kl}$ without $M$-dependence is to use the random-sampling implementation of controlled time evolution operators $e^{-itL_k^{\rm R}/\alpha_k}$ and $e^{-itL_k^{\rm I}/\alpha_k}$ (see Appendix~\ref{subsection:wan} or Refs.~\cite{Wan2022-tx,chakraborty2024implementing}).
Here, $L_k^{\rm R, I }$ are the real or imaginary parts of $L_k$ defined as
\begin{equation}
    L_k^{\rm R}=\sum_{j=1}^M {\Re}[\alpha_{kj}] P_{kj},~~~L_k^{\rm I}=\sum_{j=1}^M {\Im}[\alpha_{kj}] P_{kj},
\end{equation}
where
\begin{equation}
    \sum_{j=1}^M |{\Re}[\alpha_{kj}]|\leq \alpha_k,~~~\sum_{j=1}^M |{\Im}[\alpha_{kj}]|\leq \alpha_k.
\end{equation}
We note that the 1-qubit controlled time evolution is generated by $(I-Z)/2\otimes L_k^{\rm R,I}/\alpha_k$, that is,
\begin{equation}
    \ket{0}\bra{0}\otimes \bm{1}+\ket{1}\bra{1}\otimes e^{-itL_k^{\rm R,I}/\alpha_k}=e^{-it\frac{I-Z}{2}\otimes \frac{L_k^{\rm R,I}}{\alpha_k}}.
\end{equation}
The previous results reviewed in Appendix~\ref{subsection:wan} say that for {any one of $(n+1)$-qubit Hamiltonians} $\{(I-Z)/2\otimes L_k^{\rm R,I}/\alpha_k\}_k$ and a given positive integer $r_{\rm BE}$, 
there exists a probability distribution $p_i$ and $(n+1)+1$-qubit ($+1$ for anc) unitary gates $V_i$ such that 
\begin{equation}\label{apdx:HSforMindep}
    C_{\rm BE}\sum_{i}p_i V_i\bullet V_i^\dagger:
    \begin{pmatrix}
            A'_{00}&A'_{01}\\    
            A'_{10}&A'_{11}
        \end{pmatrix}
        \mapsto 
        \begin{pmatrix}
            *&e^{-i\frac{I-Z}{2}\otimes \frac{L_k^{\rm R,I}}{2\alpha_k}}A'_{01}e^{i\frac{I-Z}{2}\otimes \frac{L_k^{\rm R,I}}{2\alpha_k}}\\    
            e^{-i\frac{I-Z}{2}\otimes \frac{L_k^{\rm R,I}}{2\alpha_k}}A'_{10}e^{i\frac{I-Z}{2}\otimes \frac{L_k^{\rm R,I}}{2\alpha_k}}&*
        \end{pmatrix},
\end{equation}
where $C_{\rm BE}\leq {\exp}({r_{\rm BE}^{-1}/2})$ holds (the time is fixed to $t=1/2$); see around Eq.~\eqref{eq:re_expression_wan} for this expression of the previous works.
The non-Clifford gate complexity of the $V_i$ gates is given by at most $\mathcal{O}(r_{\rm BE})$.
Thus, we can simulate the controlled time evolution operators with only a single-ancilla qubit and the gate counts independent of $M$.

Once we obtain the controlled time evolution of $e^{-iL_k^{\rm R,I}/(2\alpha_k)}$, we can approximately construct the unitary Eq.~\eqref{eq:defofWL} to encode $L_k/\alpha_k$ (more precisely, $(1/2\pi)L_k/\alpha_k$) based on the standard technique of quantum matrix arithmetic e.g., Ref.~\cite{gilyen2019quantum}. 
Combining this and the same strategy as Lemma~\ref{lemma:oaa_for_B0B1}, we prove the following lemma.
\begin{lem}\label{lem:anotherBkl_imple}
    Let $k\in\{1,...,K\}$ be the index of $L_k$.
    For any positive integer $r_{\rm BE}$ and $0\leq \tau_l \leq \frac{3}{4\pi^2}$, we can construct a mixed unitary channel $\mathcal{U}'_{{\rm II},kl}$ over $n+7$ qubits such that
    \begin{equation}\label{eq:lem14_1}
        \Upsilon_{kl}(\bullet):=[C^{}_{\rm BE}]^{\mathcal{O}(\log(1/\varepsilon'))}{\rm tr}_{\overline{\rm anc,sys}}\left[\left(\bm{1}_{\rm anc,sys}\otimes I_{\rm P}\otimes\ket{0}\bra{0}^{\otimes 5}\right)\mathcal{U}'_{{\rm II},kl}( \bullet\otimes \ketbra{0}^{\otimes 6})\right]
    \end{equation}
    satisfies 
    \begin{equation}\label{eq:lem14_2}
        \Upsilon_{kl}:
        \begin{pmatrix}
            A_{00}&A_{01}\\    
            A_{10}&A_{11}
        \end{pmatrix}
        \mapsto 
        \begin{pmatrix}
            *&\mathcal{C}^{(\rm approx)}_{kl}(A_{01})\\    
            \mathcal{C}^{(\rm approx)}_{kl} (A_{10})&*
        \end{pmatrix}.
    \end{equation}
    Here, the positive value $C_{\rm BE}$ satisfies $1\leq C_{\rm BE}\leq \exp(r_{\rm BE}^{-1}/2)$, and the CPTN map $\mathcal{C}_{kl}^{\rm (approx)}$ is $\varepsilon'$-close to $\mathcal{B}_{kl}^{\rm (approx)}$ (defined in Eq.~\eqref{apdx:B_kl_approx_def}) in the diamond norm.
    The implementation of the mixed unitary channel $\mathcal{U}'_{{\rm II},kl}$ uses at most 
    $\mathcal{O}(r_{\rm BE}[\log(r_{\rm BE}/\varepsilon')]^2)$ one- or two-qubit gates including
    $\mathcal{O}(r_{\rm BE}\log(1/\varepsilon'))$ non-Clifford gates.
\end{lem}

\begin{proof}
    To simplify the proof, we first introduce the notion of block encoding~\cite{gilyen2019quantum}.
    For a given $n$-qubit operator $A$ with the norm $\|A\|\leq 1$, we say that an $(n+a)$-qubit unitary $U$ is a block-encoding of $A$, if 
    \begin{equation}
        (\bra{0}^{\otimes a}\otimes \bm{1})U(\ket{0}^{\otimes a}\otimes \bm{1})=A.
    \end{equation}
    Assuming access to the time evolution operators $e^{-iL_k^{\rm R,I}/(2\alpha_k)}$ with 1-qubit control, the circuit in Fig.~\ref{Fig:BE_sin} becomes a block-encoding of $\sin(L_k^{\rm R,I}/(2\alpha_k))$ with a single-ancilla qubit.
    Then, Ref.~\cite{gilyen2019quantum} proves that there exists an odd real ${d := }\mathcal{O}(\log(1/\varepsilon'))$-degree polynomial $P(x)$ satisfying 
    \begin{equation}
        \max_{|x|\leq 1/2}\left|P(x)-\frac{2}{\pi}\arcsin(x)\right|\leq \varepsilon',~~\mbox{and}~~~\max_{-1\leq x\leq 1}|P(x)|\leq 1.
    \end{equation}
    Therefore, by using the quantum singular value transformation~\cite{gilyen2019quantum} for the polynomial $P(x)$, we have the block encoding of $P^{\rm R,I}_k\equiv P\left(\sin(L_k^{\rm R,I}/(2\alpha_k))\right)\approx (2/\pi)L_k^{\rm R,I}/(2\alpha_k)$.
    The circuit for the block-encoding of $P_k^{\rm R,I}$ consists of ${\mathcal{O}(d)= }\mathcal{O}(\log(1/\varepsilon'))$ queries to the circuit in Fig.~\ref{Fig:BE_sin},
    and ${\mathcal{O}(d)= }\mathcal{O}(\log(1/\varepsilon'))$ one- or two-qubit elementary gates.
    Moreover, using the LCU method for adding $P^{\rm R}_k$ and $P^{\rm I}_k$, we obtain a circuit for a block-encoding of 
    \begin{equation}
        P_k:=\frac{P^{\rm R}_k}{2}+i\frac{P^{\rm I}_k}{2}\approx \frac{1}{2\pi}\frac{L_k}{\alpha_k}
    \end{equation}
    using 3 ancilla qubits in total.

    Now we obtain the block-encoding of $L_k$, we then consider the circuit in Fig.~\ref{Fig:BE_channel},
    similar to the proof of Lemma~\ref{lemma:oaa_for_B0B1}.
    In this circuit, the PREPARE circuits are defined as
    \begin{equation}\label{eq:prepare_prime_1}
        \mbox{PRE}_1\ket{0}=\sqrt{\frac{(2\pi)^2\tau_l}{4}}\ket{0}+\sqrt{1-\frac{(2\pi)^2\tau_l}{4}}\ket{1},
        ~~~\mbox{PRE}_2\ket{0}=\frac{1}{\sqrt{1+(2\pi)^2\tau_l}}(\ket{0}+2\pi\sqrt{\tau_l}\ket{1}), 
    \end{equation}
    and
    \begin{equation}\label{eq:prepare_prime_2}
        R_y\ket{0}=\frac{\sqrt{1+(2\pi)^2\tau_l}}{2}\ket{0}+\sqrt{\frac{{3-(2\pi)^2\tau_l}}{4}}\ket{1}.
    \end{equation}
    Note that all of these states are valid quantum states due to the assumption of $0\leq\tau_l\leq \frac{3}{4\pi^2}$.
    Denoting this circuit as $W'_{kl}$, it satisfies 
    \begin{equation}\label{eq:U'_kl_channelenc}
        \widetilde{\Pi}'W'_{kl}\Pi'=\frac{1}{2}\sum_{i=0,1}\ket{0}\bra{0}\otimes \ket{i}\bra{0}_{\rm P}\otimes (\ket{0}\bra{0})^{\otimes 4} \otimes C_{kl,i},
    \end{equation}
    where $\widetilde{\Pi}'=\ket{0}\bra{0}\otimes I_{\rm P}\otimes (\ket{0}\bra{0})^{\otimes 4}\otimes \bm{1}$, ${\Pi}'=(\ket{0}\bra{0})^{\otimes 6}\otimes \bm{1}$, and 
    \begin{equation}
        C_{kl,0}:= \bm{1}-\frac{1}{2}P_k^\dagger P_k\cdot 4\pi^2\tau_l,~~~C_{kl,1}=P_k\cdot 2\pi\sqrt{\tau_l}.
    \end{equation}
    Thus, considering the circuit $U'_{kl}:=W'_{kl}(2{\Pi}'-\bm{1})(W'_{kl})^\dagger (2\widetilde{\Pi}'-\bm{1})W'_{kl}$,
    we arrive at
    \begin{equation}
        {\rm tr}_{\overline{\rm sys}}[\widetilde{\Pi}'U'_{kl}((\ket{0}\bra{0})^{\otimes 6}\otimes A)(U'_{kl})^\dagger] = \sum_{i=0,1} C'_{kl,i}A (C'_{kl,i})^\dagger \equiv \mathcal{C}_{kl}^{\rm (approx)}(A)
    \end{equation}
    for any operator $A$, where
    \begin{equation}
        C'_{kl,i}=C_{kl,i}\left(\bm{1}-\frac{\tau_l^2}{8}\cdot (4\pi^2P_k^\dagger P_k)^2\right).
    \end{equation}
    The unitary channel for $U_{kl}'$ is a sequence of unitary channels with $\mathcal{O}(\log(1/\varepsilon'))$ uses of the controlled $e^{-iL_k^{\rm R,I}/(2\alpha_k)}$ and its inverse {because $U_{kl}'$ contains $\mathcal{O}(1)$ controlled block-encodings of $P_k$ and its inverse from Fig.~\ref{Fig:BE_channel}}. 
    Thus, by replacing all of the controlled $e^{\pm iL_k^{\rm R,I}/(2\alpha_k)}$ with the mixed unitary channel in the form of $\sum_i p_i V_i\bullet V_i^\dagger$ satisfying Eq.~\eqref{apdx:HSforMindep}, the unitary channel for $U_{kl}'$ becomes a mixed unitary channel $\mathcal{U}'_{{\rm II},kl}$ over the $(1+6)$-qubit ancilla system and the target system. 
    This $\mathcal{U}'_{{\rm II},kl}$ satisfies Eqs.~\eqref{eq:lem14_1} and~\eqref{eq:lem14_2}. 
    For a given positive integer $r_{\rm BE}$, $1\leq C_{\rm BE}\leq \exp[r_{\rm BE}^{-1}/2]$ holds as mentioned before.
    We note that in the unitary sequence of $U_{kl}'$, $e^{-iL_k^{\rm R,I}/(2\alpha_k)}$ appears at most $\mathcal{O}(d)$ times, resulting into the sampling overhead ${(C_{\rm BE})^{\mathcal{O}(d)}} = (C_{\rm BE})^{\mathcal{O}(\log(1/\varepsilon'))}$.
    As for the gate complexity, $\mathcal{U}_{{\rm II},kl}'$ uses $\mathcal{O}(r_{\rm BE}[\log(1/\varepsilon')\log(r_{\rm BE}/\varepsilon')]/\log\log(r_{\rm BE}/\varepsilon')) =\mathcal{O}(r_{\rm BE}[\log(r_{\rm BE}/\varepsilon')]^2)$ elementary gates including $\mathcal{O}(r_{\rm BE} \log(1/\varepsilon'))$
    non-Clifford gates, where we ignored a $\log\log\log$ term.
    This is because Eq.~\eqref{apdx:HSforMindep}, which is a special case of Algorithm~\ref{alg_circuit_generation}, uses $\mathcal{O}(r_{\rm BE}\log(r_{\rm BE}/\varepsilon'')/\log\log(r_{\rm BE}/\varepsilon''))$ gates {including $\mathcal{O}(r_{\rm BE}$)} non-Clifford gates for the truncation error $\varepsilon''$ of time evolution.
    Since $\mathcal{U}_{{\rm II},kl}'$ has ${\mathcal{O}(d) =} \mathcal{O}(\log(1/\varepsilon'))$ uses of Eq.~\eqref{apdx:HSforMindep}, $\varepsilon''$ should be taken as $\mathcal{O}(\varepsilon'/\log(1/\varepsilon'))$, leading to the complexity.

    Finally, we describe the error between $\mathcal{C}_{kl}^{\rm (approx)}$ and the desired $\mathcal{B}_{kl}^{\rm (approx)}$.
    For any operator $A$ on the target system and any environment $\mathcal{R}$,
    \begin{align}
        \|(\mathcal{C}_{kl}^{\rm (approx)}\otimes\mathcal{I}_{R})(A)-(\mathcal{B}_{kl}^{\rm (approx)}\otimes\mathcal{I}_{R})(A)\|_{1}
        &\leq \sum_{i=0,1}\|C'_{kl,i}A(C'_{kl,i})^\dagger - B'_{kl,i}A (B'_{kl,i})^\dagger\|_1\notag\\
        &\leq \sum_{i=0,1}2\|A\|_{1}\|C'_{kl,i}-B'_{kl,i}\|\notag\\
        &=\|A\|_1\cdot \mathcal{O}(\varepsilon'),
    \end{align}
    where in the final equality, we used $\|C'_{kl,i}-B'_{kl,i}\|= \mathcal{O}(\varepsilon')$ which can be proved by expanding the terms in $B'_{kl,i}$ and $C'_{kl,i}$ and using the fact $\|P_k-(1/2\pi )L_k/\alpha_k\|\leq \varepsilon'$.

\begin{figure}[htb]
\centering
\begin{quantikz}[column sep=0.6em, row sep=1em]
  & \gate{H} & \ctrl{1} & \gate{Y} & \ctrl{1} & \gate{H} & \qw \\
  & \qw      & \gate{e^{i L_k^{\mathrm{R,I}}/(2\alpha_k)}} & \qw
            & \gate{e^{-i L_k^{\mathrm{R,I}}/(2\alpha_k)}}  & \qw & \qw
\end{quantikz}
\caption{Block-encoding of $\sin\!\left(L_k^{\mathrm{R,I}}/(2\alpha_k)\right)$}
\label{Fig:BE_sin}
\end{figure}

\begin{figure}[htb]
\centering
\begin{quantikz}[column sep=0.6em, row sep=1em]
  \lstick{$\ket{\psi}$}& \gate[2]{\mbox{BE}~\mbox{of}~P_k}&\qw&\qw&\gate[2]{\mbox{BE}^\dagger~\mbox{of}~P_k}&\qw\\ 
  \lstick{$\ket{0}^{\otimes 3}$}& &\gate{\rm REF}&\qw&&\qw\\ \lstick{$\ket{0}^{\otimes 1}$}&\gate{\mbox{PRE}_1}&\octrl{-1}&\gate{\mbox{PRE}_1^\dagger}&\qw&\qw\\ 
  \lstick{$\ket{0}_{\rm P}^{\otimes 1}$}&\gate{\mbox{PRE}_2}&\octrl{-1}&\qw&\octrl{-2}&\qw\\ \lstick{$\ket{0}^{\otimes 1}$}&\gate{R_y}&\qw&\qw&\qw&\qw\\
\end{quantikz}
\caption{Quantum circuit $W_{kl}'$ for Eq.~\eqref{eq:U'_kl_channelenc}. The PREPARE circuits $\mathrm{PRE}_1,\mathrm{PRE}_2$, and $R_y$ are provided in Eqs.~\eqref{eq:prepare_prime_1} and ~\eqref{eq:prepare_prime_2}.}
\label{Fig:BE_channel}
\end{figure}

\end{proof}

Combining Lemma~\ref{lem:anotherBkl_imple} and Algorithm~\ref{alg:obs_estimaion}, we can prove the $M$-independent version of Theorem~\ref{thm: main2}.
In the step 5 of Algorithm~\ref{alg:obs_estimaion}, we run Algorithm~\ref{alg_circuit_generation} and obtain a quantum circuit $\widetilde{\mathcal{W}}_i$. 
This quantum circuit contains $r^{(i)}_{{\rm shot}}\leq r=\lceil 2\|\mathcal{L}\|^2_{\rm pauli}t^2\rceil$ uses of the unitary gate $U_{kl}$ and $3+\lceil \log_2 M\rceil$ additional ancilla qubits for the dissipation $\mathcal{B}_{kl}^{(\rm approx)}$.
Now, we replace all $U_{kl}$ with $\Upsilon_{kl}$ with $6$ additional ancilla qubits.
More precisely, we replace all $U_{kl}$ with $\mathcal{U}'_{{\rm II},kl}$ and finally multiply the positive value $[C_{\rm BE}]^{r^{(i)}_{{\rm shot}}\mathcal{O}(\log(1/\varepsilon'))}$ with the single-shot measurement outcomes.
Then, we perform the step 6 with the modified circuit and finally obtain an estimate.
The final estimate is originally given by $\varphi_N=\frac{C}{N}\sum_{i=1}^N b_{X}^{(i)}b_{O}^{(i)}\delta_{\bm{0}^{(i)},\bm{0}}$, but when the modified circuit is used, the estimate is also modified as
\begin{equation}
    \varphi'_{N'}=\frac{C}{N'}\sum_{i=1}^{N'}  [C_{\rm BE}]^{r^{(i)}_{{\rm shot}}\mathcal{O}(\log(1/\varepsilon'))}\cdot b_{X}^{(i)}b_{O}^{(i)}\delta_{\bm{0}^{(i)},\bm{0}},
\end{equation}
where $N'\geq N$ is the number of samples we specified below.

The parameters $\varepsilon',r_{\rm BE},N'$ are specified as follows.
First, to ensure the estimation bias is sufficiently small, the error $\varepsilon'$ should be set to $\mathcal{O}(\varepsilon/r\|O\|)$.
Then, the additional sampling overhead from the use of Eq.~\eqref{eq:lem14_1} becomes 
\begin{equation}
    [C_{\rm BE}]^{r^{(i)}_{{\rm shot}}\mathcal{O}(\log(r\|O\|/\varepsilon))}\leq \exp(r^{-1}_{\rm BE}{\mathcal{O}(r\log(r\|O\|/\varepsilon))}).
\end{equation}
By choosing $r_{\rm BE}=\mathcal{O}(r\log(r\|O\|/\varepsilon))$, the additional sampling overhead becomes a constant.
Under this choice, the number of samples $N'$ can be taken as $N'=\mathcal{O}(N)$.
Therefore, the above modified algorithm can estimate the target value $\Tr[Oe^{t\mathcal{L}}(\rho_0)]$ within the additive error $\varepsilon$ with at least $1-\delta$ probability using the $N'=\mathcal{O}(\|O\|^2\log(1/\delta)/\varepsilon^2)$ samples.
All quantum circuits in this algorithm have the non-Clifford gate complexity
\begin{align}
    \mathcal{O}\left[rr_{\rm BE}\log(1/\varepsilon')\right]&=\mathcal{O}\left[r^2\log^2(r\|O\|/\varepsilon)\right]=\mathcal{O}\left[(\|\mathcal{L}\|_{\rm pauli}t)^4 \log^2\left(\frac{\|O\|\|\mathcal{L}\|_{\rm pauli}t}{\varepsilon}\right)\right],
\end{align}
and the elementary gate complexity
\begin{align}
    \mathcal{{O}}\left[r\left(\frac{\log(\|O\|r/\varepsilon)}{\log \log(\|O\|r/\varepsilon)} +r_{\rm BE}\frac{\log^2(r_{\rm BE}/\varepsilon')}{\log\log(r_{\rm BE}/\varepsilon')}\right)\right]&=\mathcal{{O}}\left[r^2\frac{\log^3(\|O\|r/\varepsilon)}{\log\log(\|O\|r/\varepsilon)}\right]=\mathcal{{O}}\left[\tau^4\frac{\log^3(\|O\|\tau/\varepsilon)}{\log\log(\|O\|\tau/\varepsilon)}\right],
\end{align}
where $\tau=\|\mathcal{L}\|_{\rm pauli}t$.
And the circuits use at most $7$ ancilla qubits ($n+7$ in total).
This completes the proof of Theorem~\ref{thm:mMKindep_simulator}.

\clearpage

\section{Detailed setup of gate complexity analysis}\label{apdx:numerical}
In this section, we provide the details of the gate complexity analysis shown in Section~\ref{sec:gatecomplexity-analysis}.
Also, we denote the number of sites by $n$ and the number of qubits by $N_{\rm q}$.
\subsection{Numerical results for scaling}
This section provides supplementary numerical results extending those presented in the main text.
In particular, we consider more complicated systems and investigate dependencies on the other parameters.
The detailed procedure for estimating the gate counts is deferred to the next subsection.

We begin by reviewing the setup studied in the main text.
In the main text, we introduce two common instances: (i) a one-dimensional dissipative transverse field Ising model (TFIM), and (ii) a two-dimensional Fermi-Hubbard model (FHM) with the two-body loss, with $n$ sites.
TFIM is described by the following Hamiltonian and jump operators:
\begin{equation}
    H_{\mathrm{TFIM}}  = -J \sum_{i=1}^n Z_i Z_{i+1} - h \sum_{i=1}^{n} X_i, ~~~
    L_{\mathrm{TFIM}, k} = \sqrt{\gamma} \frac{X_k - iY_k}{2}
\end{equation}
where $J$ is the nearest-neighbor coupling, $h$ is a transverse field and $\gamma$ is a decay rate. 
FHM is described by:
\begin{equation}
    H_{\mathrm{FHM}}  = -J \sum_{\Braket{i,j}, \alpha}\left( c_{i, \alpha}^\dagger c_{j, \alpha}  + c_{j, \alpha}^\dagger c_{i, \alpha} \right)
    + U \sum_{i, \alpha \ne \beta}  n_{i, \alpha} n_{i, \beta}
    , ~~~
    L_{\mathrm{FHM}, k} = \sqrt{\gamma} \sum_{\alpha \ne \beta} c_{k, \alpha }c_{k, \beta}
\end{equation}
where $c_{i,\sigma}$ and $c_{i,\sigma}^\dagger$ are fermionic annihilation and creation operators at the site $i$ on a square lattice with the length $\sqrt{n}$, and for the spin state $\sigma = 1, \cdots, s$.
$n_{i,\sigma} = c_{i,\sigma}^\dagger c_{i,\sigma}$ is a number operator.
$\sum_{\Braket{i,j}}$ denotes summation over the nearest-neighbor pairs $i,j$.
Let $J$ and $\gamma$ be the hopping amplitude and the two-body loss rate, respectively. For the Pauli decomposition of fermionic operators, we employ the Jordan-Wigner transformation.

The number of gates and ancilla qubits for the systems are presented in the following figures;
Figures~\ref{fig:complexity-eps}, \ref{fig:complexity-tau10}, \ref{fig:complexity-tau100}, and \ref{fig:complexity-n}
show the $\varepsilon$-dependency, $\tau$-dependency with $n=10$, $\tau$-dependency with $n=100$, and $n$-dependency, respectively.
These figures provide the comparison of our algorithm (\textit{this work} colored in red), channel LCU (\textit{CLCU} colored in blue)~\cite{Cleve2016-yj}, and 
first order Hamiltonian simulation-based algorithm (\textit{HS} colored in green)~\cite{Ding2024-SDE}.
Note that the gate and ancilla count can be evaluated by the three parameters $(\varepsilon, \tau, n)$. Thus, we do not need to specify the physical parameters of Hamiltonian and the jump operators when $\tau = \|\mathcal{L}\| t$ is given, similarly, we do not need $(m, M, K)$, which is deduced from $n$.
The Hamiltonian simulation-based algorithm require an additional information $\alpha_0/\|\mathcal{L}\|$ and we assume $\alpha_0/\|\mathcal{L}\| = 1/4$.
In the main text, we considered the TFIM and the FHM with $s=2$, whereas here we also include an example with $s=6$.

Here, we highlight the insights of the results:
\begin{itemize}
    \item Figure~\ref{fig:complexity-eps} (a-c) show $\varepsilon$-scaling of the gate requirement. Our gate requirement shows an exponential advantage over first order HS, and even better than channel LCU. 
    While the gate counts of channel LCU scale polylogarithmically, we reach much smaller values.
    \item Figures~\ref{fig:complexity-tau10},\ref{fig:complexity-tau100}(a-c) show $\tau$-scaling of the gate requirement.
    The number of gates of our algorithm scales in $\mathcal{O}(\tau^2)$, while channel LCU does in $\mathcal{O}(\tau)$. Despite the unfavorable scaling, in a wide range of $\tau$, our gate requirement is fewer than the channel LCU's. Furthermore, complicated systems (larger $n$ or larger $s$) require additional gates to reach the crossing points.
    \item Figure~\ref{fig:complexity-n} (a-c) show $n$-scaling of the gate requirement. Our algorithm has a better scaling for $n$ even though the $\| \mathcal{L} \|_{\mathrm{pauli}}$-scaling is quadratically worse. This is because our algorithm removes the dependence of $m$ and $K$.
    The independence of $m$ and $K$ leads to the scalability of the algorithm.
    \item Figures~\ref{fig:complexity-eps},\ref{fig:complexity-tau10},\ref{fig:complexity-tau100},\ref{fig:complexity-n} (d-f) show the ancilla qubit requirement. Our algorithm needs the small constant ancilla qubits in all setups we investigate.
\end{itemize}

These numerical analysis provides concrete evidence of the advantage of our algorithm in certain but wide parameter regimes.
In particular, when the FTQC devices are not mature, that is, they do not have logical qubits and executable gates enough to reach the crossing points, our algorithm will be the best option.

\begin{figure}[htbp]
    \centering
    \begin{tabular}{ccc}
         \includegraphics[width=0.32\textwidth]{fig/GateCountVsEpsTFIM.pdf} & \includegraphics[width=0.32\textwidth]{fig/GateCountVsEpsFHM2.pdf} & \includegraphics[width=0.32\textwidth]{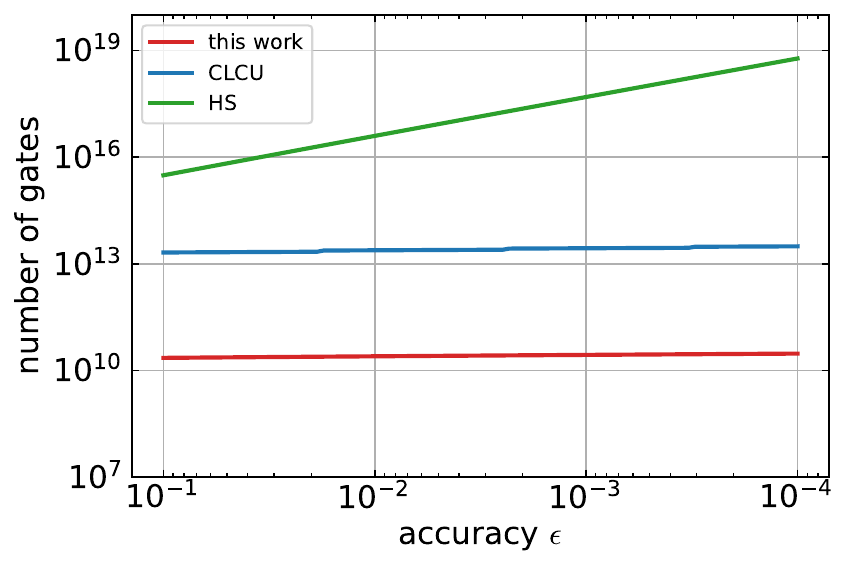} \\
         (a) Gate counts of TFIM &  (b) Gate counts of FHM ($s=2$) & (c) Gate counts of FHM ($s=6$)
         \\\\
         \includegraphics[width=0.32\textwidth]{fig/AncillaCountVsEpsTFIM.pdf} & \includegraphics[width=0.32\textwidth]{fig/AncillaCountVsEpsFHM2.pdf} & \includegraphics[width=0.32\textwidth]{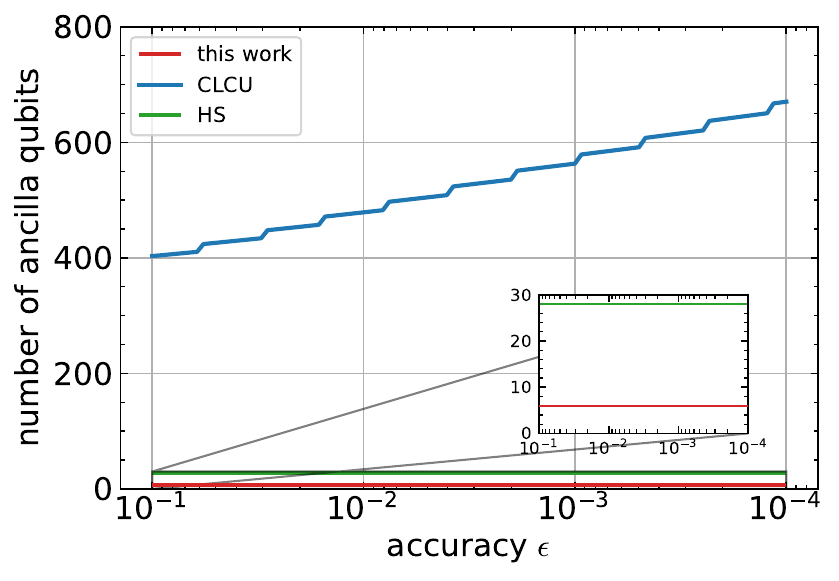} \\
         (d) Ancilla counts of TFIM &  (e) Ancilla counts of FHM ($s=2$) & (f) Ancilla counts of FHM ($s=6$)\\
    \end{tabular}
    \caption{Accuracy $\varepsilon$-dependencies of gate and additional ancilla qubit requirements, calculated by our algorithm (\textit{this work}), channel LCU (\textit{CLCU})~\cite{Cleve2016-yj} and first order Hamiltonian simulation based (\textit{HS})~\cite{Ding2024-SDE}. Target systems are $n$-site TFIM and FHM ($s=2$ and $s=6$), with $\tau/n = 10$ and $n = 10^2$.}
    \label{fig:complexity-eps}
\end{figure}

\begin{figure}[htbp]
    \centering
    \begin{tabular}{ccc}
         \includegraphics[width=0.32\textwidth]{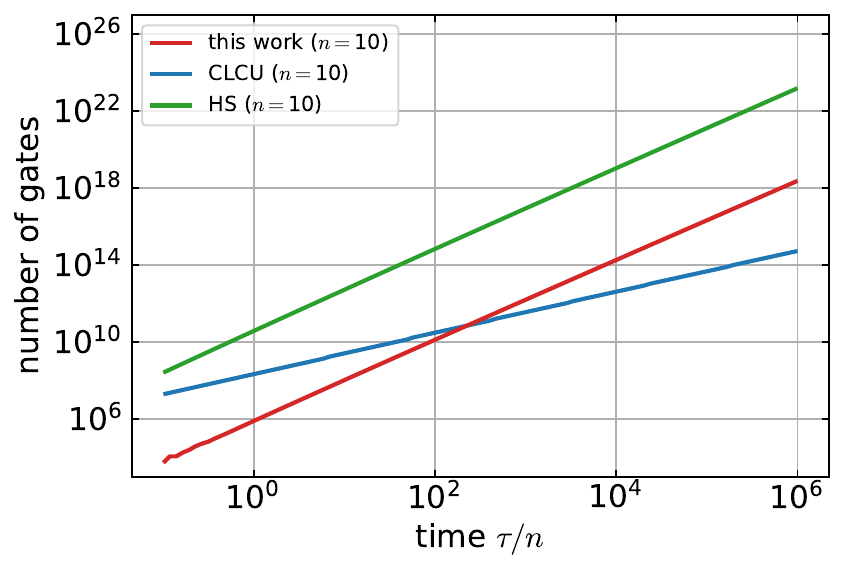} & \includegraphics[width=0.32\textwidth]{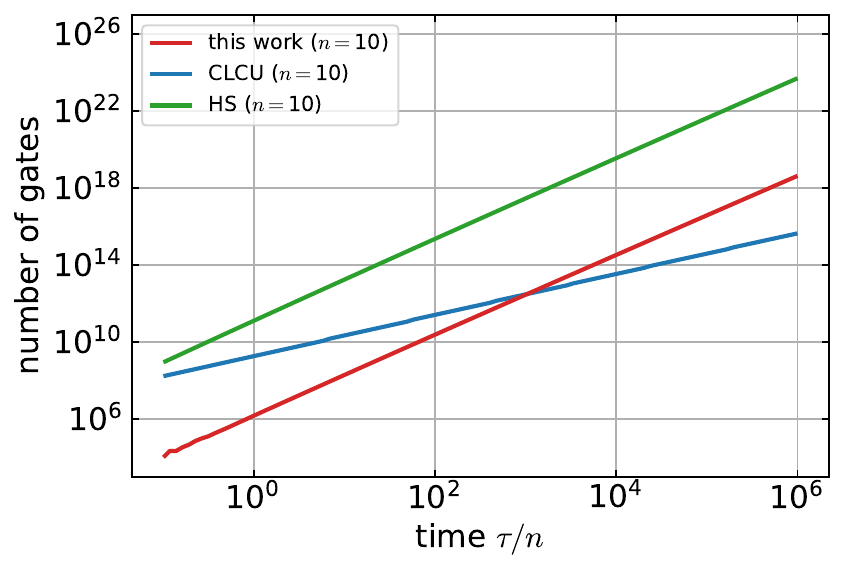} & \includegraphics[width=0.32\textwidth]{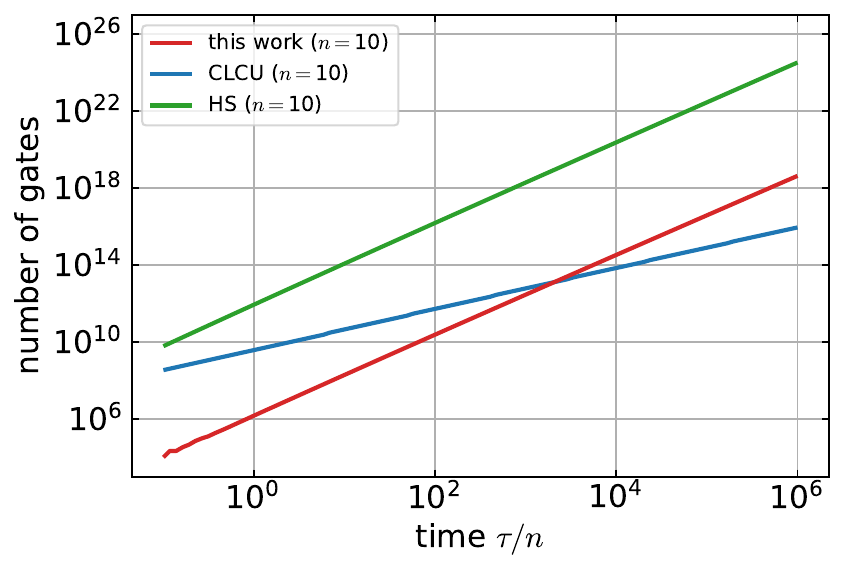} \\
         (a) Gate counts of TFIM &  (b) Gate counts of FHM ($s=2$) & (c) Gate counts of FHM ($s=6$)
         \\\\
         \includegraphics[width=0.32\textwidth]{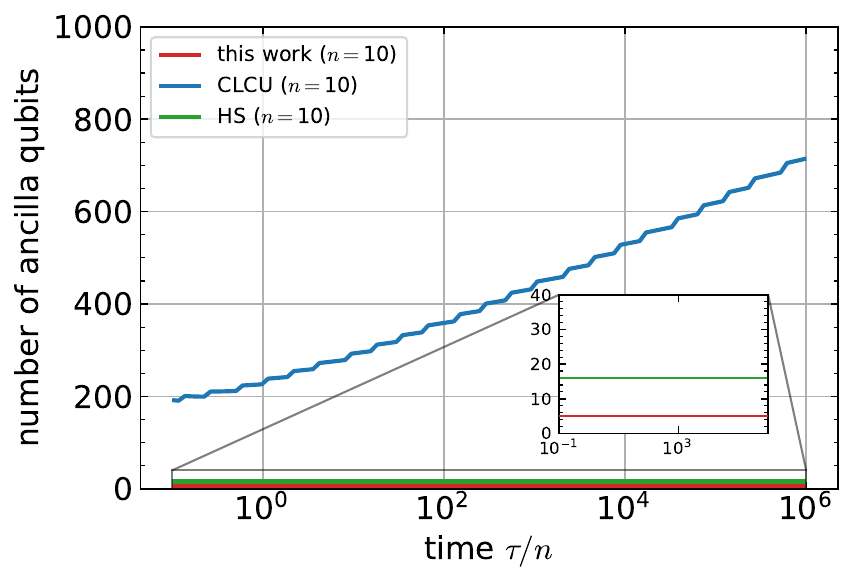} & \includegraphics[width=0.32\textwidth]{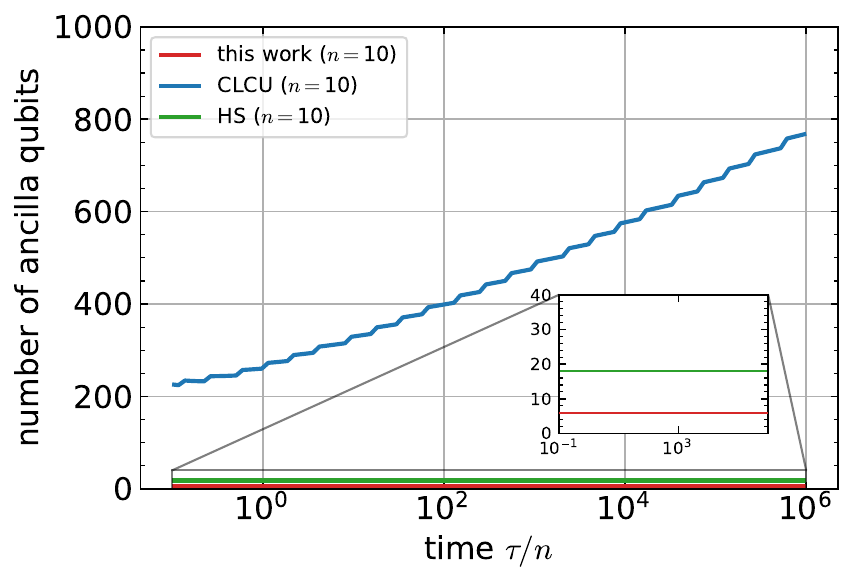} & \includegraphics[width=0.32\textwidth]{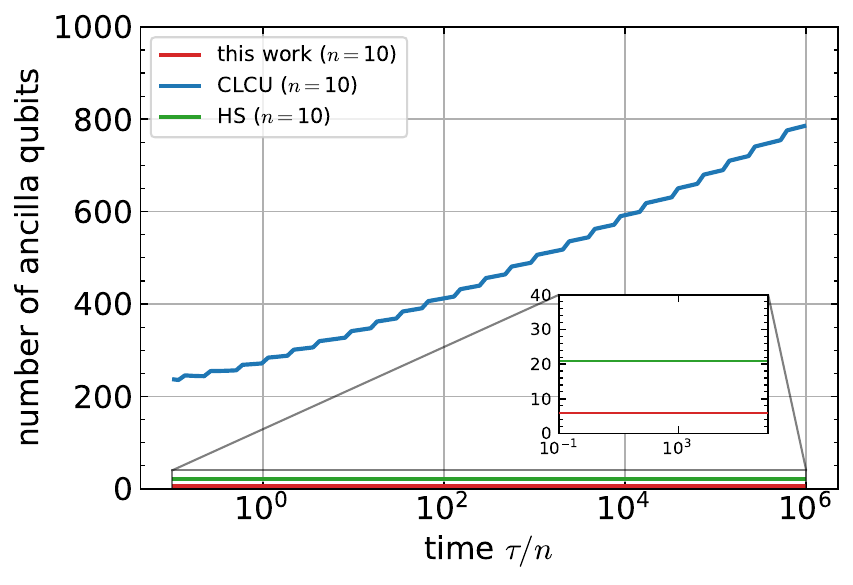} \\
         (d) Ancilla counts of TFIM &  (e) Ancilla counts of FHM ($s=2$) & (f) Ancilla counts of FHM ($s=6$)\\
    \end{tabular}
    \caption{$\tau/n$-dependencies of gate and additional ancilla qubit requirements, calculated by our algorithm (\textit{this work}), channel LCU (\textit{CLCU})~\cite{Cleve2016-yj} and first order Hamiltonian simulation based (\textit{HS})~\cite{Ding2024-SDE}. Target systems are $n$-site TFIM and FHM ($s=2$ and $s=6$), with $(n, \varepsilon) = (10, 10^{-2})$.}
    \label{fig:complexity-tau10}
\end{figure}

\begin{figure}[htbp]
    \centering
    \begin{tabular}{ccc}
         \includegraphics[width=0.32\textwidth]{fig/GateCountVsTauTFIM_100.pdf} & \includegraphics[width=0.32\textwidth]{fig/GateCountVsTauFHM2_100.pdf} & \includegraphics[width=0.32\textwidth]{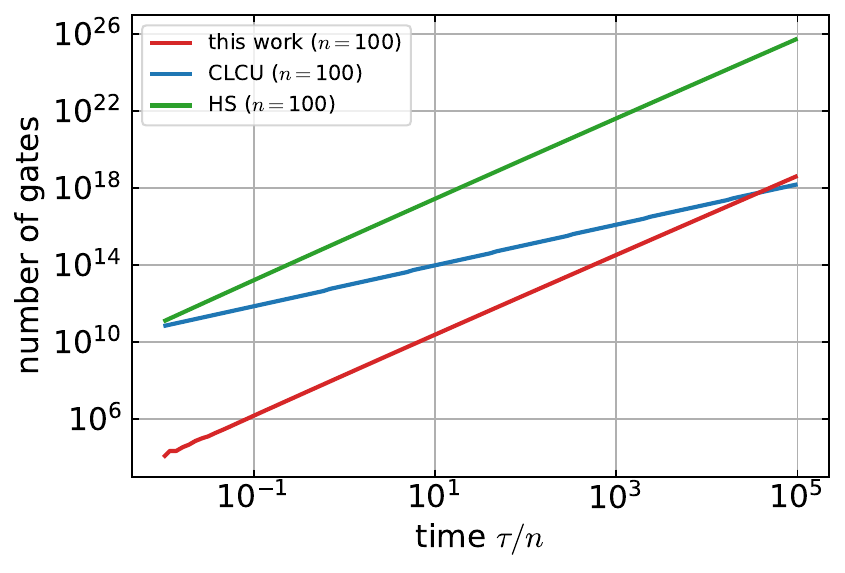} \\
         (a) Gate counts of TFIM &  (b) Gate counts of FHM ($s=2$) & (c) Gate counts of FHM ($s=6$)
         \\\\
         \includegraphics[width=0.32\textwidth]{fig/AncillaCountVsTauTFIM_100.pdf} & \includegraphics[width=0.32\textwidth]{fig/AncillaCountVsTauFHM2_100.pdf} & \includegraphics[width=0.32\textwidth]{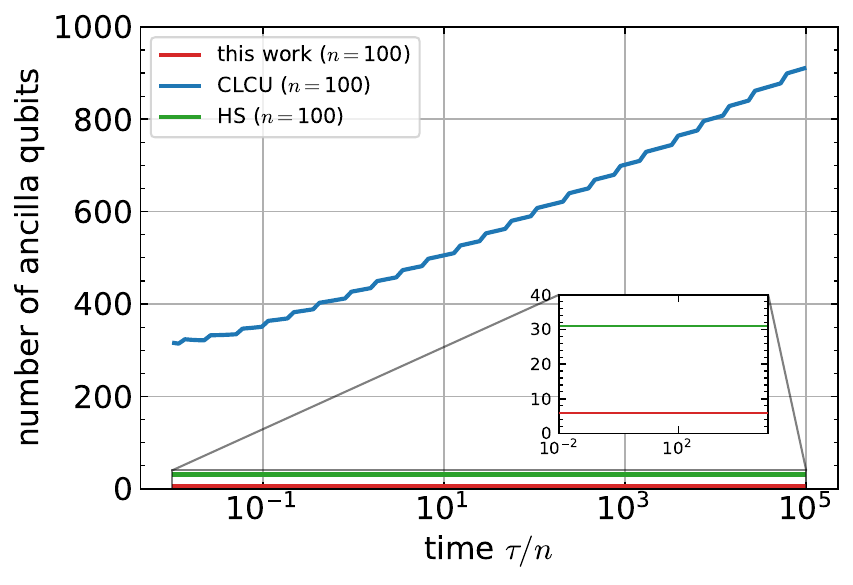} \\
         (d) Ancilla counts of TFIM &  (e) Ancilla counts of FHM ($s=2$) & (f) Ancilla counts of FHM ($s=6$)\\
    \end{tabular}
    \caption{$\tau/n$-dependencies of gate and additional ancilla qubit requirements, calculated by our algorithm (\textit{this work}), channel LCU (\textit{CLCU})~\cite{Cleve2016-yj} and first order Hamiltonian simulation based (\textit{HS})~\cite{Ding2024-SDE}. Target systems are $n$-site TFIM and FHM ($s=2$ and $s=6$), with $(n, \varepsilon) = (10^2, 10^{-2})$.}
    \label{fig:complexity-tau100}
\end{figure}

\begin{figure}[htbp]
    \centering
    \begin{tabular}{ccc}
        \includegraphics[width=0.32\textwidth]{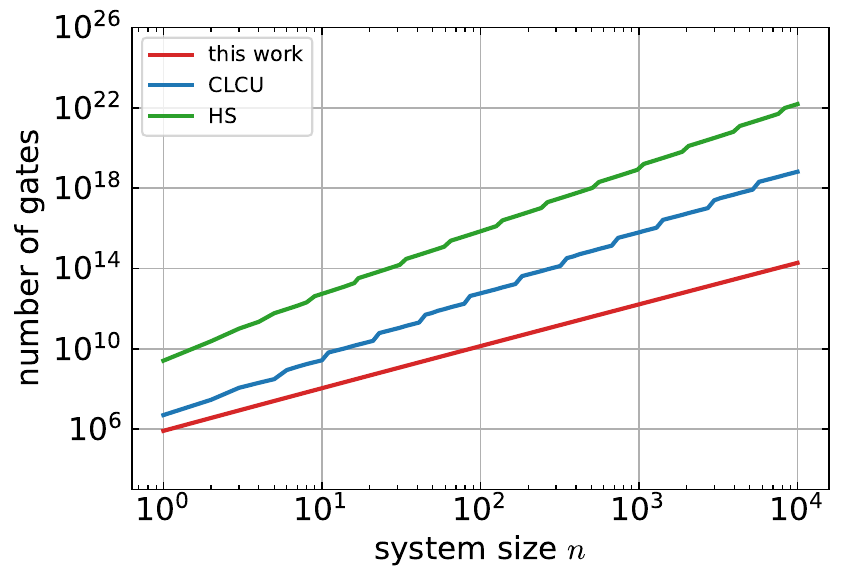} & \includegraphics[width=0.32\textwidth]{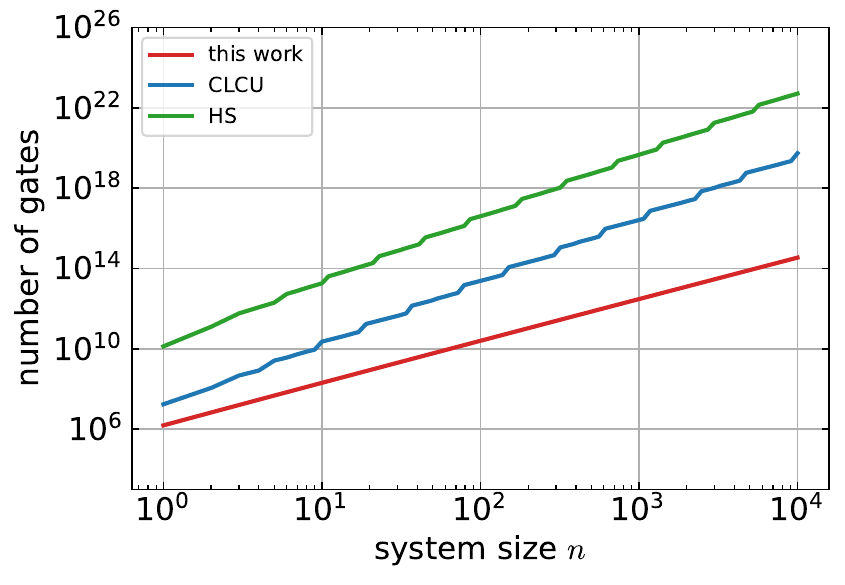} & \includegraphics[width=0.32\textwidth]{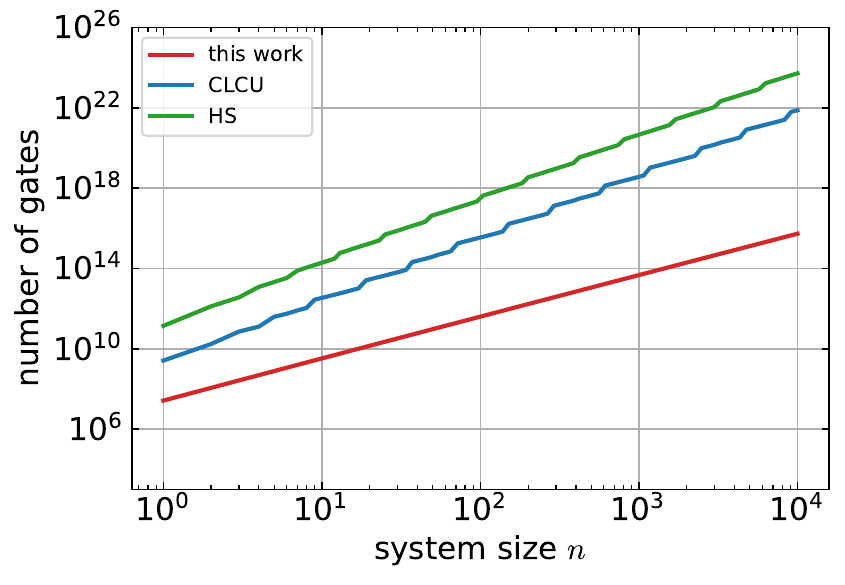} \\
         (a) Gate counts of TFIM &  (b) Gate counts of FHM ($s=2$) & (c) Gate counts of FHM ($s=6$)
         \\\\
         \includegraphics[width=0.32\textwidth]{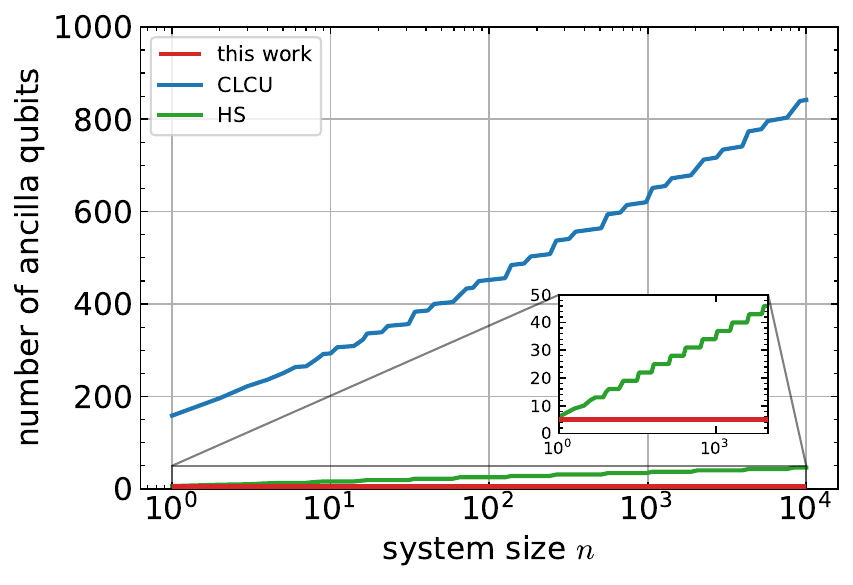} & \includegraphics[width=0.32\textwidth]{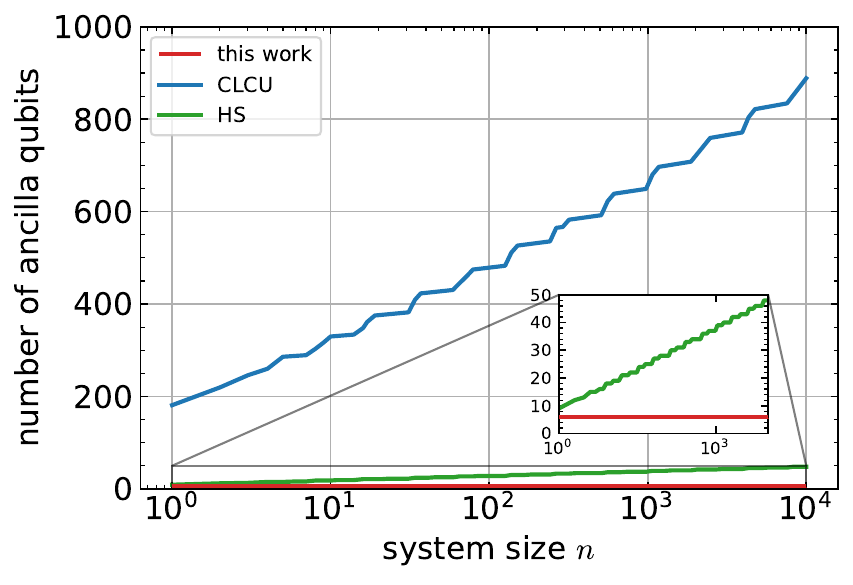} & \includegraphics[width=0.32\textwidth]{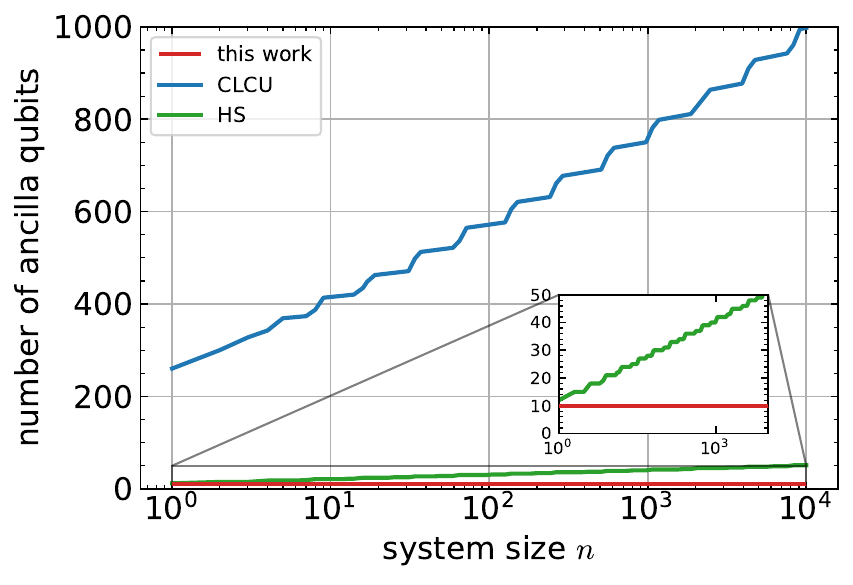} \\
         (d) Ancilla counts of TFIM &  (e) Ancilla counts of FHM ($s=2$) & (f) Ancilla counts of FHM ($s=6$)\\
    \end{tabular}
    \caption{System size $n$-dependencies of gate and additional ancilla qubit requirements, calculated by our algorithm (\textit{this work}), channel LCU (\textit{CLCU})~\cite{Cleve2016-yj} and first order Hamiltonian simulation based (\textit{HS})~\cite{Ding2024-SDE}. Target systems are $n$-site TFIM and FHM ($s=2$ and $s=6$), with $\tau / n = 10$ and $\varepsilon = 10^{-2}$.}
    \label{fig:complexity-n} 
\end{figure}
\clearpage

\subsection{Gate complexity of algorithms}\label{apdx:gata-complexity-procedure}
Here, we provide the detailed computation procedure to estimate the number of T gates, or T-count. T-count is a common metric for the cost estimation of FTQC algorithms~\cite{babbush2018encoding, childs2018toward,rhodes2024exponential, yoshioka2024hunting}, under the assumption of Clifford-T formalism on the surface code. 
As a primary cost model, T-count provides a good approximation to the overall execution time, as executing a T gate typically involves a time-consuming procedure (e.g., magic-state distillation).
Throughout the analysis, we employ the notation; [CircuitComponent] denotes the number of T-gates of CircuitComponent.

First, we provide the T-count for the following primitive components, $a$-qubit controlled X, SWAP, Hadamard, and rotation gate R:
\begin{align}
     [a\text{-qubit-ctrl~X}] &= \begin{cases} 0 &  a = 0, 1 \\ 4(a-1) & a \ge 2\end{cases} \notag \\ 
     [a\text{-qubit-ctrl~SWAP}] &= \begin{cases} 0 &  a = 0 \\ 4a & a \ge 1\end{cases} \notag \\
     [a\text{-qubit-ctrl~Hadamard}] &= \begin{cases} 0 &  a = 0 \\ 4a-2 & a \ge 1\end{cases} \notag \\
     [a\text{-qubit-ctrl~R}] &= 8(a-1) + 2 \Gamma \log_2 (1/\delta_{\mathrm{SS}}) + 2 \Xi
\end{align}
where $\Gamma$ and $\Xi$ are constants, which we set 3.3 averaging 1.03 and 5.1, following the complexity analysis~\cite{yoshioka2024hunting}. $\delta_\mathrm{SS}$ is determined to maintain the target accuracy $\varepsilon$.
Using these results, we can estimate the T-count for the controlled LCU, which is the core subroutine for open system simulation algorithms.
We follow the result of the circuit to prepare the arbitral amplitude~\cite{mottonen2005transformation}, but slightly modify it for a controlled linear combination of Pauli strings, of our interest.
\begin{align}
     [a\text{-qubit-ctrl}~b\text{-qubit PRE}] &= 2^{b+1} [a\text{-qubit-ctrl~R}] \notag\\ 
     [a\text{-qubit-ctrl}~b\text{-qubit SEL}] &= 2^{b} [(a+b)\text{-qubit-ctrl~X}] \notag \\
     [a\text{-qubit-ctrl}~b\text{-qubit LCU}] &= 2^{b+2} [a\text{-qubit-ctrl~R}] + 2^{b} [(a+b)\text{-qubit-ctrl~X}]
\end{align}
This can be confirmed from Figure~\ref{fig:lcu-circuits-synthesis} directly.
\begin{figure}[ht]
    \begin{tabular}{c}
    \begin{quantikz}
        \\
        \lstick{$\ket{\bm{0}}$} & \qwbundle{b} &  \gate{\text{PRE}} & \\
        \lstick{ctrl} & \qwbundle{a} & \ctrl{-1} & \\
        \\
    \end{quantikz}
    =
    \begin{quantikz}[row sep=0.3cm, column sep=0.2cm]
        \\
        \lstick[3]{$\ket{\bm{0}}$} & \qwbundle{1} & \gate{R_z(\theta_1)} & \targ{} &  \gate{R_z(\theta_2)} & \targ{} & & \ctrl{1} & \midstick{$\cdots$} & &\\
         & \qwbundle{1} & & \ctrl{-1} & &  & \gate{R_z(\theta_3)} & \targ{} &  \midstick{$\cdots$} & &\\
         & \qwbundle{1} & & & & \ctrl{-2} & & & \midstick{$\cdots$} & \gate{R_y(\theta_{2^{b+1}})} & \\
        \lstick{ctrl} & \qwbundle{a} & \ctrl{-3} & &   \ctrl{-3} & &   \ctrl{-2} & &  \midstick{$\cdots$} & \ctrl{-1} &\\
        \\
    \end{quantikz}
    \\
    (a) $a$-qubit controlled PREPARE gate for $b$-qubit.
    \\
    \begin{quantikz}
        \\
        \lstick{sys} & \qwbundle{n} & \gate[2]{\text{SEL}}& \\
        \lstick{$\ket{\bm{0}}$} & \qwbundle{b} &  \ghost{\text{SEL}} & \\
        \lstick{ctrl} & \qwbundle{a} & \ctrl{-1} & \\
        \\
    \end{quantikz}
    =
    \begin{quantikz}[column sep=0.2cm]
        \\
        \lstick{sys} & \qwbundle{n} & \gate{C_1} & \targ{} & \gate{C_1^\dagger} & \gate{C_2} & \targ{} & \gate{C_2^\dagger} & \midstick{$\cdots$} & \gate{C_{2^{b}}} & \targ{} & \gate{C_{2^{b}}^\dagger} &\\
         \lstick{$\ket{\bm{0}}$}& \qwbundle{b} & & \ctrl{-1}  & & &  \ctrl{-1} & &  \midstick{$\cdots$} & &  \ctrl{-1}& & \\
        \lstick{ctrl} & \qwbundle{a} & & \ctrl{-1} & &  & \ctrl{-1} & &  \midstick{$\cdots$} & & \ctrl{-1} & & \\
        \\
    \end{quantikz}
    \\
    (b) $a$-qubit controlled SELECT gate for $b$-qubit.
    \end{tabular}
    \caption{Basic Circuit components for LCU. (a) A controlled PREPARE gate ($b=3$ case). This has two alternative sequence structures: (i) $a$-qubit controlled $R_z$ gates and CX gates, and (ii) $a$-qubit controlled $R_y$ gates and CX gates. Since LCU has an uncomputational structure, CX does not need to be controlled. (b) A controlled SELECT gate for Pauli strings, involving $(a + b)$-qubit controlled X gates $2^{b}$ times, and Clifford gates $C_1, \cdots, C_{2^b}$,  $2^{b+1}$ times in total.}
    \label{fig:lcu-circuits-synthesis}
\end{figure}

Then, we show the number of the T gates of three algorithms: our work, channel LCU~\cite{Cleve2016-yj}, and first-order Hamiltonian simulation~\cite{Ding2024-SDE}.
By leveraging the result of Lemma~\ref{lemma:oaa_for_B0B1}, our T-count are calculated following:
\begin{align}\label{eq:non-cliffrod-ours}
 [\text{Total}] &= r ([\text{LCS}] + \text{[Dissipation}]) \notag\\
 [\text{LCS}] &= [1\text{-qubit-ctrl~R}] \notag\\
 [\text{Dissipation}] &= 3 [W_L] + 3 [1\text{-qubit-ctrl}~W_L] +
4[(1+\log_2M)\text{-qubit-ctrl~X}] + [(2+\log_2M)\text{-qubit-ctrl~X}] + 12 [\text{R}] \notag\\
[W_L] &=  [\log_2 M \text{-qubit~LCU} ] = 4M [R] + M [\log_2 M\text{-qubit-ctrl~X}]\notag\\
[1\text{-qubit-ctrl~}W_L] & =  [1\text{-qubit-ctrl~}\log_2 M \text{-qubit~LCU}] = 
4M [1\text{-qubit-ctrl~R}] + M [(\log_2 M+1)\text{-qubit-ctrl~X}]
\end{align}
where $r$ is the number of time slices by setting $r= \mathcal{O}(\tau^2)$.
To construct $W_L$, we utilize the standard LCU procedure.
It is worth noting that our T-count does not depend on the accuracy $\varepsilon$ explicitly, since type-(B) operators for the higher-order approximation are Clifford gates, while the rotation gate synthesis requires $\log (1/\varepsilon)$ T gates implicitly.

We can compute the gate counts of the channel LCU approach~\cite{Cleve2016-yj} from the circuit diagram in Figure~\ref{fig:cleve-circuit}.
\begin{align}\label{eq:non-clifford-channellcu}
    [\mathrm{Total}] &= r [\mathcal{M}] \notag\\
    [\mathcal{M}] &= 3 [W] +
    [(hb-1)\text{-qubit-ctrl X}] + [(ha+hb+hc-1)\text{-qubit-ctrl X}] \notag\\
    [W] &= [E] + h[a\text{-qubit-ctrl }\mu'] + h [B] + h [B'] + h[U] \notag\\
    [E] &= ha [\text{R}] \notag\\
    [a\text{-qubit-ctrl }\mu'] &=  [1\text{-qubit-ctrl }c\text{-qubit PRE}] + 2[a\text{-qubit-ctrl X}]\notag\\
    [B] &= [B']= (1+K)[(c+1)\text{-qubit-ctrl~} b\text{-qubit PRE}] + 2[a\text{-qubit-ctrl X}]\notag \\
     [U] &= (m + KM + KM^2) [(b+c+1)\text{-qubit-ctrl X}] + 2[a\text{-qubit-ctrl X}] 
\end{align}
where $r$ is the number of time slices by setting $r = \tau$, $a = \log_2(r/\varepsilon)$, $b=\log_2(KM^2 + m)$, $c=\log_2(1+K)$ are parameters of the ancilla qubits, and $h =2\log(r/\varepsilon)/(\log\log(r/\varepsilon)) + \log(4/3) -1)$ is a hamming weight cutoff parameter.

Then, we estimate the T-count of first order Hamiltonian simulation approach~\cite{Ding2024-SDE}.
Ref.~\cite{Ding2024-SDE} provides the block encoding $W_{\tilde{H}}$ of the dilated Hamiltonian $\tilde{H}$ given the block encoding of Hamiltonian $H_0$ and the jump operators $H_i$ for $i\in [0,K]$, but they do not specify the block encoding algorithm or 
the Hamiltonian simulation algorithm for $\tilde{H}$.
In order to estimate the gate cost, we assume that (i) the block encodings of $H_i$ for $i\in [0,K]$ are constructed by the standard LCU procedure, and (ii) for Hamiltonian simulation of block encoded $\tilde{H}$, QSVT technique~\cite{gilyen2019quantum} is employed (see, e.g., Ref.~\cite{toyoizumi2024hamiltonian} for a detailed circuit diagram).
Under these assumptions, we can construct the circuit in Fig.~\ref{fig:sse-circuit},
and directly count the T gates from the diagram.
\begin{align}
    [\mathrm{Total}] =& r [e^{-i \tilde{H} \sqrt{\Delta t}}] \notag\\
    [e^{-i \tilde{H} \sqrt{\Delta t}}] =& 6d [W_{\tilde{H}}] + 3[1\text{-qubit-ctrl~}W_{\tilde{H}}] + (12d+6)[(b +2c + 1)\text{-qubit-ctrl~X}] \notag \\
    & + 2[(b +2c + 2)\text{-qubit-ctrl~X}] + 12d[1\text{-qubit-ctrl~R}] \notag\\
    [l\text{-qubit-ctrl~} W_{\tilde{H}}] =& 2[(1+l)\text{-qubit-ctrl~} U] + 2[(1+l)\text{-qubit-ctrl~} S]\notag\\
    [(1+l)\text{-qubit-ctrl~} U] =& [(1 +l + c)\text{-qubit-ctrl~} \log_2 m\text{-qubit LCU}] + K [(1 + l +  c)\text{-qubit-ctrl~} \log_2  M\text{-qubit LCU}] \notag\\
    [(1+l)\text{-qubit-ctrl~} S] =& c[(1+l)\text{-qubit-ctrl SWAP}] + c[(1+l)\text{-qubit-ctrl Hadamard}] \notag
\end{align}
where parameters are determined as follows: time slicing $r = \tau^2 /\varepsilon$, small time step $\Delta t = t/r$, polynomial degree of Hamiltonian simulation $d = \sqrt{2(1+K) \Delta t} D  + \log(r /\varepsilon)$, a block encoding norm $D = \max(\alpha_0 \sqrt{\Delta t}, \alpha_1 ,\cdots, \alpha_K)$ for $\alpha_i$ defined in Eq.~\eqref{main:H_L}, and parameters related to ancilla count $b=\log_2(\max(m, M))$, $c=\log_2(1+K)$.

\begin{figure}[ht]
    \begin{tabular}{c}
    \begin{quantikz}
        \lstick{sys} & \qwbundle{n} & & \gate[3]{\mathcal{M}}  & \qw & \qw & \gate[3]{\mathcal{M}} &  & \midstick{$\cdots$} & \gate[3]{\mathcal{M}} & \\
        \lstick{$\ket{\bm{0}}$}  & \qwbundle{hb} &  &  &  \meter{} &  \midstick{$\ket{\bm{0}}$}  \setwiretype{n} &   \setwiretype{q} &  \meter{} & \midstick{$\ket{\bm{0}}$}  \setwiretype{n} &   \setwiretype{q} & \meter{} \\
        \lstick{$\ket{0}_{\mathrm{P}}$}  &  \qwbundle{ha+hc} & &  &  \meter{} &  \midstick{$\ket{0}_{\mathrm{P}}$}  \setwiretype{n} &   \setwiretype{q} &  \meter{} & \midstick{$\ket{0}_{\mathrm{P}}$}  \setwiretype{n} &   \setwiretype{q} & \meter{} \\
        \\
    \end{quantikz}
    \\
    (a) An algorithm overview
    \\
    \begin{quantikz}
        \\
        \lstick{sys} &\qwbundle{n} & \qw & \qw & \gate[3]{\mathcal{M}} & \qw & \qw\\
        \lstick{$\ket{\bm{0}}$} & \qwbundle{hb} & \qw  & \qw & \ghost{\mathcal{M}} & \qw &\meter{} \\
        \lstick{$\ket{0}_{\mathrm{P}}$} & \qwbundle{ha + hc} & \qw  & \qw & \ghost{\mathcal{M}} & \qw &\meter{}\\
        \\
    \end{quantikz}
    ~~=
    \begin{quantikz}
        \\
        \lstick{sys}            & \qwbundle{n}                       & \qw  & \qw & \gate[3]{W}& \qw & \gate[3]{W^\dagger} &\qw & \gate[3]{W} & \qw \\
        \lstick{$\ket{\bm{0}}$} & \qwbundle{hb} & \qw  & \qw & \ghost{W}  & \gate{2\tilde{\Pi}-\bm{1}} &\ghost{W^\dagger} &\gate[2]{\bm{1}-2{\Pi}}& \ghost{W} &\meter{} \\
        \lstick{$\ket{0}_{\rm P}$}      & \qwbundle{ha+ hc}                       & \qw  & \qw & \ghost{W}  & \qw &\ghost{W^\dagger}&\ghost{\bm{1}-2{\Pi}}& \ghost{W}&\meter{} \\
        \\
    \end{quantikz}
    \\
    (b) A Kraus map circuit $\mathcal{M}$ involving channel LCU and OAA
    \\
    \begin{quantikz}
            \\
        \lstick{sys}  & \qwbundle{n} & \qw  & \gate[3]{W}  &\\
        \lstick{$\ket{\bm{0}}$}  & \qwbundle{hb} &  \qw  &  & \\
        \lstick{$\ket{{0}}_{\rm P}$}  & \qwbundle{ha+hc} &  \qw  &  &\\
            \\
    \end{quantikz}
    ~~=
    \begin{quantikz}
        \\
        \lstick{sys}  & \qwbundle{n} &  \qw  & \qw  & \qw & \gate{U^{\otimes h}} & \qw & \\
        \lstick{$\ket{{0}}_{\rm P}$}  & \qwbundle{ha} &  \gate{E} & \ctrl{2} & \ctrl{1} & \ctrl{-1} \wire[d][2]{q}  &  \ctrl{1}  &\\
        \lstick{$\ket{\bm{0}}$}  & \qwbundle{hb} &  \qw &  \qw &  \gate{B'^{\otimes h}}& \ctrl{-1}  & \gate{(B^\dagger)^{\otimes h}} &\\
        \lstick{$\ket{0}_{\rm P}$}  & \qwbundle{hc} &  \qw & \gate{\mu'^{\otimes h}} & \ctrl{-1} & \ctrl{-1} & \ctrl{-1} &  \\
        \\
    \end{quantikz}
    \\
    (c) Detailed construction of the LCU circuit
    \end{tabular}
    \caption{Circuit components of Ref.~\cite{Cleve2016-yj}. (a) The algorithm repeats the Kraus map circuit $\mathcal{M}$ and mid-circuit measurements and resets for $r$ times. $h(a+b+c)$-ancilla qubits are required (b) The circuit $\mathcal{M}$ consists LCU and OAA. (c) The LCU circuit contains two kinds of gates: PREPARE gates ($E$, $\mu'$, $B$, $B'$) and SELECT gates ($U$). $B$ is a multiplex PREPARE gate in Ref.~\cite{Cleve2016-yj}, and $B'$ is a slightly modified version of $B$ for the compression scheme. $\mu'$ is also modified gate of $\mu$ in original definition for the compression scheme.}
    \label{fig:cleve-circuit}
\end{figure}

\begin{figure}[ht]
    \begin{tabular}{ccc}
    \multicolumn{3}{c}{
    \begin{quantikz}
        \lstick{sys} & \qwbundle{n} & & \gate[3]{e^{-i\tilde{H} \sqrt{\Delta t}}}  & \qw & \qw & \gate[3]{e^{-i\tilde{H}\sqrt{\Delta t}}} &  & \midstick{$\cdots$} & \gate[3]{e^{-i\tilde{H} \sqrt{\Delta t}}} & \\
        \lstick{$\ket{\bm{0}}$}  & \qwbundle{3+b+c} &  &  &  \meter{} &  \midstick{$\ket{\bm{0}}$} \setwiretype{n} &   \setwiretype{q} &  \meter{} & \midstick{$\ket{\bm{0}}$}  \setwiretype{n} &   \setwiretype{q} & \meter{} \\
        \lstick{$\ket{0}_{\mathrm{P}}$}  &  \qwbundle{c} & &  &  \meter{} &  \midstick{$\ket{0}_{\mathrm{P}}$}  \setwiretype{n} &   \setwiretype{q} &  \meter{} & \midstick{$\ket{0}_{\mathrm{P}}$}  \setwiretype{n} &   \setwiretype{q} & \meter{} 
        \\
    \end{quantikz}
    }
    \\
    \multicolumn{3}{c}{(a) An algorithm overview}
    \\
    \multicolumn{3}{c}{
    \begin{quantikz}
            \\
        \lstick{sys}  & \qwbundle{n} & \qw  & \gate[3]{W_{\tilde{H}}}  &\\
        \lstick{$\ket{0}_{\mathrm{P}}$} & \qwbundle{c} &  \qw  &  & \\
        \lstick{$\ket{\bm{0}}$}  & \qwbundle{1+b+c} &  \qw  &  &\\
            \\
    \end{quantikz}
    =
    \begin{quantikz}
        \\
        \lstick{sys}  & \qwbundle{n} & \qw & \gate[3]{U}  & \qw & \qw & \gate[3]{U} & & \\
        \lstick{$\ket{\bm{0}}$}  & \qwbundle{b} &  \qw  &  & & & & & \\
        \lstick{$\ket{0}_{\mathrm{P}}$} & \qwbundle{c} & \qw &  \qw  & \gate[2]{S} & \gate[2]{S^\dagger}  & & &\\
        \lstick{$\ket{\bm{0}}$}  & \qwbundle{c} &  \qw &  \qw & \qw & & & & \\
        \lstick{$\ket{\bm{0}}$}  & \qwbundle{1} & \gate{H}  &  \octrl{-2}  & \octrl{-1} & \ctrl{-1} & \ctrl{-2} & \gate{H} & \\
        \\
    \end{quantikz}
    }
    \\
    \multicolumn{3}{c}{(b) A block encoding of $\tilde{H}$}
    \\
    \begin{quantikz}
            \\
        \lstick{sys}  & \qwbundle{n}  & \gate[3]{U}  &\\
        \lstick{$\ket{\bm{0}}$} & \qwbundle{b} &  \qw  & \\
        \lstick{$\ket{0}_{\mathrm{P}}$} & \qwbundle{c} &  &\\
            \\
    \end{quantikz}
    =
    \begin{quantikz}
            \\
        \lstick{sys}  & \qwbundle{n}  & \gate[2]{W_{H_i}}\wire[d][2]{q}  &\\
        \lstick{$\ket{\bm{0}}$} & \qwbundle{b} &    & \\
        \lstick{$\ket{0}_{\mathrm{P}}$} & \qwbundle{c} & |[operator]| &\\
            \\
    \end{quantikz}
    &~~&
    \begin{quantikz}
            \\
        \lstick{$\ket{0}_{\mathrm{P}}$}  & \qwbundle{c}  & \gate[2]{S}  &\\
        \lstick{$\ket{\bm{0}}$} & \qwbundle{c} &  \qw  & \\
            \\
    \end{quantikz}
    =
    \begin{quantikz}
        \\
        \lstick{$\ket{0}_{\mathrm{P}}$}  & \qwbundle{c}  & \swap{1}  & \gate{H^{\otimes c}} &\\
        \lstick{$\ket{\bm{0}}$} & \qwbundle{c} & \targX{}  & &\\
        \\
    \end{quantikz}
    \\
    (c) $U$ gate &~~&  (d) $S$ gate 
    \\
    \end{tabular}
    \caption{Circuit components of Ref.~\cite{Ding2024-SDE}. (a) The algorithm repeats the Hamiltonian simulation circuit $e^{-i\tilde{H}\sqrt{\Delta t}}$ and mid-circuit measurements and resets for $r$ times. $(3 + b + 2c)$-ancilla qubits are required. (b) A block encoding of $\tilde{H}$, which incorporates $U$ and $S$ gates. (c) $U$ gate, which apply multiplex block encoding $W_{H_{0}}, \cdots , W_{H_{K}}$, where $W_{H_{0}}$ and $W_{H_{k}}$ represent the block encoding of $H \sqrt{\Delta t}$, and $L_k$ respectively. (d) $S$ gate for shuffling the index from $\ketbra{k}$ to $\ket{0}\bra{k}$.}
    \label{fig:sse-circuit}
\end{figure}

\clearpage

\subsection{Further reduction of gate complexity}\label{apdx:t-gate-reduction} 
In Appendix~\ref{apdx:gata-complexity-procedure}, we provided the general procedure to evaluate the T-count, yielding baseline results.
Beyond this baseline, we introduce additional circuit-level optimization techniques to reduce T-count.
Our circuit structure is much simpler than those of previous approaches with sophisticated multiple ancilla uses. By exploiting this simplicity, we can directly evaluate the benefits from the techniques of (1)~optimized implementation of a single jump process, (2)~assumptions based on average-case estimation, and (3)~improved rotation synthesis, for further reducing the total T-gate count (depth).
The breakdown of the reduction is provided in Table~\ref{tab:benchmark}~(e) in the main text.

Below, we explain these three techniques.
In principle, other algorithms may benefit from these techniques, but their evaluation is not straightforward unlike our method due to complicated controlled operations; we leave this for future work.

(1)
In our algorithm, we construct an encoding of the dissipation process $W_{\mathcal{B}}$ using block-encoding $W_L$ of jump operators constructed via the standard LCU procedure.
Since our circuit does not need the encoding of all jump operators, but needs only an encoding of the single jump operator separately, 
we can simplify the PREPARE circuit in $W_L$, and thus the circuit for a single dissipation $W_{\mathcal{B}}$, thereby reducing the number of rotation gates.

(2) While we focus on the worst-case complexity throughout the analysis, the average-case cost evaluation can be well justified in typical settings.
In our random algorithm, the entire circuit has $r=\mathcal{O}(\tau^2)$ layers, and in each layer, either Hamiltonian evolution (i.e., a single Pauli rotation) or dissipative processes (i.e., a single Pauli rotation or a single jump process) is randomly selected according to the ratio of the norm of Hamiltonian and jump operators.
Thus, from the central limit theorem, the total gate count in a single circuit converges to the average gate cost as $r$ or $\tau$ increases.
In the current setting, the probability of selecting the more costly dissipation is small, $1.3$\% for the boundary-driven dissipative system and $4.5$\% for the dissipative system with two-body loss, leading to a significant reduction.

(3) Our algorithm uses only a small number of ancilla qubits, so it is allowed to add a small number of additional ancilla qubits for rotation gate synthesis. Thus, we can use the 1 T-depth implementation of a rotation gate with a catalytic state~\cite{kim2025t-synthesis}.

\section{Estimation of non-linear quantities of quantum states}\label{apdx:nonlinear}
In this work, we propose the quantum algorithm for estimating the expectation value $\mathrm {Tr}[O\rho]$ rather than preparing the full state $\rho$.
This approach called the "effective simulation", enables the estimation of $\mathrm{Tr}[O \rho]$ without explicitly preparing the state $\rho$.
Although the effective simulation restricts access to $\rho$, this does not impose a significant limitation on its applications.
Many subsequent quantum algorithms can still be executed via the effective simulation as if the state $\rho$ were given.
In this section, we provide an example to estimate entropy, which is a non-linear quantity of the state $\rho$ beyond the expectation value of observables.

The $\alpha$-R\'{e}nyi entropy $R_\alpha (\rho)$ is defined as
\begin{equation}
    R_\alpha (\rho) = \frac{1}{1-\alpha} \log (\mathrm{Tr}[\rho^\alpha]).
\end{equation}
The simplest example is R\'{e}nyi entropy with $\alpha = 2$.
In this case, $\mathrm{Tr} [\rho^2]$ is estimated by the swap test, which employs the swap gate $S = \sum_{i,j} \ket{i}\bra{j} \otimes \ket{j}\bra{i}$, where $\ket{i}$ and $\ket{j}$ denote the computational basis. 
Actually, if the state $\rho^{\otimes 2}$ is explicitly prepared, the swap test (see Fig.~\ref{fig:swaptest} (a)) operates as 
\begin{equation}
    \mathrm{Tr}[(\bm{1} \otimes \bm{1} \otimes X)\frac{1}{2}\left(
    \rho^{\otimes 2}\otimes \ket{0}\bra{0}
    + \rho^{\otimes 2} S\otimes \ket{0}\bra{1}
    + S \rho^{\otimes 2} \otimes \ket{1}\bra{0}
    + S\rho^{\otimes 2}S \otimes \ket{1}\bra{1}
    \right)]
    =\mathrm{Tr}[\rho^2],
\end{equation}
by using the swap trick $\mathrm{Tr}[S\rho^{\otimes 2}] = \mathrm{Tr}[\sum_{i,j} \ket{i}\bra{j}\rho \otimes \ket{j}\bra{i}\rho] = \mathrm{Tr}[\rho^2]$.

Subsequently, we investigate the scenario where the state $\rho$ is effectively prepared by our algorithm.
Recall that our algorithm is compactly described as 
\begin{equation}\label{eq:effectivesim}
\frac{1}{C}\rho = \sum_{v \in \mathrm{S}} \frac{c_v}{C} \mathrm{Tr}_{\mathrm{anc}}[(X_{\rm anc} \otimes \bm{1}) {\widetilde{\mathcal{W}}}_v (\ket{+}\bra{+}_{\rm anc} \otimes \rho_0 )],
\end{equation}
where $\rho_0$ and $\rho = \rho_t$ represent the initial and final states, respectively, and CPTN maps $\widetilde{\mathcal{W}}_v$ are randomly applied with the probability $c_v/C$.
By incorporating the swap test with our algorithm, we can utilize the circuit illustrated in Fig.~\ref{fig:swaptest} (b), which provides the estimation of $\mathrm{Tr}[\rho^2]/C^2$.
We can recover the original estimation up to a constant factor using the circuit by substituting the state $\rho$ with the effective state preparation algorithm~\eqref{eq:effectivesim} and measuring the composite observable $X_{\rm anc} \otimes \bm{1} \otimes X_{\rm anc}  \otimes \bm{1} \otimes X$.
Thus, the R\'{e}nyi entropy $R_2 (\rho)$, non-linear quantity of $\rho$, can be estimated with the effective simulation.
\begin{figure}[htb]
    \begin{tabular}{cc}
    \begin{quantikz}
        \\
        \lstick{$\rho$}  & & \swap{1} & \gate[3]{\bm{1} \otimes \bm{1} \otimes X} \\
        \lstick{$\rho$}  & & \targX &  &  \\
        \lstick{$\ketbra{0}$} &\gate{H} &\ctrl{-1}  &  \\
    \end{quantikz}
    &
    \begin{quantikz}
            \\
        \lstick{$\ketbra{+}_{\mathrm{anc}}$}  & \gate[2]{\sum_v c_v \widetilde{\mathcal{W}}_v/ C} & & \gate[5]{X_{\rm anc} \otimes \bm{1} \otimes X_{\rm anc}  \otimes \bm{1} \otimes X}\\
        \lstick{$\rho_0$}  & & \swap{2} &  \\
        \lstick{$\ketbra{+}_{\mathrm{anc}}$}  & \gate[2]{\sum_v c_v \widetilde{\mathcal{W}}_v/ C} & &\\
        \lstick{$\rho_0$}  & & \targX &  &  \\
        \lstick{$\ketbra{0}$} &\gate{H} &\ctrl{-1}  &  \\
    \end{quantikz}
    \\
    (a) The estimation from explicitly prepared $\rho^{\otimes 2}$
    &
    (b) The estimation from effectively prepared $\rho^{\otimes 2}$
    \end{tabular}
    \caption{The circuit to estimate $\mathrm{Tr}[\rho^2]$. $\rho_0$ and $\sum_v c_v \widetilde{\mathcal{W}}_v/ C$ are described in Eq.~\eqref{eq:effectivesim}}\label{fig:swaptest}
\end{figure}

Note finally that we can straightforwardly extend the estimation procedure with effective states for $\alpha$-R\'{e}nyi entropy for arbitrary $\alpha$ and von Neumann entropy, as follows.
Ref.~\cite{wang2023quantum} shows that the $\alpha$-R\'{e}nyi entropy and von Neumann entropy are expressed as a linear expansion involving $\mathrm{Tr}[\rho \cos \rho t]$.
Similar to $\mathrm{Tr}[\rho^2]$, the quantity $\mathrm{Tr}[\rho \cos \rho t]$ can be estimated through the density matrix exponentiation technique with the effective simulation inputs.

\end{document}